\documentclass[a4paper,accepted=2023-10-11]{quantumarticle}
\pdfoutput=1

\usepackage{amssymb,amsmath,amsthm,amsfonts,dsfont, mathrsfs, soul}
\usepackage[numbers,comma,sort&compress]{natbib}
\usepackage[colorlinks,backref]{hyperref}
\usepackage{hypcap, float}
\usepackage{bbm, soul}
\usepackage{slashed}
\usepackage{braket}
\usepackage{url}
\usepackage{xkeyval, etoolbox, fancyhdr, tikz, ltxgrid, ltxcmds, lmodern, natbib, bbm}
\usepackage{pgfplotstable}
\usepackage{array}

\usepackage{graphicx}
\usepackage[font=small]{caption}
\usepackage[labelformat=simple]{subcaption}
\usepackage{float}
\usepackage{algorithmic}
\usepackage[ruled,vlined,linesnumbered]{algorithm2e}
\usepackage{footnote}
\usepackage{enumitem, pgfplotstable, booktabs, longtable}
\usepackage[margin=1in]{geometry}
\usepackage{authblk}
\usepackage{authblk}

\usepackage{bm}
\usepackage{adjustbox}
\usepackage{geometry}
\usepackage{tikz}
\usepackage{quantikz}
\usepackage{tikz-cd}
\usepackage[titletoc]{appendix}

\setlist{listparindent=\parindent,
  parsep=0pt}
  
\usepackage{xcolor}
\usepackage{tikz, pgfplots}

\usepackage{mathtools}
\mathtoolsset{centercolon}

\newtheorem{theorem}{Theorem}[section]

\newtheorem{lemma}[theorem]{Lemma}

\newcommand{\pushright}[1]{\ifmeasuring@#1\else\omit\hfill$\displaystyle#1$\fi\ignorespaces}
\newcommand{\pushleft}[1]{\ifmeasuring@#1\else\omit$\displaystyle#1$\hfill\fi\ignorespaces}

\newcommand{\invsqrt}[1]{ \frac{1}{\sqrt{ #1}}}

\DeclareFontFamily{U}{wncy}{}
    \DeclareFontShape{U}{wncy}{m}{n}{<->wncyi10}{}
    \DeclareSymbolFont{mcy}{U}{wncy}{m}{n}
    \DeclareMathSymbol{\CyrillicEl}{\mathord}{mcy}{"4C}
\newcommand{\lshladder}{\CyrillicEl \,}
\renewcommand{\lshladder}{\Gamma}
\newcommand{\one}{\mathcal{I}}
\newcommand{\had}{\mathsf{H}}
\newcommand{\reg}[1]{\mathtt{#1}}
\newcommand{\DSB}{\mathcal{D}^{\mathrm{SB}}}
\newcommand{\DLSH}{\mathcal{D}^{\mathrm{LSH}}}
\newcommand{\SVDSB}{\mathscr{U}_{\mathrm{SVD}}^{\mathrm{SB}}}
\newcommand{\SVDLSH}{\mathscr{U}_{\mathrm{SVD}}^{\mathrm{LSH}}}
\newcommand{\UfSB}{\mathscr{U}_{\widetilde{D}}^{\mathrm{SB}}}
\newcommand{\UfLSH}{\mathscr{U}_{\widetilde{D}}^{\mathrm{LSH}}}
\newcommand{\gSB}{g^{\mathrm{SB}}}
\newcommand{\gLSH}{g^{\mathrm{LSH}}}

\begin{document}

\title{General quantum algorithms for Hamiltonian simulation with applications to a non-Abelian lattice gauge theory}

\author{Zohreh Davoudi}
\email{davoudi@umd.edu}
\orcid{0000-0002-7288-2810}
\affiliation{Department of Physics, University of Maryland, College Park, MD 20742, USA}
\affiliation{Maryland Center for Fundamental Physics, University of Maryland, College Park, MD 20742, USA}
\affiliation{Joint Center for Quantum Information and Computer Science, National Institute of Standards and Technology and University of Maryland, College Park, MD 20742, USA}
\affiliation{The NSF Institute for Robust Quantum Simulation, University of Maryland, College Park, Maryland 20742, USA}

\author{Alexander F. Shaw}
\email{alexandershaw.f@gmail.com}
\affiliation{Department of Physics, University of Maryland, College Park, MD 20742, USA}
\affiliation{Joint Center for Quantum Information and Computer Science, National Institute of Standards and Technology and University of Maryland, College Park, MD 20742, USA}

\author{Jesse R. Stryker}
\email{jstryker@lbl.gov}
\orcid{0000-0002-4968-7988}
\affiliation{Department of Physics, University of Maryland, College Park, MD 20742, USA}
\affiliation{Maryland Center for Fundamental Physics, University of Maryland, College Park, MD 20742, USA}
\affiliation{Physics Division, Lawrence Berkeley National Laboratory, Berkeley, CA 94720, USA}

\begin{abstract}
\onecolumn
   With a focus on universal quantum computing for quantum simulation, and through the example of lattice gauge theories, we introduce rather general quantum algorithms that can efficiently simulate certain classes of interactions consisting of correlated changes in multiple (bosonic and fermionic) quantum numbers with non-trivial functional coefficients. In particular, we analyze diagonalization of Hamiltonian terms using a singular-value decomposition technique, and discuss how the achieved diagonal unitaries in the digitized time-evolution operator can be implemented. The lattice gauge theory studied is the SU(2) gauge theory in 1+1 dimensions coupled to one flavor of staggered fermions, for which a complete quantum-resource analysis within different computational models is presented. The algorithms are shown to be applicable to higher-dimensional theories as well as to other Abelian and non-Abelian gauge theories. The example chosen further demonstrates the importance of adopting efficient theoretical formulations: it is shown that an explicitly gauge-invariant formulation using loop, string, and hadron degrees of freedom simplifies the algorithms and lowers the cost compared with the standard formulations based on angular-momentum as well as the Schwinger-boson degrees of freedom. The loop-string-hadron formulation further retains the non-Abelian gauge symmetry despite the inexactness of the digitized simulation, without the need for costly controlled operations. Such theoretical and algorithmic considerations are likely to be essential in quantumly simulating other complex theories of relevance to nature. 
\end{abstract}
\tableofcontents

\section{Introduction}
\noindent
\emph{Motivation and brief overview of the work.}---A strong case for the promised quantum advantage offered by quantum computing is the simulation of physical systems at an exponentially reduced cost~\cite{Feynman:1981tf,lloyd1996universal, preskill2018quantum, georgescu2014quantum}. Possibilities are countless for advancing various disciplines of theoretical and applied sciences if robust large-scale fault-tolerant universal quantum hardware becomes a reality. Such possibilities in the area of quantum chemistry and material science have led to a vigorous program in quantum-algorithm design and implementation~\cite{wecker2015solving, mcardle2020quantum, cao2019quantum, babbush2018low, bauer2020quantum, von2021quantum, Ma:2020nsa}. Furthermore, the promise of substantially speeding up computations with quantum-computing resources has been driving a plethora of quantum-based research and development in nuclear and high-energy physics in recent years~\cite{NSAC-QIS-2019-QuantumInformationScience,Bauer:2022hpo,Catterall:2022wjq,Humble:2022klb}, from first-principles approaches rooted in quantum field theories of nature~\cite{Byrnes:2005qx,Jordan:2011ne, Jordan:2011ci, Zohar:2011cw, Tagliacozzo:2012vg, Banerjee:2012pg, Zohar:2012xf, Zohar:2013zla, Jordan:2014tma, Zohar:2014qma, Marshall:2015mna, Mezzacapo:2015bra, Martinez:2016yna, Zohar:2016wmo, Zohar:2016iic, Moosavian:2017tkv, Zache:2018jbt, Gorg:2018xyc, schweizer2019floquet, Klco:2018kyo, Lu:2018pjk, Bhattacharyya:2018bbv, Stryker:2018efp, Raychowdhury:2018osk, Luo:2019vmi, Surace:2019dtp, Mil:2019pbt, Klco:2019evd, Klco:2018zqz, Bauer:2019qxa, Davoudi:2019bhy, Klco:2019yrb, Lamm:2019uyc, Mueller:2019qqj, Lamm:2019bik, Alexandru:2019nsa, Klco:2020aud, Yang:2020yer, Shaw:2020udc, Chakraborty:2020uhf, Liu:2020eoa, Kreshchuk:2020dla, Haase:2020kaj, Paulson:2020zjd, Dasgupta:2020itb, Mathis:2020fuo, Atas:2021ext, ARahman:2021ktn, Davoudi:2021ney, Barata:2020jtq, deJong:2021wsd, Ciavarella:2021lel,Ciavarella:2021nmj, Kan:2021xfc, Cohen:2021imf, Andrade:2021pil,Alam:2021uuq, Nguyen:2021hyk, Zhang:2021bjq, Honda:2021aum,Zhou:2021kdl,Gonzalez-Cuadra:2022hxt,Osborne:2022jxq,Davoudi:2022uzo,Mueller:2022xbg,Murairi:2022zdg,Farrell:2022wyt,Farrell:2022vyh, Clemente:2022cka, Pardo:2022hrp, Banuls:2019bmf, Klco:2021lap,Zohar:2021nyc,Bauer:2022hpo} to effective (field) theory descriptions of strongly interacting systems such as nuclei~\cite{Dumitrescu:2018njn, Lu:2018pjk, Shehab:2019gfn, Roggero:2018hrn, Roggero:2019myu, Du:2020glq, Du:2021ctr, Roggero:2020sgd, Holland:2019zju, Kharzeev:2020kgc, Kreshchuk:2020aiq, Bauer:2019qxa, Bepari:2020xqi, Bauer:2021gup}.

Simulation algorithms on digital quantum computers, regardless of the theory under study, share a number of general features. Most importantly, these require a digitized approximation to the system's evolution in real time. An immediate advantage of using quantum bits (qubits) is an exponentially more compact encoding of the degrees of freedom (DOFs) compared with classical encodings. Nonetheless, the efficiency of the simulation relies on how the number of costly operations varies as a function of error tolerance, the system's size, and model's parameters. While general statements can be made regarding the efficiency of algorithms for local or nearly local interactions~\cite{lloyd1996universal,Jordan:2011ne, childs2019nearly}, only an exact account of the type and the number of operations in connection to the accuracy goal of the computation can determine the viability of the algorithms. Such an analysis is particularly important in light of the limited capacity of the hardware and the imperfect fidelity of quantum operations (gates) in any practical implementation.

With the ultimate goal of quantifying quantum-resource requirements of complex theories of nature described by gauge field theories, in this paper we tackle the following questions: i) how to ``best'' digitize time evolution in theories with simultaneous changes in various types of quantum numbers due to interactions, ii) how to avoid quantumly evaluating certain operations to facilitate the simulation, and iii) how to take advantage of better formulations of the simulated theory to simplify operations and retain symmetries? The answers, as will be demonstrated, lie in a geometrical intuition, algebraic tricks, classical pre-processing, and rethinking a theory's formulation of its DOFs and constraints.

\vspace{0.1 cm}
\noindent
\emph{Framework and concrete objectives.}---Product formulas use the Trotter-Suzuki expansion to decompose the time-evolution operator $e^{-iHt}$ (for a system with Hamiltonian $H$ and evolution time $t$) into products of efficiently implementable exponentials in various ways~\cite{suzuki1991general, wiebe2010higher, childs2021theory}. For example, the (first-order) Lie-Trotter formula amounts to implementing $e^{-iHt}$ with $H=\sum_{j=1}^\Upsilon H_j$ as
\begin{align}
V_1(t) \equiv \left[\prod_{j=1}^\Upsilon e^{-it H_j/s}\right]^s,
\end{align}
up to an error that scales as $\mathcal{O}(t^2/s)$.
Product formulas, given their simplicity, have been the primary digitization method in quantum simulation, and they require no extra qubits beyond what is needed to encode the physical DOFs. There exists a range of other quantum-simulation algorithms that take advantage of techniques such as Taylor series expansion and linear combinations of unitaries~\cite{childs2012hamiltonian, berry2015simulating}, quantum signal processing~\cite{Low:2016sck}, qubitization and block encodings~\cite{low2019hamiltonian, chakraborty2018power}, singular-value transformation~\cite{gilyen2019quantum}, off-diagonal Hamiltonian expansion~\cite{kalev2021quantum}, and hybrid algorithms~\cite{Rajput:2021khs}. Compared with product formulas, the gate complexities of these algorithms generally scale better with the error tolerance and the system's size and parameters asymptotically, but they involve more complicated circuits and often have a non-negligible ancilla-qubit overhead. It is, in fact, known empirically that in certain problems, product formulas perform better than what theoretical bounds on them indicate~\cite{childs2019nearly, Stetina:2020abi,Ostmeyer:2022lxs}, pointing to the fact that such bounds are generally not tight, or that they involve small pre-factors, which can impact concrete resource estimates. For concreteness, we focus on product formulas in analyzing the time-evolution operator in this work, but many of the ideas to be introduced are applicable to other quantum-simulation methods as well.

Now given the time-evolution algorithm, a first estimate of resource requirements is obtained by counting the number of costly operations needed to guarantee an error tolerance $\epsilon >0$, defined as $||V(t)-e^{-iHt}|| \leq \epsilon$, given the system's size and parameters. Here, $V(t)$ is the product-formula approximation to the exact time-evolution operator and $||\cdot||$ denotes the spectral norm. This estimate is not complete as the cost of state preparation and observable measurement need to be included subsequently. However, as many preparation and observable-evaluation routines require implementing $e^{-iHt}$ as an ingredient, the cost estimate of time evolution is a useful indicator of the gate complexity of the full simulation. Two computational models are often considered. In the near term, when fault tolerance is out of reach but noise-mitigation strategies may ameliorate the accuracy loss~\cite{preskill2018quantum}, the number of two-qubit entangling gates (e.g., CNOT gates) needs to be minimized as such gates exhibit lower fidelities. In the far-term, when the gate errors meet certain thresholds, it becomes possible to correct anticipated qubit errors~\cite{shor1996fault}. However, T gates are known to require costly error-correction encodings and hence it is the T-gate count that needs to be minimized in the far term.

\vspace{0.1 cm}
\noindent
\emph{Diagonalization, shearing transformation, and singular-value decomposition.}---The question of how to best decompose exponentiated operators representing steps of time evolution, namely the ``propagators'', in theories with interactions involving numerous DOFs is, for example, of paramount importance in the context of (lattice) gauge theories. There, multiple types of fermionic and bosonic fields may interact, and further the interactions are locally constrained by Gauss's laws. Consider the propagator $e^{-iH_j t}$ where $H_j$ is a term or appropriate collection of terms in the Hamiltonian, chosen such that $e^{-iH_j t}$ can be decomposed exactly to a universal set of gates. If $H_j$ is a diagonal operator when expressed in the computational basis of qubit registers (that represent the DOFs in the original theory), then single- and two-qubit Pauli-Z rotations provide a complete basis for implementing $e^{-iH_j t}$, making its circuit synthesis rather straightforward. On the other hand, for non-diagonal $H_j$ one needs to proceed with a simultaneous diagonalization in the basis states of all types of quantum registers involved in $H_j$, which is a non-trivial task in general, and may lead to approximations, hence a potential violation of original symmetries. An example of this is the implementation of the fermion-gauge-boson hopping propagator in the lattice Schwinger model, where the decomposition proposed in Ref.~\cite{Shaw:2020udc}, while being efficient, introduces violations of Gauss's law.

A subsequent work~\cite{Stryker:2021asy} guided by a geometrical interpretation of the transitions in the space of quantum numbers showed that a better diagonalization in the lattice Schwinger model is achievable via a ``shearing'' transformation, requiring negligible additional quantum resources over the earlier approach, but with the advantage of preserving the local constraints. In Sec.~\ref{sec:methods}, we introduce an algebraic procedure that is a generalization of the shearing transformation. The approach reduces the diagonalizing of $H_j$ and subsequently $e^{-iH_j t}$ to a modest distortion of the exact propagator (if any) through the use of singular-value decompositions (SVDs) of relevant matrices in the space of quantum numbers. While in the example of the Schwinger model, a fully gauge-invariant implementation of the local hopping propagator is achievable, in the example of the non-Abelian SU(2) lattice gauge theory (LGT) coupled to fermions that is studied in this work, the hopping propagator can be efficiently circuitized while maintaining some but not necessarily all of the local constraints, which is still an improvement over existing algorithms~\cite{Kan:2021xfc}. This diagonalization methodology, as will be shown, is applicable to more complex propagators such as those corresponding to magnetic interactions in gauge theories.\footnote{When the gauge DOFs are digitized in the group's irreducible representation (irrep) basis.}  Furthermore, the procedure is relevant beyond implementing the time-evolution operator. For example, to measure a Hamiltonian's expectation value, one may diagonalize its summands and estimate the expectation value of each summand individually, then add the results. Minimizing the length of this sum reduces the number of measurements required.

\vspace{0.1 cm}
\noindent
\emph{Phase evaluation and classical pre-processing.}---After transforming the propagators to a computational basis, the remaining diagonal operations involve, at their core, a set of $Z$-rotations. In scenarios where the rotation angle (phase) is a constant, implementation amounts to a circuit synthesis with known counts of ancilla qubits and CNOT gates or T gates, given the algorithm used and the synthesis accuracy aimed. On the other hand, in many instances, the phases are non-trivial functions of dynamical quantum numbers, which despite being diagonal in the computational basis of their respective registers, need to be evaluated at each step of the evolution. As is known, while such function evaluations can proceed via a ``phase-kickback'' algorithm~\cite{childs2010quantum} and embedded quantum arithmetic routines~\cite{Haner:2018yea, haener2018quantum}, they are prohibitively costly in the near term, requiring abundant ancillary registers and controlled operations to evaluate, store, and reprocess function values. This will likely be the biggest bottleneck to quantum simulations of non-Abelian gauge theories, as the non-Abelian algebra (in the irrep basis) involves Clebsch-Gordan coefficients (and generalizations of), which find their way to the exponent of diagonal propagators for both fermion-gauge hopping and the magnetic interactions.

Instead of quantumly evaluating these functions on the go, as a near-term strategy to be described in this work, one can locally decompose them to strings of Pauli-$Z$ operations with now fixed coefficients, and within a fixed tolerance, systematically neglect rotations with small angles. The classical pre-processing involved in obtaining the Pauli decomposition of diagonal functions scales exponentially with the number of qubits per lattice site. It is, therefore, only a function of the cutoff on the gauge-field quantum numbers in the LGT example studied, but is independent of system's size. Our strategy is similar in spirit to that proposed in Ref.~\cite{Ciavarella:2021nmj} which uses the pre-processed local operations in simulating propagators in a pure SU(3) LGT. Nonetheless, despite the strategy in Ref.~\cite{Ciavarella:2021nmj}, our method does not require a hard encoding of the Clebsch-Gordan coefficients in the circuit \emph{a priori} using controlled rotations with pre-set angles, but rather pre-evaluates the functional form for any input, as will be discussed in Sec.~\ref{sec:neardiag}.

\vspace{0.1 cm}
\noindent
\emph{Basis/formulation considerations and symmetry imposition.}---It is known that equivalent expressions of the Hilbert space may lead to different algorithmic complexity in quantum simulation. A known example is the use of position-space versus momentum-space wave functions (see. e.g., Ref.~\cite{Stetina:2020abi}), and the first- or second-quantized formulations (see e.g., Refs.~\cite{kivlichan2017bounding, Su:2021lut, babbush2016exponentially, jorgensen2012second,  moll2016optimizing, babbush2017exponentially}) in simulating electronic systems and quantum chemistry. In the context of this work, the question of what is the best basis to represent the DOFs not only concerns the potential reduction of the qubit cost of encoding, but also the simplicity of interactions (e.g., the type and the number of quantum registers to be operated on simultaneously), of phase-function evaluations, and of the expression of symmetries and constraints.

We make the need for wiser choices of basis evident through the SU(2) LGT example considered. In the irrep basis, where the electric Hamiltonian is diagonal at the cost of off-diagonal hopping and magnetic Hamiltonians, the states are characterized by their fermionic and gauge-flux content. The electric-field hence gauge-link DOFs can be represented, for example, in the standard angular-momentum formulation~\cite{Kogut:1974ag}, the Schwinger-boson (prepotential) formulation~\cite{Schwinger:1952dse, Mathur:2004kr, Anishetty:2009ai, Mathur:2010wc}, or the recently-developed loop-string-hadron (LSH) formulation~\cite{Raychowdhury:2019iki}. As will be demonstrated, while the qubit count and the gate complexity of the diagonalization routine for the various formulations are found to be comparable, the absolute cost of performing diagonal functions is reduced substantially in the LSH formulation.

Furthermore, as already discussed, decomposing the various propagators to a set of gates in product formulas amounts to breaking up a collection of terms in the Hamiltonian that are only gauge invariant together, and this can potentially break the Gauss's laws throughout the evolution. In the case of the SU(2) theory, this problem is circumvented all together in the LSH formulation that builds the physical Hilbert space \emph{a priori} using a complete set of gauge-invariant local operators~\cite{Raychowdhury:2019iki}, eliminating the need for encoding a large unphysical Hilbert space~\cite{Davoudi:2020yln}. A remaining link-local Abelian constraint can be easily retained by our diagonlization algorithm, making it possible to achieve an evolution that satisfies both non-Abelian and Abelian Gauss's law constraints, improving upon quantum algorithms of Ref.~\cite{Kan:2021xfc}. Such a reformulation of the theory also eliminates the need for various symmetry-protection protocols~\cite{Halimeh:2020ecg, Tran:2020azk, Kasper:2020owz, Lamm:2020jwv, Halimeh:2021vzf}, Gauss's law verification circuits~\cite{Stryker:2018efp, Raychowdhury:2018osk}, or expensive controlled operations~\cite{Ciavarella:2021nmj} to enforce the gauge symmetries. For example, while Ref.~\cite{Ciavarella:2021nmj} introduces a valuable strategy in simulating an SU(3) LGT by constraining local transitions to those satisfying the Gauss's law using controlled operations, it is conceivable that generalizations of the LSH formulation to the case of SU(3)~\cite{Kadam:2022ipf} will simplify the evolution there as well.  

\vspace{0.1 cm}
\noindent
\emph{Generalizations and future applications.}---Finally, in light of lessons learned from the investigation of this work for a non-Abelian LGT in 1+1 dimensions (D), and given the rather general algorithms and strategies proposed, one can explore future directions, with an eye on generalization to theories in 3+1 D, and other gauge groups such as quantum chromodynamics (QCD). Such discussions will follow in Sec.~\ref{sec:conclusion}, including remarks on the applicability of our approach to other quantum-simulation algorithms, and ideas for improving upon the algorithmic error bounds obtained in this work.

\section{Methods: Strategies for simulating product formulas
\label{sec:methods}}
\noindent
The goal of this section is to introduce generic methods to derive quantum circuits that approximate the unitary $e^{-itH}$, where $H$ is a Hermitian operator acting simultaneously on multiple qubit registers and $t>0$ is a real parameter. $H$ could represent the full Hamiltonian of a physical system or a term (or collection of terms) that is exponentiated separately in a product-formula approximation to the evolution operator. 

Given a product formula, the quantum circuit that implements it can be derived using the identity
\begin{equation}
    e^{-itH} = \mathscr{U}^\dagger e^{-it\mathcal{D}} \mathscr{U},
    \label{eq:diagonalization}
\end{equation}
where $\mathscr{U}$ is a unitary that diagonalizes $H$ and $\mathcal{D}$ is the diagonalized form of $H$ in the computational basis of the qubit registers. This reduces the task of simulating $H$ to the task of providing quantum circuits for $\mathscr{U}$ and $e^{-it\mathcal{D}}$, which turned out to be more tractable. An overarching task of quantum simulation is to implement Eq.~(\ref{eq:diagonalization}) to within some fixed accuracy $\epsilon$. The error can be considered as the spectral norm of the difference between the approximate and ideal unitaries, which is a useful quantity to bound the error in observables given an arbitrary initial state~\cite{Shaw:2020udc}. The error bound can be improved by assuming a specific input state or an input-state distribution~\cite{Su:2020gzf,Sahinoglu:2020dwp,yi2022spectral} but such improvements will not be considered here.

When deriving quantum circuits for $e^{-itH}$ via Eq.~(\ref{eq:diagonalization}) for a Hamiltonian $H = \sum_j H_j = \sum_j \mathscr{U}_j^\dagger \mathcal{D}_j \mathscr{U}_j$, a balance must be found in minimizing i) the cost of implementing exponential of each summand $H_j$, ii) the number of costly summand exponentials, and iii) the Trotter error. This is because some $\mathscr{U}_j$ and $e^{-it\mathcal{D}_j}$ could dominate the quantum-computational cost of the simulation, so it is important to split the Hamiltonian to the summands $H_j$ properly. This split, on the other hand, should be informed by Trotter-error considerations since this error depends on the commutators among the summands. For example, consider $H_j = H_j^{(1)} + H_j^{(2)}$. It may be that the sum cost of implementing $\mathscr{U}_{j}^{(1)}$, $e^{-it\mathcal{D}_j^{(1)}}$, $\mathscr{U}_{j}^{(2)}$, and $e^{-it\mathcal{D}_j^{(2)}}$ is higher than that of $\mathscr{U}_j$ and $e^{-it\mathcal{D}_j}$, in which case, it is beneficial to not split $H_j$, as long as diagonalization of $H_j$ and an inexpensive implementation of $\mathscr{U}_j$ and $e^{-it\mathcal{D}_j}$ can be found. If instead, $\mathscr{U}_{j}^{(1)}$, $e^{-it\mathcal{D}_j^{(1)}}$, $\mathscr{U}_{j}^{(2)}$, and $e^{-it\mathcal{D}_j^{(2)}}$ can be implemented more economically, it may make sense to consider the split to $ H_j^{(1)}$ and $H_j^{(2)}$, notwithstanding this will increase the Trotter error if $[H_j^{(1)},H_j^{(2)}] \neq 0$. Another consideration in choosing the summands is minimizing the symmetry violation in Trotter evolution, which may offer some benefits. For example, the more symmetries kept throughout the evolution, the more error diagnostics and noise-mitigation tools at one's disposal to (partially) verify and/or correct noisy intermediate-scale quantum simulations.

Analytically optimizing this problem in search for the best choice of summands is not trivial in general, nor is numerically finding the optimized choice given the sheer dimensionality of the operators. While systematic strategies such as efficient numerical approaches for navigating this complex optimization problem are desired, finding the balance among the goals mentioned here may be possible on a case-by-case basis. As will be discussed in the example of the non-Abelian LGT studied in Sec.~\ref{sec:application}, the knowledge of the Hamiltonian structure and symmetries, and the commutator algebra involved, allow for decomposition choices that are more optimal than the others.

Given the summand $H_j$, one may wonder if classical-computing methods can be used to circuitize either $e^{-itH_j}$ directly or $\mathscr{U}_j$ and $e^{-it\mathcal{D}_j}$. Circuitizing $e^{-itH_j}$ directly generally amounts to finding the Pauli decomposition of $H_j$. Since, in general, $H_j$ is not diagonal in the computational basis of the qubits, one would need to determine $4^{\left\lceil \log_{2}(d)\right\rceil}$ coefficients, with $d$ being the dimensionality of the relevant Hilbert space. Thus, this method is costly for large Hilbert-space sizes, and is not scalable. Furthermore, simulating $e^{-itH_j}$ by Pauli decomposition introduces significant Trotter error, as $e^{-itH_j}$ must be simulated through applications of a product formula, and Pauli strings do not necessarily commute. As a result, the Pauli-decomposition method should not be considered as the method of choice whenever alternative methods can be found that are exact and use circuits of comparable or lower cost.

The method based on the diagonalized form of $H_j$, first of all, requires diagonalizing $H_j$ to find $\mathscr{U}_j$ and $\mathcal{D}_j$, a problem that classically scales poorly with the dimensionality of the Hilbert space. For $k$-sparse Hamiltonians with $k \ll d$, which is the case for most physical Hamiltonians of interest, $H_j$ only acts on a small part of the Hilbert space and efficient numerical methods can ameliorate the scaling problem. With the proper choice of $H_j$ for local or semi-local Hamiltonians, such a diagonalization may be achievable far more efficiently numerically, or even analytically, as will be demonstrated in the method of this work. Implementing $\mathscr{U}_j$ can take advantage of classical circuit-synthesis methods but the bottleneck is dealing with large dense matrices. $e^{-it\mathcal{D}_j}$ can be similarly circuitized using a Pauli decomposition but since it is a diagonal unitary, it requires obtaining $2^{\left\lceil \log_{2}(d)\right\rceil}$ coefficients, as only tensor products of Pauli-$Z$ and identity operators are needed, see Sec.~\ref{sec:phase-evaluation}. More importantly, the implementation of these strings introduces no Trotter error. Numerical strategies for finding a circuit that simulates diagonal Hermitian operators are simpler than they are for arbitrary Hermitian matrices, and so in the far-term, simpler logic synthesis\footnote{We are using the term `circuit synthesis' to indicate a classical routine that decomposes a unitary matrix defined on a set of qubits to a set of quantum gates that effect the exact (or an approximation of the) matrix, while the term `logic synthesis' is meant to imply a program that converts an abstract specification of the desired circuit behavior into a circuit implementation in terms of a set of logic gates~\cite{wiki-logic-synthesis}. An example of the former is the Pauli decomposition of diagonal unitaries for near-term applications, and an example of the latter is the implementation of diagonal unitaries using a phase-kickback algorithm and Newton's method-based function evaluation for fault-tolerant applications, both discussed in Sec.~\ref{sec:circuits}.} can be used to find a circuit decomposition of $e^{-it\mathcal{D}_j}$, as will be detailed in the LGT example in Sec.~\ref{sec:application}. In summary, if the structure of the Hamiltonian and the proper choice of summands allow an analytical determination of $\mathscr{U}_j$, $\mathcal{D}_j$, and the circuit decomposition of $\mathscr{U}_j$, the only remaining task is circuitizing $e^{-it\mathcal{D}_j}$, which can benefit from less demanding classical pre-processing approaches and more straightforward circuit-synthesis methods in both near and far terms.

The following section details an approach to selecting $H_j$ in order to derive analytical circuits that diagonalize them. It is a useful tool for navigating the goals stated earlier in this section for certain Hamiltonians.

\subsection{Summand diagonalization}\label{sec:SVD}
In this section, we introduce a strategy based on linear algebra for breaking given Hamiltonians into a sum of terms that can be diagonalized by simple quantum circuits. It will become clear through the examples that will follow which class of Hamiltonians can benefit from the proposed method.

First suppose that a summand $H_j$ can be written as $A+A^\dagger$ such that $A^2=A^{\dagger \, 2} = 0$. Equivalently, let the Hilbert space $\mathbb{H}$ associated with the linear map $A$ be expressed as the direct sum $\mathbb{H}_0 \oplus \mathbb{H}_1$, where $\mathbb{H}_0$ is the nullspace of $A$, i.e., $\mathbb{H}_0 \equiv \ker(A) $, and $\mathbb{H}_1$ is its complement space, i.e., $\mathbb{H}_1 \equiv (\ker(A))^\perp$. Then $A$ can be written as $A=\mathcal{P}_0 A \mathcal{P}_1$, where $\mathcal{P}_{b}$ is the projector to subspace $\mathbb{H}_{b}$ for $b\in \{0,1\}$. Similarly, $A^\dagger=\mathcal{P}_1 A^\dagger \mathcal{P}_0$. If deducing the unitary $\mathscr{U}$ that diagonalizes $A+A^\dagger$ via $A+A^\dagger=\mathscr{U}^\dagger\mathcal{D}\, \mathscr{U}$ is not straightforward, or is known but costly to implement in a quantum circuit, one may alternatively proceed by using an SVD of $A$ instead, which for square operators that are considered here is a diagonal square matrix of real non-negative elements, and should be easier to circuitize than the original non-diagonal form. This is provided that the SVD unitaries $\mathscr{V}$ and $\mathscr{W}$, defined as $\mathcal{S}=\mathscr{V}^\dagger A \mathscr{W}(=\mathcal{S}^\dagger=\mathscr{W}^\dagger A^\dagger \mathscr{V})$ for the singular-value matrix $\mathcal{S}$, can be found easily and their circuit implementation is efficient. As will be demonstrated shortly, such efficient SVD can be worked out via simple quantum circuits for certain (common) Hamiltonians. 

Assuming that the SVD matrices $\mathscr{V}$ and $\mathscr{W}$ are found and can be implemented straightforwardly, the general quantum circuit that leads to $\mathcal{S}$ can be formed as follows.
One may first introduce an ancillary qubit $\reg{x}$ which is prepared in the state $\ket{0}_\reg{x}$. The operator $\mathscr{P}$, defined via $\mathscr{P} \ket{0}_\reg{x} \ket{\psi}_{\mathbb{H}_b} = \ket{b}_\reg{x} \ket{\psi}_{\mathbb{H}_b}$ for $b\in \{0,1\}$, can be introduced to apply the transformation $\mathscr{P}\left[\ket{0}\bra{0}_\reg{x} (A+A^\dagger)\right]\mathscr{P}^\dagger=\ket{0}\bra{1}_\reg{x} A+\ket{1}\bra{0}_\reg{x} A^\dagger$.
Next the operator $\mathscr{Q}$, defined as $\mathscr{Q}=\ket{0}\bra{0}_\reg{x} \mathscr{V}^\dagger + \ket{1}\bra{1}_\reg{x} \mathscr{W}^\dagger$, applies the transformation $\mathscr{Q}\left(\ket{0}\bra{1}_\reg{x} A+\ket{1}\bra{0}_\reg{x} A^\dagger\right)\mathscr{Q}^\dagger=X_\reg{x} \mathcal{S}$, where $X_\reg{x}=\ket{1}\bra{0}_\reg{x}+\ket{0}\bra{1}_\reg{x}$ is the Pauli-$X$ operator acting in the Hilbert space of the ancillary qubit.
Finally, a Hadamard transformation $\mathsf{H}_\reg{x}$ on the ancillary qubit leads to $\mathsf{H}_\reg{x} X_\reg{x} \mathcal{S} \, \mathsf{H}_\reg{x} = Z_\reg{x} \mathcal{S}$, where $Z_\reg{x}=\ket{0}\bra{0}_\reg{x}-\ket{1}\bra{1}_\reg{x}$ is the Pauli-$Z$ operator acting on the ancillary qubit.
This final form is, therefore, diagonal in both the ancillary-qubit Hilbert space and in $\mathbb{H}$, as desired. These steps are schematically shown in Fig.~\ref{fig:sumdiag}.
\begin{figure*}
\centering
\includegraphics[width=0.9\textwidth]{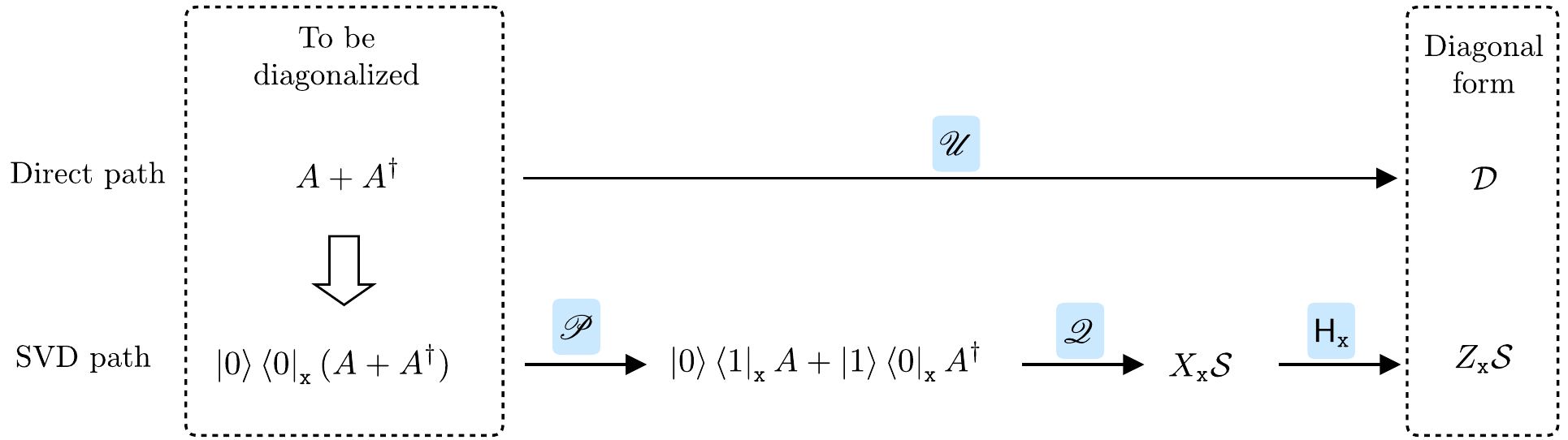}
\caption{
Diagonalization of $A+A^\dagger$, where $A^2=0$.
`Direct diagonalization' of $A+A^\dagger $ is expressed by $A+A^\dagger \xrightarrow[]{\mathscr{U}} \mathcal{D}$, where $\mathcal{D}$ is diagonal in the computational basis, $\mathscr{U}$ is unitary, and the notation $ O \xrightarrow[]{\mathscr{U}} O' $ means $O' = \mathscr{U} O \mathscr{U}^\dagger $.
We suppose that an SVD of $A$ is given by $A = \mathscr{V} \mathcal{S} \mathscr{W}^\dagger$.
The property $A^2=0$ also implies $A = \mathcal{P}_0 A \mathcal{P}_1$, where $\mathcal{P}_0$ and $\mathcal{P}_1$ are projectors onto orthogonal subspaces $\mathbb{H}_0$ and $\mathbb{H}_1$ of the full Hilbert space $\mathbb{H}=\mathbb{H}_0 \oplus \mathbb{H}_1$.
Introducing an ancillary qubit labelled as $\reg{x}$ and tensoring it with the original space $\mathbb{H}$,
unitaries $\mathscr{P}$ and $\mathscr{Q}$ are defined such that $\mathscr{P}\ket{0}_\reg{x} ( \mathcal{P}_b \ket{\psi} ) = \ket{b}_\reg{x} ( \mathcal{P}_b \ket{\psi} )$ (for $b\in \{0,1\}$) and $\mathscr{Q}=\ket{0}\bra{0}_\reg{x} \mathscr{V}^\dagger + \ket{1}\bra{1}_\reg{x} \mathscr{W}^\dagger$.
The operators at the right end of the diagram are diagonal in the computational bases.
}
\label{fig:sumdiag}
\end{figure*}

If the SVD unitaries are unknown, or not implementable by cheap quantum circuits, it may be possible to further split $A + A^\dagger$ to $\sum_k (A_k + A_k^\dagger)$ such that simple diagonalizing circuits can be found for each $A_k+A_k^\dagger$. However, such further splitting increases the total number of unitary operations that are required to construct the product formula. It may also increase the Trotter error. As discussed before, an optimized splitting would balance circuit cost and error tolerance.

In the following, we demonstrate two examples from an Abelian LGT that can take advantage of the SVD algorithm above, as well as an example of a banded Hamiltonian, to give an idea of the general characteristics of the Hamiltonians that may benefit from the method of this section. The case of SU(2) LGT that is examined thoroughly in Sec.~\ref{sec:application} provides another example, suggesting that the SVD algorithm will have wide applicability to simulating LGT Hamiltonians.
\begin{itemize}
\item[$\diamond$]{\emph{Example 1: Coupled bosonic and fermionic incrementers.} 
This interaction type describes the hopping Hamiltonian in the U(1) LGT in the staggered formulation~\cite{Kogut:1974ag}. Its form can be generalized to other Abelian and non-Abelian LGTs upon appropriate modifications to the type and/or the number of fermionic and bosonic operators, see e.g., Sec.~\ref{sec:application} for the example of a SU(2) LGT. After mapping the nearest-neighbor fermionic interaction in the U(1) LGT to qubits, the Hamiltonian on a single link can be written as 
\begin{equation}
    H_{\mathrm{hop}} = \ket{0}\bra{1}_\reg{x} \ket{1}\bra{0}_\reg{y} U_\reg{p} + {\rm H.c.},
    \label{eq:U(1)hopping}
\end{equation}
where $\reg{x}$ and $\reg{y}$ denote the two-dimensional Hilbert spaces of the qubit registers associated with fermions at two adjacent sites, and $\reg{p}$ denotes a collection of qubits that encode the bosonic Hilbert space. $U_\reg{p}$ is a ladder operator that acts on the bosonic space $\reg{p}$, and is defined as $U_\reg{p} \equiv \sum_{j=-\infty}^{\infty} c_j \ket{j-1}\bra{j}_\reg{p}$, with $j$ being an integer and $c_j=1$ for the U(1) LGT. These coefficients may generally depend on $j$ as is the case in the SU(2) LGT. They do not affect the diagonalization procedure to be outlined here but are relevant when the resulting diagonal operator is to be implemented, as will be discussed in Sec.~\ref{sec:phase-evaluation}.

Since $(\ket{1}\bra{0}_\reg{y} U_\reg{p})^2 = 0$, one can use the procedure of this section to diagonalize $e^{-itH_{\rm hop}}$. Here, the ancillary qubit and the $\mathscr{P}$ operator are not needed, since the presence of $\ket{0}\bra{1}_\reg{x}$ and its Hermitian conjugate in the Hamiltonian provides the form needed for the application of the SVD routine. Noting that $X_\reg{y}\ket{1}\bra{0}_\reg{y} \one_\reg{y}=\ket{0}\bra{0}_\reg{y}$ and $\lambda_\reg{p}^+\ket{j-1}\bra{j}_\reg{p} \one_\reg{p}=\ket{j}\bra{j}_\reg{p}$, with the incrementer/decrementor operator defined as $\lambda_\reg{p}^{\pm} \ket{j}_\reg{p}= \ket{j\pm 1}_\reg{p}$, and $\one$ being the identity operator on the corresponding registers, one arrives at
\begin{equation}
\ket{1}\bra{0}_\reg{y} U_\reg{p}+{\rm H.c.} =  X_\reg{y} \lambda^-_\reg{p} \bigl(\ket{0}\bra{0}_\reg{y} \mathcal{D}_\reg{p} \bigr) \one_\reg{y} \one_\reg{p}+{\rm H.c.},
\end{equation}
where $ \mathcal{D}_\reg{p} \equiv \sum_{j=-\infty}^\infty c_j \ket{j}\bra{j}_\reg{p}$. As a result, the singular-value unitaries $\mathscr{V}$ and $\mathscr{W}$ can be identified as $\mathscr{V}= X_\reg{y} \lambda^-_\reg{p}$ and $\mathscr{W} = \one_\reg{y} \one_\reg{p}$. Then according to the procedure depicted in Fig.~\ref{fig:sumdiag}, the Hamiltonian in Eq.~\eqref{eq:U(1)hopping} is diagonalized as:
\begin{equation}
H_{\rm hop} = \mathscr{U}^\dagger (Z_\reg{x} \ket{0}\bra{0}_\reg{y} \mathcal{D}_\reg{p}) \, \mathscr{U},
\end{equation}
with the diagonalizing transformation
\begin{align}
  \mathscr{U}&=\mathsf{H}_\reg{x} (\ket{0}\bra{0}_\reg{x} X_\reg{y}\lambda^+_\reg{p} + \ket{1}\bra{1}_\reg{x} \one_\reg{y} \one_\reg{p}) . \label{eq:unitaryForSVD}
\end{align}
On a quantum computer, the $\mathscr{U}$ operation involves basic addition primitives and Pauli gates.\footnote{For the projectors in Eq.~\eqref{eq:unitaryForSVD}, note that $Z=\ket{0} \bra{0} - \ket{1} \bra{1}$, so $\ket{0}\bra{0}=\frac{\one+Z}{2}$ and $\ket{1}\bra{1}=\frac{\one-Z}{2}$.} In practice, the bosonic Hilbert space must  be truncated and the incrementer (and decrementors) must be modified at the edge of the Hilbert space, requiring implementing modular additions (and subtractions), which are known operations in quantum circuitry. Such details will be dealt with more closely when we present algorithms for the SU(2) LGT in various formulations. 

The decomposition for the U(1) hopping term presented in this example is exact, and hence the $e^{-itH_{\rm hop}}$ operator can be implemented without violating the local Gauss's law. This feature was absent in the algorithm of Ref.~\cite{Shaw:2020udc}, where the $U_\reg{p}$ operator was split to the sum of two terms, upon a (non-exact) periodic wrapping, and the exponential of each of the terms was shown to be implementable on a quantum computer efficiently. While the periodic wrapping can be mitigated at the cost of introducing multi-controlled operations, the splitting of the hopping term introduces Gauss's-law-violating Trotter errors even for a single-link term. On the other hand, the improved algorithm of Ref.~\cite{Stryker:2021asy}  introduces shear transformations in the space of fermion and boson quantum numbers to effectively move the non-trivial operation onto one of the three registers ($\reg{x}$, $\reg{y}$, or $\reg{p}$), and implements the U(1) hopping term locally in an exact manner. This geometric picture inspired the present algorithm which finds the shear transformations systematically using an SVD.
}

\item[$\diamond$]{\emph{Example 2: Multiple coupled bosonic incrementers along the edges of a square.} This type of interaction, called a plaquette interaction, corresponds to the magnetic Hamiltonian in a U(1) LGT in higher than one spatial dimensions. The strategy for diagonalizing this term can be generalized straightforwardly to other LGTs. The single-plaquette Hamiltonian is:
\begin{equation}
    H_{\mathrm{plaq.}} = U_\reg{p} U_\reg{q} U_\reg{s}^\dagger U_\reg{t}^\dagger + {\rm H.c.},
    \label{eq:U(1)plaquette}
\end{equation}
where $\reg{p},\reg{q},\reg{s},\reg{t}$ denote the bosonic qubit registers associated with the four links of the plaquette, and the link operators are defined as in the previous example. Here, the condition $(U_\reg{p} U_\reg{q} U_\reg{s}^\dagger U_\reg{t}^\dagger)^2 =0$ does not hold in general, and so an exact simulation of this terms via the SVD algorithm will not be possible. However, as mentioned before, one can keep splitting the operator such that the resulting subterms satisfy the required condition. In the case of the plaquette interaction, only one such splitting is needed to achieve the desired form for the application of the SVD algorithm. 

Let $\mathcal{E}_\reg{p} = \sum_{k = -\infty}^\infty \ket{2k} \bra{2k}_\reg{p}$ (with integer $k$) and $\mathcal{O}_\reg{p} = \one_\reg{p} - \mathcal{E}_\reg{p}$ be the projection operators onto the even and odd quantum numbers in the $\reg{p}$ register, respectively. With these projectors, the plaquette Hamiltonian can be written as
\begin{eqnarray}
     H_{\mathrm{plaq.}}&=&H_{\mathrm{plaq.}}^{(e)}+H_{\mathrm{plaq.}}^{(o)},
\end{eqnarray}
with
\begin{subequations}
\begin{align}
     &H_{\mathrm{plaq.}}^{(e)} =
     (\mathcal{O}_\reg{p} U_\reg{p} \mathcal{E}_\reg{p}) U_\reg{q} U_\reg{s}^\dagger U_\reg{t}^\dagger + {\rm H.c.},
     \label{eq:Hplaqe}
     \\
     &H_{\mathrm{plaq.}}^{(o)} = (\mathcal{E}_\reg{p} U_\reg{p}\mathcal{O}_\reg{p}) U_\reg{q} U_\reg{s}^\dagger U_\reg{t}^\dagger + {\rm H.c.}
     \label{eq:Hplaqo}
\end{align}
\end{subequations}
$H_{\mathrm{plaq.}}^{(e)}$ and $H_{\mathrm{plaq.}}^{(o)}$ then will be exponentiated separately in the product formula. Let us apply the SVD algorithm to diagonalize $H_{\mathrm{plaq.}}^{(e)}$. Diagonalization of $H_{\mathrm{plaq.}}^{(o)}$ follows analogously upon $\mathcal{E} \leftrightarrow \mathcal{O}$.

First note that $\mathcal{O}_\reg{p} U_\reg{p} \mathcal{E}_\reg{p}$ in the first term of Eq.~\eqref{eq:Hplaqe} and $(\mathcal{O}_\reg{p} U_\reg{p} \mathcal{E}_\reg{p})^\dagger$ in the Hermitian conjugate term act, respectively, on two disjoint Hilbert spaces, $\mathbb{H}_1$ for the even quantum numbers and $\mathbb{H}_0$ for the odd quantum numbers of the $\reg{p}$ register. Hence, we identify $\ker({\mathcal{O}_\reg{p} U_\reg{p} \mathcal{E}_\reg{p}})=\mathbb{H}_0$. After introducing an ancillary qubit $\reg{x}$, the operator $\ket{0}\bra{0}_\reg{x}H_{\mathrm{plaq.}}^{(e)}$ is in exactly the form on which the SVD algorithm can be applied, starting with the transformation $\mathscr{P}$ as defined before. The $\mathscr{Q}$ transformation depends on the singular-value unitaries $\mathscr{V}$ and $\mathscr{W}$, which are easy to guess given an SVD of the link operator as obtained in the previous example. One then finds that
\begin{align}
    (\mathcal{O}_\reg{p} U_\reg{p} \mathcal{E}_\reg{p}) U_\reg{q} U_\reg{s}^\dagger U_\reg{t}^\dagger &= \lambda^-_\reg{p} \lambda^-_\reg{q} [ (\mathcal{D}_\reg{p}\mathcal{E}_\reg{p}) \mathcal{D}_\reg{q} \mathcal{D}_\reg{s} \mathcal{D}_\reg{t} ] \lambda^+_\reg{s} \lambda^+_\reg{t} ,
\end{align}
with the incrementer and decrementor operators defined in the previous example. With the identification of 
$\mathcal{V}=\lambda^-_\reg{p} \lambda^-_\reg{q}$ and $\mathcal{W}=\lambda^-_\reg{s} \lambda^-_\reg{t}$, the diagonalization proceeds as:
\begin{equation}
\ket{0}\bra{0}_\reg{x} H_{\rm plaq}^{(e)} = \mathscr{U}^\dagger (Z_\reg{x}(\mathcal{D}_\reg{p}\mathcal{E}_\reg{p}) \mathcal{D}_\reg{q} \mathcal{D}_\reg{s} \mathcal{D}_\reg{t}) \, \mathscr{U},
\end{equation}
with the diagonalizing operator $\mathscr{U}$ fully specified:
\begin{align}
    \mathscr{U}=\mathsf{H}_\reg{x}(\ket{0}\bra{0}_\reg{x} \lambda^+_\reg{p} \lambda^+_\reg{q} \one_\reg{s} \one_\reg{t} + \ket{1}\bra{1}_\reg{x} \one_\reg{p} \one_\reg{q} \lambda^+_\reg{s} \lambda^+_\reg{t} )\mathscr{P}.
 \end{align}
This unitary can be implemented on a quantum computer using standard operations. The modifications arising from the truncated Hilbert space of the links can be dealt with similar to the hopping-term example, which will be discussed in detail in Sec.~\ref{sec:application}.

There exists another method for diagonalizing $e^{-itH_{\rm plaq}}$ in the U(1) LGT~\cite{Haase:2020kaj,Kan:2021xfc}. The U(1) LGT, when truncated in the irrep basis, maintains the group structure in the group-element basis and is isomorphic to some $Z_n$ group. Therefore, one can use a quantum Fourier transform (QFT) over $Z_n$ (for which many quantum circuits exist) to diagonalize the truncated U(1) plaquette, but only if one allows direct non-zero transitions between the positive and negative cutoff states. When the proper Fourier transform is not known, like with the continuous non-Abelian groups, or the unphysical transitions are to be avoided, the strategy presented here will be advantageous. Note that the algorithm of this work requires breaking only one of the link operators and, therefore, introduces fewer Trotter commutators compared with an algorithm presented in Ref.~\cite{Kan:2021xfc}, in which all four link operators are split and implemented separately.

}

\item[$\diamond$]{\emph{Example 3: Banded Hamiltonians.}
Suppose Hamiltonian $H$ is a $2^N\times 2^N$ matrix in a basis denoted $\ket{n}$, for $n=0,1,\cdots,2^N-1$.
Furthermore, suppose that $H$ has non-zero entries only along the off-diagonal bands up to a modular distance $k$ from the main diagonal,
that is, $k$ is the minimum non-negative integer such that $\bra{m}H\ket{n}=0$ whenever $|(m-n) \text{ mod }2^N|>k$.
Then $H$ can be decomposed into a sum of $2k+1$ terms which are systematically diagonalizable.
To see this, consider that
\begin{equation}
    H = \mathcal{D}_0 + \sum_{j=1}^{k} \mathcal{D}_j (\lambda^+)^j + {\rm H.c.},
\end{equation}
for some diagonal matrices $\mathcal{D}_j.$ Since the dimension is $2^N$, one can argue that, for any $j>0$, the operator $(\lambda^+)^j$ maps computational basis states across some bipartition of the Hilbert space. This bipartition can be seen in the action of addition by $j$ modulo $2^N$ in the qubit position corresponding to the least significant, non-zero bit of $j$. In graph-theoretical terms, let the computational basis vectors define a set of vertices, where vertices $m$ and $n$ are connected by an edge if $m = (n\pm j) \text{ mod } 2^N$. For all $j>0$, addition modulo $2^N$ generates an even-length cycle, since a theorem from number theory guarantees the cycle length must divide $2^N$. 

The even and odd steps along this cycle correspond to the bipartition. This bipartition is similar to the bipartition in Example 2, and results in a similar diagonalization procedure, where each $ \mathcal{D}_j (\lambda^+)^j + {\rm H.c.}$ for $j>0$ is expressed as a sum of two diagonalizable summands. Therefore, $H$ can be expressed as a sum of $2k+1$ systematically diagonalizable summands.

}
\end{itemize}

\subsection{Implementing diagonal unitaries
\label{sec:phase-evaluation}}
Once $e^{-iH_j t}$ is diagonalized via the strategy of the previous section, the remaining task is to implement a diagonal operator of the form
\begin{equation}
    e^{-it\mathcal{D}(\hat n_1, \hat n_2,\cdots, \hat n_\gamma)},
\end{equation}
where each $\hat n_\reg{j}$ is a number operator, which is diagonal in the computational basis. Two generic avenues can be explored to construct quantum circuits which implement such a unitary. They may be identified as near-term and far-term strategies, as the former emphasizes the use of multi-qubit diagonal rotations, while the latter involves reversibly computing $\mathcal{D}(\hat n_1, \hat n_2,\cdots, \hat n_\gamma)$ using ancillary registers. Nonetheless, these methods are not mutually exclusive. One may blend them depending on the target unitary and available computational resources.

\subsubsection{Near term}\label{sec:neardiag}
With the computational register in binary, the number operator on register $\reg{j}$ with $\eta_\reg{j}$ qubits can be written as $\hat n_\reg{j} = \sum_{k=0}^{\eta_\reg{j} -1} 2^k \hat n_{\reg{j},k}$, where $\hat n_{\reg{j},k}$ is the number operator of the $k^\mathrm{th}$ qubit of register $\reg{j}$, returning a value $0$ or $1$. 
Arbitrary Hermitian diagonal matrices in $\mathbb{C}_{2^N} \times \mathbb{C}_{2^N}$ can be decomposed into tensor products of Pauli $Z_{\reg{j},k}(=\one_{\reg{j},k}-2n_{\reg{j},k})$ matrices,
\begin{eqnarray}
    &&\mathcal{D}(\hat n_1, \hat n_2, \cdots, \hat n_\gamma) 
    \nonumber \\
    && \hspace{0.5 cm} \equiv c_{I} \one + c_{\{1,1\}} Z_{1,1} + \cdots + c_{\{1,2\},\{1,3\}} Z_{1,2}Z_{1,3} + \cdots+c_{\{1,2\},\{1,3\},\{1,4\}} Z_{1,2}Z_{1,3}Z_{1,4}+\cdots,
    \label{eq:diagonalPauliDecomp}
\end{eqnarray}
which is equivalent to the Walsh series of the function~\cite{golubov2012walsh}.
With at most $2^N$ non-zero constant coefficients $c_{\{j,k\},\cdots}$, the diagonal unitary $e^{-it\mathcal{D}(\hat n_1, \hat n_2,\cdots, \hat n_\gamma)}$ can be implemented straightforwardly by exponentiating terms in this expansion individually. Such an implementation does not introduce any Trotter error as any pairs of terms in Eq.~(\ref{eq:diagonalPauliDecomp}) commute. The coefficients $c_{\{\reg{j},k\},\cdots}$, sometimes known as Walsh coefficients, are rotation angles that can be determined via a classical pre-processing by solving the following equation
\begin{equation}
    c_{\{\reg{j},k\},\{\reg{i},l\},\cdots} = \frac{1}{2^N}\mathrm{Tr}\bigl( \mathcal{D} Z_{\reg{j},k}Z_{\reg{i},l}\cdots \bigr).
    \label{eq:cCoeff}
\end{equation}
A fast-Walsh-transform algorithm, for example, can compute the coefficients in $\mathcal{O}(N  \,2^N)$ floating-point operations~\cite{yarlagadda2012hadamard}.

Finally, the exponentiation of each term is performed by realizing the following identity:
\begin{equation}
    e^{ic Z_{\reg{j}_1} \otimes \cdots \otimes Z_{\reg{j}_p}} = e^{i c \pi_{\reg{j}_1,\reg{j}_2,\cdots, \reg{j}_p}}.
\end{equation}
Here, $\pi_{\reg{j}_1,\reg{j}_2,\cdots, \reg{j}_p}$ is 1 ($-1$) if the state of the qubit string composed of registers $\{\reg{j}_1,\reg{j}_2,\cdots,\reg{j}_p\}$ has an even (odd) number of $\ket{1}$ states. So this operation is done by a single $Z$ rotation on a qubit storing the parity of $p$ qubits.\footnote{This qubit can be one of the main qubit registers and so no ancillary qubit is necessary.} However, computing the parity of $p$ qubits takes $p-1$ entangling CNOT gates. Additionally, at worst, there are $2^{N}$ such rotations, where $N$ is the total number of qubits, i.e., $N=\sum_{\reg{j}=1}^
{\gamma}\eta_\reg{j}$. Therefore, while optimized algorithms to reduce the number of entangling gates exist in certain cases~\cite{welch2014efficient,Shaw:2020udc,Kane:2022ejm}, quantum circuits for a diagonal operator remain costly if $N$ is large and $\mathcal{D}$ involves many non-zero coefficients in Eq.~\eqref{eq:diagonalPauliDecomp}. This means that if one wants to avoid saturating the upper bound of $2^N$ sufficient rotations, either the function $\mathcal{D}$ better have a small number of terms in its $Z$-string decomposition, or an approximation to $\mathcal{D}$ is made such that small rotations are dropped, trading accuracy for lower cost of implementation on a quantum computer.
The problem of finding the minimal-length Walsh-series approximation to a function with discrete argument, and the analytic bound on the error made, has been addressed in the context of quantum simulation in recent years, see e.g., Refs.~\cite{welch2014efficient,Kane:2022ejm}. In Sec.~\ref{sec:application}, we provide empirical conclusions regarding such approximations in the context of simulating time dynamics of a non-Abelian LGT.

\subsubsection{Far term}
\label{sec:fardiag}
In the far term, one may avoid the potentially exponential scaling of $Z$ rotations (which are assumed to be costly in this scenario) by unitarily computing $\ket{ \mathcal{D}(n_1, n_2, \cdots, n_\gamma)}$ from the registers $\ket{n_1}\ket{n_2}\cdots \ket{n_\gamma}$. One then extracts the computed phase with a number of single-qubit $Z$ rotations equal to the bit precision desired, and then finishes by uncomputing $\ket{ \mathcal{D}(n_1, n_2, \cdots, n_\gamma)}$. This procedure is well known as ``phase kickback'', and a detailed application of it is presented in Sec.~\ref{sec:application}. Depending on the form of the phase functions, a range of classical arithmetic algorithms can be generalized to quantum algorithms to enable the phase-kickback protocol. Newton's method, addition, and multiplication-table algorithms are among the routines used to construct the diagonal phase functions in the non-Abelian LGT example of the next section. A brief description of these routines is presented in Appendix~\ref{app:arithmetic}.

\section{Application: SU(2) lattice gauge theory in 1+1 D
\label{sec:application}}
\noindent
Hamiltonian simulation of the SU(2) LGT has been the focus of theory and algorithmic developments in recent years, from theoretical studies to cast it in more suitable representations~\cite{Anishetty:2009ai, Mathur:2010wc,Raychowdhury:2019iki, Davoudi:2020yln,Mathur:2016cko, Mathur:2021vbp,Ligterink:2000ug,Silvi:2016cas, Brower:1997ha}, to the first tensor-network simulations of its static and dynamical properties~\cite{Kuhn:2015zqa, Banuls:2017ena, Sala:2018dui}, to the first quantum-simulation algorithms and experiments to study its spectrum and evolution~\cite{Klco:2019evd, Atas:2021ext, ARahman:2021ktn, Kan:2021xfc}. In the context of this work, the SU(2) LGT provides an ideal example for the application of the diagonalization and phase-evaluation methods introduced in the previous section: it exhibits interactions involving changes in several quantum numbers with non-trivial (functional) coefficients.

The technical details of any quantum simulation algorithm for a LGT are foremost decided by the choice of formulation, digitization of continuous bosonic DOFs, and the implementation of lattice fermions. Among the formulations of the SU(2) LGT~\cite{Davoudi:2020yln} is the Kogut-Susskind formulation~\cite{Kogut:1974ag} and the Schwinger-bosons and LSH forms which are derived from it. Here, we limit our analysis to an electric eigenbasis given the naturally discrete nature of the eigenvalues in the electric basis, and their suitability in expressing Gauss's laws, which are constraints on the local flux of electric fields.

For the SU(2) LGT in the Kogut-Susskind formulation, the electric basis is also known as the irrep or angular-momentum basis. The angular-momentum basis has been digitized and studied in Ref.~\cite{Kan:2021xfc} following algorithms developed for the U(1) LGT in Ref.~\cite{Shaw:2020udc}. Here, we apply the new algorithmic approach of this work, with added benefits, to two competing formulations: the Schwinger-boson formulation and its derivative, LSH. Only a 1+1-dimensional theory will be studied in detail although the algorithms of this work are equally applicable to higher-dimensional theories, as was demonstrated for the example of a U(1) magnetic Hamiltonian in the previous section. The ultimate goal of this section is to arrive at a rigorous comparison of resource requirements in each formulation, and to evaluate to what degree the symmetries are preserved in each simulation.

\subsection{The Kogut-Susskind framework
\label{sec:KS}}
In 1+1 D, the Kogut-Susskind Hamiltonian describing SU(2) gauge fields interacting with one flavor of staggered fermions is given by~\cite{Kogut:1974ag}
\begin{align}
\hat{H}=\hat{H}_M+\hat{H}_E+\hat{H}_I.
\label{eq:HKS}  
\end{align}
The simplest contribution is the fermion self-energy,
\begin{align}
    \label{eq:HM}
    \hat{H}_M &= \mu \sum_{r=0}^{L-1} (-1)^r \hat{\psi}^\dagger(r) \hat{\psi}(r),
\end{align}
where $\hat{\psi}=\big(\begin{smallmatrix}
  \hat{\psi}_1 \\
  \hat{\psi}_2
\end{smallmatrix}\big)$ is an SU(2) doublet in the fundamental representation and each component of $\hat{\psi}$ is a one-component field that satisfies fermionic statistics.\footnote{Here and in the following, the site dependence of operators, states, and quantum numbers will be dropped for brevity, unless its specification is necessary for clarity.} 
The alternating sign $(-1)^r$ corresponds to the usage of staggered fermions.
Next, the electric Hamiltonian $\hat{H}_E$ is associated with the energy stored in the left, $\hat{E}_i^L(r)$, and right, $\hat{E}_i^R(r)$, electric fields defined on the link connecting site $r$ to $r+1$.
The index $i=1,2,3$ corresponds to the three generators of SU(2).
On each link, these satisfy an ``Abelian Gauss's Law'' (AGL) condition:\footnote{Strictly speaking, this is not a `Gauss's law' as it does not concern the flux of electric field at sites, and rather enforces the SU(2) group property at the links connecting the sites. For convenience, we choose to call this an `AGL' throughout.} $
\sum_{i=1}^3\big(\hat{E}^L_i(r)\big)^2 \equiv \big(\hat{E}^L(r)\big)^2=\sum_{i=1}^3\big(\hat{E}^R_i(r)\big)^2 \equiv \big(\hat{E}^R(r)\big)^2$. 
In words, the Casimirs at each end of the link are equal.
Finally, the electric energy is directly expressed in terms of these Casimirs as
\begin{align}
    \label{eq:HE}
    \hat{H}_E &= \sum_{r=0}^{L-2} \big(\hat{E}(r)\big)^2.
\end{align}
The left and right electric fields are the conjugate variables to the gauge-link variable,
\begin{subequations}
\begin{align}
\label{eq:EUcomm}
    &[\hat{E}^L_i,\hat{U}]=\hat{T}_i\hat{U},
    \\
    &{[\hat{E}^R_i,\hat{U}]}= \hat{U}\hat{T}_i,
\end{align}
\end{subequations}
and further satisfy the commutation relations of the SU(2) Lie algebra at each link,
\begin{subequations}
\begin{align}
    &[\hat{E}^L_i,\hat{E}^L_j]=-i\epsilon_{ijk} \hat{E}^L_k,
    \\
    &{[\hat{E}^R_i,\hat{E}^R_j]}=i\epsilon_{ijk} \hat{E}^R_k,
    \\
    &{[\hat{E}^L_i,\hat{E}^R_j]}=0.
    \label{eq:ERELcomm}
\end{align}
\end{subequations}
Here, $T_i=\frac{1}{2} \sigma_i$, $\sigma_i$ is the $i^{\rm th}$ Pauli matrix, and $\epsilon_{ijk}$ is the Levi-Civita tensor. The commutation relations for operators at different links vanish.
The gauge-matter interaction Hamiltonian
\begin{align}
     \label{eq:HI}
     \hat{H}_I &= x \sum_{r=0}^{L-2} \hat{\psi}^\dagger(r) \hat{U}(r) \hat{\psi}(r+1) + \text{H.c.},
\end{align}
consists of the hopping of a staggered fermion at site $r$, $\hat{\psi}(r)$, to an adjacent site via interactions with the gauge link $\hat{U}(r)$ originating from site $r$ (and its Hermitian conjugate). $\hat{U}$ is therefore realized as an element of the SU(2) group in the fundamental representation: it is a $2\times2$ matrix consisting of bosonic-field-operator elements. $L$ denotes the number of lattice points and $x$ is the hopping strength. Open boundary conditions (OBCs) are assumed here and throughout this work. The case of periodic boundary conditions (PBCs) requires minimal modifications to the algorithms presented.

For convenience, the Hamiltonian in Eq.~(\ref{eq:HKS}) is written in dimensionless form, in which the original dimensionfull Hamiltonian is rescaled by $\frac{2}{a_sg^2}$, with $a_s$ being the spatial lattice spacing and $g$ being the gauge coupling. The dimensionless parameters $x$ and $\mu$ are related to the original dimensionfull parameters via $x=\frac{1}{a_s^2g^2}$ and $\mu=\frac{2m}{a_s g^2}$, where $m$ is the fermion mass. The ``strong-coupling vacuum'' is associated with the ground state of the theory in the limit $x \to 0~(a_sg \to \infty)$. The continuum limit is achieved by taking the double-ordered limit $\lim_{x \to \infty}\lim_{L \to \infty}$ for a fixed $\frac{m}{g}$~\cite{Hamer:1997dx}. The electric Hamiltonian is diagonal in electric or irrep basis. It is, therefore, a more suitable basis in the strong-coupling limit, and may be less efficient towards the continuum limit, that is achieved in the weak-coupling limit. Alternative bases for SU(N) LGTs such as group-element basis~\cite{Zohar:2014qma} and dual bases~\cite{Mathur:2016cko, Mathur:2021vbp} are either not fully developed for continuous non-Abelian groups such as SU(2) and/or are not suitable for representing the electric Hamiltonian.

The Hamiltonian in Eq.~(\ref{eq:HKS}) commutes with the Gauss's law operators,
\begin{equation}
\hat{G}_i(r)=-\hat{E}^L_i(r)+\hat{E}^R_i(r-1)+\hat{\psi}^\dagger(r) T_i \hat{\psi}(r).
\label{eq:Ga}
\end{equation}
The physical sector of the Hilbert space corresponds to the zero eigenvalue of $\hat{G}_i(r)$ at every site $r$. Since $[\hat{G}_i(r),\hat{G}_j(r)] \neq 0$, specifying the physical sector of the Hilbert space is more complex than in the Abelian theories~\cite{Davoudi:2020yln}. Moreover, checking the non-Abelian Gauss's laws through a checker subroutine in the quantum circuits~\cite{Stryker:2018efp} will be non-trivial and costly.

Given the commutation relations in Eq.~(\ref{eq:ERELcomm}), the left and right electric fields can be mapped to the body-frame ($\hat{\bm{J}}^b$) and space-frame ($\hat{\bm{J}}^s$) angular momenta of a rigid body. Explicitly, $\hat{\bm{E}}^L=-\hat{\bm{J}}^b(\equiv -\hat{\bm{J}}^L)$ and $\hat{\bm{E}}^R=\hat{\bm{J}}^s(\equiv \hat{\bm{J}}^R)$, satisfying $\bm{J}^2 \equiv (\hat{\bm{J}}^L)^2=(\hat{\bm{J}}^R)^2$ on each link $r$. In such an angular-momentum basis, the Hilbert space of each link can be characterized by the basis state:~\footnote{Matrix elements for the other angular-momentum components, $\hat{J}^{L/R}_1$ and $\hat{J}^{L/R}_2$, can be derived as in any standard treatment of angular momentum but are not needed below.}
\begin{equation}
\ket{J,m^L,m^R},~~ J=0,\frac{1}{2},1,\frac{3}{2},\cdots,~-J \leq m^L,m^R \leq J.
\label{eq:Jmket}
\end{equation}
The angular-momentum operators act on the basis states through the standard relations
\begin{subequations}
\label{eq:J-EV}
\begin{align}
&(\hat{\bm{J}}^L)^2\ket{J,m^L,m^R}=J(J+1)\ket{J,m^L,m^R},
\\
&(\hat{\bm{J}}^R)^2\ket{J,m^L,m^R}=J(J+1)\ket{J,m^L,m^R},
\\
&\hat{J}^L_3\ket{J,m^L,m^R}=m^L\ket{J,m^L,m^R},
\\
&\hat{J}^R_3\ket{J,m^L,m^R}=m^R\ket{J,m^L,m^R}.
\end{align}
\end{subequations}
The link operator present in the hopping term acts on the basis states at link $r$ as:
\begin{eqnarray}
\hat{U}^{(\alpha,\beta)}\ket{J,m^L,m^R}
 &=& \sum_{j=\{0,\frac{1}{2},1,...\}} \sqrt{\frac{2J+1}{2j+1}}
\braket{J,m^L;\frac{1}{2},\alpha | j,m^L+\alpha}
\nonumber\\
&& \hspace{1.75 cm} \braket{J,m^R;\frac{1}{2},\beta | j,m^R+\beta} \ket{j,m^L+\alpha,m^R+\beta},
\label{eq:UonState}
\end{eqnarray}
where $\alpha,\beta=\pm \frac{1}{2}$ and $\hat{U}_{11}=\hat{U}^{(\frac{1}{2},-\frac{1}{2})},~\hat{U}_{12}=\hat{U}^{(-\frac{1}{2},-\frac{1}{2})},~\hat{U}_{21}=U^{(\frac{1}{2},\frac{1}{2})},~\hat{U}_{22}=\hat{U}^{(-\frac{1}{2},\frac{1}{2})}$.

The fermions Hilbert space is defined on site,
\begin{equation}
\ket{f_1,f_2},~~f_1=0,1,~f_2=0,1,
\label{eq:ketf}
\end{equation}
consisting of two fermionic quantum numbers $f_1$ and $f_2$ corresponding to the occupation number of the two components of the (anti)matter field, $\psi_1$ and $\psi_2$, each taking values 0 and 1. These correspond to the absence and presence of (anti)matter, respectively:
\begin{subequations}
\begin{align}
\label{eq:psionketf}
\hat{\psi}_1\ket{f_1,f_2}&=(1-\delta_{f_1,0})\ket{f_1-1,f_2}, \\
\hat{\psi}_1^\dagger\ket{f_1,f_2}&=(1-\delta_{f_1,1})\ket{f_1+1,f_2}, \\
\hat{\psi_2}\ket{f_1,f_2}&=(-1)^{f_1}(1-\delta_{f_2,0})\ket{f_1,f_2-1}, \\
\hat{\psi}_2^\dagger\ket{f_1,f_2}&=(-1)^{f_1}(1-\delta_{f_2,1})\ket{f_1,f_2+1}.
\end{align}
\end{subequations}
Here, $\delta$ denotes the Kronecker-delta symbol.
 
The physical states are those that can be represented as a direct product of proper linear combinations of the local basis states $\ket{J,m^R}_{r-1} \ket{f_1,f_2}_r \ket{J,m^L}_r$ such that each linear combination satisfies Gauss's laws at site $r$, and that the AGL is satisfied so that the left and right total angular momenta on the link are the same. Gauss's laws in this basis amounts to ensuring the net angular momentum at each site is zero: $\bm{J}^{L}(r)+\bm{J}^{R}(r-1)+\bm{J}^{f}(r)=0$, where $J^{f}(r)=\frac{1}{2}$ if $f_{1}(r)+f_{2}(r)=1$ and $J^{f}(r)=0$ if $f_{1}(r)+f_{2}(r)=0 \mod 2$.

The Hilbert space of each link should be truncated to allow for simulations with finite capacity. The truncation can be implemented by imposing $J\leq \Lambda_J$, where $2\Lambda_J$ is an integer. The limit of $\Lambda_J \to \infty$ must be realized via an extrapolation procedure from finite but sufficiently large values of $\Lambda_J$. Different observables will have different sensitivity to $\Lambda_J$ but previous work reveals that all observables eventually fall into a scaling region in which they asymptote to the $\Lambda_J \to \infty$ limit exponentially fast~\cite{Davoudi:2020yln,Ciavarella:2021nmj,Tong:2021rfv}.
For any finite $\Lambda_J$, any raising operators in the Hamiltonian must be redefined in order to ensure they cannot raise a state beyond the $\Lambda_J$ irrep.
In later sections, this will often be done by introducing appropriate projection operators.

\subsubsection{Schwinger-boson formulation
\label{sec:SB}}
In an equivalent representation of the Kogut-Susskind theory, one may consider the single rotor in the body and fixed frames as two uncoupled rotors in one frame, with the requirement that $J(r) \equiv J^L(r)= J^R(r)$.
Applying Schwinger's oscillator model of angular momentum, the left (right) rotor can be imagined as a collection of $n^L_1(r)+n^L_2(r)~(n^R_1(r)+n^R_2(r))$ spin-$\frac{1}{2}$ particles, with $n^L_1(r)~(n^R_1(r))$ spin-up and $n^L_2(r)~(n^R_2(r))$ spin-down particles.
The benefit of this representation is that transitions between states are expressed using simple combinations of creation and annihilation operators acting on the spin-up and spin-down populations. The annihilation operators can be conveniently organized in SU(2) doublets as $\big(\begin{smallmatrix}
  \hat{a}_1(r) \\
  \hat{a}_2(r)
\end{smallmatrix}\big)$ for the left oscillators and $\big(\begin{smallmatrix}
  \hat{b}_1(r) \\
  \hat{b}_2(r)
\end{smallmatrix}\big)$ for the right oscillators (with the creation operators being Hermitian conjugates of these doublets).

\vspace{0.2 cm}
\noindent
\emph{Degrees of freedom.}---The site-local fermionic DOFs are carried over unchanged from the Kogut-Susskind formulation, but the link Hilbert space becomes the tensor product space of two simple harmonic oscillators per side of the link (four bosonic modes per link).
Thus, associated to the left (right) end of each link, there is a complete set of commuting observables $\hat{n}^{L}_1$ and $\hat{n}^{L}_2$ ($\hat{n}^{R}_1$ and $\hat{n}^{R}_2$) that each have a bosonic spectrum (with occupation 0, 1, 2, $\cdots$).
The $\hat{n}^{L}_1$, $\hat{n}^{L}_2$, $\hat{n}^{R}_1$, and $\hat{n}^{L}_2$ occupation-number operators are associated with harmonic-oscillator annihilation operators $\hat{a}_1$, $\hat{a}_2$, $\hat{b}_1$, and $\hat{b}_2$, respectively, 
\begin{subequations}
\begin{align}
&\hat{n}^L_1(r)=\hat{a}_1^\dagger(r) \hat{a}_1(r),~~ \hat{n}^L_2(r)=\hat{a}_2^\dagger(r) \hat{a}_2(r),
\\
&\hat{n}^R_1(r)=\hat{b}_1^\dagger(r) \hat{b}_1(r),~~ \hat{n}^R_2(r)=\hat{b}_2^\dagger(r) \hat{b}_2(r),
\end{align}
\end{subequations}
which satisfy the usual commutation relations for bosons.
The tensored fermionic-bosonic Hilbert space at each site can be constructed from an orthonormal basis defined by 
\begin{align}
\ket{ n^R_1, n^R_2}_{r-1} \ket{f_1 , f_2 }_r \ket{ n^L_1, n^L_2 }_r =&
\bigl( \hat{\psi}_1^\dagger(r) \bigr)^{f_1(r)} 
\bigl( \hat{\psi}_2^\dagger(r) \bigr)^{f_2(r)}
\nonumber\\ 
&
\frac{ \bigl(\hat{a}_1^\dagger(r) \bigr)^{n^L_1(r)}}{\sqrt{n^L_1(r) \, !}} 
\frac{ \bigl(\hat{a}_2^\dagger(r) \bigr)^{n^L_2(r)}}{\sqrt{n^L_2(r) \, !}} 
\frac{ \bigl(\hat{b}_1^\dagger(r-1) \bigr)^{n^R_1(r-1)}}{\sqrt{n^R_1(r-1) \, !}} 
\frac{ \bigl(\hat{b}_2^\dagger(r-1) \bigr)^{n^R_2(r-1)}}{\sqrt{n^R_2(r-1) \, !}} 
\ket{0}.
\label{eq:ketSchBoson}
\end{align}
where $\ket{0}$ is the normalized simultaneous vacuum ket of all six modes. The Schwinger-boson DOFs are depicted in Fig.~\ref{fig:latticeDOFs_SB}(a).
\begin{figure*}[t!]
\centering
\includegraphics[scale=0.74]{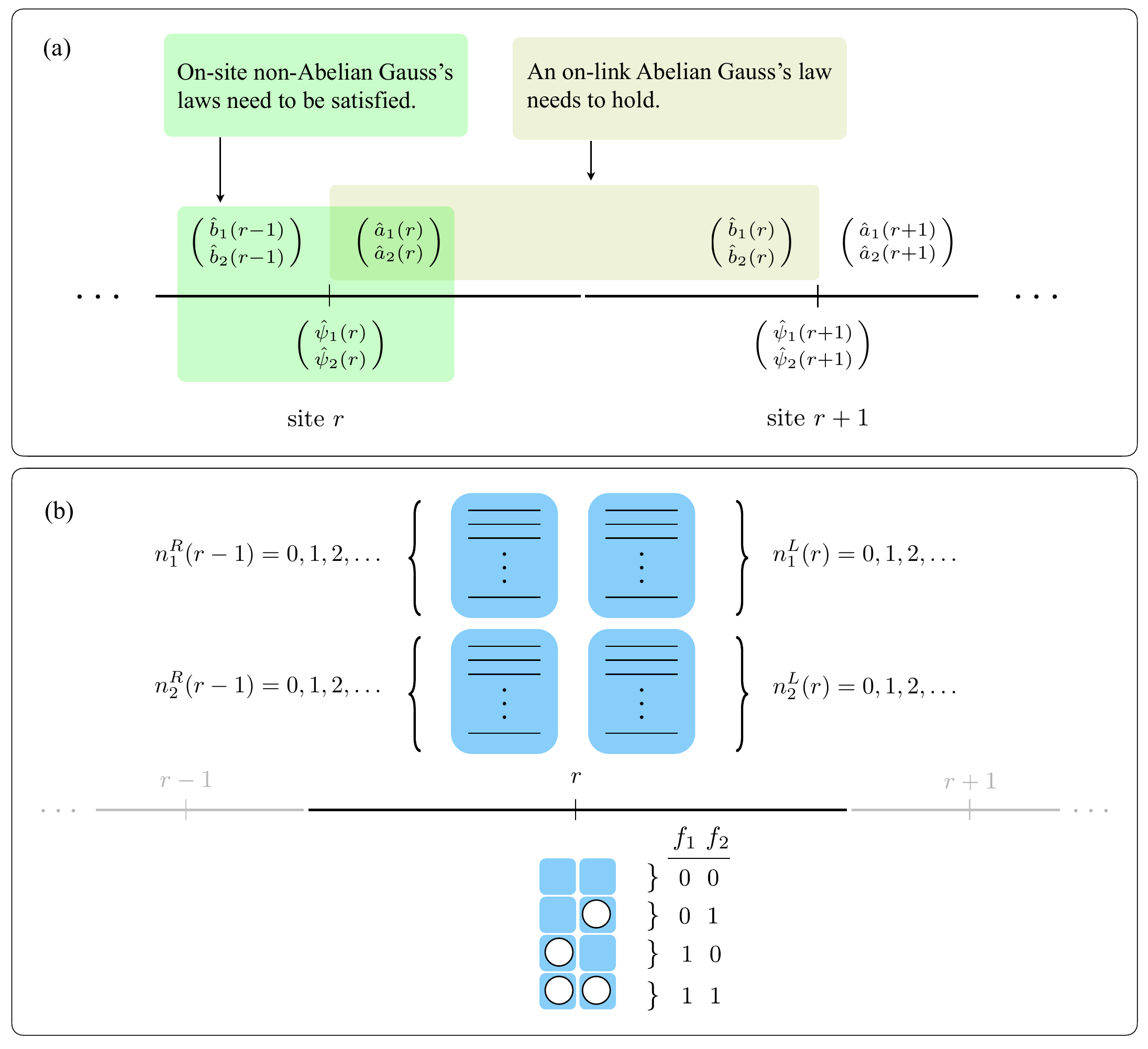}
\caption{\label{fig:latticeDOFs_SB}
a) A physical site along the spatial direction is split to two staggered sites in the Kogut-Susskind Hamiltonian. These sites are connected by a gauge link. Corresponding to each staggered site, there is a two-component fermionic field. Furthermore, in the Schwinger-boson formulation, there is a SU(2)-doublet of oscillators associated with the left of the link to the right of the site, and  a SU(2)-doublet of oscillators associated with the right of the link to the left of the next site. The operators involved in the non-Abelian and Abelian Gauss's laws are enclosed with corresponding boxes. b) The Schwinger-boson Hilbert space at each site consists of four distinct bosonic and two distinct fermionic Hilbert spaces that can then be mapped to the corresponding qubit registers, upon truncating the bosonic occupations. 
}

\end{figure*}

For future reference, the standard matrix elements of the number and ladder operators are collected below (suppressing the site index):
\begin{subequations}
\label{eqs:SBEV}
\begin{align}
    &\hat{n}^L_{i} \ket{n^L_1,n^L_2} \ket{n^R_1,n^R_2}=n^L_i \ket{n^L_1,n^L_2} \ket{n^R_1,n^R_2}, \\
    &\hat{n}^R_{i} \ket{n^L_1,n^L_2} \ket{n^R_1,n^R_2}=n^R_i \ket{n^L_1,n^L_2} \ket{n^R_1,n^R_2}, \\
    &\hat{a}_{i} \ket{n^L_1,n^L_2} \ket{n^R_1,n^R_2}=\sqrt{n^L_{i}} \ket{n^L_1-\delta_{i,1},n^L_2-\delta_{i,2}} \ket{n^R_1,n^R_2}, \\
    &\hat{b}_{i} \ket{n^L_1,n^L_2} \ket{n^R_1,n^R_2}=\sqrt{n^R_{i}} \ket{n^L_1,n^L_2}\ket{n^R_1-\delta_{i,1},n^R_2-\delta_{i,2}}, \\
    &\hat{a}^\dagger_{i} \ket{n^L_1,n^L_2} \ket{n^R_1,n^R_2}=\sqrt{n^L_{i}+1} \ket{n^L_1+\delta_{i,1},n^L_2+\delta_{i,2}} \ket{n^R_1,n^R_2}, \\
    &\hat{b}^\dagger_{i} \ket{n^L_1,n^L_2} \ket{n^R_1,n^R_2}=\sqrt{n^R_{i}+1} \ket{n^L_1,n^L_2}\ket{n^R_1+\delta_{i,1},n^R_2+\delta_{i,2}},
\end{align}
\end{subequations}
for $i=1,2$, and
\begin{subequations}
\label{eqs:SBEV-psi}
\begin{align}
    \hat{\psi}_{1} \ket{f_1 , f_2 } \ket{ n^L_1, n^L_2 } \ket{ n^R_1, n^R_2} &= (1-\delta_{0, f_1}) \ket{f_1-1, f_2 } \ket{ n^L_1, n^L_2 } \ket{ n^R_1, n^R_2} , \\
    \hat{\psi}^\dagger_{1} \ket{f_1 , f_2 } \ket{ n^L_1, n^L_2 } \ket{ n^R_1, n^R_2} &= (1-\delta_{1, f_1}) \ket{f_1+1, f_2 } \ket{ n^L_1, n^L_2 } \ket{ n^R_1, n^R_2}, \\
    \hat{\psi}_{2} \ket{f_1 , f_2 } \ket{ n^L_1, n^L_2 } \ket{ n^R_1, n^R_2} &= (1-\delta_{0, f_2}) (-1)^{f_1} \ket{f_1, f_2-1 } \ket{ n^L_1, n^L_2 } \ket{ n^R_1, n^R_2}, \\
    \hat{\psi}^\dagger_{2} \ket{f_1 , f_2 } \ket{ n^L_1, n^L_2 } \ket{ n^R_1, n^R_2} &= (1-\delta_{1, f_2}) (-1)^{f_1} \ket{f_1, f_2+1 } \ket{ n^L_1, n^L_2 } \ket{ n^R_1, n^R_2}.
\end{align}
\end{subequations}

\vspace{0.2 cm}
\noindent
\emph{Composite fields.}---The Hamiltonian of the Kogut-Susskind formulation is constructed in terms of $\hat{\psi}$, $\hat{E}^{L/R}$, and $\hat{U}$ fields. The fermionic field $\hat{\psi}$ in the Schwinger-boson formulation is the same as in the the Kogut-Susskind formulation, but there is a translation of the bosonic fields. The left and right electric-field operators are given by
\begin{align}
    \hat{E}^L_i(r) = -\hat{a}^\dagger(r) T_i \hat{a}(r),~~\hat{E}^R_i(r) = \hat{b}^\dagger(r) T_i \hat{b}(r),
\end{align}
for $i=1,2,3$, while the gauge-link operator can be written as
\begin{align}
    \label{eq:USchBoson}
    \hat{U}(r) &= \invsqrt{\hat{N}_{L}(r)+1} \left( \begin{array}{cc} - \hat{a}_1(r) \hat{b}_2(r) + \hat{a}_2^\dagger(r) \hat{b}_1^\dagger(r) & \hat{a}_1(r) \hat{b}_1(r) + \hat{a}_2^\dagger(r) \hat{b}_2^\dagger(r) \\ - \hat{a}_2(r) \hat{b}_2(r) - \hat{a}_1^\dagger(r) \hat{b}_1^\dagger(r) & \hat{a}_2(r) \hat{b}_1(r) - \hat{a}_1^\dagger(r) \hat{b}_2^\dagger(r) \end{array} \right) \invsqrt{\hat{N}_{R}(r)+1},
\end{align}
with
\begin{subequations}
\begin{align}
    &\hat{N}^L(r) \equiv \hat{n}^L_1(r)+\hat{n}^L_2(r) , \\
    &\hat{N}^R(r) \equiv \hat{n}^R_1(r)+\hat{n}^R_2(r).
\label{eq:AGLSchBosons}
\end{align}
\end{subequations}

\vspace{0.2 cm}
\noindent
\emph{Constraints.}---The Kogut-Susskind identity $J^{L}(r) = J^{R}(r)$ translates into a constraint along links in the Schwinger-boson formulation:
\begin{align}
\hat{N}^R(r) - \hat{N}^L(r) = 0.
\end{align}
Any basis state not satisfying this constraint is considered a part of the unphysical Hilbert space and has no counterpart in the Kogut-Susskind formulation.
The constraint is known as the \emph{Abelian} Gauss's law because the associated generators $\hat{N}^R(r) - \hat{N}^L(r)$ all commute.
In addition to the Abelian generators, there are also the ordinary non-Abelian Gauss's law generators in Eq.~\eqref{eq:Ga}, which carry over to the Schwinger-boson formulation with $\hat{E}^{L/R}$ translated as above.
As before, this restricts the physical space to appropriate linear combinations of the basis states (including the fermionic DOFs) at each site such that the net angular momentum is zero. Furthermore, because of the AGL, the number of independent DOFs associated with the two $\hat{a}$-type and two $\hat{b}$-type oscillators at each link reduces from four to three. These can be identified as $n^L_1(r)+n^L_2(r) (=n^R_1(r)+n^R_2(r))$, $n^L_1(r)-n^L_2(r)=2J^L_3(r)$, and $n^R_1(r)-n^R_2(r)=2J^R_3(r)$.

\vspace{0.2 cm}
\noindent
\emph{Hamiltonian.}---The Hamiltonian can be expressed in terms of the Schwinger-boson operators as well as fermionic operators.
The fermion self-energy, being in all ways identical to the Kogut-Susskind formulation, is
\begin{align}
    \hat{H}_M^{\rm SB} &= \mu \sum_{r=0}^{L-1} (-1)^r \bigl( \hat{\psi}_1^\dagger(r) \hat{\psi}_1(r) + \hat{\psi}_2^\dagger(r) \hat{\psi}_2(r) \bigr) .
\end{align}
The electric Hamiltonian can be derived by inserting the Schwinger-boson expressions into $\sum_{i=1}^{3} \hat{E}_i^L(r) \hat{E}_i^L$ and using the canonical commutation relations, with the end result
\begin{eqnarray}
\hat{H}_E^{\rm SB}=\sum_{r=0}^{L-2}\frac{\hat{N}^L(r)}{2}\bigg(\frac{\hat{N}^L(r)}{2}+1\bigg)=\sum_{r=0}^{L-2}\frac{\hat{n}^L_1(r)+\hat{n}^L_2(r)}{2}\left(\frac{\hat{n}^L_1(r)+\hat{n}^L_2(r)}{2}+1\right),
\label{eq:HE-SB}
\end{eqnarray}
where the left Casimirs for the electric energy are used by choice.
The interacting (hopping) Hamiltonian can be written as before with the identification of the link operator as in Eq.~(\ref{eq:USchBoson}), leading to eight distinct terms each with two types of fermionic operators and two types of Schwinger-boson operators:
\begin{align}
    \hat{H}_I^{\rm SB} = x \sum_{r=0}^{L-2} \tfrac{1}{\sqrt{\hat{N}^L(r)+1}} \bigl[ &
    \hat{\psi}_{1}^\dagger(r) \hat{\psi}_{2}(r+1) \hat{a}_{1}(r) \hat{b}_{1}(r)
    - \hat{\psi}_{1}^\dagger(r) \hat{\psi}_1(r+1) \hat{a}_{1}(r) \hat{b}_{2}(r) \nonumber \\
    &- \hat{\psi}_{2}^\dagger(r) \hat{\psi}_1(r+1) \hat{a}_{2}(r) \hat{b}_{2}(r) + \hat{\psi}_{2}^\dagger(r) \hat{\psi}_{2}(r+1) \hat{a}_{2}(r) \hat{b}_{1}(r) \nonumber \\
    &+ \hat{\psi}_{2}(r) \hat{\psi}_{1}^\dagger(r+1) \hat{a}_{1}(r) \hat{b}_{1}(r)
    - \hat{\psi}_{1}(r) \hat{\psi}_{1}^\dagger(r+1) \hat{a}_{2}(r) \hat{b}_{1}(r) \nonumber \\
    &- \hat{\psi}_{1}(r) \hat{\psi}_{2}^\dagger(r+1) \hat{a}_{2}(r) \hat{b}_{2}(r)
    + \hat{\psi}_{2}(r) \hat{\psi}_{2}^\dagger(r+1) \hat{a}_{1}(r) \hat{b}_{2}(r) \bigr] \tfrac{1}{\sqrt{\hat{N}^L(r)+1}} + \mathrm{H.c.}
    \label{eq:full-HI-SB}
\end{align}
Here, the AGL is used to replace $\tfrac{1}{\sqrt{\hat{N}^R(r)+1}}$ on the RHS of the link operator $\hat{U}$ in Eq.~\eqref{eq:USchBoson} with $\tfrac{1}{\sqrt{\hat{N}^L(r)+1}}$. The reason is that each term in the square bracket in Eq.~\eqref{eq:full-HI-SB} preserves the AGL.

\vspace{0.2 cm}
\noindent
\emph{Truncation.}---The infinite-dimensional Hilbert space of the Schwinger bosons needs to be truncated to allow simulations with finite resources.
One can introduce an upper cutoff integer $\Lambda$ for each bosonic-oscillator quantum number, i.e., $0 \leq n_i^L,n_i^R \leq \Lambda$ for $i=1,2$.
The exact SU(2) LGT is recovered as $\Lambda \to \infty$.
Various relations need to be modified for any finite $\Lambda$ by inclusion of an appropriate projection to ensure ${\hat{a}_{1}^{\dagger}}$ (${\hat{a}_{1}}$) acting on the kets (bras) with $n_1^L=\Lambda$ vanishes, and similarly for the other three modes.

The cutoff in the angular-momentum basis, $\Lambda_J$, is related to that in the Schwinger-boson basis, $\Lambda$, via the relation $2J=n^L_1+n^L_2~(=n^R_1+n^R_2)$.
However, there is a subtlety involved, in that our truncation scheme will leave incomplete irreps near the cutoff.
The problem can be illustrated by considering $\Lambda=1$ and focusing on one end (e.g., left) of a link.
In this case, there is a Schwinger-boson configuration $(n^L_1,n^L_2)=(1,1)$ corresponding to $(J^L,m^L)=(1,0)$, but no configurations corresponding to $(J^L,m^L)=(1,\pm 1)$.
Consequently, one can define ``the'' cutoff angular momentum to be the highest $J$ for which the complete irrep is accounted for. This turns out to be $J=\tfrac{\Lambda}{2}$.

\subsubsection{Loop-string-hadron formulation
\label{sec:lshintro}}
The Schwinger-boson formulation can be used to derive a closely-related, but distinct, formulation known as the LSH formulation.
Considering that both the Schwinger bosons and the matter fields associated with a given site transform in the fundamental representation of the local SU(2) group, various bilinear operators can be formed each transforming as a singlet under SU(2). These provide gauge-invariant operators that generate the physical Hilbert space out of the vacuum, that is the state with no matter and gauge excitations. These operators include segments of electric flux loops, quarks starting or ending at bosonic strings, and hadrons.

\vspace{0.2 cm}
\noindent
\emph{Degrees of freedom.}---%
To each site of the 1D lattice, the LSH formulation associates a complete set of commuting observables $\{ \hat{n}_\ell, \hat{n}_i,\hat{n}_o \}$, where $\hat{n}_\ell$ is a Hermitian operator with bosonic eigenvalues (0, 1, 2, $\cdots$), and $\hat{n}_i$ and $\hat{n}_o$ are Hermitian operators with fermionic eigenvalues (0 and 1 only).
Normalized ladder operators $\hat{\lshladder}$, $\hat{\chi}_{i}$, $\hat{\chi}_{o}$, and their adjoints are introduced such that
\begin{subequations}
\begin{align}
    &[\hat{n}_\ell , \hat{\lshladder} ] = - \hat{\lshladder} , \hspace{2 cm} [\hat{n}_\ell , \hat{\lshladder}^\dagger ] = \hat{\lshladder}^\dagger , \\
   & [\hat{n}_{q'}  , \hat{\chi}_{q} ] = - \hat{\chi}_{q} \, \delta_{q',q}, \hspace{1.0 cm} [\hat{n}_{q'}  , \hat{\chi}_{q}^\dagger ] = \hat{\chi}_{q}^\dagger \, \delta_{q',q},
\end{align}
\end{subequations}
for $q,q' \in \{i,o\}$, in addition to the usual anticommuting statistics of quark modes:
\begin{subequations}
\begin{align}
    &\{\hat{\chi}_{q'},\hat{\chi}_{q}^\dagger\} = \delta_{q' , q} , \\
    &\{\hat{\chi}_{q'},\hat{\chi}_{q}\} = \{\hat{\chi}_{q'}^\dagger,\hat{\chi}_{q}^\dagger\} = 0.
\end{align}
\end{subequations}
The Hilbert space can be constructed from on-site, orthonormal basis states defined by
\begin{align}
    \label{eq:LSHBasis}
    \ket{n_{\ell}, n_{i}, n_{o}} &= \bigl( \hat{\lshladder}^{\dagger} \bigr)^{n_{\ell}} \bigl( \hat{\chi}_{i}^{\dagger} \bigr)^{n_{i}} \bigl( \hat{\chi}_{o}^{\dagger} \bigr)^{n_{o}} \ket{0} ,
\end{align}
where $\ket{0}$ is the normalized vacuum ket of all three modes. The LSH DOFs are depicted in Fig.~\ref{fig:latticeDOFs_LSH}(a).

For future reference, the matrix elements of the number and normalized ladder operators are collected below:
\begin{subequations}
\label{eq:LSH-EV-Eqs}
\begin{align}
&\hat n_\ell|n_\ell, n_i, n_o\rangle = n_\ell|n_\ell, n_i, n_o\rangle, \\
&\hat n_i|n_\ell, n_i, n_o\rangle = n_i|n_\ell, n_i, n_o\rangle, \\
&\hat n_o|n_\ell, n_i, n_o\rangle = n_o|n_\ell, n_i, n_o\rangle, \\
&\hat \lshladder^{\dagger}|n_\ell, n_i, n_o\rangle = |n_\ell+ 1, n_i, n_o\rangle, \\
&\hat \lshladder|n_\ell, n_i, n_o\rangle = (1-\delta_{n_\ell,0})|n_\ell- 1, n_i, n_o\rangle, \\
&\hat \chi_i^{\dagger}|n_\ell, n_i, n_o\rangle = (1-\delta_{n_i,1}) |n_\ell, n_i+ 1, n_o\rangle, \\
&\hat \chi_i^{}|n_\ell, n_i, n_o\rangle = (1-\delta_{n_i,0}) |n_\ell, n_i- 1, n_o\rangle, \\
&\hat \chi_o^{\dagger}|n_\ell, n_i, n_o\rangle = (-1)^{n_i}(1-\delta_{n_0,1}) |n_\ell, n_i, n_o+ 1\rangle,\\
&\hat \chi_o^{}|n_\ell, n_i, n_o\rangle = (-1)^{n_i}(1-\delta_{n_0,0}) |n_\ell, n_i, n_o- 1\rangle. 
\end{align}
\end{subequations}
\begin{figure}[t!]
\centering
\includegraphics[scale=0.705]{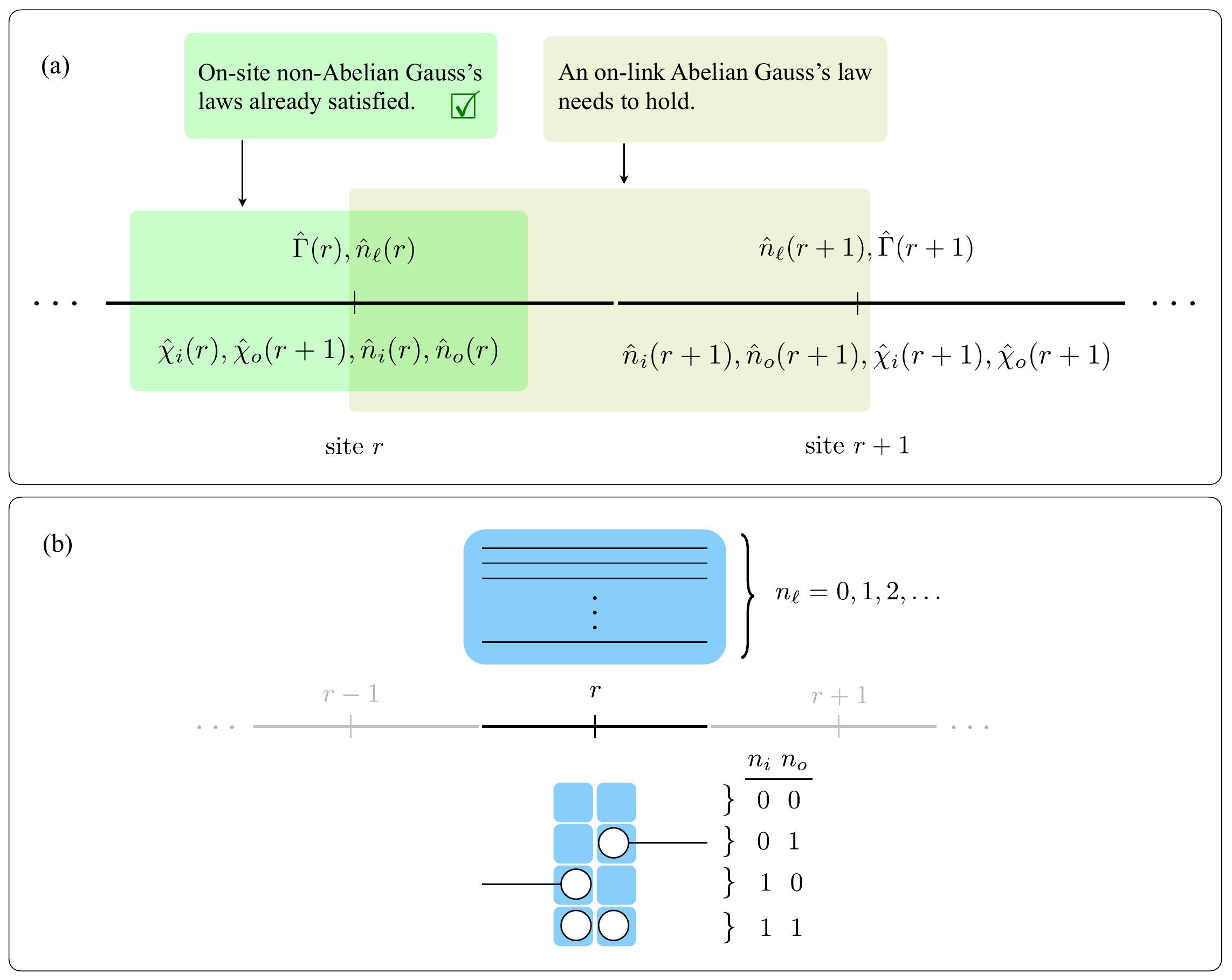}
\caption{\label{fig:latticeDOFs_LSH}
(a) Corresponding to each staggered site, there are a number of SU(2)-invariant fermionic and bosonic operators that act on the on-site Hilbert space of the LSH theory. The operators involved in the AGL are enclosed with a corresponding box. The non-Abelian Gauss's laws have already been satisfied by construction, as all operators enclosed in the corresponding box are individually gauge invariant. (b) The LSH Hilbert space at each site consists of one distinct bosonic and two distinct fermionic Hilbert spaces that can then be mapped to the corresponding qubit registers, upon truncating the bosonic occupation.
}
\end{figure}

\vspace{0.2 cm}
\noindent
\emph{Composite fields.}---%
Unlike the Kogut-Susskind and Schwinger-boson formulations, the LSH formulation does not have realizations of the $\hat{E}^{L/R}_i$ and $\hat{U}$ fields.
Instead, dynamics is generated by site-localized ``loop,'' ``string,'' and ``hadron'' operators based on which the formulation is named.
The string operators are given in factorized form by
\begin{subequations}
\label{eq:S-def}
\begin{align}
 &\hat{S}_{o(i)}^{++} =  \hat{\chi}_{o(i)}^\dagger (\hat{\lshladder}^{\dagger})^{\hat{n}_{i(o)}}\sqrt{ \hat{n}_{\ell}+2- \hat{n}_{i(o)}}, \\
 &\hat{S}_{o(i)}^{--} =  \hat{\chi}_{o(i)} \hat{\lshladder}^{ \, \hat{n}_{i(o)}}\sqrt{ \hat{n}_{\ell}+2(1- \hat{n}_{i(o)})}, \\
 &\hat{S}_{o(i)}^{+-(-+)} =  \hat{\chi}_{i(o)}^\dagger \hat{\lshladder}^{\,1- \hat{n}_{o(i)}}\sqrt{ \hat{n}_{\ell}+2 \hat{n}_{o(i)}}, \\
 &\hat{S}_{o(i)}^{-+(+-)} =  \hat{\chi}_{i(o)} (\hat{\lshladder}^\dagger)^{1- \hat{n}_{o(i)}}\sqrt{ \hat{n}_{\ell}+1+ \hat{n}_{o(i)}}. 
\end{align}
\end{subequations}
We omit the respective relations for the ``loop'' and ``hadron'' operators because the Hamiltonian, to be discussed shortly, does not make explicit use of them.

\vspace{0.2 cm}
\noindent
\emph{Constraints.}---%
For this site-local representation to recover the original Kogut-Susskind theory, an AGL among the quantum numbers of two adjacent sites must be satisfied:
\begin{align}
    \hat{N}^L(r) - \hat{N}^R(r) &= 0,
\end{align}
where
\begin{subequations}
\begin{align}
    \hat{N}^L(r) &\equiv \hat{n}_\ell(r) + \hat{n}_{o}(r) (1-\hat{n}_{i}(r)), \\
    \hat{N}^R(r) &\equiv \hat{n}_\ell(r+1) + \hat{n}_{i}(r+1) (1-\hat{n}_{o}(r+1)) .
\end{align}
\end{subequations}
Incidentally, $\hat{N}^L$ and $\hat{N}^R$ correspond exactly to the total left and right occupation numbers defined in the Schwinger-boson formulation.
The non-Abelian Gauss's laws are built in to the LSH Hilbert space such that no constraints beyond the AGL need to be imposed on the Hilbert space \emph{a posteriori}.
This makes the expression of physical states simple in the basis of Eq.~(\ref{eq:LSHBasis}), leading to locally 1-sparse interactions, contrary to the Hamiltonians expressed in the angular-momentum and Schwinger-boson bases.

\vspace{0.2 cm}
\noindent
\emph{Hamiltonian.}---%
The Hamiltonian of the SU(2) theory in 1+1 D with staggered quarks can then be expressed in terms of the LSH number and ladder operators.
The mass and electric Hamiltonians are simply related to gauge-invariant operators $\hat{\psi}^\dagger (r) \cdot \hat{\psi}(r)=\hat{n}_i(r)+\hat{n}_o(r)$ and $\hat{N}^L(r)=\hat{N}^R(r)$~\cite{Raychowdhury:2019iki}:
\begin{align}
    \hat{H}_M^{\rm LSH} =\mu \sum_{r=0}^{L-2}(-1)^r(\hat{n}_{i}(r)+\hat{n}_{o}(r)),
    \label{eq:HE-LSH}
\end{align}
and
\begin{align}
    \hat{H}_E^{\rm LSH}
    =\sum_{r=0}^{L-2}\frac{\hat{N}^L(r)}{2}\bigg(\frac{\hat{N}^L(r)}{2}+1\bigg)
    =\sum_{r=0}^{L-2}\frac{\hat{n}_\ell(r)+\hat{n}_o(r)(1-\hat{n}_i(r))}{2}\bigg[\frac{\hat{n}_\ell(r)+\hat{n}_o(r)(1-\hat{n}_i(r))}{2}+1 \bigg] ,
\end{align}
where the left Casimirs for the electric energy are used by choice.
The gauge-matter interaction Hamiltonian in the LSH formulation, which is constructed using the string operators, can be taken as
\begin{align}
\label{eq:HI-subterms-LSH}
    \hat{H}^{\rm LSH}_I =& x \sum_{r=0}^{L-2} \bigg\{  \left[ \hat{\chi}_{i}^\dagger(r+1) \, \hat{\lshladder}^{\dagger}(r+1)^{\hat{n}_o(r+1)} \right] \left[ \hat{\chi}_{i}(r) \, \hat{\lshladder}^{\dagger}(r)^{1-\hat{n}_o(r)} \right] \sqrt{\frac{\hat{n}_\ell(r)  + 1 + \hat{n}_o(r)}{\hat{n}_\ell(r) + 1 + \hat{n}_o(r+1)}} + 
    \nonumber \\
    &\left[ \hat{\chi}^\dagger_{o}(r) \, \hat{\lshladder}^{\dagger}(r)^{\hat{n}_i(r)} \right] \left[ \hat{\chi}_{o}(r+1)  \, \hat{\lshladder}^{\dagger}(r+1) ^{1-\hat{n}_i(r+1)} \right] \sqrt{\frac{\hat{n}_\ell(r+1)  + 1 + \hat{n}_i(r+1)}{\hat{n}_\ell(r+1) + 1 + \hat{n}_i(r)}} + \mathrm{H.c.} \bigg\}.
\end{align}
The above expressions are not directly expressed using the string operators in Eqs.~(\ref{eq:S-def}) as prescribed in Ref.~\cite{Raychowdhury:2019iki}.
Instead, they have been derived from the original prescription by further invoking the AGL and simplifying the square-root functions, and
the details of this manipulation are provided in Appendix~\ref{app:LSHSubtermSimplification}.
The above forms are preferred in this work because they yield cost-saving benefits from the viewpoint of quantum algorithms.

\vspace{0.2 cm}
\noindent
\emph{Truncation.}---%
As with the Schwinger-boson formulation, it is necessary to truncate the infinite-dimensional Hilbert spaces of bosons for simulation with finite resources.
In the LSH formulation, this truncation can be realized by introducing an upper cutoff $\Lambda_\ell$ on the flux quantum numbers $n_{\ell}$, i.e., $0 \leq n_\ell \leq \Lambda_\ell$.
The limit of $\Lambda_\ell \to \infty$ must be taken to recover the Kogut-Susskind theory.
With a finite cutoff on the $n_\ell$ quantum numbers in effect, one must ensure $\hat{\lshladder}^\dagger$ ($\hat{\lshladder}$) acting on a ket (bra) with $n_\ell=\Lambda_\ell$ vanishes with the use of corresponding projectors.

The relation between the cutoff $\Lambda_\ell$ and those introduced in the angular-momentum and Schwinger-boson formulations is again realized through the equations for $\hat{N}^L$ ($\hat{N}^R$).
The subtlety in this formulation is a matter of some missing states near the cutoff, although it has nothing to do with incomplete irreps.
Rather, there are certain angular momenta for which some but not all states are present.
A simple example is illustrated by taking $\Lambda_\ell=0$.
This truncation admits states with $J=\tfrac{1}{2}$ arising from the string-end configurations such as $(n_{\ell},n_{i},n_{o})=(0,0,1)$, since $N^L=1=2J$ for this configuration.
But not all $J=\tfrac{1}{2}$ states are accounted for, such as $(n_{\ell},n_{i},n_{o})=(1,0,0)$, which produces the same value of $N^L$.
As with the Schwinger-boson formulation, ``the'' cutoff angular momentum can be defined to be the highest $J$ for which all states up to $J$ are accounted for.
This again turns out to be $\tfrac{\Lambda_\ell}{2}$. Our strategy in the following is to fix the cutoff in the angular-momentum basis, which given the definitions above, leads to the same cutoff on the Schwinger-boson and the loop excitations. For this reason, $\Lambda_\ell = \Lambda$ and no distinction will be needed in the following between the cutoffs in the two formulations.

To conclude this section, we emphasize that among all the formulations, only the LSH formulation constructs a Hilbert space with non-Abelian gauge invariance built in. Both the site-local LSH and Schwinger-boson bases require an AGL to be imposed on the quantum numbers of two adjacent sites, but this constraint is more easier to satisfy in the  quantum-simulation algorithms than the (non-commuting) non-Abelian Gauss's laws, as will be demonstrated shortly.

\subsection{Circuit decomposition of the propagators
\label{sec:circuits}}
The general methods of Sec.~\ref{sec:methods}  can be applied straightforwardly to decompose the SU(2) LGT within the formulations introduced in Sec.~\ref{sec:KS}, once the mapping of the Hilbert space to qubits is specified and the simulable unitaries are identified. Since a quantum algorithm for simulating SU(2) dynamics within the angular-momentum formulation of the Kogut-Susskind theory has already been developed~\cite{Kan:2021xfc} (albeit using a different strategy than that of this work), in the following we exclusively discuss algorithms for simulating  the time-evolution operator in the Schwinger-boson and the LSH bases. The application of our general method within the angular-momentum basis will be briefly discussed in Section \ref{sec:conclusion} to provide a qualitative cost comparison. For notational brevity, the overhead hat symbols for operators will be dropped from this point on.

\subsubsection{Schwinger-boson propagators
\label{sec:circuits_SB}}
To decompose time evolution in the gate model of quantum computation, the field operators introduced in the Hamiltonian must be replaced by qubit operators. The fermionic modes on each lattice site require two qubits, one for each SU(2) component of the fermions, with no Hilbert space truncation needed. In 1+1 D, Jordan-Wigner transformation is the simplest way to replace the fermionic operators with operators that act on the two-dimensional qubit Hilbert space $\{\ket{0},\ket{1}\}$. The transformation maps the $L$ fermionic doublets onto $2L$ spin variables. Explicitly, we choose a Jordan-Wigner transformation that maps the fermionic modes $\psi(r)$ and $\psi^\dagger(r)$ according to\footnote{Note that our convention set in Sec.~\ref{sec:methods} leads (counter-intuitively) to the relations $\sigma^+ = \ket{0} \bra{1} $ and $\sigma^- = \ket{1} \bra{0}$ for the spin-$\tfrac{1}{2}$ raising and lowering operators.}
\begin{subequations}
\label{eq:JW-SB}
\begin{align}
    \label{eq:JW-SB-1}
    &\psi_1(r) \rightarrow \left( \prod_{k=0}^{2r-1} Z_k \right) \sigma^+_{2r},~~\psi_2(r) \rightarrow \left( \prod_{k=0}^{2r} Z_k   \right) \sigma^+_{2r+1}, \\
    \label{eq:JW-SB-2}
    &{\psi_1^\dagger(r)} \rightarrow \left( \prod_{k=0}^{2r-1} Z_k   \right) \sigma^-_{2r},~~{\psi_2^\dagger(r)} \rightarrow \left( \prod_{k=0}^{2r} Z_k   \right) \sigma^-_{2r+1}.
\end{align}
\end{subequations}

The Hilbert space of the Schwinger-boson operators is built from basis states defined in Eq.~(\ref{eq:ketSchBoson}), with the corresponding operators defined in Eqs.~(\ref{eqs:SBEV}) and (\ref{eqs:SBEV-psi}). With a qubit register $\reg{p}$ of finite size $\eta$, only the first $2^\eta$ modes can be encoded, putting a truncation $\Lambda \equiv 2^\eta-1 $ on the occupation number of each bosonic mode:
$\hat{N}_\reg{p}\ket{p}=p\ket{p}$ with $0 \leq p \leq \Lambda$. 
Moreover, the corresponding raising operator must be modified by a projection operator such that it annihilates the state with occupation $\Lambda$, that is $a^\dagger_\reg{p}\ket{p} = (1-\delta_{p,\Lambda})\sqrt{p+1}\ket{p+1}$. 
The state $\ket{p}$ is mapped to qubits via a binary encoding, that is
\begin{eqnarray}
\ket{p}=\bigotimes_{j=0}^{\eta-1}\ket{p_j}\quad\text{with}\quad p=\sum_{j=0}^{\eta-1}p_j2^j,
\label{eq:bosonic-register}
\end{eqnarray}
where $p_j=0,1$ are the coefficients of the binary representation of the integer $p$. Note that the $j=0$ index is designated as the least significant bit in the binary representation of $p$. For the Schwinger bosons on a 1D chain with OBCs, there are $4(L-1)$ such bosonic operators, associated with a bosonic register for each SU(2) component of the left and right oscillators on the link, see Sec.~\ref{sec:SB}. These registers are all assumed to have size $\eta$. Figure~\ref{fig:latticeDOFs_SB}(b) depicts the Hilbert spaces associated with the DOFs of the Schwinger-boson formulation.

Next, the Hamiltonian needs to be represented in terms of the operators acting on these qubit registers. Since the Hamiltonian is the sum of (site- or link-)local terms, in this section we focus on deriving efficient circuit decompositions for exponentiating each of such local terms, and will consider product formulas for approximating the full time-evolution operator of the system within a fixed error tolerance later in Sec.~\ref{sec:bounds}. 

\begin{center}
\textit{---Near term---}
\end{center}

\begin{center}
\emph{Implementing mass propagators}
\end{center}

The mass Hamiltonian is $\sum_{r=0}^{L-1} H_M^{\rm SB}(r)$, with site-local mass terms decomposed into two commuting ``subterms'' as
\begin{subequations}
\label{eqs:HM-subterms_SB}
\begin{align}
    &H_M^{\rm SB}(r)=\sum_{j=1}^2 H_{M}^{{\rm SB}(j)}(r), \\
    &H_{M}^{{\rm SB}(1)}(r)=
    \frac{\mu}{2}(-1)^{r+1} Z_{2r}, \\
    &H_{M}^{{\rm SB}(2)}(r)=
    \frac{\mu}{2}(-1)^{r+1} Z_{2r+1},
\end{align}
\end{subequations}
after the Jordan-Wigner transformation in Eq.~\eqref{eq:JW-SB}, and upon neglecting constant terms that only introduce time-dependent but otherwise constant phases in the dynamics. The circuit decomposition of this propagator amounts to single-qubit $Z$ rotations on each fermionic register
. With two qubits indexed by $2r$ and $2r+1$, corresponding to `$\psi_1(r)$' and `$\psi_2(r)$' registers, respectively,  the circuit implements
\begin{eqnarray}
R^Z_{2r}\left((-1)^{r+1}\mu\,t \right)R^Z_{2r+1}\left((-1)^{r+1}\mu\,t \right),
\label{eq:propMr}
\end{eqnarray}
with $R_j^Z(\theta)$ defined as $R_j^Z \equiv e^{-i \theta Z_j/2}$.
\begin{lemma}
Using a $2$-qubit register with no ancilla qubits, $e^{-itH^{{\rm SB}}_M(r)}$ can be implemented without approximation, up to a phase, with no CNOT gates required.
\end{lemma}
\begin{proof}
The circuit implementing the mass-term propagator in the near-term scenario is simply a straightforward implementation of Eq.~(\ref{eq:propMr}) with exact rotation angles applied, and hence requires no additional ancilla qubits. As a result, the mass-term propagator is essentially `free' in a near-term implementation (costing no entangling gates).
This trivial cost is recorded in Table \ref{tab:diag-costs_near_SB}.
\end{proof}
\begin{center}
\emph{Implementing electric propagators}
\end{center}

The electric Hamiltonian is $\sum_{r=0}^{L-2} H_E^{\rm SB}(r)$, with link-local electric terms (Casimir operators) decomposed into two commuting subterms as
\begin{subequations}
\label{eqs:HE-subterms_SB}
\begin{align}
    &H_E^{\rm SB}(r) = \sum_{j=1}^2 H_{E}^{{\rm SB}(j)}(r) , \\
    &H_{E}^{{\rm SB}(1)}(r) = \frac{1}{2} N^L(r) , \\
    &H_{E}^{{\rm SB}(2)}(r) = \frac{1}{4} \left({N^L(r)}\right)^2 .
\end{align}
\end{subequations}
Note that the Casimir operators $H_E^{\rm SB}(r)$ are already diagonalized in the chosen basis, so other divisions of $H_E(r)$ into subterms are possible with terms that all commute. In other words, our splitting into $H_{E}^{{\rm SB}(1)}(r)$ and $H_{E}^{{\rm SB}(2)}(r)$ is not unique, but a choice is made to provide the concrete gate-count analysis that follows.
Nonetheless, a naive Pauli decomposition of $H_E(r)$ only involves $O(\eta^2)$ terms, so any choice of subterms will not have a CNOT-gate scaling worse than this.
Regardless of the division into subterms, there are two conclusions one can draw:
1) $e^{-itH^{\rm SB}_E(r)}$ is readily circuitized without any theoretical approximation, and using established algorithms (see below).
2) Over-optimizing the circuitization of $e^{-itH^{\rm SB}_E(r)}$ is unnecessary because, as will be seen later, its cost will ultimately be dwarfed by that of $e^{-itH^{\rm SB}_I(r)}$.

The division into subterms as in Eqs.~\eqref{eqs:HE-subterms_SB} involves two distinct types of subterms, one linear and one quadratic in $N^L=N^L_1+N^L_2$, where the number operator for either individual oscillator is related to the binary register's qubits via
\begin{eqnarray}
N_\reg{p} \ket{p}=\bigg[ \frac{1}{2}(2^\eta-1)\one_\reg{p} -\frac{1}{2}\sum_{j=0}^{\eta-1}2^j Z_{\reg{p}_j}\bigg] \ket{p}.
\label{eq:Nponketp}
\end{eqnarray}
Here, $\one_\reg{p}=\prod_{j=0}^{\eta-1}\one_{\reg{p}_j}$ is the identity operator on the state of $\reg{p}$ register, i.e., $\one_\reg{p}\ket{p}=\ket{p}$. 

\begin{lemma}
Using a $2\eta$-qubit register with no ancilla qubits, $e^{-itH^{{\rm SB}}_E(r)}$ can be implemented without approximation, up to a phase, using $3\eta^2+\eta-2$ CNOT gates (and a number of single-qubit rotations).
\end{lemma}
\begin{proof}
This cost is obtained as follows.
i) The operator $N^L$ is a linear combination of the identity and every single-qubit Pauli-$Z$ operators across the `$a_1$' and `$a_2$' registers.
Therefore, simulating $e^{-iN^L}$ amounts to a global phase and $2\eta$ single-qubit $Z$ rotations, but no CNOT gates.
ii) The operator $(N^L)^2$ may be further decomposed as 
    $$ (N^L)^2 = (n^L_1)^2 + (n^L_2)^2 + 2 n^L_1 n^L_2 .$$
    The $e^{-it(n^L_i)^2/4}$ term can be simulated like an `$e^{-itE^2}$' term in the U(1) theory, with the difference between the two operators being terms linear in single-qubit $Z$ rotations (which cost no CNOT gates).
    According to Lemma 2 of Ref.~\cite{Shaw:2020udc}, this can be done with $(\eta+2)(\eta-1)/2$ CNOT gates.
    The remaining $e^{-itn^L_1 n^L_2/2}$ contribution can be simulated by naively circuitizing its Pauli decomposition consisting of $\eta^2$ distinct two-qubit Z rotations ($Z \otimes Z$ operators), each requiring two CNOT gates.
The gate counts of i) and ii) are summarized in Table \ref{tab:diag-costs_near_SB}, adding up to the stated total CNOT count of $2\times (\eta+2)(\eta-1)/2 + 2 \eta^2 = 3\eta^2 +\eta -2$.
\end{proof}
\begin{table*}
\centering
\begin{tabular}
{>{\arraybackslash}p{7.5cm}  >{\centering\arraybackslash}p{4cm}  c }
Schwinger-boson mass propagator subroutine & {CNOT} count \\
\hline
\hline
$e^{-itH^{{\rm SB}(1)}_M(r)}$ and $e^{-itH^{{\rm SB}(2)}_M(r)}$ gate & 0 \\
\hline
Full $e^{-itH^{{\rm SB}}_M(r)}$ & 0 \\
\\
\\
Schwinger-boson electric propagator subroutine & {CNOT} count \\
\hline
\hline
$e^{-itN^L_1/2}$ or $e^{-itN^L_2/2}$ & 0 \\
$e^{-it(N^L_1)^2/4}$ or $e^{-it(N^L_2)^2/4}$ 
& $(\eta+2)(\eta-1)/2$ \\
$e^{-itN^L_1 N^L_2/2} $ & $2\eta^2$ \\
\hline
Full $e^{-itH^{{\rm SB}}_E(r)}$ & $3\eta^2 +\eta -2$ 
\end{tabular}
\caption{\label{tab:diag-costs_near_SB}
Summary of the cost associated with the near-term implementation of diagonal operators in the Schwinger-boson Hamiltonian.}
\end{table*}
\begin{center}
\emph{Implementing hopping propagators}
\end{center}

The gauge-matter interaction Hamiltonian is $\sum_{r=0}^{L-2} H_I^{\rm SB}(r)$, with link-local hopping terms $H_I^{\rm SB}(r)$ that are off-diagonal in the electric basis, see Eqs.~\eqref{eq:full-HI-SB}.
The hopping terms can be decomposed into a set of subterms but as with the mass and electric terms, the division of $H_I^{\rm SB}(r)$ into subterms is not unique.
However, unlike the mass and electric terms, the hopping terms are off-diagonal and do not readily split into commuting subterms.
For this reason, we resort to an approximation $e^{-i t H_I^{\rm SB}(r) } \approx \Pi_{j} \, e^{-itH^{{\rm SB}(j)}_I(r)}$, in which each subterm is simulated without further splitting. Clearly, one important consequence of the non-commuting subterms is that the Trotter error bound will grow, as will be seen in Sec.~\ref{sec:bounds}.
Furthermore, different choices of $H^{{\rm SB}(j)}_I(r)$ could have different simulation costs, which could significantly change the total cost of the simulation as implementing $e^{-itH^{{\rm SB}(j)}_I(r)}$ is seen to dominate the gate count.
Even for a similar cost of each individual $E^{-itH^{{\rm SB}(j)}_I(r)}$, the total number of hopping subterms has immediate implications for the Trotter error and gate count, so having less number of simulable terms is desirable.
The number of hopping subterms will be denoted by $\nu$ in the following, with a `SB' superscript when the Schwinger-boson formulation is concerned.

A choice of splitting is to the eight individual terms shown in Eq.~\eqref{eq:full-HI-SB}, along with their Hermitian conjugates. To translate the hopping terms into the language of a quantum computer, one applies the Jordan-Wigner transformations in Eqs.~(\ref{eq:JW-SB}), along with the truncation to $\eta$-bit registers for each bosonic oscillator mode. Then, the interaction Hamiltonian acting in the qubit space is:
\begin{subequations}
\label{eqs:HI-subterms-JW_SB}
\begin{align}
    &H_I^{\rm SB}(r) = \sum_{j=1}^{8} H_{I}^{{\rm SB}(j)}(r), \\
    &H_I^{{\rm SB}(1)}(r)
    =  x\,\sigma_{2r}^- Z_{2r+1} Z_{2r+2} \sigma_{2r+3}^+ \lambda^{-}_{L1} \lambda^{-}_{R1} \DSB ( n^L_{1} , n^R_{1} , n^L_{2} ) + {\rm H.c.}, \\
    &H_I^{{\rm SB}(2)}(r)
    = - x\,\sigma_{2r}^- Z_{2r+1} \sigma_{2r+2}^+ \lambda^{-}_{L1} \lambda^{-}_{R2} \DSB ( n^L_{1} , n^R_{2} , n^L_{2} ) +{\rm H.c.}, \\
    &H_I^{{\rm SB}(3)}(r)
    = - x\,\sigma_{2r+1}^- \sigma_{2r+2}^+ \lambda^{-}_{L2} \lambda^{-}_{R2} \DSB ( n^L_{2} , n^R_{2} , n^L_{1} ) +{\rm H.c.}, \\
    &H_I^{{\rm SB}(4)}(r)
    = x\,\sigma_{2r+1}^- Z_{2r+2} \sigma_{2r+3}^+ \lambda^{-}_{L2} \lambda^{-}_{R1} \DSB ( n^L_{2} , n^R_{1} , n^L_{1} ) +{\rm H.c.}, \\
    &H_I^{{\rm SB}(5)}(r)
    =  - x\,\sigma_{2r+1}^+ \sigma_{2r+2}^- \lambda^{-}_{L1} \lambda^{-}_{R1} \DSB ( n^L_{1} , n^R_{1} , n^L_{2} ) + {\rm H.c.}, \\
    &H_I^{{\rm SB}(6)}(r)
    = x\,\sigma_{2r}^+ Z_{2r+1} \sigma_{2r+2}^- \lambda^{-}_{L2} \lambda^{-}_{R1} \DSB ( n^L_{2} , n^R_{1} , n^L_{1} ) +{\rm H.c.}, \\
    &H_I^{{\rm SB}(7)}(r)
    = x\,\sigma_{2r}^+ Z_{2r+1} Z_{2r+2} \sigma_{2r+3}^- \lambda^{-}_{L2} \lambda^{-}_{R2} \DSB ( n^L_{2} , n^R_{2} , n^L_{1} ) +{\rm H.c.}, \\
    &H_I^{{\rm SB}(8)}(r)
    = - x\,\sigma_{2r+1}^+ Z_{2r+2} \sigma_{2r+3}^- \lambda^{-}_{L1} \lambda^{-}_{R2} \DSB ( n^L_{1} , n^R_{2} , n^L_{2} ) +{\rm H.c.}
\end{align}
\end{subequations}
The diagonal function $\DSB$ is defined as\footnote{$\DSB (0,q,0) \equiv 0$ because this would correspond to lowering one of the oscillators beyond its lower cutoff.}
\begin{align}
    \label{eq:DSB}
    \DSB( p, q, p') \equiv \sqrt{\frac{p \, q}{(p+p')(p+p'+1)}} .
\end{align}
Here for brevity, we have adopted the shorthand notation $p \equiv N_\reg{p}$, where $N_\reg{p}$ is the occupation-number operator on the bosonic register $\reg{p}$\footnote{$p$ as an eigenvalue versus an operator can be deduced from the context.}
(and similarly for $q$ and $p'$).

The action of the cyclic ladder operators $\lambda^\pm$ on the bosonic modes should be realized as $\lambda^+_\reg{p}\ket{p}=(1-\delta_{p,\Lambda})\ket{p+1}+\delta_{p,\Lambda} \ket{0}$ and $\lambda^-_\reg{p}\ket{p}=(1-\delta_{p,0})\ket{p-1}+\delta_{p,0} \ket{\Lambda}$.
The reason such modified ladder operators still reproduce the original Schwinger-boson interaction Hamiltonian is the presence of $D(p,q,p')$ factors. Consider two situations that are affected by the new definition of the ladder operators: i) Terms with $\lambda^-_\reg{p}\lambda^-_\reg{q} D(p,q,p')$ operator structure: consider, for example, $\lambda^-_\reg{p} D(p,q,p') \ket{0}_\reg{p}$. While $\lambda^-_\reg{p}\ket{0}_\reg{p}=\ket{\Lambda}_\reg{p}$ for the cyclic operator, since $D(0,q,p')=0$, the amplitude for this process remains zero, which reproduces what is expected from the non-cyclic ladder operator. Similarly, the same argument holds for $p \to q$. ii) Terms with $D(p,q,p')\lambda^+_\reg{q}\lambda^+_\reg{p}$ operator structure: consider, for example, $D(p,q,p')\lambda^+_\reg{p} \ket{\Lambda}_\reg{p}$. While $\lambda^+_\reg{p}\ket{\Lambda}_\reg{p}=\ket{0}_\reg{p}$ for the cyclic operator, since $D(0,q,p')=0$, the amplitude for this process remains zero, which reproduces what is expected from the non-cyclic but truncated ladder operator. Similarly, the same argument holds for $p \to q$.\footnote{The situation is different for the U(1) theory, where the action of the ladder operators on any state is never equivalent to zero in the untruncated theory, and hence the cyclic wrapping will create a mixing between the lower and upper cutoff states, unless extra operations are introduced to prevent it, see e.g., Ref.~\cite{Shaw:2020udc}.}

The terms in Eqs.~(\ref{eqs:HI-subterms-JW_SB}) have been put into the form of (shifting operators) $\times$ (diagonal operators) to match the example Hamiltonians studied in Sec.~\ref{sec:SVD}, so that the SVD diagonalization methods are readily applied.
In the language of Sec.~\ref{sec:SVD}, each term $H_I^{\text{SB}(j)}$ involves two spin raising and lowering operators (that act on fermionic modes at two adjacent sites), two bosonic raising or lowering operators (acting at opposite ends of a common link), and a non-trivial diagonal function (of three bosonic quantum numbers). The general structure of the Schwinger-boson hopping subterms is encapsulated by
\begin{align}
    H_I^{\mathrm{SB}(j)}
    = \pm x \, \sigma_{\reg{x}}^+ \color{gray}( Z_\reg{x'} ) ( Z_\reg{y'} )\color{black} \sigma_\reg{y}^- \DSB( p, q, p') \lambda^+_\reg{p} \lambda^+_\reg{q} + \mathrm{H.c.},
    \label{eq:generalizedHISubterm-JW_SB}
\end{align}
with the labels $\{ \reg{x}$, $\reg{x'}$, $\reg{y}$, $\reg{y'} \}$ introduced for fermionic modes, and $\{ \reg{p}$, $\reg{q}$, $\reg{p'} \}$ for bosonic modes. The precise mapping to modes for each subterm is displayed in Table \ref{tab:HI-subterm-labels_SB}. Here and from now on, we denote the two Pauli-$Z$ operators on registers $\reg{x}'$ and $\reg{y}'$ in parentheses in gray color, to remind that one or both of them may be absent in some of the hopping terms, as is seen in Eqs.~\eqref{eqs:HI-subterms-JW_SB}.
\begin{table}
\centering
\begin{tabular}{cccccccccc}
\vspace{0.5mm} Label & ~~ & \multicolumn{8}{c}{$H_I^{\mathrm{SB}(j)}$ translation} \\
\hline
\hline
\vspace{1mm} $j$ & ~~ & 1 & 2 & 3 & 4 & 5 & 6 & 7 & 8 \\
$\reg{x }$ & ~~ & $2r  $ & $2r  $ & $2r+1$ & $2r+1$ & $2r+2$ & $2r+2$ & $2r+3$ & $2r+3$ \\
$\reg{y }$ & ~~ & $2r+3$ & $2r+2$ & $2r+2$ & $2r+3$ & $2r+1$ & $2r  $ & $2r  $ & $2r+1$ \\
$\reg{x'}$ & ~~ & $2r+1$ & $2r+1$ & $2r  $ & $2r  $ & $2r+3$ & $2r+3$ & $2r+2$ & $2r+2$ \\
$\reg{y'}$ & ~~ & $2r+2$ & $2r+3$ & $2r+3$ & $2r+2$ & $2r  $ & $2r+1$ & $2r+1$ & $2r  $ \\
$\reg{p }$ & ~~ & ${L1}$ & ${L1}$ & ${L2}$ & ${L2}$ & ${L1}$ & ${L2}$ & ${L2}$ & ${L1}$ \\
$\reg{q }$ & ~~ & ${R1}$ & ${R2}$ & ${R2}$ & ${R1}$ & ${R1}$ & ${R1}$ & ${R2}$ & ${R2}$ \\
$\reg{p'}$ & ~~ & ${L2}$ & ${L2}$ & ${L1}$ & ${L1}$ & ${L2}$ & ${L1}$ & ${L1}$ & ${L2}$ \\
\end{tabular}
\caption{\label{tab:HI-subterm-labels_SB}
Label associations between registers of a generalized Schwinger-boson subterm, Eq.~(\ref{eq:generalizedHISubterm-JW_SB}), and the eight subterms of Eqs.~(\ref{eqs:HI-subterms-JW_SB}). For terms that do not involve the Pauli-Z operators on both registers $\reg{x}'$ and $\reg{y}'$, the assignment of these registers to the qubit labels is arbitrary, but a choice is being made for concreteness.
}
\end{table}

Now recalling the discussion of Sec.~\ref{sec:SVD}, the combination
$\color{gray}( Z_\reg{x'} ) ( Z_\reg{y'} ) \color{black}\sigma_\reg{y}^- \DSB( p, q, p') \lambda^+_\reg{p} \lambda^+_\reg{q} $ squares to zero, allowing it to be identified with the `$A$' operator of Sec.~\ref{sec:SVD}. There is also no need to introduce an ancilla qubit because the $\reg{x}$ qubit can function as the control to the SVD-transformation gates ($\mathscr{V}$, $\mathscr{W}$, or their adjoints).
A solution to the SVD transformations is $\mathscr{V}=\mathcal{I}_\reg{y}\mathcal{I}_\reg{p}\mathcal{I}_\reg{q}$ and $\mathscr{W} = X_\mathtt{y} \lambda^-_\reg{p} \lambda^-_\reg{q} $.
The full diagonalizing unitary for the generalized subterm in Eq.~\eqref{eq:generalizedHISubterm-JW_SB} thus works out to be
\begin{align}
    \SVDSB &\equiv \had_\reg{x} \left( \ket{0}\bra{0}_\reg{x} + \ket{1}\bra{1}_\reg{x} X_\reg{y} \lambda^+_\reg{p} \lambda^+_\reg{q} \right) .
\end{align}
When applied to the hopping subterms, one ultimately obtains
\begin{align}
    \SVDSB & H_I^{\mathrm{SB}(j)} {\mathscr{U}_{\mathrm{SVD}}^{\mathrm{SB}}}^\dagger = \pm x \, Z_\reg{x} \color{gray}( Z_\reg{x'} ) ( Z_\reg{y'} )\color{black} \ket{1}\bra{1}_{\reg{y}} \DSB( p, q, p') . \label{eq:diagonalized-subterm_SB}
\end{align}
\begin{lemma}
Using a $(4\eta+4)$-qubit register and no ancilla qubits, 
$e^{-itH^{{\rm SB}}_I(r)}$
can be simulated approximately as 
$\Pi_{j=1}^8 e^{-itH^{{\rm SB}(j)}_I(r)}$ using at most 
$16\times 8^\eta+64\eta^2+64\eta+32$
CNOT gates (and a number of single-qubit Z rotations).
\end{lemma}
\begin{proof}
This cost consists of the (near-term) cost of $\SVDSB $, plus the (near-term) cost of implementing a diagonalized subterm.
The $\SVDSB $ circuit, shown in Fig.~\ref{fig:HI_high-level_SB}, consists of one CNOT, two controlled-$\lambda^+$ ($C(\lambda^+)$) gates, and a Hadamard gate.
A single $C(\lambda^+)$ gate is identical to the uncontrolled incrementer on $\eta+1$ qubits (followed by a bit flip on the ``control'' qubit), which requires
$2\eta(\eta+1)$ CNOT gates, according to Lemma \ref{lem:inc-near}.
The total CNOT-gate cost of a single $\SVDSB $ execution is then
$1+2\times2\eta(\eta+1)=4\eta^2+4\eta+1$,
as reported in Table \ref{tab:offdiag-costs_near_SB}.
The dominant CNOT-gate cost, however, is associated with implementing the diagonalized subterms, which will be discussed next.

The diagonalized subterms to be simulated, according to Eq.~\eqref{eq:diagonalized-subterm_SB}, are of the form:
\begin{align*}
    \pm x \, Z_\reg{x} \color{gray} ( Z_\reg{x'} ) ( Z_\reg{y'} ) \color{black} \ket{1}\bra{1}_{\reg{y}} \DSB( p, q, p').
\end{align*}
The function $\DSB(p,q,p')=\sqrt{p q/(( p +  p')( p +  p' +1))}$, in general, cannot be expressed by a low-degree polynomial, and the exact Pauli decomposition requires up to $(\Lambda+1)^3=2^{3\eta}$ terms, which can be solved for numerically, see Eq.~(\ref{eq:cCoeff}). Note that if a fixed accuracy is sought in evaluating $\DSB$, one may consider replacing the complete expansion by an approximate one that uses less terms. Later we will consider two approaches to such approximation: ``small-angle truncation'' and ``input truncation'', to be discussed below. For now, we will cost the simulation of the diagonal hopping propagator in the worst-case scenario, that is assuming the Pauli decomposition of the phase function $\DSB$ is not truncated, and that all the $2^{3\eta}$ terms in the expansion come with non-zero coefficients. 
\begin{table*}[t!]
\centering
\begin{tabular}
{>{\arraybackslash}p{11cm}  >{\centering\arraybackslash}p{4cm}  }
Schwinger-boson subterm SVD subroutine & {CNOT} count \\
\hline
\hline
`Explicit' {CNOT} & 1  \\
$C(\lambda^{\pm})$ & $2\eta(\eta+1)$  \\
\hline
Overall $\SVDSB$ circuit for each subterm & $4\eta^2+4\eta+1$  \\ 
\\
\\
Schwinger-boson subterm diagonal-rotations subroutine 
& {CNOT} count \\
\hline
\hline
Compute J-W parity of $\tau=0,1,2$ number of Pauli-Z operators & $\tau$  \\
Hamiltonian cycle through `hypercube' of Pauli parities & $2^{3\eta+1}$  \\
Uncompute J-W parity & $\tau$ \\
\hline
Overall diagonal-rotations circuit & $2^{3\eta+1}+2\tau$ \\
\\
\\
Schwinger-boson hopping-propagator routine & {CNOT} count \\
\hline
\hline
Subterm SVD & $4\eta^2+4\eta+1$ \\
Subterm diagonal rotations (Pauli-$Z$-string length of $\tau$) & $2^{3\eta+1}+2\tau$ \\
Subterm SVD${}^{-1}$ & $4\eta^2+4\eta+1$ \\
\hline
Full $e^{-itH^{{\rm SB}}_I(r)}$ circuit & $16\times 8^\eta+64\eta^2+64\eta+32$ \\ 
\end{tabular}
\caption{\label{tab:offdiag-costs_near_SB}
Summary of the CNOT-gate cost associated with near-term simulation of off-diagonal operators in the Schwinger boson Hamiltonian, as explained in the text.}
\end{table*}

To find an efficient gate implementation of a full Pauli decomposition, first consider an individual Pauli-$Z$ string $ \prod_{i=1}^m Z_i$, mapped to an $\mathsf{N}$-bit binary string ($m\leq \mathsf{N}$), with a `0' (`1') for bits that are acted on by `1' ($Z$). In order to implement the rotation $e^{-i\alpha \prod_{i=1}^m Z_i}$, one has to compute the parity of the corresponding collection of bits labelled as 1. This can be done by computing that parity using $m$ CNOT gates, each one controlled by one of the `1' qubits and all with a shared target ancilla (initialized to 0), then applying $e^{-i \alpha Z}$ on that ancilla, and finally uncomputing the parity by reapplying the same $m$ CNOT gates. To effect the full Pauli decomposition, it is sufficient to implement $e^{-i\alpha \prod_{i=1}^m Z_i}$ back to back for all possible collections of $Z$ rotations on $m$ qubits and for all values of $m$. In practice, the order in which the various Pauli-$Z$ string operators are simulated is carefully chosen such that the overwhelming majority of CNOT gates involved with computing and uncomputing parities cancel each other out, leading to significant savings over applying each rotation, with all the CNOT gates involved, in isolation.

This cost reduction can be explained and quantified in terms of a Hamiltonian cycle on a hypercube.
Consider $B=\{0,1\}^{ \mathsf{N}}$ to be the set of bit strings, where the bit string $(b_{\mathsf{N}-1},b_{\mathsf{N}-2},\cdots ,b_{1},b_{0})$ is associated with the diagonal operator $Z_{\mathsf{N}-1}^{b_{\mathsf{N}-1}} Z_{\mathsf{N}-2}^{b_{\mathsf{N}-2}} \cdots Z_1^{b_{1}} Z_0^{b_{0}} $.
Each bit string $b\in B$ can be visualized as a corner on an $\mathsf{N}$-dimensional hypercube.
For any $\mathsf{N}$-dimensional hypercube, there always exists a Hamiltonian cycle starting and ending at the all-zeroes vertex.
Such a Hamiltonian cycle through the hypercube provides a Gray code for iterating through all the bit strings in $B$, and by extension a ``Gray code'' of the diagonal Pauli operators.\footnote{%
Gray codes for $\mathsf{N}$-bit strings can be generated efficiently using, e.g., the recursive algorithm presented in Ref.~\cite{gray_1953}. The Gray code for $\mathsf{N}=1$ bit is $(0,1)$.
The $\mathsf{(N+1)}$-bit code is obtained by concatenating one copy of the $\mathsf{N}$-bit code with a reversed copy of the $\mathsf{N}$-bit code, prepending a 0 to the first copy, and prepending a 1 to the reverse copy.
Hence, for $\mathsf{N}=2$ one has $(0,1) \to (0,1,1,0) \to (00,01,11,{\color{red}1}0)$, for $\mathsf{N}=3$ one obtains (000,001,011,010,110,111,101,100), and so on.}
Now, each vertex is associated with a diagonal Pauli operator, and to simulate that operator one needs the associated parity.
The starting vertex of all zeroes corresponds to the identity, for which one applies a global phase.
Then, each step along the Hamiltonian cycle between adjacent corners $b$ and $b'$ corresponds to a CNOT gate in the following way:
if the bit strings $b$ and $b'$ differ in the $k^{\text{th}}$ bit, then the parity of $b'$ is derived from that of $b$ by applying a CNOT controlled by qubit $k$ onto the target ancilla.
Between each step (i.e., each CNOT implementation), one applies $e^{-i \alpha Z}$ on the ancilla to effect the rotation of a distinct operator.
The cost of traversing all $2^\mathsf{N}$ bit strings and returning to the origin will then be $2^\mathsf{N}$ CNOT gates and $2^{\mathsf{N}}-1$ non-trivial ancilla rotations.
Note that one can adapt the above procedure for when an ancilla is not desired, but the overall cost will be identical asymptotically.\footnote{For similar approaches on efficient and exact implementation of diagonal functions, see e.g., Ref.~\cite{bullock2003smaller}, in which an optimal cost is reported to be $2^\mathsf{N}-2$ CNOT gates and $2^\mathsf{N}-1$ single-$R^Z$ rotations.}

In the case at hand, the diagonal operator $\ket{1}\bra{1}_\reg{y} \DSB(p,q,p')$ depends on $3\eta+1$ qubits, implying a cost of $2^{3\eta+1}$ CNOT gates.
However, this has not yet accounted for the $Z_\reg{x}\color{gray}( Z_\reg{x'} ) ( Z_\reg{y'})$ factor on the diagonalized subterm.
The $\reg{x}$ qubit can simply act as the ancilla for all parity evaluations, costing no additional CNOT gates.
If $\tau=0,1,2$ number of operators in $\{ Z_\reg{x'} , Z_\reg{y'} \}$ are called for, then an additional overhead of $2 \tau$ CNOT gates is needed.
This leads to the total cost of $2^{3\eta+1}+2\tau$ for a given diagonalized Schwinger-boson subterm, as reported in Table \ref{tab:offdiag-costs_near_SB}.

The cumulative CNOT cost associated with a complete hopping term is obtained as 16 times the cost of a single $\mathscr{U}_{\mathrm{SVD}}^{\mathrm{(SB)}}$ circuit, plus the individual costs of all eight diagonalized subterms.
This works out to $16\times 8^\eta+64\eta^2+64\eta+32$, also reported in Table \ref{tab:offdiag-costs_near_SB}.
\end{proof}
\begin{figure*}[t!]
\centering
    \includegraphics[width=0.51\textwidth]{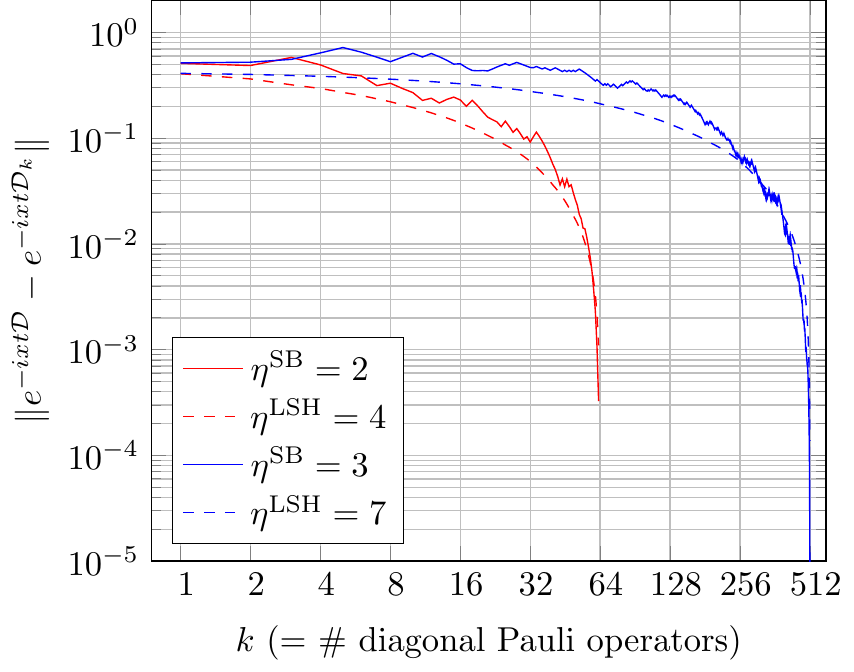}
    \caption{\label{fig:diagonal-hopping-Pauli-truncation}
    The spectral-norm error of approximating $e^{-ixt \mathcal{D}}$ at $xt=1$ as function of $k$, the number of operators 
    kept from the full Pauli decomposition of $\mathcal{D}$, defined in Eq.~(\ref{eq:DSB}) ((Eq.~(\ref{eq:DLSH})) for the Schwinger-boson (LSH) hopping propagator, after truncating the smallest rotations.
    The Schwinger-boson diagonal function is modified for these plots to vanish for quantum numbers that cannot satisfy the constraint $q\leq p + p'$ imposed by the AGL.
    The truncated diagonal functions containing the $k$ largest rotations are called $\mathcal{D}_k$, and the error here is defined as the spectral norm of $e^{-ixt \mathcal{D}}-e^{-ixt \mathcal{D}_k}$.
    $\eta$ denotes the number of qubits per bosonic quantum number in either formulation.
    The diagonal functions considered are six-qubit (red) or nine-qubit (blue) operators. In other words, unequal $\eta$ values (or different cutoffs) are chosen among the two formulations in order to enable direct comparisons of the accuracy achieved for truncating Pauli decompositions of equal sizes.}
\end{figure*}

Given the significant cost of implementing $e^{-itZ_\reg{x} \color{gray}(Z_\reg{x'}Z_\reg{y'})\color{black}\DSB(p,q,p')}$ exactly, one may seek approximations to the exact implementation in the near-term scenario, such as the cases considered in Refs.~\cite{welch2014efficient,Kane:2022ejm} where the gate-count scaling is polynomial in the number of qubits and the reciprocal error tolerance.
To that end, we explored the following two possible means of truncating the Pauli decompositions.

\vspace{0.2 cm}
\noindent
\emph{Small-angle truncation.}---Let $\DSB_k(p,q,p')$ denote the Pauli decomposition of $\DSB( p, q,  p')$ where only the first $k$ terms with the largest coefficients are retained, a method that we here term as ``small-angle truncation''. If the goal is to match the ideal evolution up to a fixed spectral-norm error, the Pauli decomposition usually can only be truncated minimally, or not at all. Figure~\ref{fig:diagonal-hopping-Pauli-truncation} plots the error trade-off per number of Pauli-$Z$ strings retained (dropping the fermionic $Z$-rotations for simplicity) for $\eta=2,3$ (or $\eta = 3,7$ for the LSH formulation) and at a fixed simulation time $t=x^{-1}$, where $x$ is the hopping strength.
Note that for the Schwinger-boson formulation (only), the diagonal function can be modified to incorporate knowledge of the AGL constraint:
input states with $q>p+p'$ are inconsistent with the AGL, so $\DSB$ can be set to zero on such states. The results obtained in Fig.~\ref{fig:diagonal-hopping-Pauli-truncation} correspond to such a modified function.
It is observed that the spectral norm of the difference $e^{-itx \DSB( p, q, p')}-e^{-itx \DSB_k( p, q, p')}$ decreases rather slowly as $k$ is increased until $k$ is close to its maximal value $2^{3\eta}$, hence making it challenging to achieve high accuracy for any $k \ll 2^{3\eta}$. 
For example, consider a lattice with $L=4$ and $tx=1$. There are $48 \, s$ applications of the diagonal hopping propagator required in a second-order product formula with $s$ Trotter steps. For $\eta=2$ and for only a single Trotter step ($s=1$), if the overall budget for the simulation is equal to $0.1$, then the error budget per individual hopping subterm is only $0.1/48 \approx 2.08\times 10^{-3}$. According to Fig.~\ref{fig:diagonal-hopping-Pauli-truncation}, the Pauli decomposition would then have to be truncated at $k\geq62$ (out of 64 possible terms). Note that in this example, error tolerance does not even account for the Trotterization error.

\begin{figure*}[t!]
\centering
    \includegraphics[width=0.51\textwidth]{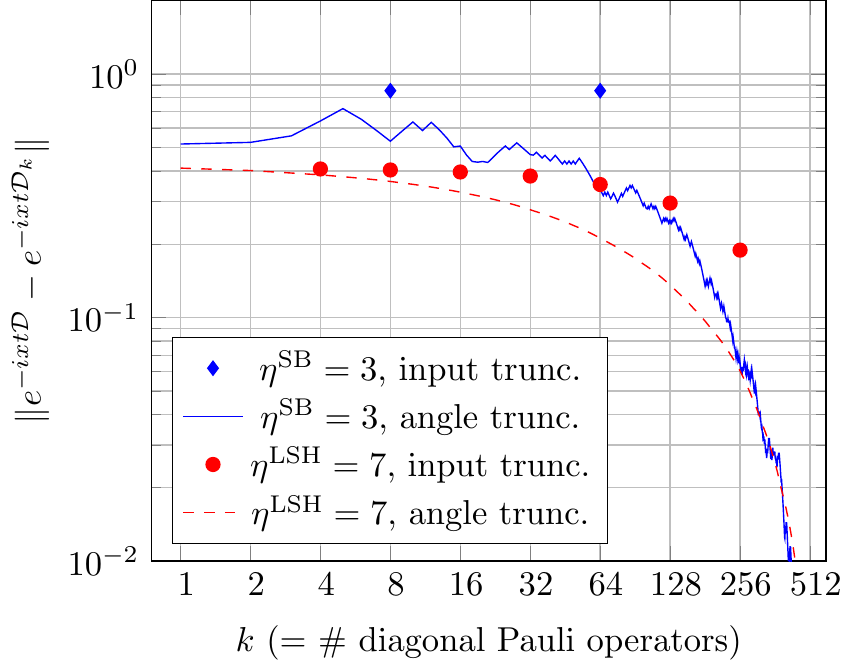}
    \caption{\label{fig:diagonal-hopping-input-truncation}
    The spectral-norm error of approximating $e^{-ixt \mathcal{D}}$ at $xt=1$ as function of $k$, the number of operators in the Pauli decomposition of $\mathcal{D}$ after the least significant bits of its bosonic arguments are truncated.
   The Schwinger-boson diagonal function was modified for these plots to vanish for quantum numbers that cannot satisfy the constraint $q\leq p + p'$ imposed by the AGL.
    In the input-truncation scheme, $k$ is a power of two.
    The truncated diagonal functions depending on only $\log_2(k)$ qubits are called $\mathcal{D}_k$, and the error here is defined as the spectral norm of $e^{-ixt \mathcal{D}}-e^{-ixt \mathcal{D}_k}$.
    $\eta$ denotes the number qubits per bosonic quantum number in either formulation.
    The diagonal functions considered are all nine-qubit operators for both the Schwinger-boson and the LSH formulations. In other words, unequal $\eta$ values (or different cutoffs) are chosen among the two formulations in order to enable direct comparisons of the accuracy achieved for truncating Pauli decompositions of equal sizes.
}
\end{figure*}

Besides the fact that the truncation of small-angle rotations may not be an efficient way of reducing the simulation cost for the examples of this work, as was demonstrated in Fig.~\ref{fig:diagonal-hopping-Pauli-truncation},
there exist other problems with this method. Importantly, truncation of small angles alone does not specify \emph{a priori} what Pauli operators will have to be simulated. The angles will have to be pre-computed classically,\footnote{Nonetheless, efficient (polynomial-time) algorithms have been developed to obtain a finite number of terms in the series expansion of functions~\cite{kushilevitz1991learning}---algorithms that may be employed in this context to mitigate the classical-computational cost of finding the truncated Pauli expansion of functions.} and the circuits that will have to be generated may, in principle, involve all qubits.
For general functions, no systematic structure is guaranteed for the retained $k$ terms, making it difficult to place a tight upper bound on the number of CNOT gates called for, as well as to give an explicit algorithm for cycling through all $k$ rotations.

\vspace{0.2 cm}
\noindent
\emph{Input truncation.}---To address these problems, an alternative truncation scheme may be used, which we term ``input truncation''.
This amounts to taking the integer arguments to the diagonal function, e.g., $p$, $q$, and $p'$ in the Schwinger-boson formulation, and truncating some number of the least significant bits of those inputs.
The resulting Pauli decomposition can only involve Pauli-$Z$ operators for the retained bits of the input, with each truncated bit giving a factor of two reduction in the total set of operators.
As a simple example, consider the matrix $\text{diag}(a,b,c,d)$, a diagonal function on two qubits whose Pauli decomposition involves $\mathcal{I}_1 \otimes \mathcal{I}_2$, $Z_1 \otimes \mathcal{I}_2$, $\mathcal{I}_1 \otimes Z_2$, and $Z_1 \otimes Z_2$.
If dependence on the least significant bit is dropped, this matrix becomes $\text{diag}(a,a,c,c)$, which is a linear combination of $\mathcal{I}_1 \otimes \mathcal{I}_2$ and $Z_1 \otimes \mathcal{I}_2$ only.

Following input truncation, one can use the same Hamiltonian-cycle method previously described for the case of all qubits, but now being limited to a hypercube in a lower-dimensional space.
The $2^\mathsf{N}$ scaling of the CNOT gates for this method implies that the number of required CNOT gates will be halved for each dropped bit of the input.
In Fig.~\ref{fig:diagonal-hopping-input-truncation}, the spectral-norm errors resulting from the input-truncation method are shown for the case of diagonal operators acting on nine total qubits for the Schwinger-boson (blue diamonds) and the LSH (red circles) diagonal functions .
To obtain the results in Fig.~\ref{fig:diagonal-hopping-input-truncation}, as was the case for Fig.~\ref{fig:diagonal-hopping-Pauli-truncation}, the diagonal function of the Schwinger-boson formulation (only) is modified to incorporate knowledge of the AGL constraint:
input states with $q>p+p'$ are inconsistent with the AGL, so $\DSB$ was modified to be zero on such states.
For ease of comparison, data from small-angle truncation of the nine-qubit operators in Fig.~\ref{fig:diagonal-hopping-Pauli-truncation} are also displayed in Fig.~\ref{fig:diagonal-hopping-input-truncation}.
It is evident that, at least for the examples considered, the error bounds tend to be looser from this method, indicating a trade-off between tightening the theoretical error bound and lowering the number of the required CNOT gates systematically. Better characterization of the pros and cons of the two methods outlined here will be needed in future work on the near-term implementations of the SU(2) and similar theories, in which the computation of a non-trivial dynamical phase is necessary.
\begin{center}
\textit{---Far term---}
\end{center}
\begin{center}
\emph{Implementing mass propagators}
\end{center}

\begin{lemma}
The mass propagator $e^{-itH^{{\rm SB}}_M(r)}$ can be implemented, without approximation, using no ancilla qubits, two $R^Z$ gates, and no additional T gates.
\end{lemma}
\begin{proof}
The $R^Z$ gate count follows directly from Eq.~\eqref{eq:propMr}.
It must be noted that in the far-term scenario, the synthesis of each $R^Z$ calls for a number of T gates that depends on the desired accuracy in the rotation angle.
We will assume there exists a cost function $\mathcal{C}_z(\epsilon)$ giving the number of T gates required to simulate an $R^Z$ gate (of any input angle) within a spectral-norm error of $\epsilon$.
The repeat-until-success (RUS) method of Ref.~\cite{bocharov2015efficient} has an average cost $\mathcal{C}_z(\epsilon) \approx 1.15 \log_2 (2/\epsilon) + 9.2 $, giving a T-count scaling of $\log(1/\epsilon)$ in the limit of high precision.
The trivial cost of a mass propagator is recorded in Table \ref{tab:diag-costs_far_SB}.
\end{proof}
\begin{table*}[t!]
\centering
\begin{tabular}{l >{\centering\arraybackslash}p{4.1cm} c c }
Schwinger-boson mass propagator subroutine & T gates & Workspace~~ & Scratch space \\
\hline
\hline
Each $E^{-itH_M^{{\rm SB}(j)}}$ rotation & $\mathcal{C}_z (\epsilon)$ & 0 & 0 \\
\hline
Full $e^{-itH^{{\rm SB}}_M(r)}$ circuit & $ 2 \, \mathcal{C}_z (\epsilon) $ & 0 & 0 \\ 
\\
\\
Schwinger-boson electric propagator subroutine & T gates & Workspace~~ & Scratch space \\
\hline
\hline
(Un)compute $N^L$ & $4\eta$ & $\eta$ & 1 \\
Copy & 0 & 0 & $\eta+1$ \\
$(\eta+1)\times(\eta+1)$ mult. & $8\eta^2 + 12\eta + 4$ & $2\eta+2$ & $2\eta+2$ \\
$H_E^{{\rm SB}(1)}$ rotations & $(\eta+1) \mathcal{C}_z (\epsilon)$ & 0 & 0 \\
$H_E^{{\rm SB}(2)}$ rotations & $(2\eta+2) \mathcal{C}_z (\epsilon)$ & 0 & 0 \\
\hline
Full $e^{-itH^{{\rm SB}}_E(r)}$ circuit & $ 16\eta^2 + 32\eta + 8 + (3\eta+3) \mathcal{C}_z (\epsilon)$ & $2\eta+2$ & $3\eta+4$ \\ 
\end{tabular}
\caption{\label{tab:diag-costs_far_SB}
Summary of the costs associated with the far-term simulation of the diagonal operators in the Schwinger-boson Hamiltonian, as explained in the text.}
\end{table*}
\begin{center}
\emph{Implementing electric propagators}
\end{center}
\begin{lemma}
The electric propagator $e^{-itH^{{\rm SB}}_E(r)}$ can be implemented, without approximation, using $5\eta + 6$ ancilla qubits, $3\eta+3$ $R^Z$ gates, and an additional $ 16\eta^2 + 32\eta + 8$ T gates. 
\end{lemma}
\begin{proof}
The procedure leading to this cost consists of the following steps:
\begin{enumerate}
    \item Compute $N^L$.
    \item Compute $(N^L)^2$.
    \item Effect phase kickback via the registers containing $N^L$ and $(N^L)^2$.
    \item Uncompute $(N^L)^2$ and $N^L$.
\end{enumerate}
In step (1), $N^L$ can be computed as 
\begin{equation}
\ket{n^L_1} \ket{n^L_2} \ket{0}^{\otimes (1+\eta)} \mapsto \ket{n^L_1} \ket{N^L} \ket{0}^{\otimes \eta}
\end{equation}
using an $\eta$-bit in-place adder.
The adder, according to Lemma~\ref{lem:inc-far}, calls for $\eta$ workspace qubits and costs $4\eta$ T gates.
For (2), $(N^L)^2$ is computed as $\ket{N^L}\ket{0}^{\otimes(5\eta+5)} \mapsto \ket{N^L}  \ket{N^L} \ket{0}^{\otimes(4\eta+4)} \mapsto \ket{N^L} \ket{N^L} \ket{(N^L)^2} \ket{0}^{\otimes(2\eta+2)}$, by first copying $N^L$ to an $(\eta+1)$-bit register using CNOT gates, and then multiplying the two copies  of $N^L$.
According to Lemma~\ref{lem:multi}, the multiplier costs $8\eta^2 + 12\eta + 4$ T gates and $2\eta+2$ bits of workspace.
In step (3), the $N^L(r)$ and $(N^L(r))^2$ terms of $H_E(r)$ are effectively simulated by applying single-qubit $R^Z$ gates across the $\ket{N^L}$ and $\ket{(N^L)^2}$ registers (up to global phases that are dropped).
Finally, in step (4), uncomputation involves reversing the gates of steps (2) and (1) and the associated costs are the same.
Steps (2)-(4) are shown in Fig. \ref{fig:HE_phase-kickback}.
\begin{figure}[t]
    \centering
    \adjustbox{scale=1.0,center}{%
 \begin{tikzcd}[row  sep={1cm,between  origins},transparent]  
 \lstick{$\ket{N^L}$} & \qwbundle{\eta+1} & \gate[2]{\rotatebox{90}{\text{~Copy}}} & \gate[3,label  style={yshift=0.4cm}]{\rotatebox{90}{\text{Mult.}}} & \gate{ \prod_{k=0}^{\eta} \exp( i \, t \, 2^{k-2} Z_k )} & \gate[3,label  style={yshift=0.4cm}]{\rotatebox{90}{\text{Mult.}${}^\dagger$}} & \gate[2]{\rotatebox{90}{\text{~Copy}${}^\dagger$}} & \qw & \qw & \rstick{$\ket{N^L}$}  \\ 
 \lstick{$\ket{0}$} & \qwbundle{\eta+1} & \gateoutput{$>$} & \qw & \qw & \qw & \gateinput{$<$} & \qw & \qw & \rstick{$\ket{0}$}  \\ 
 \lstick{$\ket{0}$} & \qwbundle{2\eta+2} & \qw & \gateoutput{$>$} & \gate{ \prod_{k=0}^{2\eta+1} \exp( i \, t \, 2^{k-3} Z_k )} & \gateinput{$<$} & \qw & \qw & \qw & \rstick{$\ket{0}$} 
 \end{tikzcd}  
    }
    \caption{\label{fig:HE_phase-kickback}
    A quantum circuit to realize the phase kickback for $H_E(r)$. The same circuit is applicable to both the Schwinger-boson and the LSH formulations, with the only difference being the evaluation of $N^L$ (the subcircuit that evaluates $N_L$ is not shown).    The $>$ symbol denotes that the obtained value in the subcircuit is stored in the corresponding qubit register, and the $<$ symbol indicates that the corresponding register is cleared from the stored values as a result of the action of the inverse subcircuit.
    }
\end{figure}
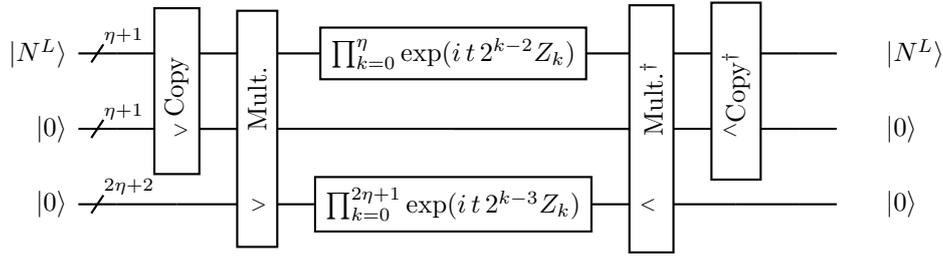
In total, the procedure outlined above involves $3\eta+3$ $R^Z$ gates, which can each be done to the desired precision using the RUS method mentioned above \cite{bocharov2015efficient}.
The costs associated with all the subroutines are summarized in Table~\ref{tab:diag-costs_far_SB}.
The final, quoted costs are obtained by adding up the workspace and scratch-space sizes.
\end{proof}
\begin{center}
\emph{Implementing hopping propagators}
\end{center}

\begin{lemma}
Let $\eta=\log_2(\Lambda+1)$ be the number of qubits per Schwinger boson mode, $n>0$ be the number of Newton's method iterations, and $m>0$ be a fixed binary arithmetic precision.
Then $\Pi_{j=1}^8 e^{-itH^{{\rm SB}(j)}_I(r)}$ can be implemented within an additive spectral-norm error of 

\begin{align}
    \label{eq:IPropError_SB}
    8 x t \left [ 2^n \left(\sqrt{2}-1\right)^{2^n} + 2^{2-m} \left( \left(\frac{3}{2}\right)^n - 1\right) \right ]
\end{align}
using 
\begin{equation}
4 n \eta+27 \eta+\max (2\eta+2,m)+3 m n+10 m+3 n+15
\end{equation}
ancilla qubits,
$16(m+1)$ $R^Z$ gates, and
\begin{align}
    & 2560 \eta m n + 512 m^2 n + 768 \eta^2 + 256 \eta m + 256 \eta n + 1408 m n - 128 n \max (4\eta+2,2m) \nonumber\\
    & \hspace{8.6 cm} + 896 \eta - 64 \max (2\eta,m) - 64 n 
\end{align}
T gates.
\end{lemma}
\begin{proof}
To obtain this result, one must circuitize the propagators associated with each of the eight subterms of the hopping term.
One subterm is simulated by applying an appropriate diagonalizing transformation $\SVDSB$, implementing the diagonalized subterm, and finally changing basis back via ${\SVDSB}^\dagger$.
In the far-term scenario, the diagonalized subterm will be implemented via phase kickback.
The sequence for simulating a hopping subterm is shown schematically in Fig.~\ref{fig:HI_high-level_SB}(a).
\begin{figure*}[t]
    \centering
    \adjustbox{width=1.0\linewidth,center}{%
 \begin{tikzcd}[row  sep={7mm,between  origins},transparent]  
 \lstick{$\color{gray}(\mathtt{y^\prime})\color{black}$} & \qw & \qw & \gate[7,style={fill=green!100,fill  opacity=0.2},disable  auto  height][2.6cm]{\text{phase  kickback}}   & \qw & \qw &&&&&&&&&&&&   \\ 
 \lstick{$\color{gray}(\mathtt{x^\prime})\color{black}$} & \qw & \qw && \qw & \qw &&&&&&&&&&&&   \\ 
 \lstick{$\mathtt{y}$} & \qw & \gate[4,style={fill=blue!100,fill  opacity=0.2},disable  auto  height][1.4cm]{\mathscr{U}_{\mathrm{SVD}}^{\mathrm{SB}}} && \gate[4,style={fill=blue!100,fill  opacity=0.2},disable  auto  height][1.4cm]{\mathscr{U}_{\mathrm{SVD}}^{\mathrm{SB}  \dagger}} & \qw &&&& \lstick{$\mathtt{y}$} & \qw & \targ{}  \gategroup[4,steps=4,style={solid,fill=blue!20,  inner  xsep=2pt},label  style={label  position=above,anchor=south},background]{} & \qw & \qw & \qw & \qw & \qw &   \\ 
 \lstick{$\mathtt{x}$} & \qw &&&& \qw &&&& \lstick{$\mathtt{x}$} & \qw & \ctrl{-1} & \ctrl{1} & \ctrl{2} & \gate{\mathsf{H}} & \qw & \qw &   \\ 
 \lstick{$\mathtt{p}$} & \qwbundle{\eta} &&&& \qw &&&& \lstick{$\mathtt{p}$} & \qwbundle{\eta} & \qw & \gate{\lambda^+} & \qw & \qw & \qw & \qw &   \\ 
 \lstick{$\mathtt{q}$} & \qwbundle{\eta} &&&& \qw &&&& \lstick{$\mathtt{q}$} & \qwbundle{\eta} & \qw & \qw & \gate{\lambda^+} & \qw & \qw & \qw &   \\ 
 \lstick{$\mathtt{p^\prime}$} & \qwbundle{\eta} & \qw && \qw & \qw &&&&&&&&&&&&   \\ 
&&& \text{(a)} &&&&&&&&&& \text{(b)} &&&&
 \end{tikzcd}  
    }
    \adjustbox{width=1.0\linewidth,center}{%
 \begin{tikzcd}[row  sep={7mm,between  origins},transparent]  
& \lstick{$\color{gray}(\mathtt{y^\prime})\color{black}$} & \qw & \qw \gategroup[9,steps=9,style={solid,fill=green!20,  inner  xsep=2pt},  background,label  style={label  position=below,anchor=north,yshift=-3mm}]{(c)} & \qw & \ctrl{3} & \qw & \qw & \qw & \ctrl{3} & \qw & \qw & \qw & \qw &&&   \\ 
& \lstick{$\color{gray}(\mathtt{x^\prime})\color{black}$} & \qw & \qw & \ctrl{2} & \qw & \qw & \qw & \qw & \qw & \ctrl{2} & \qw & \qw & \qw &&&   \\ 
& \lstick{$\mathtt{y}$} & \qw & \qw & \qw & \qw & \ctrl{1} & \qw & \ctrl{1} & \qw & \qw & \qw & \qw & \qw &&&   \\ 
& \lstick{$\mathtt{x}$} & \qw & \qw & \targ{} & \targ{} & \gate{\text{\large $e^{\mp  i  tx  (1-2^{-m})  Z_\text{\footnotesize $\reg{x}$}  }$}} & \qw & \gate[2]{\prod_{k=0}^{m-1} \text{\large $e^{\pm  i  tx   2^{k-m}  Z_\text{\footnotesize $\reg{x}$} Z_\text{\scriptsize $k$} }$}} & \targ{} & \targ{} & \qw & \qw & \qw &&&   \\ 
& \lstick{$\ket{0}$} & \qwbundle{m} & \gate[5]{\mathscr{U}_{\widetilde{D}}}  \gateoutput{$>$} & \qw & \qw & \qw & \qw && \qw & \qw & \gate[5]{\mathscr{U}_{\widetilde{D}}^\dagger}  \gateinput{$<$} & \qw & \qw && \lstick{$\ket{0}  \  \  $} &   \\ 
& \lstick{$\mathtt{p}$} & \qwbundle{\eta} && \qw & \qw & \qw & \qw & \qw & \qw & \qw && \qw & \qw &&&   \\ 
& \lstick{$\mathtt{q}$} & \qwbundle{\eta} && \qw & \qw & \qw & \qw & \qw & \qw & \qw && \qw & \qw &&&   \\ 
& \lstick{$\mathtt{p^\prime}$} & \qwbundle{\eta} && \qw & \qw & \qw & \qw & \qw & \qw & \qw && \qw & \qw &&&   \\ 
 \lstick[wires=1]{junk  $\left\{  \right.$  } & \lstick{$\ket{0}$} & \qwbundle{\ell} && \qw & \qw & \qw & \qw & \qw & \qw & \qw && \qw & \qw && \lstick{$\ket{0}  \  \  $} &
 \end{tikzcd}  
    }
    \caption{\label{fig:HI_high-level_SB}
    The circuit that implements the Schwinger-boson hopping propagator corresponding to each subterm in Eq.~\eqref{eqs:HI-subterms-JW_SB}.
    (a) A high-level representation of the (far-term) circuit that realizes $e^{-itH^{{\rm SB}(j)}_I(r)}$.
    (b) The diagonalization circuit $\SVDSB$ for a Schwinger-boson hopping subterm.
    This circuit applies to both the near-term and far-term algorithms.
    (c) The phase-kickback circuit used in the far-term scenario.
    All circuits in (a-c) call for ancilla qubits that are not explicitly drawn but are discussed in the text and counted in the cost tables.
    }
\end{figure*}
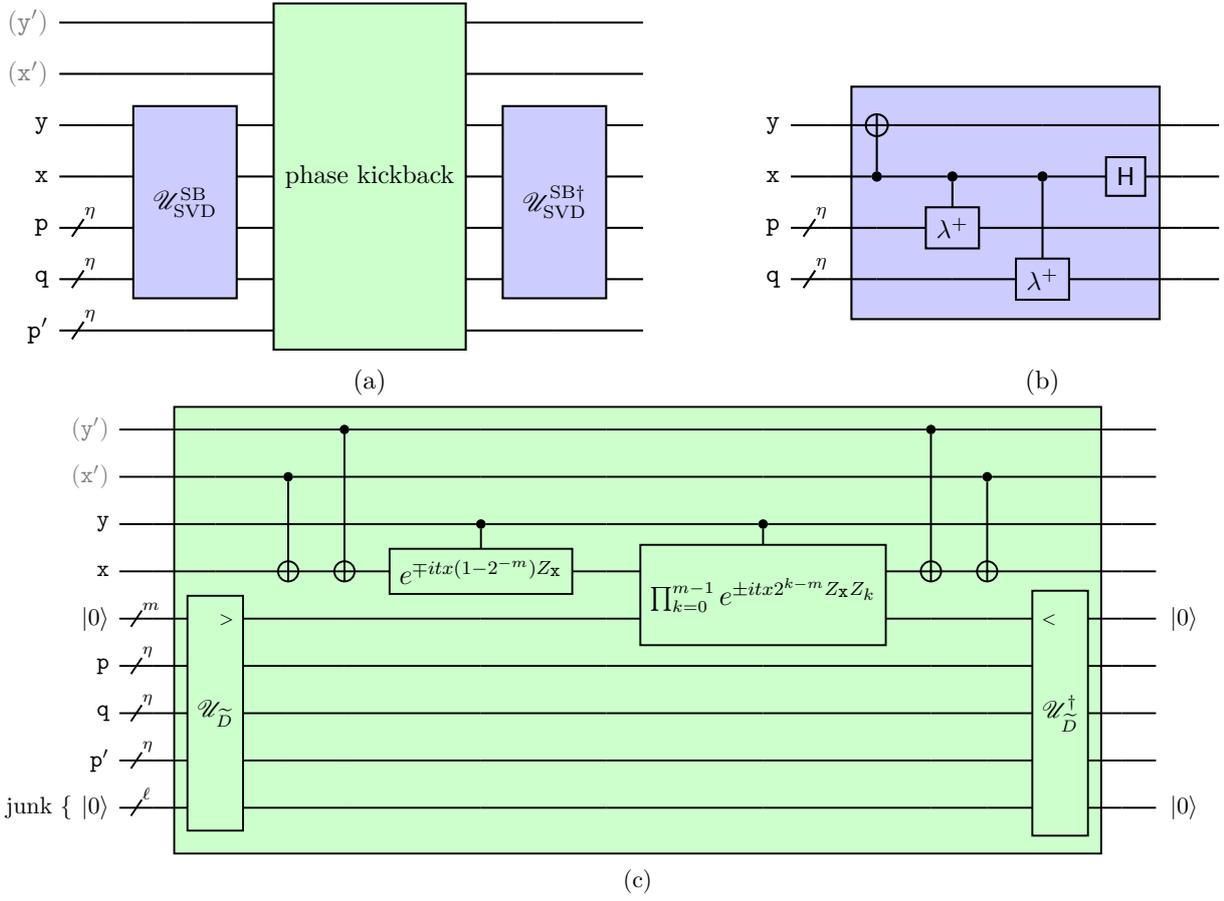

\vspace{0.2 cm}
\noindent
\emph{Diagonalization.}---The SVD transformation is as shown in Fig. \ref{fig:HI_high-level_SB}(b):
one CNOT plus two $C(\lambda^+)$ gates and a Hadamard gate.
The $C(\lambda^+)$ gate, that is a controlled $\eta$-qubit incrementer, can proceed as an uncontrolled $(\eta+1)$-qubit incrementer (followed by a bit flip on the ``control'' qubit).
According to Lemma \ref{lem:inc-far}, for an input of size $\eta+1$, the cyclic incrementer costs $4\eta-4$ T gates and $\eta-2$ ancilla qubits, where the ancilla qubits can be reused for all subsequent implementations.
Doubling this, the cost of $\SVDSB$ (or ${\SVDSB}^\dagger$) is $8\eta-8$ T gates and $\eta-2$ workspace qubits, as reported in Table \ref{tab:offdiag-costs_far_SB}.
Covering all eight subterms will further multiply the T-gate count by 16, for a total of $128\eta-128$ T gates per hopping term associated with diagonalizing transformations.
The dominant cost of simulating a hopping term, however, comes from implementing the diagonalized subterms, which will be discussed next.

\vspace{0.2 cm}
\noindent
\emph{Phase kickback.}---The implementation of the diagonalized subterm in Eq.~\eqref{eq:diagonalized-subterm_SB} proceeds via a phase-kickback algorithm whose circuit implementation is depicted in Fig.~\ref{fig:HI_high-level_SB}(c).
First, the value of the non-trivial diagonal bosonic operator $\DSB$ is computed, up to a known accuracy, via a unitary $\UfSB$, and the value is stored in an ancillary register of $m$ bits, where  $\widetilde{\mathcal{D}}^{\mathrm{SB}}$ denotes a fixed-precision approximation to $\DSB$ and $m$ is the desired bit-precision. We will construct $\UfSB$ via Newton's method shortly and will discuss its cost as a function of the approximation error.
The fixed-precision arithmetic is another source of error, independent of the Trotter and $R^Z$ synthesis errors, to be taken into account in the full error budget.
Since the ancillary register holds the (approximate) value of $\DSB$ in binary form, the diagonal operator $e^{-itx \ket{1}\bra{1}_{\reg{y}} \color{gray}( Z_\reg{x'} ) ( Z_\reg{y'} )\color{black} Z_\reg{x} \DSB( p, q, p') }$ can be straightforwardly implemented by $R^Z$ gates on the corresponding qubits, with the angles shown in the circuit.
The ancillary register is set back to an all-$\ket{0}$ state by the inverse operations, and is used in subsequent implementations.

There are multiple ways of formulating the calculation of $\DSB( p, q, p')= \sqrt{ \frac{p  q}{( p +  p')( p +  p' +1)} }$.
In view of the identity
$$ \sqrt{\frac{\mathrm{numerator}}{\mathrm{denominator}}} = \mathrm{numerator} \times \frac{1}{\sqrt{\mathrm{numerator}\times\mathrm{denominator}}} ,$$
one way of formulating the calculation goes as follows:
\begin{enumerate}
    \item Multiply $p$ and $q$ to get the numerator.
    \item Compute $p+p'$ with an adder.
    \item Multiply $p\,q$ with $p+p'$ to obtain $pq(p+p')$.
    \item Compute $p+p'+1$ with an incrementer.
    \item Multiply $pq(p+p')$ with $(p+p'+1)$ to obtain $pq(p+p')(p+p'+1)$.
    For later convenience, this product is refered to by
\begin{align}
    \gSB(p,q,p') &\equiv pq(p+p')(p+p'+1).
\end{align}
    \item Compute the inverse square root of $\gSB(p,q,p')$.
    \item Multiply $p\,q$ with $\gSB(p,q,p')^{-1/2}$.
\end{enumerate}
In this protocol, the most complicated and resource-intensive step is (6):
the inverse-square-root evaluation.%
\footnote{
In Ref.~\cite{Kan:2021xfc}, an alternative protocol was proposed that, when translated to the Schwinger-boson formulation, consists of the following steps:
($1'$) Multiply $p$ and $q$ to get the numerator.
($2'$) Compute $p+p'$ and $p+p'+1$ using adders.
($3'$) Multiply $p+p'$ and $p+p'+1$ to get the denominator.
($4'$) Compute the reciprocal of $( p +  p')( p +  p' +1)$ from step (3').
($5'$) Multiply the results from ($1'$) and ($4'$) to get the argument of the square root.
($6'$) Take the square root of the fractional quantity.
In this protocol, the most complicated and resource-intensive steps are ($4'$) and ($5'$):
the integer reciprocal and the square root. Our route has the benefit of evaluating the inverse square root in a single function-estimation step using Newton's method.
}
The unitary that implements the steps above is called $\UfSB$, and the evaluation only produces a fixed-precision approximation $\widetilde{\mathcal{D}}^{\mathrm{SB}}(p,q,p')\approx\DSB(p,q,p')$.
We will return shortly to address the details of using fixed-precision arithmetic and the associated errors.
A circuit diagram for $\UfSB$ is provided in Fig.~\ref{fig:Uf_SB}.
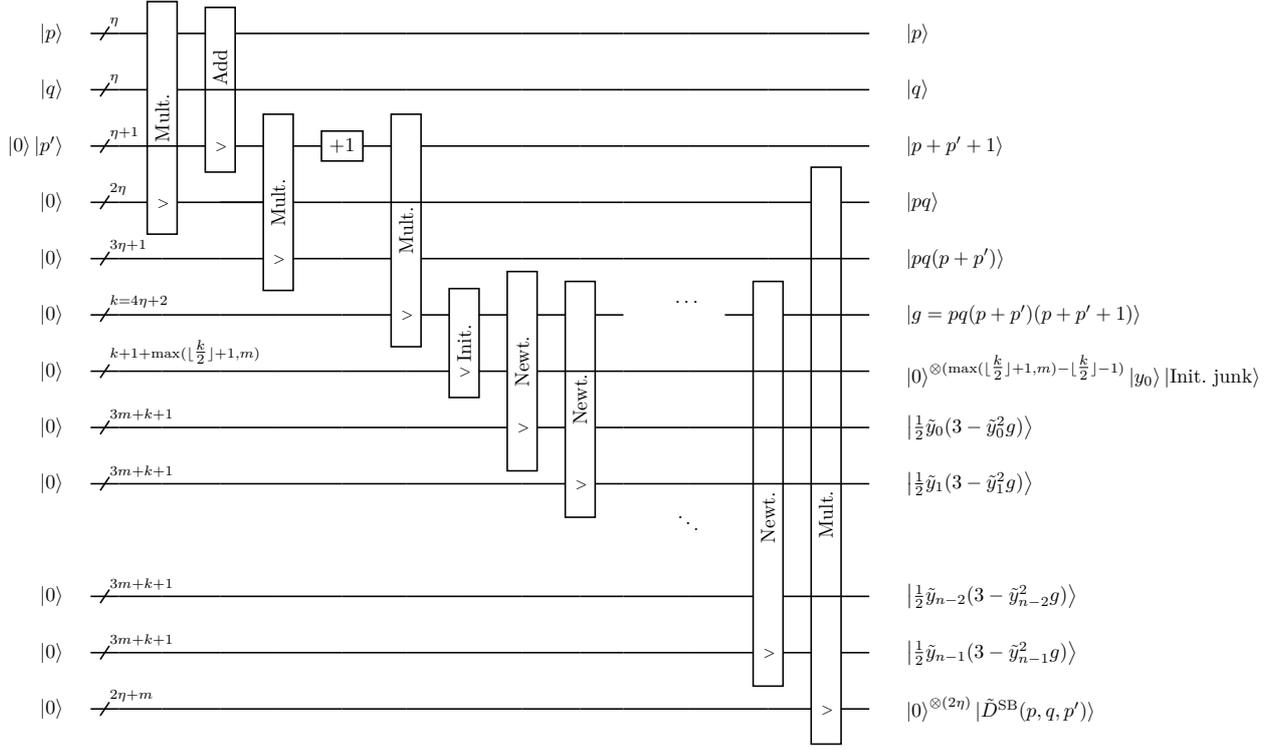
\begin{figure*}
\centering
    \begin{adjustbox}{width=0.99\linewidth}\begin{tikzcd}[row sep={1cm,between origins},transparent]
\lstick{$\ket{p}$}\quad & \qwbundle{\eta} &   \gate[4]{\rotatebox{90}{\text{Mult.}}}    &   \gate[3,label style={yshift=0.4cm}]{\rotatebox{90}{\text{Add}}} & \qw & \qw &    \qw &      \qw &  \qw & \qw& \qw &  \qw    & \qw & \qw  &\qw & \rstick{$\ket{p}$}\\
\lstick{$\ket{q}$}\quad &\qwbundle{\eta} &          & \linethrough & \qw & \qw &   \qw & \qw & \qw &\qw&   \qw  & \qw & \qw & \qw  &\qw & \rstick{$\ket{q}$}\\
\lstick{$\ket{0} \ket{p'}$}\quad & \qwbundle{\eta+1} &   \linethrough      & \gateoutput{$>$} & \gate[3]{\rotatebox{90}{\text{Mult.}}} & \gate{+1} & \gate[4]{\rotatebox{90}{\text{Mult.}}}   & \qw  & \qw &\qw& \qw & \qw   & \qw & \qw  &\qw & \rstick{$\ket{p+p'+1}$}\\
\lstick{$\ket{0}$} \quad& \qwbundle{2\eta} &      \gateoutput{$>$}    &    \qw & \qw & \qw &   \linethrough   &     \qw &  \qw &\qw& \qw &   \qw   & \qw & \gate[10,nwires={7},label style={yshift=-1.1cm}]{\rotatebox{90}{\text{
Mult.}}}  &\qw & \rstick{$\ket{pq}$}\\
\lstick{$\ket{0}$}\quad &  \qwbundle{3\eta+1} &    \qw     &     \qw & \gateoutput{$>$} & \qw &     &     \qw &  \qw &\qw& \qw &   \qw   & \qw &  \linethrough &\qw & \rstick{$\ket{pq(p+p')}$}\\
\lstick{$\ket{0}$}\quad & \qwbundle{k=4\eta+2} &     \qw     &  \qw & \qw & \qw &   \gateoutput{$>$} & \gate[2]{\rotatebox{90}{\text{Init.}}} & \gate[3]{\rotatebox{90}{\quad \text{
Newt.}}} & \gate[4]{\rotatebox{90}{\text{
Newt.}}} & \qw &  \midstick{$\cdots$}    & \gate[7,nwires={5},label style={yshift=-0.6cm}]{\rotatebox{90}{\text{
Newt.}}} & \linethrough  &\qw & \rstick{$\ket{g=pq(p+p')(p+p'+1)}$}\\
\lstick{$\ket{0}$} \quad&  \qwbundle{k+1+\max(\lfloor  \tfrac{k}{2}  \rfloor+1,m)} &    \qw     &  \qw & \qw & \qw & \qw &     \gateoutput{$>$}   &   & \linethrough & \qw& \qw  & \linethrough & \linethrough  &\qw & \rstick{$\ket{0}^{\otimes(\max(\lfloor  \tfrac{k}{2}  \rfloor+1,m)-\lfloor  \tfrac{k}{2}  \rfloor-1)}\ket{y_0
}\ket{\rm Init.~junk}$}\\
\lstick{$\ket{0}$}\quad & \qwbundle{3m+k+1} &     \qw     &  \qw & \qw & \qw & \qw &    \qw & \gateoutput{$>$}  & & \qw &  \qw  & \linethrough &\linethrough  &\qw & \rstick{$\Ket{\tfrac{1}{2} \tilde{y}_0 (3-\tilde{y}_0^2 g)}$}\\
\lstick{$\ket{0}$}\quad & \qwbundle{3m+k+1} &      \qw     &  \qw & \qw & \qw & \qw &    \qw &  \qw &  \gateoutput{$>$} &\qw&  \qw  & \linethrough &\linethrough  &\qw & \rstick{$\Ket{\tfrac{1}{2} \tilde{y}_1 (3-\tilde{y}_1^2 g)}$}\\
                    &           &   &  & &   &    &       &       &                   & & \midstick{$\ddots$}  &   &   &  & \\
\lstick{$\ket{0}$} \quad&  \qwbundle{3m+k+1} &    \qw    &   \qw & \qw & \qw &  \qw &    \qw &   \qw & \qw& \qw & \qw  &          & \linethrough  &\qw & \rstick{$\Ket{\tfrac{1}{2} \tilde{y}_{n-2} (3-\tilde{y}_{n-2}^2 g)}$}\\
\lstick{$\ket{0}$}\quad & \qwbundle{3m+k+1} &     \qw     &  \qw & \qw & \qw & \qw &    \qw &  \qw &\qw& \qw & \qw  &    \gateoutput{$>$}   &   &  \qw & \rstick{$\Ket{\tfrac{1}{2} \tilde{y}_{n-1} (3-\tilde{y}_{n-1}^2 g)}$}\\
\lstick{$\ket{0}$}\quad & \qwbundle{2\eta + m} &     \qw     &  \qw & \qw & \qw & \qw &    \qw &  \qw &\qw& \qw & \qw  &   \qw     &  \gateoutput{$>$} &  \qw & \rstick{$\ket{0}^{\otimes (2\eta)}\ket{\tilde{D}^{\rm SB}(p,q,p')}$}\\
\end{tikzcd}
    \end{adjustbox}
    \caption{A quantum circuit for computing an $m$-bit approximation to diagonal function $\DSB( p, q, p')$ defined in Eq.~(\ref{eq:DSB}) with $\eta$-qubit bosonic registers $p$, $q$, and $p'$. Multiplication (Mult.), addition (Add), incrementers ($+1$), and initial-guess (Init.) subcircuits are described in Appendix~\ref{app:arithmetic}. The decomposition of a Newton-iteration step (Newt.) is depicted in Fig.~\ref{fig:newton-iteration}. For each subcircuit, the line passing through the box indicates that the corresponding qubit register does not participate in the operations. The output of each unitary is marked by $>$ (or $<$ for the inverses). For example, for multiplication, $>$ indicates the register into which the product of the other two input registers is output. Workspace qubits are left implicit. $k \equiv 4\eta +2$, $n$ indexes the final Newton-iteration step, $y_0=2^{-s}$ with $s=\left\lfloor \log_4(pq(p+p')(p+p'+1)) + \frac{1}{2} \right\rfloor$, and the state $\ket{\rm Init.~junk}$ is stored in $k+1$ qubits, see Lemma~\ref{lem:guess}.
    }
    \label{fig:Uf_SB}
\end{figure*}

To evaluate $g^{-1/2}$ given some non-negative integer $g$, one may proceed with the following choice
\begin{equation}
\label{eq:newtinp}
    F(y) \equiv \frac{1}{y^2} - g,
\end{equation}
for the function to be used in the implementation of Newton's method for inverse square root. Here, in the case at hand, $g = g^{SB}(p,q,p')$.
The positive zero, $y_z$, of $F(y)$ is $g^{-1/2}$, which when is input to a final multiplier for the product $p \, q \, y_z$ gives $\DSB(p,q,p')$.
Our method of calculation is motivated by the goal of computing $\DSB(p,q,p')$ to high accuracy using the fewest number of quantum multiplication circuits.
This choice also results in iterations that require only addition and multiplication (and division by two, which has no T-gate and ancilla-qubit cost in binary).

Given some approximation $\tilde{y}$ of $g^{-1/2}$, the next approximation generated by Newton's method is given by
\begin{equation}\label{eq:iterate}
    \mathcal N_g(\tilde{y}) \equiv \frac{\tilde{y} (3 - \tilde{y}^2 g)}{2}.
\end{equation}
With an approximation $\tilde{y}$ to $g^{-1/2}$ already stored in some qubit register, a single Newton iteration to generate the next approximation $\mathcal N_g(\tilde{y})$ is performed via the circuit shown in Fig.~\ref{fig:newton-iteration}.
\begin{figure*}
\centering
    \includegraphics[width=0.85\linewidth]{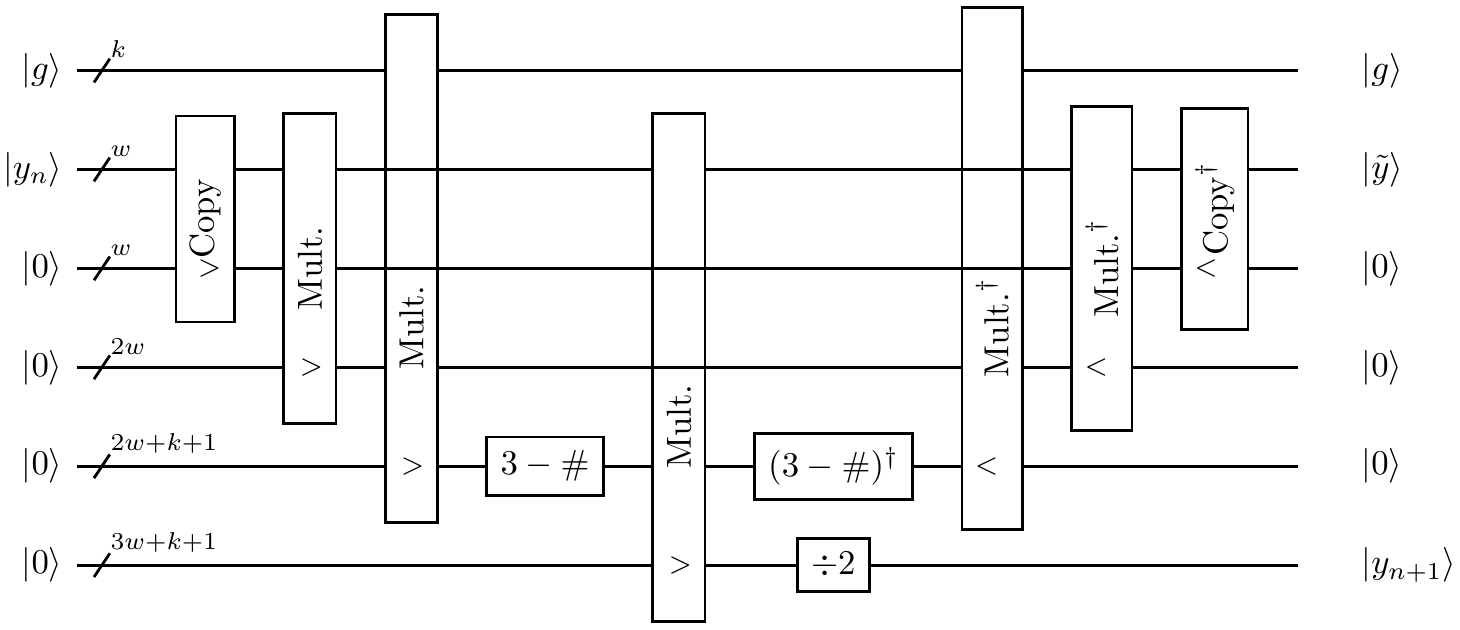}
    \caption{\label{fig:newton-iteration}
    Newton-iteration circuit for computing the inverse square root of $g$. Workspace qubits are left implicit and the rest of the notation is as in Fig.~\ref{fig:Uf_SB}.
    }
\end{figure*}
It uses known quantum-circuit routines such as copy, addition, and multiplication that are described in more detail in Appendix~\ref{app:arithmetic}. The total T-gate count of each iteration, as analyzed in Lemma~\ref{lem:newtit}, is $32w^2+40kw+4k+8w{-8\max(k,2w)}-12$, and the circuit calls for $9w+3k+3$ ancilla qubits, where $k$ is the number of bits specifying $g$ ($k=4\eta+2$ for the Schwinger-boson case) and $w$ is the number of bits specifying $\tilde{y}$.

Conventionally, Newton's method involves generating some initial value $y_0$, followed by successive iterates $y_n = \mathcal N_g (y_{n-1}) $, up to some prescribed number of rounds.
Defining the $n$th-order relative error $\delta_n = (y_n - y_z)/y_z,$ it is straightforward to see that
\begin{equation}
    \delta_{n+1} = -\frac{3}{2}\delta_n^2 - \frac{1}{2} \delta_n^3.
    \label{eq:convergence}
 \end{equation}
This illustrates that as long as $|\delta_0|< 1$, Newton's method converges fast and the relative error is bounded by $|\delta_n| \leq 2^n |\delta_0|^{2^n}$. Since the convergence rate depends on the initial guess $y_0$, in order to guarantee quadratic convergence for any possible input $g$, instead of some arbitrary initial value, $y_0$ needs to be computed from $g$. A suitable choice is the inverse square root of the $s^{\rm th}$ power of four such that $s$ is the closest integer to the logarithm of the input $g$ in base four.
A circuit implementation of the initial-value evaluation is presented in Appendix~\ref{app:arithmetic}.
According to Lemma \ref{lem:guess}, the initial value can be computed into a $(1+\lfloor k/2 \rfloor)$-bit register at a cost of $4k$ T gates and $k+1$ `junk' qubits, where again $k$ is the number of bits specifying $g$.\footnote{In Ref.~\cite{Haner:2018yea}, it is noted that the number of Newton iterations may be reduced by one for a small overhead in computing a better initial guess. We did not pursue this improvement due to its asymptotically identical cost.
}

A practical issue with Newton's method as described above is that the cost of each Newton iteration grows extremely fast.
It is more realistic to specify some moderate number of bits $m$ for the precision that suffices for Newton's method to converge to the exact answer within a bounded error.
Thus, we will instead seek approximations $\bar{y}_0$, $\bar{y}_1 $, $\cdots$, $\bar{y}_n$ where $\bar{y}_{i}$ (for $i \geq 1$) is $\mathcal N_g(\bar{y}_{i-1})$ truncated to $m-1$ bits after the binary point and similarly for the initial value $\bar{y}_0$ (note that as argued in Lemma~\ref{lem:bound}, $y_i \leq 1/\sqrt{g} \leq 1$ for $i \geq 0$.).
For concreteness, we will always apply the circuit of Fig. \ref{fig:newton-iteration} with the input being implicitly truncated/padded to $m$ bits as necessary, making the input size $w$ effectively equal to $m$.

To summarize the phase evaluation, there exists a quantum circuit that implements the unitary $\UfSB$ defined as
\begin{equation}
    \ket{p}\ket{q}\ket{p'}\ket{0}^{\otimes (1+m + \ell + \beta)} \xrightarrow[]{\UfSB}  \ket{p}\ket{q}\ket{p + p' + 1}\ket{\widetilde{D}^{\mathrm{SB}}(p, q , p')}\ket{\rm junk}\ket{0}^{\otimes \beta}.  
\end{equation}
The state $\ket{\widetilde{D}^{\mathrm{SB}}(p, q , p')}$ is an $m$-bit fixed-precision encoding of $\DSB(p, q , p')$ within absolute error $|\delta_{n,m}^{\rm Abs.}|\leq 2^n (\sqrt{2}-1)^{2^n} + 2^{2-m} ( (3/2)^n - 1)$ for $n$ iterations of Newton's method.
The state $\ket{\rm junk}$, which holds intermediate values, occupies an $\ell$-bit register. Furthermore, the circuit can be implemented using at most 
\begin{equation}
160 \eta m n + 32m^2 n + 48\eta^2 + 16\eta m + 16\eta n + 88mn - 8n\max(4\eta+2,2m) + 48\eta - 4\max(2\eta,m) - 4n + 8
\end{equation}
T gates. The circuit consumes at most
\begin{equation}
 \ell =  4\eta n + 3mn + 15\eta + \max(2\eta+2,m) + 3n + 6
\end{equation}
in output junk qubits, and borrows
\begin{equation}
\beta=12\eta+9m+9
\end{equation}
reusable workspace qubits.
Lemma~\ref{lem:kineticphasekickSB} of Appendix~\ref{app:arithmetic} provides a derivation of these results.\\

Finally, with $\pm x \widetilde{\mathcal{D}}^{\rm SB}$ evaluated, the diagonal operator $\pm x \ket{1}\bra{1}_{\reg{y}}\color{gray} ( Z_\reg{x'} ) ( Z_\reg{y'} ) \color{black} Z_\reg{x} \DSB( p, q, p')$ from Eq.~(\ref{eq:diagonalized-subterm_SB}) is effectively replaced by $\ket{1}\bra{1}_{\reg{y}} \color{gray} ( Z_\reg{x'} ) ( Z_\reg{y'} ) \color{black} Z_\reg{x} N_\reg{aux}$, where $N_\reg{aux}= 2^{1-m} \sum_{k=0}^{m-1} 2^{k-1} (1-Z_\reg{aux,k})$ is the fixed-point number operator on the $m$-bit auxiliary register holding $\widetilde{D}^{\mathrm{SB}}$ (see Eq.~\eqref{eq:Nponketp}).
With $N_\reg{aux}$ expanded out into $m$ single-qubit $Z$ operators plus a constant, the rotations can be split into $m+1$ distinct commuting rotations.
As remarked above, the far-term $R^Z$ gates incur a T-gate cost associated with the tunable precision of the desired rotation angles, and one can use the RUS construction of Ref.~\cite{bocharov2015efficient}. These rotations are recorded in Table~\ref{tab:offdiag-costs_far_SB}, where an additional factor of two is included to account for the control factor $\ket{1}\bra{1}_{\reg{y}}= \frac{1}{2}(1-Z_\reg{y})$.

\vspace{0.2 cm}
\noindent
\emph{Final tally.}---Putting everything together, the circuit in Fig.~\ref{fig:HI_high-level_SB}(a) costs
\begin{align}
\label{eq:ancilla-hop-SB}
4 \eta  n+27 \eta + \max (2 \eta +2,m)+3 m n+10 m+3 n+15
\end{align}
ancilla qubits for work and scratch spaces (where we have added $m$ to account for the primary phase-kickback register), $2(m+1)$ $R^Z$ gates, and
\begin{align}
\label{eq:T-gate-hop-SB}
    &320 \eta m n + 64 m^2 n + 96 \eta^2 + 32 \eta m + 32 \eta n + 176 m n - 16 n \max (4\eta+2,2m) \nonumber\\
    & \hspace{7.6 cm} + 112 \eta - 8 \max (2\eta,m) - 8 n 
\end{align}
T gates.
Adding the costs of doing eight structurally equivalent subterms, the full upper bound on implementing each Trotterized hopping propagator in the Schwinger-boson formulation comes to eight times the T-gate cost given in Eq.~\eqref{eq:T-gate-hop-SB} and no more ancilla cost than what is stated in Eq.~\eqref{eq:ancilla-hop-SB}, as stated originally and summarized in Table \ref{tab:offdiag-costs_far_SB}.
\end{proof}
\begin{table}[t!]
\centering
\begin{tabular}{l >{\arraybackslash}p{5.5cm} l >{\arraybackslash}p{3.5cm}}
Routine & T gates & Workspace & Scratch space \\
\hline
\hline
$\SVDSB$ or ${\SVDSB}^\dagger$ & $8\eta-8$ & $\eta-2$ & 0 \\
$\UfSB$ or ${\UfSB}^\dagger$ & $160 \eta m n + 32m^2 n + 48\eta^2 + 16\eta m + 16\eta n + 88mn - 8n\max(4\eta+2,2m) + 48\eta - 4\max(2\eta,m) - 4n + 8$ & $12\eta+9m+9$ & $4\eta n + 3mn + 15\eta + \max(2\eta+2,m) + m + 3n + 6$ \\
Diagonal rotations & $2(m+1) \mathcal{C}_z(\epsilon)$ & 0 & 0 \\ 
\hline
\\ 
Full $e^{-itH^{{\rm SB}(j)}_I(r)}$ circuit & $ 
320 \eta m n + 64 m^2 n + 96 \eta^2 + 32 \eta m + 32 \eta n + 176 m n - 16 n \max (4\eta+2,2m) + 112 \eta - 8 \max (2\eta,m) - 8 n + 2(m+1) \mathcal{C}_z(\epsilon) 
$ & $12\eta+9m+9$ & $4\eta n + 3mn + 15\eta + \max(2\eta+2,m) + m + 3n + 6$ \\ 
\end{tabular}
\caption{\label{tab:offdiag-costs_far_SB}
Summary of the costs associated with the far-term simulation of the off-diagonal operators in the Schwinger-boson Hamiltonian, as explained in the text. For itemized cost contributing to the cost of $\UfSB$, see Table~\ref{tab:Uf_SB} of Appendix~\ref{app:arithmetic}.}
\end{table}

\subsubsection{Loop-string-hadron propagators
\label{sec:circuits_LSH}
}

In the LSH formulation, the fermionic operators are $\chi_{i(o)}$ and $\chi^\dagger_{i(o)}$, and these can be directly mapped to Pauli spin operators via a Jordan-Wigner transformation. With an additional staggering analogous to what was introduced in the Schwinger-boson formulation, we choose the mapping:
\begin{subequations}
\label{eq:JW-LSH}
\begin{align}
    \label{eq:JW-LSH-1}
    &\chi_i(r) \rightarrow \left( \prod_{k=0}^{2r-1} Z_k \right) \sigma^+_{2r},~~\chi_o(r) \rightarrow \left( \prod_{k=0}^{2r} Z_k   \right) \sigma^+_{2r+1},\\
    &{\chi_i^\dagger(r)} \rightarrow \left( \prod_{k=0}^{2r-1} Z_k   \right) \sigma^-_{2r},~~{\chi_o^\dagger(r)} \rightarrow \left( \prod_{k=0}^{2r} Z_k   \right) \sigma^-_{2r+1}.
    \label{eq:JW-LSH-2}
\end{align}
\end{subequations}

The bosonic operators in the LSH formulation consist of $n_\ell$ (occupation-number) and $ \lshladder$ (ladder) operators, respectively. Assuming that there are $\eta$ qubits available to encode the local occupation of bosons in binary, the $n_\ell$ quantum number at each site must be truncated to a finite value, $\Lambda \equiv 2^\eta-1$. The only modification to the action of bosonic operators is that the raising operator $ \lshladder^\dagger$ must annihilate the state if the occupation of the state at the corresponding site is equal to $\Lambda$, that is: $ \lshladder^\dagger \ket{n_\ell,n_i,n_o}=(1-\delta_{n_\ell,\Lambda})\ket{n_\ell+1,n_i,n_o}$. The $n_\ell$ quantum number can be expressed in a binary representation and be mapped to qubit registers as in Eq.~(\ref{eq:bosonic-register}). Figure~\ref{fig:latticeDOFs_SB}(b) depicts the Hilbert spaces associated with the DOFs of the LSH formulation.

An efficient circuit decomposition for exponentiating each site- or link-local term in the Hamiltonian will be presented in the following, and product formulas for approximating the full time-evolution operator of the system within a fixed error will be discussed later in Sec.~\ref{sec:bounds}. 

\begin{center}
\textit{---Near term---}
\end{center}
\begin{center}
\emph{Implementing mass propagators}
\end{center}

The mass Hamiltonian is $\sum_{r=0}^{L-1} H_M^{\rm LSH}(r)$, and each site-local mass term can be decomposed into two commuting subterms as
\begin{subequations}
\label{eqs:HM-subterms_LSH}
\begin{align}
    &H_M(r) = \sum_{j=1}^2 H_{M}^{{\rm LSH}(j)}(r), \\
    &H_{M}^{{\rm LSH}(1)}(r)
    =\frac{\mu}{2}(-1)^{r+1} Z_{2r}, \\
    &H_{M}^{{\rm LSH}(2)}(r)
    =\frac{\mu}{2}(-1)^{r+1} Z_{2r+1},
\end{align}
\end{subequations}
after the Jordan-Wigner transformation in Eq.~\eqref{eq:JW-LSH}, and upon neglecting constant terms that only introduce time-dependent but otherwise constant phases in the dynamics.

\begin{lemma}
Using a $2$-qubit register with no ancilla qubits, $e^{-itH^{{\rm LSH}}_M(r)}$ can be implemented without approximation, up to a phase, with no CNOT gates required.
\end{lemma}
\begin{proof}
The expressions in Eqs.~\eqref{eqs:HM-subterms_LSH} are identical to those of the Schwinger-boson formulation in Eqs.~\eqref{eqs:HM-subterms_SB}, and so the circuit decomposition of this propagator is trivial as with the Schwinger-boson case.
Two qubits are required, each indexed by $2r$ and $2r+1$, corresponding to `$n_i(r)$' and `$n_o(r)$' registers, respectively. 
The circuit implements $R^Z_{2r}\left((-1)^{r+1}\mu\,t \right)R^Z_{2r+1}\left((-1)^{r+1}\mu\,t \right)$.
\end{proof}
\begin{center}
\emph{Implementing electric propagators}
\end{center}

The electric Hamiltonian is $\sum_{r=0}^{L-2} H_E^{\rm LSH}(r)$, with link-local Casimir operators decomposed into three commuting subterms as
\begin{subequations}
\label{eqs:HE-subterms_LSH}
\begin{align}
    &H_E^{\rm LSH}(r) = \sum_{j=1}^3 H_{E}^{{\rm LSH}(j)}(r) , \\
    &H_{E}^{{\rm LSH}(1)}(r) = \frac{1}{2} n_\ell(r) , \\
    &H_{E}^{{\rm LSH}(2)}(r) = \frac{1}{4} {n_\ell(r)}^2 , \\
    &H_{E}^{{\rm LSH}(3)}(r) = \frac{1}{2} \big(n_\ell(r)+\frac{3}{2}\big) n_o(r) \big(1-n_i(r)\big) .
\end{align}
\end{subequations}
As remarked for the Schwinger-boson electric propagators, there is a certain amount of arbitrariness involved with dividing $H_E(r)$ into commuting subterms, but any reasonable choice will do since slight difference in the cost associated with those choices will not matter given that simulating the interaction Hamiltonian will dominate the final cost.

The division into subterms in Eqs.~\eqref{eqs:HE-subterms_LSH} for the near-term simulation involves three kinds of subterms: one proportional to $n_\ell$, one proportional to $n_\ell^2$, and the last being the coupling of $n_\ell$ to $n_i$ and $n_o$.
The CNOT-gate and ancilla-qubit requirements of implementing $e^{-itH^{\rm LSH}_E(r)}$ in the near-term scenario can be derived following similar analysis to the Schwinger-boson case.

\begin{lemma}
Using an $(\eta+2)$-qubit register with no ancilla qubits, $e^{-itH^{{\rm LSH}}_E(r)}$ can be implemented without approximation, up to a phase, using $\tfrac{1}{2}\eta^2+\tfrac{17}{2}\eta+1$ CNOT gates (and a number of single-qubit rotations).
\end{lemma}
\begin{proof}
This cost is obtained as follows.
i) The operator $n_\ell$ is a linear combination of the identity and every single-qubit Pauli-$Z$ operator across the `$n_\ell$' register.
Simulating this term amounts to a global phase and $\eta$ single-qubit $Z$ rotations, but no CNOT gates.
ii) The $n_\ell^2$ term can be simulated like an `$E^2$' term in the U(1) theory \cite{Shaw:2020udc}, similar to what was described in the Schwinger-boson section.
Following the same reasoning leads to a CNOT-gate cost of $(\eta+2)(\eta-1)/2$ for this term.
iii) As for the remaining $(n_\ell+3/2)n_o(r) (1-n_i(r))$, a naive Pauli decomposition involves $2\eta+1$ operators of the form $Z\otimes Z$, plus $\eta$ of the form $Z\otimes Z\otimes Z$, leading to a cost of $2(2\eta+1)+4\eta=8\eta+2$ CNOT gates.
The gate costs of i), ii), and iii) are summarized in Table~\ref{tab:diag-costs_near_LSH}, adding up to the stated total CNOT-gate count of $\tfrac{1}{2}\eta^2+\tfrac{17}{2}\eta+1$.
 \end{proof}

\begin{table}
\centering
\begin{tabular}
{>{\arraybackslash}p{7.5cm}  >{\centering\arraybackslash}p{4cm}  c }
LSH mass propagator subroutine & {CNOT} count \\
\hline
\hline
$R^Z$ gate & 0 \\
\hline
Full $e^{-itH^{{\rm LSH}}_M(r)}$ circuit & 0 \\
\\
\\
LSH electric propagator subroutine & {CNOT} count \\
\hline
\hline
Simulate $e^{-itn_\ell/2}$ & 0 \\
Simulate $e^{-itn_\ell^2/4}$ 
& $\frac{1}{2}(\eta+2)(\eta-1)$ \\
Simulate $e^{-it(n_{\ell}+3/2) n_o(1-n_i)/2}$ & $8\eta+2$ \\
\hline
Full $e^{-itH^{{\rm LSH}}_E(r)}$ circuit & $\frac{1}{2}\eta^2+\frac{17}{2}\eta + 1$ 
\end{tabular}
\caption{\label{tab:diag-costs_near_LSH}
Summary of the costs associated with the near-term simulation of the diagonal operators in the LSH Hamiltonian, as explained in the text.}
\end{table}

\begin{center}
\emph{Implementing hopping propagators}
\end{center}

The gauge-matter interaction Hamiltonian is $\sum_{r=0}^{L-2} H_I^{\rm LSH}(r)$, with link-local hopping terms $H_I^{\rm LSH}(r)$ that are off-diagonal in the electric basis. Like in the Schwinger-boson formulation, the hopping terms do not readily split into a set of mutually commuting subterms and we resort to an approximation $e^{-i t H_I^{\rm LSH}(r) } \approx \Pi_{j} \, e^{-itH^{{\rm LSH}(j)}_I(r)}$, in which each subterm is simulated without further splitting.
Once again, the splitting into $\nu$ hopping subterms directly impacts the Trotter error bound and gate counts.
In the LSH formulation, one is naturally led to a splitting into $\nu^{\rm LSH}=2$ subterms. After applying the Jordan-Wigner transformations in Eqs.~(\ref{eq:JW-LSH}), along with the truncation to $\eta$-bit registers for each bosonic oscillator mode, the two subterms in the qubit space are:
\begin{subequations}
\label{eqs:HI-subterms-JW_LSH}
\begin{align}
    H_{I}^{\mathrm{LSH}(1)}(r)
    &= x \, \sigma_{2r}^+ Z_{2r+1} \sigma_{2r+2}^- \lshladder^\dagger(r)^{1-n_{2r+1}} \lshladder^\dagger(r+1)^{n_{2r+3}} \DLSH (n_\ell(r),n_{2r+1},n_{2r+3}) + \mathrm{H.c.}, \\
    H_{I}^{\mathrm{LSH}(2)}(r)
    &= x \, \sigma_{2r+1}^- Z_{2r+2} \sigma_{2r+3}^+ \lshladder^\dagger(r)^{n_{2r}} \lshladder^\dagger(r+1)^{1-n_{2r+2}} \DLSH (n_\ell(r+1),n_{2r+2},n_{2r}) + \mathrm{H.c.}
\end{align}
\end{subequations}
The action of the LSH ladder operators, $\lshladder$ and $\lshladder^\dagger$, on the corresponding LSH basis was introduced in Sec.~\ref{sec:lshintro}, with the finite-cutoff modification introduced above, while the diagonal function $\DLSH$ is defined as
\begin{align}
    \label{eq:DLSH}
    \DLSH(p,n, n' ) \equiv \sqrt{\frac{p+1+n}{p+1+n'}},
\end{align}
for $n,n'=0,1$ being fermionic quantum numbers.

The expressions in Eqs.~(\ref{eqs:HI-subterms-JW_LSH}) have been put into the form of (shifting operators) $\times$ (diagonal operators) to match the example Hamiltonians studied in Sec.~\ref{sec:SVD} so that the SVD diagonalization methods are readily applied.
In the language of Sec.~\ref{sec:SVD}, each term $H_I^{\text{LSH}(j)}$ involves two spin raising and lowering operators (acting on the $n_{i(o)}$ quantum numbers at two adjacent sites), two conditional bosonic raising or lowering operators (acting on $n_\ell$ quantum numbers at the same adjacent sites), and a non-trivial diagonal function (of one bosonic and two fermionic occupation numbers).
The general structure of the LSH hopping subterms is encapsulated by
\begin{align}
    \label{eq:generalizedHISubtermLSHTruncated}
    &H_I^{\mathrm{LSH}(j)} =  
    x\, \sigma_{\reg{y'}}^- Z_\reg{x} \sigma_\reg{x'}^+ ( \lshladder^{\dagger}_\reg{q} )^{n_\reg{y}} ( \lshladder^{\dagger}_\reg{p} )^{1-n_\reg{x}} 
    \DLSH (p,n_\reg{x},n_\reg{y})+ \mathrm{H.c.} \nonumber \\
    & \hspace{0.2 cm} =
    x\, Z_\reg{x}\, \bigl( 1 - \ket{1}\bra{1}_\reg{y} \ket{\Lambda}\bra{\Lambda}_\reg{p} \bigr) \bigl( 1 - \ket{0}\bra{0}_\reg{x} \ket{\Lambda}\bra{\Lambda}_\reg{p} \bigr) 
     \DLSH (p,n_\reg{x},n_\reg{y}) \sigma_{\reg{y'}}^+ \sigma_\reg{x'}^- ( \lambda^{-}_\reg{q} )^{n_\reg{y}} ( \lambda^{-}_\reg{p} )^{1-n_\reg{x}} + \mathrm{H.c.}
\end{align}
with the labels $\{ \reg{x}$, $\reg{x'}$, $\reg{y}$, $\reg{y'} \}$ introduced for fermionic registers, and $\{ \reg{p}$, $\reg{q} \}$ for bosonic registers.
The precise mapping to modes for each subterm is displayed in Table \ref{tab:HI-subterm-labels_LSH}. Note that $n_\ell(r)$ and $n_\ell(r+1)$ are effectively equal under the action of the Hamiltonian terms in Eqs.~\eqref{eqs:HI-subterms-JW_LSH} because of the AGL, hence the replacement $\ket{\Lambda}\bra{\Lambda}_\reg{q} \to \ket{\Lambda}\bra{\Lambda}_\reg{p}$ in the second line of Eq.~\eqref{eq:generalizedHISubtermLSHTruncated}.
\begin{table}
\centering
\begin{tabular}{crlrl}
\vspace{0.5mm} Label & \multicolumn{4}{c}{$H_I^{\mathrm{LSH}(j)}$ translation} \\
\hline
\hline
\vspace{1mm} $j$ & \multicolumn{2}{c}{$\quad$ 1} & \multicolumn{2}{c}{$\quad$ 2} \\
$\reg{x'}$ & $(r,i  )$ & $\sim 2r  $ & $(r+1,o)$ & $\sim 2r+3$ \\
$\reg{x }$ & $(r,o  )$ & $\sim 2r+1$ & $(r+1,i)$ & $\sim 2r+2$ \\
$\reg{y'}$ & $(r+1,i)$ & $\sim 2r+2$ & $(r  ,o)$ & $\sim 2r+1$ \\
$\reg{y }$ & $(r+1,o)$ & $\sim 2r+3$ & $(r  ,i)$ & $\sim 2r  $ \\
$\reg{p }$ & $(r  ,\ell)$ & & $(r+1,\ell)$ & \\
$\reg{q }$ & $(r+1,\ell)$ & & $(r  ,\ell)$ & \\
\end{tabular}
\caption{\label{tab:HI-subterm-labels_LSH}
Label associations between registers of a generalized LSH subterm, Eq.~(\ref{eq:generalizedHISubtermLSHTruncated}), and the two subterms of Eqs.~\eqref{eqs:HI-subterms-JW_LSH}.
}
\end{table}

Now recalling the discussion of Sec.~\ref{sec:SVD}, the combination
$$x \, Z_\reg{x}\, \bigl( 1 - \ket{1}\bra{1}_\reg{y} \ket{\Lambda}\bra{\Lambda}_\reg{p} \bigr) \bigl( 1 - \ket{0}\bra{0}_\reg{x} \ket{\Lambda}\bra{\Lambda}_\reg{p} \bigr) \DLSH (p,n_\reg{x},n_\reg{y}) \sigma_\reg{x'}^- ( \lambda^{-}_\reg{q} )^{n_\reg{y}} ( \lambda^{-}_\reg{p} )^{1-n_\reg{x}}$$squares to zero, allowing it to be identified with the `$A$' operator of Sec.~\ref{sec:SVD}. There is also no need to introduce an ancilla because the $\reg{y'}$ qubit can function as the control to the SVD-transformation gates ($\mathscr{V}$, $\mathscr{W}$, or their adjoints). A solution to the SVD transformations is $\mathscr{V}=\mathcal{I}_\reg{x'}\mathcal{I}_\reg{p}\mathcal{I}_\reg{q}$ and $\mathscr{W} = X_\mathtt{x'} ( \lambda^{+}_\reg{q} )^{n_\reg{y}} ( \lambda^{+}_\reg{p} )^{1-n_\reg{x}} $. 
The full diagonalizing unitary for the generalized subterm in Eq.~\eqref{eq:generalizedHISubtermLSHTruncated} therefore works out to be 
\begin{align}
    \SVDLSH &\equiv \had_\reg{y'} \bigl( \ket{0}\bra{0}_\reg{y'} + \ket{1}\bra{1}_\reg{y'} X_\reg{x'} ( \lambda^-_\reg{q})^{n_\reg{y}} ( \lambda^-_\reg{p})^{1-n_\reg{x}} \bigr) .
\end{align}
When applied to the hopping subterms, one ultimately obtains
\begin{align}
    \SVDLSH & H_I^{\mathrm{LSH}(j)} {\mathscr{U}_{\mathrm{SVD}}^{\mathrm{LSH}}}^\dagger = x\ket{1}\bra{1}_{\reg{x'}} Z_{\reg{y'}} Z_\reg{x} \bigl( 1 - \ket{1}\bra{1}_\reg{y} \ket{\Lambda}\bra{\Lambda}_\reg{p} \bigr) \bigl( 1 - \ket{0}\bra{0}_\reg{x} \ket{\Lambda}\bra{\Lambda}_\reg{p} \bigr) 
    \DLSH (p,n_\reg{x},n_\reg{y}). \label{eq:diagonalized-subterm_LSH}
\end{align}
\begin{lemma}
Using a $(2\eta+4)$-qubit register plus 1 ancilla qubit, $e^{-itH^{{\rm LSH}}_I(r)}$ can be Trotterized using at most $16\times 2^\eta+16\eta^2+16\eta +52$ CNOT gates (and a number of single-qubit Z rotations).
\end{lemma}
\begin{proof}
This cost consists of the (near-term) cost of $\SVDLSH$, plus the cost of implementing a diagonalized subterm.
The $\SVDLSH$ circuit, shown in Fig.~\ref{fig:HI_high-level_LSH}, consists of one CNOT, two $C^2(\lambda^+)$ gates, and a Hadamard gate. A single $C^2(\lambda^+)$ gate can be achieved by performing the logical AND of the two control qubits into a blank ancilla qubit, then using that ancilla qubit as the single control to a $C(\lambda^\pm)$ gate, followed by uncomputing the AND.
The AND gates can be done with 3 CNOT gates each using the construction of Ref.~\cite{Barenco:1995na} (Sec. VI-B), while the $C(\lambda^\pm)$ gate can be done with $2\eta(\eta+1)$ CNOT gates as explained previously (see Table~\ref{tab:offdiag-costs_near_SB}).
Thus, the cost of a single $C^2(\lambda^+)$ gate is $3+2\eta(\eta+1)+3=2\eta^2+2\eta+6$ CNOT gates and one ancilla qubit.
The total CNOT-gate cost of a single $\SVDLSH$ execution is then
$1+2\times(2\eta^2+2\eta+6)=4\eta^2+4\eta+13$,
as reported in Table~\ref{tab:offdiag-costs_near_LSH}.
As was the case with the Schwinger-boson formulation, the dominant CNOT cost is associated not with $\SVDLSH$ but with implementing the diagonalized subterms, which will be discussed next.

The diagonalized subterms to be simulated are as given in Eq.~\eqref{eq:diagonalized-subterm_LSH}:
\begin{align*}
    x\ket{1}\bra{1}_{\reg{x'}} Z_\reg{y'} Z_{\reg{x}} \bigl( 1 - \ket{1}\bra{1}_\reg{y} \ket{\Lambda}\bra{\Lambda}_\reg{p} \bigr) \bigl( 1 - \ket{0}\bra{0}_\reg{x} \ket{\Lambda}\bra{\Lambda}_\reg{p} \bigr) \DLSH (p,n_\reg{x},n_\reg{y}) .
\end{align*}
The factors depending on $\{\reg{x'}$, $\reg{y}$, $\reg{p}$, $\reg{x}\}$ can be treated as a single diagonal function of $\eta+3$ qubits to be decomposed to Pauli matrices.
The only remaining qubit involved, $\reg{y'}$, can serve as the target for all parity evaluations of the terms in the $(\eta+3)$-bit Pauli decomposition, and for the associated single-qubit $R^Z$ gates.
According to the Hamiltonian-cycle method described for the Schwinger-boson hopping propagators, implementing this diagonal Pauli decomposition will cost $2^{\eta+3}$ CNOT gates.
This cost of a diagonalized LSH hopping subterm is recorded in Table \ref{tab:offdiag-costs_near_LSH}.
\begin{table}
\centering
\begin{tabular}
{>{\arraybackslash}p{10cm}  >{\centering\arraybackslash}p{4cm} c }
LSH subterm SVD subroutine & CNOT count & Ancillas \\
\hline
\hline
`Explicit' {CNOT} & 1 & 0 \\
$C^2(\lambda^{\pm})$ & $2\eta^2+2\eta+6$ & 1 \\
\hline
Overall $\SVDLSH$ circuit for each subterm & $4\eta^2+4\eta+13$ & 1 \\ 
\\
\\
LSH subterm diagonal-rotations subroutine
& {CNOT} count \\
\hline
\hline
Hamiltonian cycle through `hypercube' of Pauli parities & $2^{\eta+3}$ & 0 \\
\hline
Overall diagonal-rotations circuit & $2^{\eta+3}$ & 0 \\
\\
\\
LSH hopping-propagator routine & {CNOT} count \\
\hline
\hline
Subterm SVD & $4\eta^2+4\eta+13$ & 1 \\
Subterm diagonal rotations & $2^{\eta+3}$ & 0 \\
Subterm SVD${}^{-1}$ & $4\eta^2+4\eta+13$ & 1 \\
\hline
Full $e^{-itH^{{\rm LSH}}_I(r)}$ circuit & $16\times 2^\eta+16\eta^2+16\eta+52$ & 1 \\ 
\end{tabular}
\caption{\label{tab:offdiag-costs_near_LSH}
Summary of the costs associated with the near-term simulation of the off-diagonal operators in the LSH Hamiltonian, as explained in the text.}
\end{table}

The cumulative CNOT cost associated with a complete hopping term is obtained as four times the cost of a single $\SVDLSH$ circuit, plus the individual costs of both diagonalized subterms.
This works out to $16\times 2^\eta+16\eta^2+16\eta +52$, also reported in Table \ref{tab:offdiag-costs_near_LSH}.
\end{proof}

The Pauli decomposition of $\DLSH $ can be truncated via both the small-angle and input-truncations methods described in the case of the Schwinger-boson propagators. The qualitative conclusion regarding the significant loss of  accuracy by even minimal truncation remains the same with both methods, as demonstrated by the examples in Figs.~\ref{fig:diagonal-hopping-Pauli-truncation} and \ref{fig:diagonal-hopping-input-truncation}. The implementation of the LSH diagonal function is seen to benefit more from both types of truncations than the Schwinger-model diagonal function. At severe truncations for the LSH functions, the error bounds become comparable for both input-truncation and small-angle truncation methods.

\begin{center}
\textit{---Far term---}
\end{center}
\begin{center}
\emph{Implementing mass propagators}
\end{center}
\begin{lemma}
The mass propagator $e^{-itH^{{\rm LSH}}_M(r)}$ can be implemented, without approximation, using no ancilla qubits, two $R^Z$ gates, and no additional T gates.
\end{lemma}
\begin{proof}
The $R^Z$-gate count follows directly from Eqs.~\eqref{eqs:HM-subterms_LSH}, which constitute a linear combination of two single-qubit $Z$ operators.
As discussed originally for the Schwinger-boson propagator, the far-term $R^Z$ gates have a T-gate cost that scales with the desired synthesis error.
The function $\mathcal{C}_z(\epsilon)$ has been introduced to refer to the number of T gates required to simulate an $R^Z$ gate (of any input angle) within a spectral-norm error of $\epsilon$, and the RUS method of Ref.~\cite{bocharov2015efficient} may be used for a T count scaling of $\log(1/\epsilon)$.
The trivial cost of a mass propagator is recorded in Table \ref{tab:diag-costs_far_LSH}.
\end{proof}
\begin{table}[t!]
\centering
\begin{tabular}{l >{\centering\arraybackslash}p{5cm} c c }
LSH mass-propagator subroutine & T gates & Workspace & Scratch space \\
\hline
\hline
$H_M^{{\rm LSH}(j)}$ rotation & $\mathcal{C}_z (\epsilon)$ & 0 & 0 \\
\hline
Full $e^{-itH^{{\rm LSH}}_M(r)}$ circuit & $ 2 \, \mathcal{C}_z (\epsilon) $ & 0 & 0 \\ 
\\
\\
LSH electric-propagator subroutine & T gates & Workspace & Scratch space \\
\hline
\hline
(Un)compute $N^L$ & $4\eta+4$ & $\eta-1$ & 2 \\
Copy & 0 & 0 & $\eta+1$ \\
$(\eta+1)\times(\eta+1)$ mult. & $8\eta^2 + 12\eta + 4$ & $2\eta+2$ & $2\eta+2$ \\
$H_E^{{\rm LSH}(1)}$ rotations & $(\eta+1) \mathcal{C}_z (\epsilon)$ & 0 & 0 \\
$H_E^{{\rm LSH}(2)}$ rotations & $(2\eta+2) \mathcal{C}_z (\epsilon)$ & 0 & 0 \\
\hline
Full $e^{-itH^{{\rm LSH}}_E(r)}$ & $ 16\eta^2 + 32\eta  + 16 + (3\eta+3) \mathcal{C}_z (\epsilon)$ & $2\eta+2$ & $3\eta+5$ \\ 
\end{tabular}
\caption{\label{tab:diag-costs_far_LSH}
Summary of the costs associated with the far-term simulation of the diagonal operators in the LSH Hamiltonian, as explained in the text.}
\end{table}
\begin{center}
\emph{Implementing electric propagators}
\end{center}
In the far term, we use a splitting of $H_E$ that is identical to the Schwinger-boson formulation:
\begin{subequations}
\label{eqs:HE-subterms_LSH-FT}
\begin{align}
    \label{eq:HM-1_LSH-FT}
    H_{E}^{{\rm LSH}(1)}(r)&= \frac{1}{2} N^L (r) , \\
    \label{eq:HM-2_LSH-FT}
    H_{E}^{{\rm LSH}(2)}(r)&= \frac{1}{4} \bigl( N^L (r) \bigr)^2 ,
\end{align}
\end{subequations}
where the only difference is the expression for $N^L$ in terms of the LSH quantum numbers, i.e., $N^L(r)=n_{\ell}(r) + n_{o}(r)(1-n_{i}(r))$.

\begin{lemma}
The electric propagator $e^{-itH^{{\rm LSH}}_E(r)}$ can be implemented, without approximation, using $5\eta + 6$ ancilla qubits, $3\eta+3$ $R^Z$ gates, and an additional $ 16\eta^2 + 32\eta + 16$ T gates. 
\end{lemma}
\begin{proof}
The procedure leading to this cost is essentially the same as in the Schwinger-boson formulation, with the exception of how $N^L$ is computed.
It consists of the following steps:
\begin{enumerate}
    \item Compute $N^L$.
    \item Compute $(N^L)^2$.
    \item Realize a phase kickback via the registers containing $N^L$ and $(N^L)^2$.
    \item Uncompute $(N^L)^2$ and $N^L$.
\end{enumerate}
In step (1), $N^L$ can be computed as $$\ket{n_i} \ket{n_o} \ket{0} \ket{0} \ket{n_\ell} \ket{0}^{\otimes (\eta-1)} \mapsto \ket{n_i} \ket{n_o} \ket{ n_o (1-n_i) } \ket{N^L} \ket{0}^{\otimes (\eta-1)},$$ where $\ket{n_o(1-n_i)}$ is the output of a Toffoli gate and serves as the control bit to an incrementer on the  $\ket{0}\ket{n_\ell}$ register.
The Toffoli gate costs $4$ T gates, while the controlled incrementer calls for $4\eta$ T gates and $\eta-1$ workspace qubits (resulted from applying Lemma \ref{lem:inc-far} to an $(\eta+2)$-bit register). 
For (2), $(N^L)^2$ is computed as $\ket{N^L}\ket{0}^{\otimes(5\eta+5)} \mapsto \ket{N^L}\ket{N_L}^{\otimes(4\eta+4)} \mapsto \ket{N^L} \ket{N^L} \ket{(N^L)^2} \ket{0}^{\otimes(2\eta+2)}$ by first copying $N^L$ to an $(\eta+1)$-bit register using CNOT gates and then multiplying the two copies  of $N^L$.
The multiplier costs $8\eta^2 + 12\eta + 4$ T gates and $2\eta+2$ bits of workspace.
In step (3), the $N^L(r)$ and $(N^L(r))^2$ terms of $H_E(r)$ are effectively simulated by applying single-qubit $R^Z$ gates across the entire $\ket{N^L}$ and $\ket{(N^L)^2}$ registers (up to global phases that are dropped).
Finally, in step (4), uncomputation involves reversing the gates of steps (2) and (1) and the associated costs are the same.
Steps (2)-(4) are shown in Fig.~\ref{fig:HE_phase-kickback}.
In total, the procedure outlined above involves $3\eta+3$ $R^Z$ gates, which can each be executed to the desired precision using the RUS method mentioned in previous sections~\cite{bocharov2015efficient}.
The costs associated with all the subroutines are summarized in Table~\ref{tab:diag-costs_far_LSH}.
The final, quoted costs are obtained by adding up the workspace and scratch-space sizes.
\end{proof}
\begin{center}
\emph{Implementing hopping propagators}
\end{center}
\begin{lemma}
Let $\eta=\log_2(\Lambda+1)$ be the number of qubits per LSH bosonic register, $n>0$ be the number of Newton's method iterations, and $m>0$ be a fixed binary arithmetic precision.
Then $\Pi_{j=1}^2 e^{-itH^{{\rm LSH}(j)}_I(r)}$ can be implemented within an additive spectral-norm error of 
\begin{align}
    \label{eq:IPropError_LSH}
    2 \sqrt{2} \, x\, t \left[ 2^n \left(\sqrt{2}-1\right)^{2^n} + 2^{2-m} \left( \left(\frac{3}{2}\right)^n - 1\right) \right]
\end{align}
using
\begin{align}
2 \eta  n+3 m n+12 \eta + \max (\eta +2,m)+10 m+3 n+20
\end{align}
ancilla qubits, $4(m+1)$ $R^Z$ gates, and
\begin{align}
    320 \eta m n + 128 & m^2 n + 32 \eta^2 + 32 \eta m + 32 \eta n + 352 m n - 32 n \max(2\eta+2,2m)& \nonumber\\
    & \hspace{2.5 cm} + 224 \eta + 32 m - 16 \max(\eta+1,m) - 16n + 64 
\end{align}
T gates.
\end{lemma} 
\begin{proof}
To obtain this result, one must circuitize propagators associated with each of the two subterms of the hopping term.
One subterm is simulated by applying an appropriate diagonalizing transformation $\SVDLSH$, implementing the diagonalized subterm, and finally changing basis back via ${\SVDLSH}^\dagger$.
In the LSH formulation, the diagonalized hopping subterm in Eq.~\eqref{eq:diagonalized-subterm_LSH} involves a logical control factor associated with the finite cutoff $\Lambda$, which will be computed as part of implementing the diagonalized subterm.
The remainder of the diagonalized-subterm implementation follows the phase-kickback procedure used in the Schwinger-boson section.
The full sequence for simulating a hopping subterm is shown schematically in Fig.~\ref{fig:HI_high-level_LSH}(a).
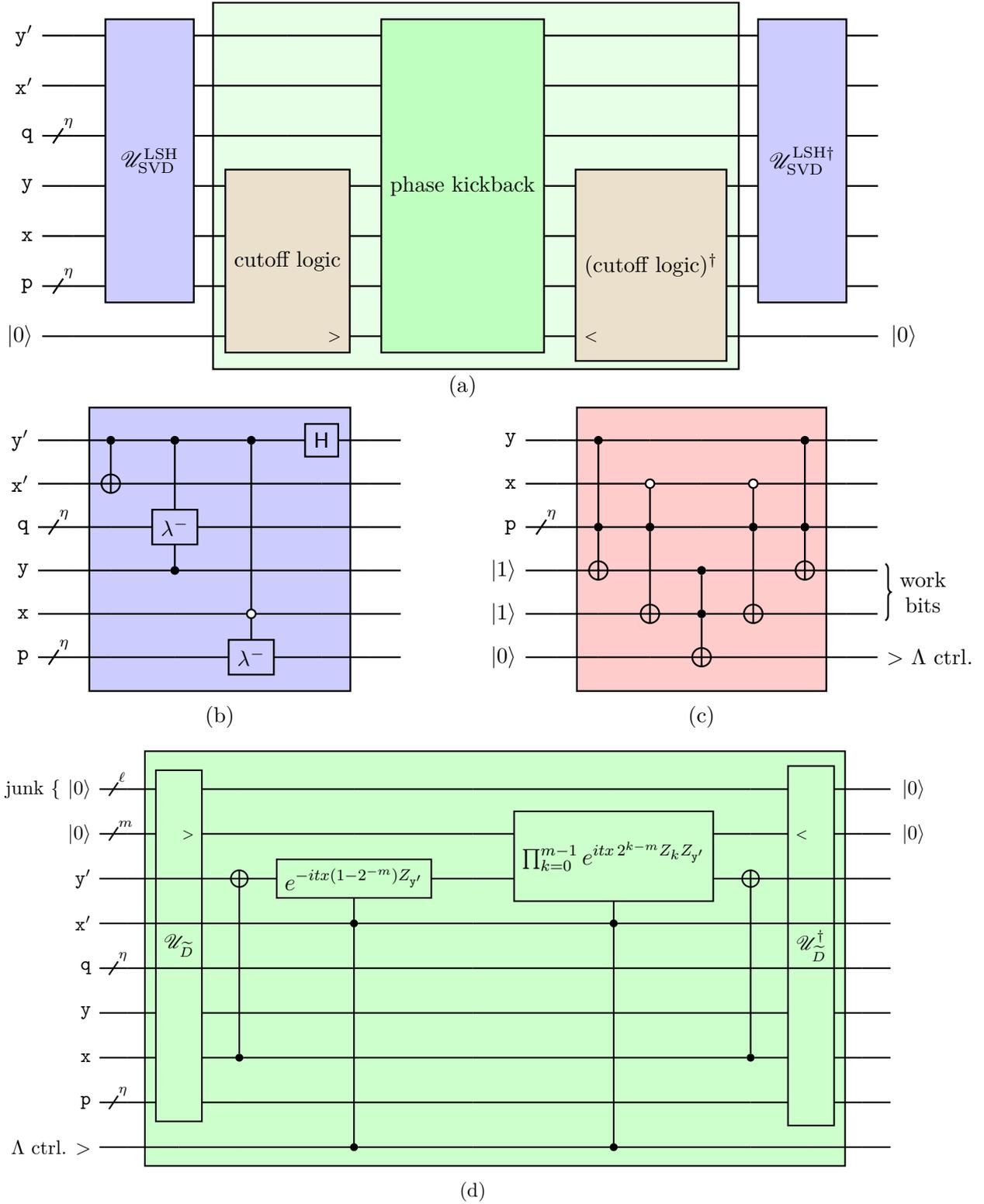
\begin{figure}
    \centering
    \adjustbox{width=1.0\linewidth,center}{%
 \begin{tikzcd}[row  sep={8mm,between  origins}]  
 \lstick{$\mathtt{y^\prime}$} & \qw & \gate[6,style={fill=blue!20},disable  auto  height][1.4cm]{\mathscr{U}_{\mathrm{SVD}}^{\mathrm{LSH}}} & \qw  \gategroup[7,steps=3,style={solid,  fill=green!10,  inner  xsep=2pt},  background]{  } & \gate[7,style={fill=green!25},disable  auto  height]{\text{phase  kickback}} & \qw & \gate[6,style={fill=blue!20},disable  auto  height][1.4cm]{\mathscr{U}_{\mathrm{SVD}}^{\mathrm{LSH}  \dagger}} & \qw &&&   \\ 
 \lstick{$\mathtt{x^\prime}$} & \qw && \qw && \qw && \qw &&&   \\ 
 \lstick{$\mathtt{q}$} & \qwbundle{\eta} && \qw && \qw && \qw &&&   \\ 
 \lstick{$\mathtt{y}$} & \qw && \gate[4,style={fill=green!40!red!20},disable  auto  height]{\text{cutoff  logic}} && \gate[4,style={fill=green!40!red!20},disable  auto  height]{(\text{cutoff  logic})^\dagger} && \qw &&&   \\ 
 \lstick{$\mathtt{x}$} & \qw &&&&&& \qw &&&   \\ 
 \lstick{$\mathtt{p}$} & \qwbundle{\eta} && \linethrough{} && \linethrough{} && \qw &&&   \\ 
 \lstick{$\ket{0}$} & \qw & \qw & \gateoutput{$>$}   && \gateinput{$<$}   & \qw & \qw && \lstick{$\ket{0}  \  \  $} &   \\ 
&&&& \text{(a)} &&&&&&
 \end{tikzcd}  
    }
    \adjustbox{width=1.0\linewidth,center}{%
 \begin{tikzcd}[row  sep={7mm,between  origins},transparent]  
& \lstick{$\mathtt{y^\prime}$} & \qw & \ctrl{1}  \gategroup[6,steps=4,style={solid,fill=blue!20,  inner  xsep=2pt},label    style={label    position=below,anchor=north,yshift=-3mm},background]{(b)} & \ctrl{2} & \ctrl{4} & \gate{\mathsf{H}} & \qw & \qw &&&& \lstick{$\mathtt{y}$} & \qw & \ctrl{2}  \gategroup[6,steps=5,style={solid,fill=red!20,  inner  xsep=2pt},label    style={label    position=below,anchor=north,yshift=-3mm},background]{(c)} & \qw & \qw & \qw & \ctrl{2} & \qw & \qw &   \\ 
& \lstick{$\mathtt{x^\prime}$} & \qw & \targ{} & \qw & \qw & \qw & \qw & \qw &&&& \lstick{$\mathtt{x}$} & \qw & \qw & \octrl{1} & \qw & \octrl{1} & \qw & \qw & \qw &   \\ 
& \lstick{$\mathtt{q}$} & \qwbundle{\eta} & \qw & \gate{\lambda^{-}} & \qw & \qw & \qw & \qw &&&& \lstick{$\mathtt{p}$} & \qwbundle{\eta} & \ctrl{1} & \ctrl{2} & \qw & \ctrl{2} & \ctrl{1} & \qw & \qw &   \\ 
& \lstick{$\mathtt{y}$} & \qw & \qw & \ctrl{-1} & \qw & \qw & \qw & \qw &&&& \lstick{$\ket{1}$} & \qw & \targ{} & \qw & \ctrl{1} & \qw & \targ{} & \qw & \qw  \rstick[2]{work\\bits}   &   \\ 
& \lstick{$\mathtt{x}$} & \qw & \qw & \qw & \octrl{1} & \qw & \qw & \qw &&&& \lstick{$\ket{1}$} & \qw & \qw & \targ{} & \ctrl{1} & \targ{} & \qw & \qw & \qw &   \\ 
& \lstick{$\mathtt{p}$} & \qwbundle{\eta} & \qw & \qw & \gate{\lambda^{-}} & \qw & \qw & \qw &&&& \lstick{$\ket{0}$} & \qw & \qw & \qw & \targ{} & \qw & \qw & \qw & \qw  \rstick[1]{$>$  $\Lambda$  ctrl.} &   \\ 
&&&&&&&&&&&&&&&&&&&&&
 \end{tikzcd}  
    }
    \adjustbox{width=1.0\linewidth,center}{%
 \begin{tikzcd}[row  sep={8mm,between  origins},transparent]  
 \lstick[wires=1]{junk  $\left\{  \right.$  } & \lstick{$\ket{0}$} & \qwbundle{\ell} & \gate[8]{\mathscr{U}_{\widetilde{D}}}  \gategroup[9,steps=7,style={solid,fill=green!20,  inner  xsep=2pt},  background]{  } & \qw & \qw & \qw & \qw & \qw & \gate[8]{\mathscr{U}_{\widetilde{D}}^\dagger} & \qw & \qw && \lstick{$\ket{0}  \  \  $} &   \\ 
& \lstick{$\ket{0}$} & \qwbundle{m} & \gateoutput{$>$} & \qw & \qw & \qw & \gate[2]{\text{\large $\prod_{k=0}^{m-1}  e^{i  tx \, 2^{k-m}  Z_\text{\scriptsize $k$}  Z_\reg{y'}  }$}} & \qw & \gateinput{$<$} & \qw & \qw && \lstick{$\ket{0}  \  \  $} &   \\ 
& \lstick{$\mathtt{y^\prime}$} & \qw & \linethrough{} & \targ{} & \gate{\text{\large $e^{-  i  tx  (1-2^{-m})  Z_\reg{y'}  }$}} & \qw && \targ{} & \linethrough{} & \qw & \qw &&&   \\ 
& \lstick{$\mathtt{x^\prime}$} & \qw & \linethrough{} & \qw & \ctrl{-1} & \qw & \ctrl{-1} & \qw & \linethrough{} & \qw & \qw &&&   \\ 
& \lstick{$\mathtt{q}$} & \qwbundle{\eta} & \linethrough{} & \qw & \qw & \qw & \qw & \qw & \linethrough{} & \qw & \qw &&&   \\ 
& \lstick{$\mathtt{y}$} & \qw && \qw & \qw & \qw & \qw & \qw && \qw & \qw &&&   \\ 
& \lstick{$\mathtt{x}$} & \qw && \ctrl{-4} & \qw & \qw & \qw & \ctrl{-4} && \qw & \qw &&&   \\ 
& \lstick{$\mathtt{p}$} & \qwbundle{\eta} && \qw & \qw & \qw & \qw & \qw && \qw & \qw &&&   \\ 
& \lstick{$\Lambda$  ctrl.  $>$} & \qw & \qw & \qw & \ctrl{-5} & \qw & \ctrl{-5} & \qw & \qw & \qw & \qw &&&   \\ 
&&&&&& \text{(d)} &&&&&&&&
 \end{tikzcd}  
    }
    \caption{\label{fig:HI_high-level_LSH}
    The circuit that implements the Schwinger-boson hopping propagator corresponding to each subterm in Eq.~\eqref{eqs:HI-subterms-JW_LSH}.
    (a) A high-level representation of the (far-term) circuit that realizes $H^{{\rm LSH}(j)}_I(r)$.
    (b) The diagonalization circuit $\SVDLSH$ for an LSH hopping subterm.
    This diagram applies to both the near-term and far-term algorithms.
    (c) The cutoff-control circuit.
    (d) The phase-kickback circuit.
    All circuits in (a-d) call for ancilla qubits that are not explicitly drawn but are discussed in the text and counted in the cost tables.
    }
\end{figure}

\vspace{0.2 cm}
\noindent
\emph{Diagonalization.}---The SVD transformation is as shown in Fig. \ref{fig:HI_high-level_LSH}(b):
one CNOT plus two $C^2(\lambda^+)$ gates and a Hadamard gate.
The $C^2(\lambda^+)$ operation, that is a doubly-controlled $\eta$-qubit incrementer, can proceed as an uncontrolled $(\eta + 2)$-qubit incrementer (followed by a CNOT and bit flips on the ``control'' qubits).
According to Lemma~\ref{lem:inc-far}, for an input of size $\eta+2$, the cyclic incrementer costs $4\eta$ T gates and $\eta-1$ ancilla qubits, where the ancilla qubits can be reused for all subsequent implementations.
Doubling this, the cost of $\SVDLSH$ is $8\eta$ T gates and $\eta-1$ workspace qubits, as reported in Table \ref{tab:offdiag-costs_far_LSH}.
Covering both subterms will further multiply the T count by four, for a total of $32\eta$ T gates per hopping term associated with diagonalizing transformations.
The dominant cost of simulating a hopping term, however, comes from implementing the diagonalized subterms, which will be discussed next.

\vspace{0.2 cm}
\noindent
\emph{Cutoff control.}---Before going into the phase-kickback procedure, the cutoff logical control (that was not explicitly needed in the Schwinger-boson formulation given the functional form of $\DSB$) will be computed.
The specific factor in Eq.~\eqref{eq:diagonalized-subterm_LSH} that is referred to here is
$ \bigl( 1 - \ket{1}\bra{1}_\reg{y} \ket{\Lambda}\bra{\Lambda}_\reg{p} \bigr) \bigl( 1 - \ket{0}\bra{0}_\reg{x} \ket{\Lambda}\bra{\Lambda}_\reg{p} \bigr) ,$
which amounts to an evaluation of
$(1-\delta_{n_\reg{y}  ,  1}  \delta_{p  ,  \Lambda}  )  (1-\delta_{n_\reg{x},0}  \delta_{p,\Lambda})$.
The cutoff logic evaluation is shown in circuit form in Fig.~\ref{fig:HI_high-level_LSH}(c).
The implementation requires four $C^{\eta+1}(X)$ gates and one Toffoli gate.
According to Ref.~\cite{HeLuoZhang:2017} (Table 3), a $C^{\eta+1}(X)$ gate can be performed using $4\eta$ T gates, $\eta+1$ workspace qubits, and one output ancilla qubit.
The outputs of the initial two $C^{\eta+1}(X)$ gates are combined into a single control using a Toffoli gate at an added cost of four T gates and one qubit for the final result.
Overall, the cutoff logic circuit calls for $16\eta+4$ T gates, $(\eta+1)+2=\eta+3$ workspace qubits, and one qubit for the final output.

\vspace{0.2 cm}
\noindent
\emph{Phase kickback.}---The remainder of the diagonalized hopping-subterm propagator consists of a phase-kickback algorithm that closely follows the one presented for the Schwinger-boson formulation.
Its circuit implementation is depicted in Fig.~\ref{fig:HI_high-level_LSH}(d).
First, the value of the non-trivial diagonal bosonic operator $\DLSH$ is computed up to a known accuracy via a unitary $\UfLSH$ and the value is stored in an ancillary register of $m$ bits, where  $\widetilde{\mathcal{D}}^{\mathrm{LSH}}$ denotes a fixed-precision approximation to $\DLSH$ and $m$ is the desired precision.
We will construct $\UfLSH$ via Newton's method shortly and will discuss its cost as a function of the approximation error. Since the ancillary register holds the (approximate) value of $\DLSH$ in binary form, the diagonal operator 
$$e^{-itx \ket{1}\bra{1}_{\reg{x'}} Z_{\reg{y'}} Z_\reg{x} \bigl( 1 - \ket{1}\bra{1}_\reg{y} \ket{\Lambda}\bra{\Lambda}_\reg{p} \bigr) \bigl( 1 - \ket{0}\bra{0}_\reg{x} \ket{\Lambda}\bra{\Lambda}_\reg{p} \bigr) \DLSH (p,n_\reg{x},n_\reg{y})}$$
can be straightforwardly implemented by $R_Z$ gates on the corresponding qubits with the angles shown in the circuit.
The ancillary register is set back to an all-$\ket{0}$ state by the inverse operations, and is used in subsequent implementations.

There are multiple ways of formulating the calculation of $\DLSH(p,n,n')= \sqrt{ \frac{p+1+n}{p+1+n'} }$.
One way of formulating it goes as follows:
\begin{enumerate}
    \item Increment $p$ by one (after augmenting the register with a qubit).
    \item Duplicate $p+1$ into an $(\eta+1)$-bit register.
    \item Use $n$ and $n'$ as controls to (separate) $C(\lambda^+)$ gates acting on the $\ket{p+1}$ registers, resulting in registers containing $p+1+n$ and $p+1+n'$.
    \item Multiply the numerator and denominator into a $(2\eta+2)$-bit register to obtain $(p+1+n)(p+1+n')$. For later convenience, this product is referred to by
\begin{align}
    \gLSH(p,n,n') &\equiv (p+1+n)(p+1+n').
\end{align}
    \item  Compute the inverse square root of $\gLSH(p,n,n')$.
    \item  Multiply $(p+1+n)$ with $\gLSH(p,n,n')^{-1/2}$.
\end{enumerate}
In this protocol, the most complicated and resource-intensive step is (5):
the inverse-square-root evaluation.
The unitary that implements the steps above is called $\UfLSH$, and the evaluation only produces a fixed-precision approximation $\widetilde{\mathcal{D}}^{\mathrm{LSH}}(p,n,n')\approx\DLSH(p,n,n')$.
We will return shortly to address the details of using fixed-precision arithmetic and the associated errors.
A circuit diagram for $\UfLSH$ is provided in Fig.~\ref{fig:Uf_LSH}.
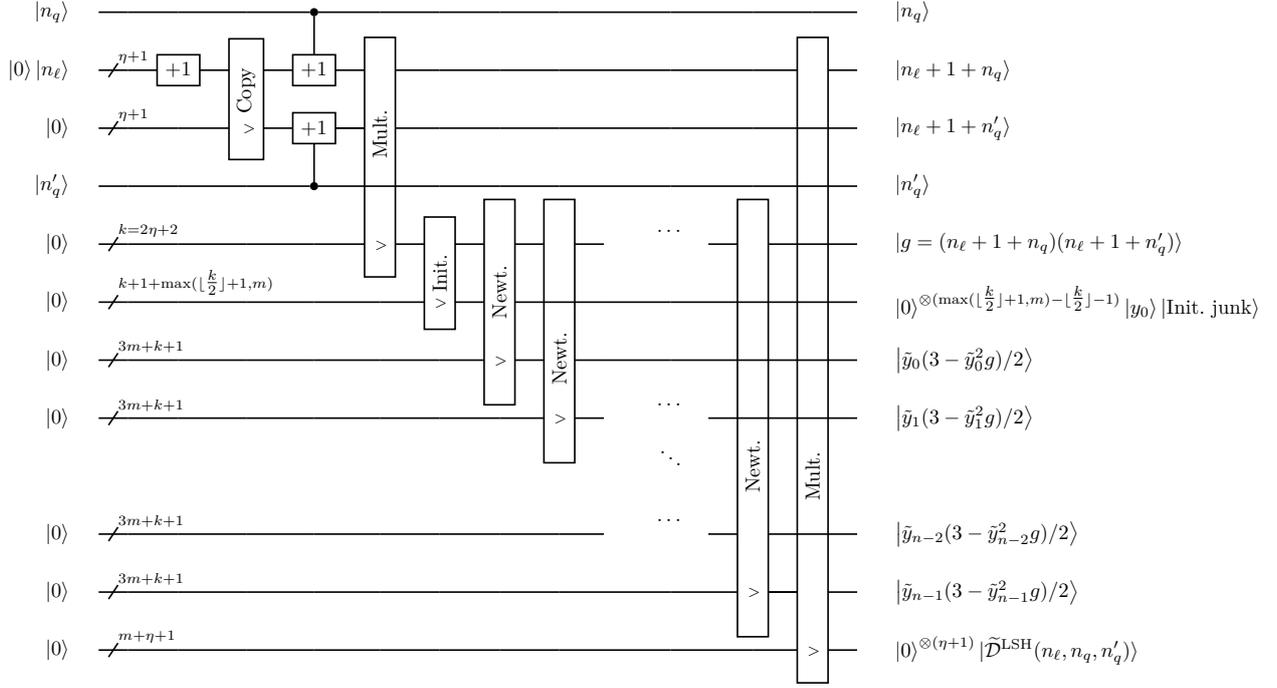
\begin{figure}
\centering
    \begin{adjustbox}{width=0.99\linewidth}
    \begin{tikzcd}[row  sep={1cm,between  origins},transparent]  
 \lstick{$\ket{n_q}$\quad} & \qw & \qw & \qw & \ctrl{1} & \qw & \qw & \qw & \qw & \qw & \qw & \qw & \qw & \qw & \rstick{$\ket{n_q}$}  \\ 
 \lstick{$\ket{0}\ket{n_{\ell}}$\quad} & \qwbundle{\eta+1} & \gate{+1} & \gate[2,label  style={yshift=0.15cm}]{\rotatebox{90}{\text{Copy}}} & \gate{+1} & \gate[4,label  style={yshift=0.4cm}]{\rotatebox{90}{\text{Mult.}}} & \qw & \qw & \qw & \qw & \qw & \qw & \gate[11,nwires={8},label  style={yshift=-2.0cm}]{\rotatebox{90}{\text{Mult.}}} & \qw & \rstick{$\ket{n_{\ell}+1+n_q}$}  \\ 
 \lstick{$\ket{0}$\quad} & \qwbundle{\eta+1} & \qw & \gateoutput{$>$} & \gate{+1} & \qw & \qw & \qw & \qw & \qw & \qw & \qw & \linethrough{} & \qw & \rstick{$\ket{n_{\ell}+1+n_q^{\prime}}$}  \\ 
 \lstick{$\ket{n_q^{\prime}}$\quad} & \qw & \qw & \qw & \ctrl{-1} & \linethrough & \qw & \qw & \qw & \qw & \qw & \qw & \linethrough & \qw & \rstick{$\ket{n_q^{\prime}}$}  \\ 
 \lstick{$\ket{0}$\quad} & \qwbundle{k=2\eta+2} & \qw & \qw & \qw & \gateoutput{$>$} & \gate[2]{\rotatebox{90}{\text{Init.}}} & \gate[3]{\rotatebox{90}{\quad  \text{Newt.}}} & \gate[4,label  style={yshift=-0.65cm}]{\rotatebox{90}{\quad  \text{Newt.}}} & \qw & \midstick{$\cdots$} & \gate[7,nwires={5},label  style={yshift=-1.0cm}]{\rotatebox{90}{\quad  \text{Newt.}}} & \linethrough & \qw & \rstick{$\ket{g=(n_{\ell}+1+n_q)(n_{\ell}+1+n_q^{\prime})}$}  \\ 
 \lstick{$\ket{0}$\quad} & \qwbundle{k+1+\max(\lfloor  \tfrac{k}{2}  \rfloor+1,m)} & \qw & \qw & \qw & \qw & \gateoutput{$>$} &          & \linethrough & \qw & \qw & \linethrough & \linethrough & \qw & \rstick{$\ket{0}^{
 \otimes(\max(\lfloor  \tfrac{k}{2}  \rfloor+1,m)-\lfloor  \tfrac{k}{2}  \rfloor-1)
 }\ket{y_0
 }\ket{\rm Init.~junk}$}  \\ 
 \lstick{$\ket{0}$\quad} & \qwbundle{3m+k+1} & \qw & \qw       & \qw & \qw & \qw & \gateoutput{$>$} &          & \qw & \qw     & \linethrough & \linethrough & \qw & \rstick{$\Ket{\tilde{y}_0  (3-\tilde{y}_0^2  g)/2}$}  \\ 
 \lstick{$\ket{0}$\quad} & \qwbundle{3m+k+1} & \qw & \qw         & \qw & \qw & \qw & \qw & \gateoutput{$>$} & \qw & \midstick{$\cdots$} & \linethrough & \linethrough & \qw & \rstick{$\Ket{\tilde{y}_1  (3-\tilde{y}_1^2  g)/2}$}  \\ 
&                        &&        &      &    &        &          &                &    & \midstick{$\ddots$} &    &    &    & \\ 
 \lstick{$\ket{0}$\quad} & \qwbundle{3m+k+1} & \qw & \qw         & \qw & \qw & \qw & \qw & \qw     & \qw   & \midstick{$\cdots$} &      & \linethrough & \qw & \rstick{$\Ket{\tilde{y}_{n-2}  (3-\tilde{y}_{n-2}^2  g)/2}$}  \\ 
 \lstick{$\ket{0}$\quad} & \qwbundle{3m+k+1} & \qw & \qw       & \qw & \qw & \qw & \qw & \qw & \qw & \qw   & \gateoutput{$>$} & \qw & \qw & \rstick{$\Ket{\tilde{y}_{n-1}(3-\tilde{y}_{n-1}^2  g)/2}$}  \\ 
 \lstick{$\ket{0}$\quad} & \qwbundle{m+\eta+1} & \qw & \qw         & \qw & \qw & \qw & \qw & \qw & \qw & \qw   & \qw & \gateoutput{$>$} & \qw & \rstick{$\ket{0}^{\otimes(\eta+1)}\ket{\widetilde{\mathcal{D}}^{\mathrm{LSH}}(n_\ell,n_q,n_q')}$} 
 \end{tikzcd}  
    \end{adjustbox}
    \caption{\label{fig:Uf_LSH}
    A quantum circuit for computing the diagonal function $\DLSH$ defined in Eq.~(\ref{eq:DLSH}) with $\eta$-qubit bosonic register $n_\ell$ and one-qubit fermionic registers $n_q$ and $n_q'$.
    Incrementers ($+1$), multiplication (Mult.), and initial-guess (Init.) subcircuits are described in Appendix~\ref{app:arithmetic}.
    The decomposition of a Newton-iteration step (Newt.) is depicted in Fig.~\ref{fig:newton-iteration}.
    For each subcircuit, a line passing through the box indicates that the corresponding qubit register does not participate in the operations.
    The output of each unitary is marked by $>$ (or $<$ for the inverses).
    Workspace qubits are left implicit.
    $k \equiv 2\eta +2$, $n$ indexes the final Newton-iteration step, $y_0=2^{-s}$ with $s=\left\lfloor \log_4((n_{\ell}+1+n_q)(n_{\ell}+1+n_q^{\prime})) + \frac{1}{2} \right\rfloor$, and the state $\ket{\rm Init.~junk}$ is stored in $k+1$ qubits, see Lemma~\ref{lem:guess}.
    }
\end{figure}

The complete discussion of relative error as a function of Newton iterations $n$ and fixed-bit precision $m$ in calculating the inverse square root $g^{-1/2}$ carries through unchanged from the Schwinger-boson formulation to the LSH formulation.
The only difference between the Schwinger-boson and LSH cases is the conversion from relative error bound to absolute error bound.
To do this conversion in a way that is agnostic to the input state, one can multiply by the norm of the diagonal function.
In Appendix~\ref{app:commutators}, it is shown that $\| \DSB \| \leq 1$ and $\| \DLSH \|=\sqrt{2}$, implying a factor of $\sqrt{2}$ needs to be included in the Newton's method approximation error for the LSH formulation.

To summarize the phase evaluation, there exists a quantum circuit that implements the unitary $\UfLSH$ defined as
\begin{equation}
    \ket{n} \ket{n^\prime} \ket{p} \ket{0}^{\otimes (1+m+\ell+\beta)} \xrightarrow[]{\UfLSH}  \ket{n} \ket{n^\prime} \ket{p + 1 + n_q} \ket{\widetilde{D}^{\mathrm{LSH}}(p, n , n')}\ket{\rm junk}\ket{0}^{\otimes \beta}.
\end{equation}
The state $\ket{\widetilde{D}^{\mathrm{LSH}}(p, n , n')}$ is an $m$-bit fixed-precision encoding of $\DLSH(p, n , n')$ within absolute error $|\delta_{n,m}^{\rm Abs.}|\leq \sqrt{2} \big[ 2^n (\sqrt{2}-1)^{2^n} + 2^{2-m} ( (3/2)^n - 1)\big]$ for $n$ iterations of Newton's method.
The state $\ket{\rm junk}$, which holds intermediate values, occupies an $\ell$-bit register. 
Furthermore, the circuit can be implemented using at most
\begin{align}
&80\eta mn + 32 m^2 n + 8\eta^2 + 8\eta m + 8 \eta n + 88mn - 8n \max(2\eta+2,2m)+32\eta
\nonumber\\
& \hspace{6.75 cm} -4\max (\eta+1,m)+8m-4n+8
\end{align}
T-gates. 
Finally, the circuit consumes at most
\begin{align}
\ell = 2\eta n + 3mn + 6\eta + \max(\eta+2,m) + 3n + 7
\end{align}
in output junk qubits and 
borrows 
\begin{align}
\beta= 6\eta+9m+9
\end{align}
reusable workspace qubits.
Lemma~\ref{lem:kineticphasekickLSH} of Appendix~\ref{app:arithmetic} provides a derivation of these results.

Finally, with the cutoff-control factor and $\DLSH$ both evaluated, the diagonal operator $$x\ket{1}\bra{1}_{\reg{x'}} Z_{\reg{y'}} Z_\reg{x} \bigl( 1 - \ket{1}\bra{1}_\reg{y} \ket{\Lambda}\bra{\Lambda}_\reg{p} \bigr) \bigl( 1 - \ket{0}\bra{0}_\reg{x} \ket{\Lambda}\bra{\Lambda}_\reg{p} \bigr) \DLSH (p,n_\reg{x},n_\reg{y})$$ from Eq.~(\ref{eq:diagonalized-subterm_LSH}) is effectively replaced by $x\ket{1}\bra{1}_{\reg{x'}} \ket{1}\bra{1}_{\text{$\Lambda$-ctrl}} Z_{\reg{y'}} Z_\reg{x} N_\reg{aux}$, where `$\Lambda$-ctrl' refers to the cutoff-control qubit and $N_\reg{aux} = 2^{2-m} \sum_{k=0}^{m-1} 2^{k-1} (1-Z_\reg{auxkj})$ is the fixed-point number operator on the $m$-bit auxiliary register holding $\widetilde{\mathcal{D}}^{\mathrm{LSH}}$.
With $N_\reg{aux}$ expanded out into $m$ single-qubit $Z$ operators plus a constant, the rotations can be split into $m+1$ doubly-controlled $R^Z$ rotations.
To simplify the rotations, a scratch ancilla is introduced into which the result of logical AND of the two controls is input, adding a T-gate cost of eight (including the uncomputation).
As remarked above, the far-term $R^Z$ gates incur a T-gate cost associated with the tunable precision of the desired rotation angles and one can use the RUS construction of Ref.~\cite{bocharov2015efficient}.
These rotations are recorded in Table \ref{tab:offdiag-costs_far_LSH}, where a factor of two is included to account for now a single control on the $R^Z$ gates.

\vspace{0.2 cm}
\noindent
\emph{Final tally.}---Putting everything together, the simulation of the diagonalized hopping subterm costs 
\begin{align}
\label{eq:ancilla-hop-LSH}
2 \eta  n+3 m n+12 \eta + \max (\eta +2,m)+10 m+3 n+18
\end{align}
ancilla qubits (where $m$ is added to account for the primary phase-kickback register), $2(m+1)$ $R^Z$ gates, and an additional
\begin{align}
\label{eq:T-gate-hop-LSH}
&160 \eta m n + 64 m^2 n + 16 \eta^2 + 16 \eta m + 16 \eta n + 176 m n - 16 n \max(2\eta+2,2m) + 112 \eta 
\nonumber\\
& \hspace{7.6 cm} + 16 m - 8 \max(\eta+1,m) - 8n  + 32
\end{align}
T gates.
Adding the costs of doing two structurally equivalent subterms, the full upper bound on implementing each Trotterized hopping propagator in the LSH formulation comes to twice the T-gate cost given in Eq.~\eqref{eq:T-gate-hop-LSH} and no more ancilla cost than what is stated in Eq.~\eqref{eq:ancilla-hop-LSH},
as stated originally and summarized in Table \ref{tab:offdiag-costs_far_LSH}.
\end{proof}
\begin{table}[t]
\centering
\begin{tabular}{l >{\arraybackslash}p{7.5cm} l >{\arraybackslash}p{2.5cm}}
Routine & T gates & Workspace & Scratch space \\
\hline
\hline
$\SVDLSH$ or ${\SVDLSH}^\dagger$ & $8\eta$ & $\eta-1$ & 0 \\
(Un)compute cutoff logic & $16\eta+4$ & $\eta+3$ & $1$ \\
$\UfLSH$ or ${\UfLSH}^\dagger$ & $80\eta mn + 32 m^2 n + 8\eta^2 + 8\eta m + 8 \eta n + 88mn - 8n \max(2\eta+2,2m)+32\eta-4\max (\eta+1,m)+8m-4n+8$ & $6\eta+9m+9$ & $2\eta n + 3mn + 6\eta + \max(\eta+2,m) + m + 3n + 7$ \\
Diagonal rotations & $8+2(m+1) \mathcal{C}_z(\epsilon) $ & 0 & 1 \\ 
\hline
Overall $H^{{\rm LSH}(j)}_i(r)$ & $ 
160 \eta m n + 64 m^2 n + 16 \eta^2 + 16 \eta m + 16 \eta n + 176 m n - 16 n \max(2\eta+2,2m) + 112 \eta + 16 m - 8 \max(\eta+1,m) - 8 n + 32 + 2 (m+1) \mathcal{C}_z(\epsilon) 
$ & $6\eta+9m+9$ & $2\eta n + 3mn + 6\eta + \max(\eta+2,m) + m + 3n + 9$ \\ 
\end{tabular}
\caption{\label{tab:offdiag-costs_far_LSH}
Summary of the costs associated with the far-term simulation of the off-diagonal operators in the LSH Hamiltonian, as explained in the text.}
\end{table}
%


\subsection{Error-bound analysis and simulation cost
\label{sec:bounds}}
With a second-order Trotter-Suzuki formula $V_2(t)$, the Trotterization error scales cubically with the time duration of each Trotter step, improving upon the quadratic scaling of the first-order formula $V_1(t)$ at only twice the gate cost. Given the decompositions of the Hamiltonian in both the Schwinger-boson and the LSH bases introduced in Sec.~\ref{sec:circuits}, and considering that in both formulations $H_E^{(j)}$ mutually commute, as do $H_M^{(j)}$, the 
second-order expansion of the time-evolution operator, $V_2(t)$, in the SU(2) theory can be formed in terms of the first-order expansion $V_1(t)$ as
\begin{subequations}
\begin{align}
\label{eq:1stTrotter}
V_1(t)&= \prod_{r=0}^{L-2} \Big[ e^{-it (H_M(r) + H_E(r))} \prod_{j=1}^\nu e^{-it H_I^{(j)}(r)} \Big]
e^{-it H_M(L-1)} ,\\
\label{eq:2ndTrotter}
V_2(t) &= V_1 (t/2) V_1 (-t/2)^\dagger ,
\end{align}
\end{subequations}
where the product in Eq.~\eqref{eq:1stTrotter} is ordered left to right.  Each exponential in the product can then be circuitized according to the near- or far-term strategies of Sec.~\ref{sec:methods}.

Furthermore, the second-order product formula in Eq.~(\ref{eq:2ndTrotter}) can be bounded by the double-commutator relation~\cite{childs2021theory}
\begin{align}
    \| V_2(t) - e^{-itH} \| \leq \frac{t^3}{24} \sum_{i = 1}^\Upsilon \biggl \| \biggl[ H_i , \biggl [H_i,\sum_{j = i + 1}^\Upsilon  H_j \biggr ] \biggr] \biggr \| + \frac{t^3}{12} &
    \sum_{i = 1}^\Upsilon \biggl \| \biggl [\sum_{k = i + 1}^\Upsilon H_k , \biggl [\sum_{j = i + 1}^\Upsilon H_j, H_i \biggr ]\biggr] \biggr \|,
\end{align}
where $H=\sum_{i=1}^\Upsilon H_i$ and $H_i \in \big\{ H_M,H_E,H_I^{(j)}\big\}$ for $j=1,\cdots,\nu$. This bound is evaluated for the Schwinger-boson and the LSH formulations in Appendix~\ref{app:commutators}. The corresponding results will be used in the following to estimate the cost of simulation in each formulation, given the algorithms of Sec.~\ref{sec:circuits} and the target error on the time-evolution operator.

In the fault-tolerant regime, in addition to the Trotterization error, there are additional known sources of error.
The first of these is associated with the synthesis of arbitrary $R^Z$ rotations, the accuracies of which are generally referred to as $\epsilon$ here.
The second of these is the error due to finite precision in the diagonal phase-function evaluations.
Specifically, this arises from truncating inverse-square-root evaluations to a finite number of steps $n$ in Newton's method, as well as the truncation of the involved calculations to a fixed $m$-bit precision.
The three sources of error will be summarized and combined in this section to derive the complete error bound in the far-term scenario.

\subsubsection{Schwinger-boson formulation
\label{sec:SBcost}}

\noindent
\emph{Trotterization error.}---In Appendix \ref{app:commutators_SB}, we apply Eq.~\eqref{eq:2ndTrotter} to the complete Hamiltonian of the Schwinger-boson formulation to obtain the following result:
\begin{subequations}
\label{eq:rho-SB}
\begin{align}
    &\| V(\theta) - e^{-i\theta H} \| \leq L \theta^3 \rho^{\rm SB} (x, \Lambda, \mu ) , \\
    &\rho^{\rm SB} (x, \eta, \mu ) \equiv  
    \frac{1658 x^3}{3} + 32 \Lambda  x^2 + \frac{218 \mu  x^2}{3} + 8 x^2 + \frac{\Lambda ^2 x}{3} + \frac{4 \Lambda  \mu  x}{3} + \frac{\Lambda  x}{6} + \frac{5 \mu ^2 x}{3} + \frac{\mu  x}{3} + \frac{x}{48},
\end{align}
\end{subequations}
with $\Lambda=2^\eta-1$.
For a total evolution time of $T$ spread over $s$ Trotter steps, this implies
\begin{align}
    \| V(T/s)^{s} - e^{-i T H} \| &\leq \frac{L T^3}{s^2} \rho^{\rm SB} (x, \eta, \mu ) .
\end{align}
When working with a fixed error budget $\Delta_{\mathrm{Trot}}$, one should take the number of Trotter steps $s$ to be
\begin{align}
    \label{eq:error-bounded-s_SB}
    s(x,\eta,\mu,L,T,\Delta_{\mathrm{Trot}}) &= \left\lceil \sqrt{ \frac{L T^3}{\Delta_{\mathrm{Trot}}} \rho^{\rm SB} (x, \eta, \mu ) } \right\rceil
\end{align}
in order to guarantee $\| V(T/s)^{s} - e^{-i T H} \| \leq \Delta_{\mathrm{Trot}} $. Note that since dimensionless Hamiltonians have been used throughout, when considering the dimensionfull evolution time $t$, it should be converted to the scaled dimensionless time $T$ via $t=2xTa_s$, see discussions in Sec.~\ref{sec:KS}.

\vspace{0.2 cm}
\noindent
\emph{Near-term simulation costs.} In the near-term, the cost metric is identified as the number of CNOT gates due to their non-negligible error.
Given some desired error $\Delta_{\mathrm{Trot}}$, one can use Eq.~\eqref{eq:error-bounded-s_SB} to determine the minimal required number of Trotter steps necessary to bring the Trotterization error within the error budget. Then given the cost analysis of Sec.~\ref{sec:circuits}, the total number of CNOT gates required for the Schwinger-boson formulation is
\begin{align}
2s(L-1)\left(16\times 8^\eta+67\eta^2+65\eta+30\right)
\end{align}
Some representative values are tabulated for a range of modest simulation parameters in Table~\ref{tab:errorBoundedNISQCosts} of Appendix \ref{app:cost-tables}. Furthermore, the CNOT-gate count as a function of $\eta$ and evolution time are plotted in Fig.~\ref{fig:CNOT-cost} for better visualization.
\begin{figure}[t!]
\centering
  \begin{subfigure}{0.43\textwidth}
    \includegraphics[width=\textwidth]{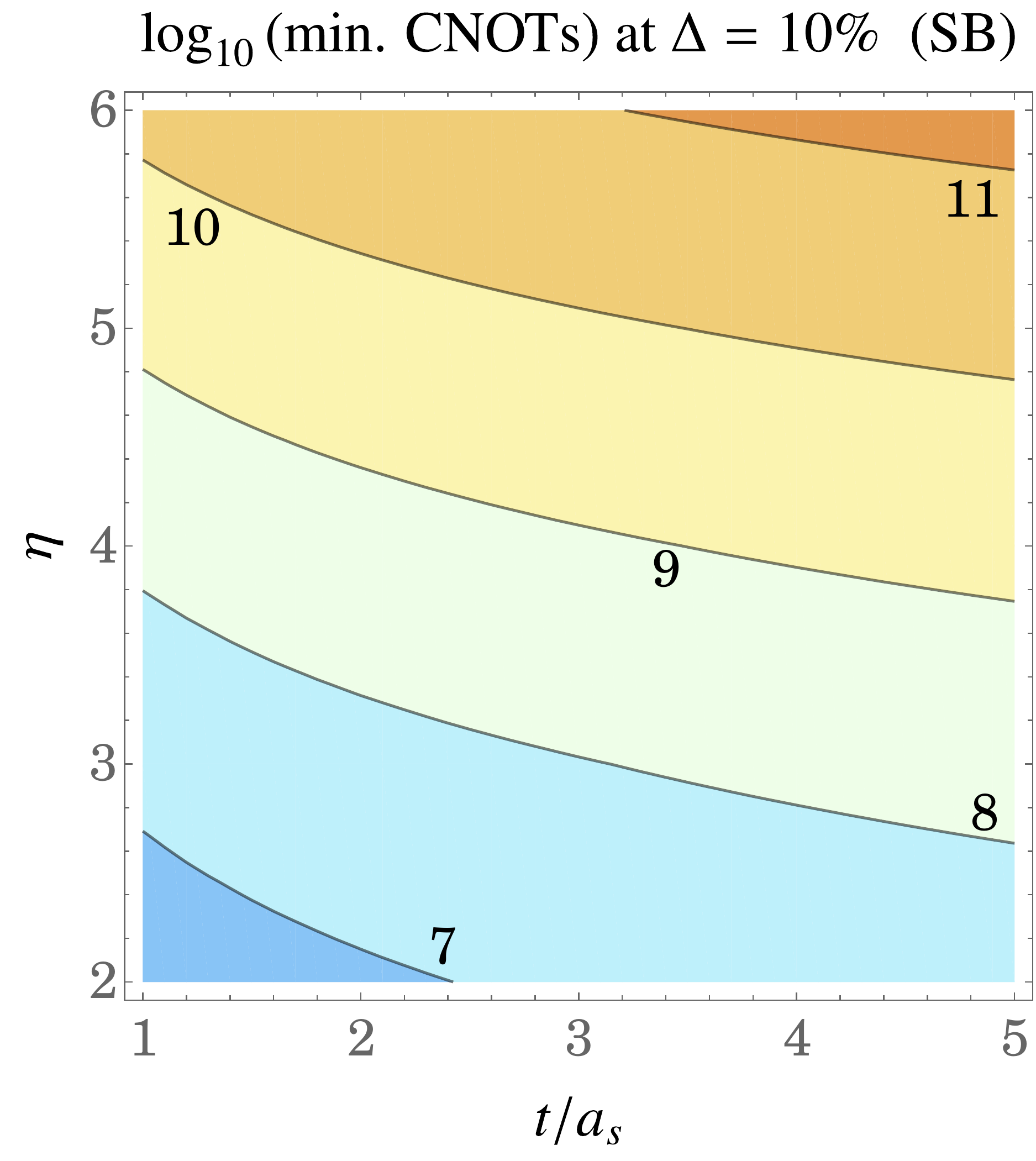}
    \caption{
    }
  \end{subfigure}
  \hfill
  \begin{subfigure}{0.43\textwidth}
    \includegraphics[width=\textwidth]{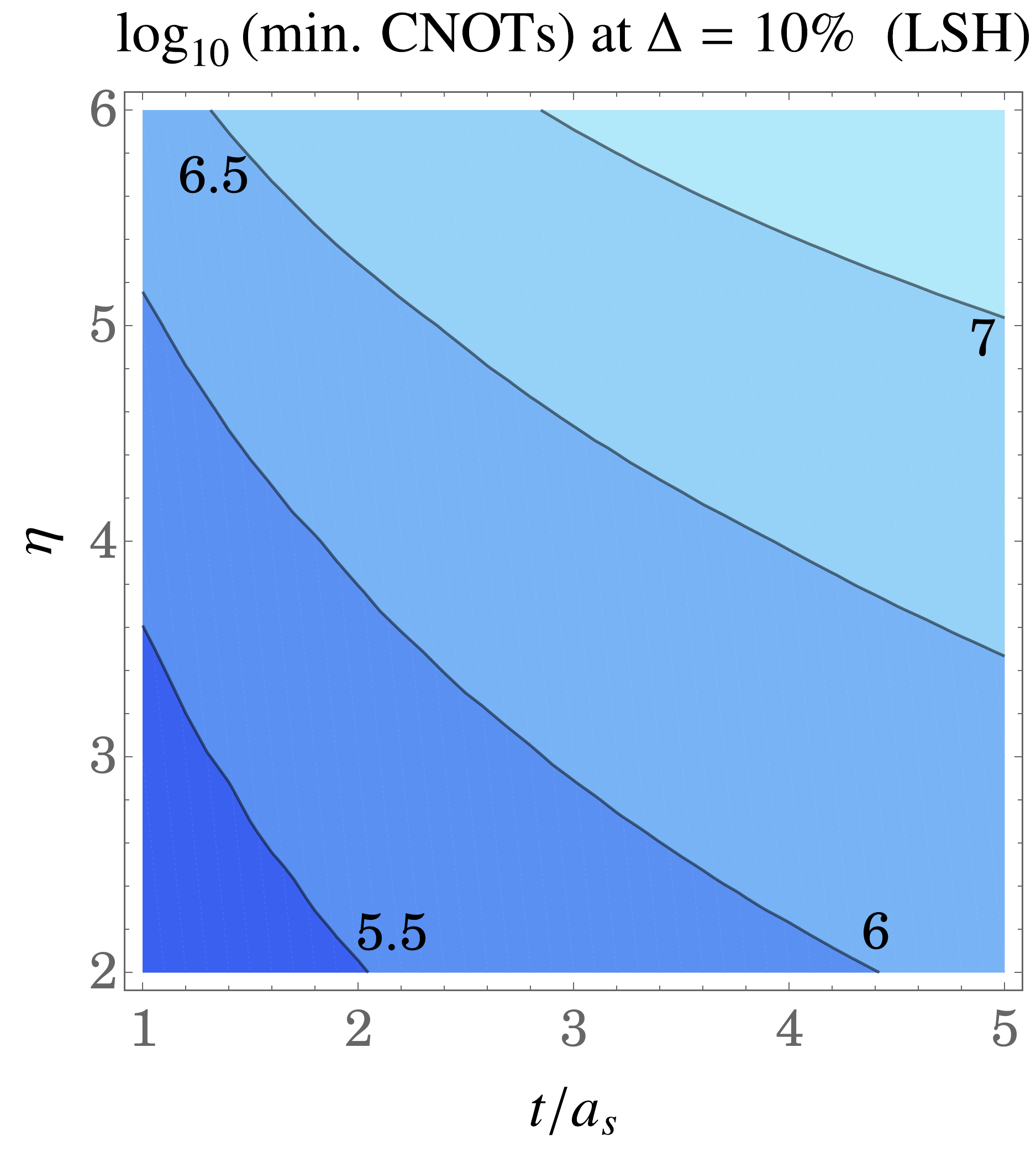}
    \caption{
    }
  \end{subfigure}
  \begin{subfigure}{0.43\textwidth}
    \includegraphics[width=\textwidth]{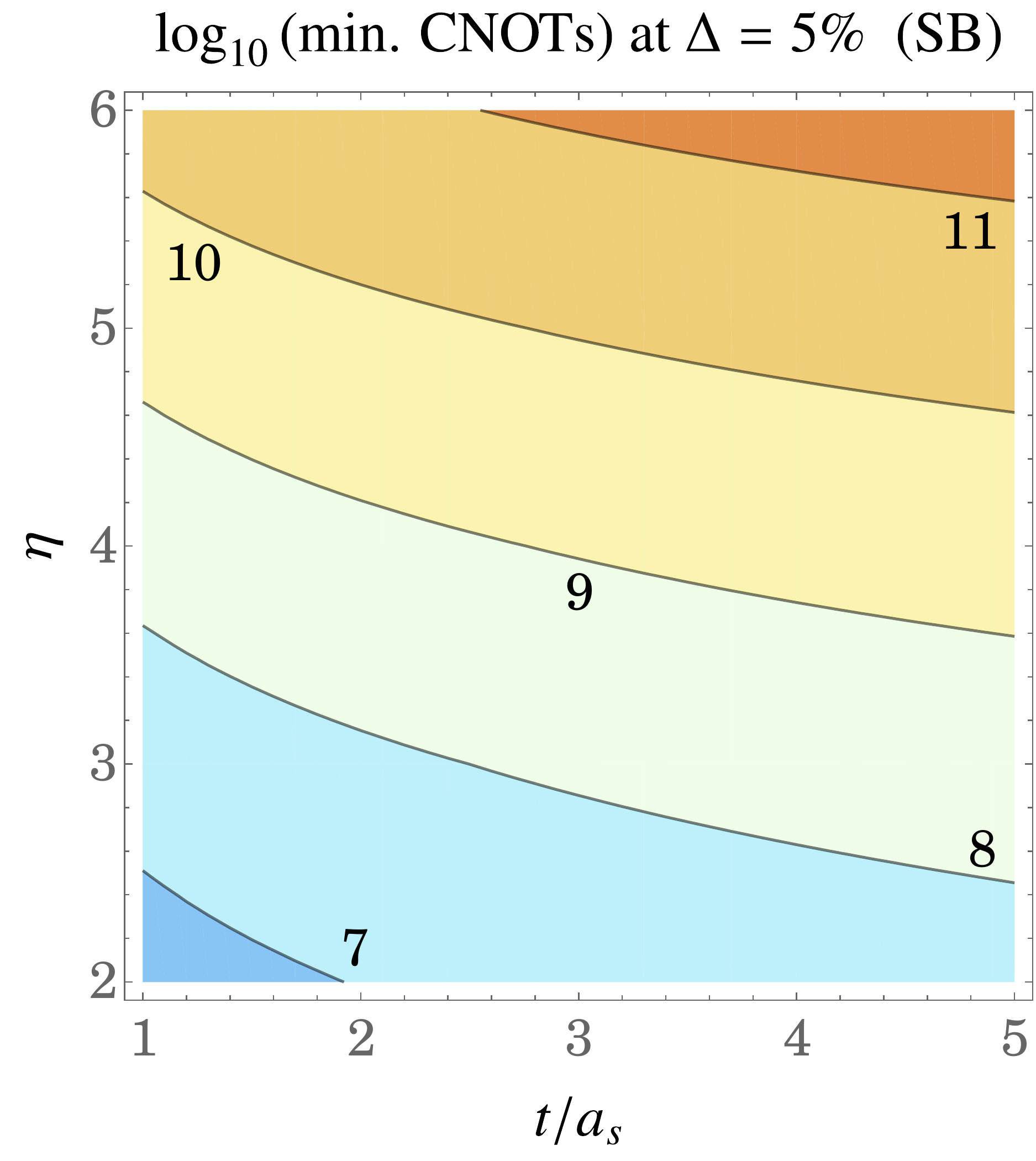}
    \caption{
    }
  \end{subfigure}
  \hfill
  \begin{subfigure}{0.43\textwidth}
    \includegraphics[width=\textwidth]{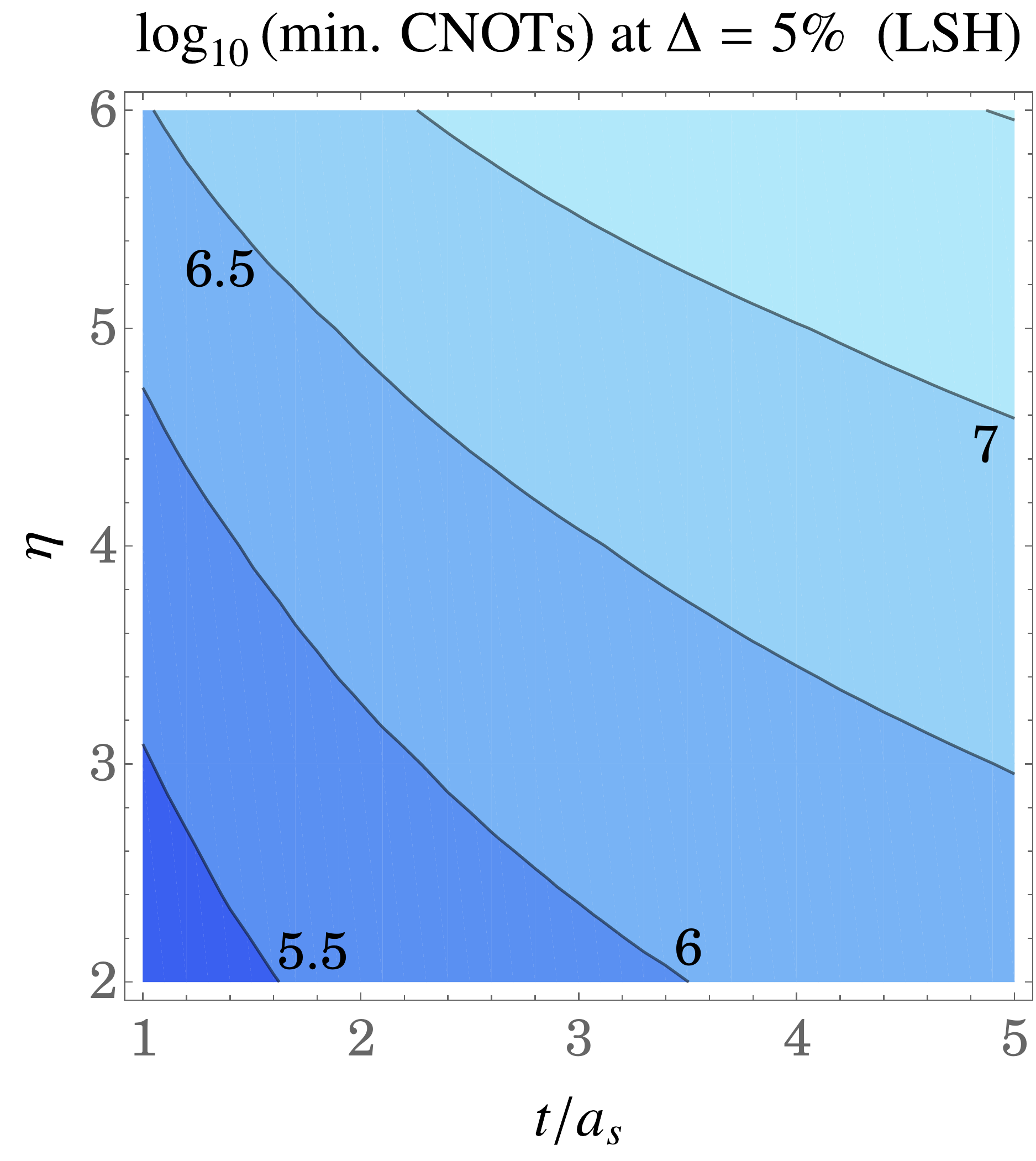}
    \caption{
    }
  \end{subfigure}
\caption{
    CNOT-gate costs at fixed $m/g=1$.
    Other simulation parameters not explicitly shown are $x=1$ and $L=10$.
    \label{fig:CNOT-cost}
    }
\end{figure}

\vspace{0.2 cm}
\noindent
\emph{Finite-precision (Newton's method) error.}---Turning to the simulation in the far-term scenario, in applying Newton's method to the evaluation of inverse-square-root functions, an error due to the finite number of steps, $n$, as well as the practical choice of using fixed $m$-bit precision occurred.
As explained in the previous sections and in Appendix \ref{app:arithmetic}, these truncations contribute an absolute additive error
\begin{align}
    x \theta \left[ 2^n \left(\sqrt{2}-1\right)^{2^n}+2^{2-m}\left(\left(\frac{3}{2}\right)^n-1\right) \right]
\end{align}
for an evolution by time $\theta$ of a single hopping subterm (cf.~Eq.~\eqref{eq:IPropError_SB}).
In the second-order Trotterization, one has $\theta \to T/(2s)$ and the total number of diagonal-function evaluations will be $2s \, \nu^{\rm SB}=16s$.
The total error is, therefore,
\begin{align}
    \label{eq:fullNewtonError_SB}
    \Delta_\mathrm{Newt.} &= x (L-1) T \nu^{\rm SB} \left[ 2^n \left(\sqrt{2}-1\right)^{2^n}+2^{2-m}\left(\left(\frac{3}{2}\right)^n-1\right) \right].
\end{align}

Given a fixed error budget in Newton's method truncation, Eq.~\eqref{eq:fullNewtonError_SB} sets minimal sizes on $n$ and $m$ necessary to meet the error budget.
However, the choices of $n$ and $m$ are not independent, since there are two correlated parameters but only one constraint. Therefore, a choice must be made for how to simultaneously choose $n$ and $m$.
To this end, we propose the following scheme:
given that the vast majority of T gates are spent on implementing function evaluations, 
$m$ and $n$ are chosen such that the T count of $\UfSB$ is minimized (subject to the constraint set by the error budget).
The T count of $\UfSB$ is as stated in Lemma \ref{lem:kineticphasekickSB} and it is a function of $\eta$, $n$, and $m$.
The advantage of this scheme is that it should optimize the overall T count of simulation.
The drawback is that these implicit definitions do not lend themselves to closed-form expressions for $n(x,\eta,L,T,\Delta_\mathrm{Newt.} )$ and $m(x,\eta,L,T,\Delta_\mathrm{Newt.} )$.
In practice, only a few $(n,m)$ pairs have to be considered.

\vspace{0.2 cm}
\noindent
\emph{Synthesis error.}---The total error due to synthesis of $R^Z$ gates, $\Delta_\mathrm{synth.}$, is the sum of individual synthesis errors $\epsilon$ of each rotation.
Thus, the main quantity of interest is the number of $R^Z$ gates associated with each type of subpropagator in the Trotter decomposition.
These have been identified in the text of previous sections, and they are collected for convenience in Table \ref{tab:RzCounts}.
The total number of $R^Z$ gates per second-order Trotter step is thus found to be
\begin{align}
    2 \big[ 2 L + (3\eta+3) (L-1) + 8\times2(m+1) (L-1) \big] &= -6 \eta +6 \eta  L+32 L m+42 L-32 m-38 .
\end{align}
The total synthesis error is this value times $s$ and $\epsilon $:
\begin{align}
    \label{eq:fullSynthesisError_SB}
    \Delta_\mathrm{synth.} &= s \epsilon(-6 \eta +6 \eta  L+32 L m+42 L-32 m-38) .
\end{align}
When there is a predefined error budget for rotation synthesis, Eq.~\eqref{eq:fullSynthesisError_SB} implies a maximum value for $\epsilon$:
\begin{align}
\label{eq:Rz mininum gate error, SB}
    \epsilon(\eta,L,s,m,\Delta_{\mathrm{synth}}) &= \frac{\Delta_\mathrm{synth.}}{s(-6 \eta +6 \eta  L+32 L m+42 L-32 m-38)} .
\end{align}
\begin{table}
\centering
\begin{tabular}{c c }
Routine & $R^Z$ gates\\
\hline
\hline
$e^{-i\theta H^{\rm SB}_M(r)}$ or $e^{-i\theta H^{\rm LSH}_M(r)}$ & 2 \\
$e^{-i\theta H^{\rm SB}_E(r)}$ or $e^{-i\theta H^{\rm LSH}_E(r)}$ & $3\eta+3$ \\
$e^{-i\theta H^{\mathrm{SB}(j)}_I(r)}$ or $e^{-i\theta H^{\mathrm{LSH}(j)}_I(r)}$ & $2(m+1)$ \\
\hline
Full Trotter step (SB)  & $-6 \eta +6 \eta  L+32 L m+42 L-32 m-38$ \\
Full Trotter step (LSH) & $-6 \eta +6 \eta  L+8 L m+18 L-8 m-14$
\end{tabular}
\caption{Summary of the number of the single-qubit $R^Z$ gates involved in implementing the three types of (sub)propagators, and of the complete Trotter-step propagator.}
\label{tab:RzCounts}
\end{table}

\vspace{0.2 cm}
\noindent
\emph{Far-term error-bounded simulation costs.}---We are now ready to provide the full costs of the algorithms for simulating the SU(2) LGT in the Schwinger-boson formulation as a function of model parameters and target accuracy.
Here, $x$, $\Lambda=2^\eta-1$, $\mu$, $L$, and $T$ are taken as given parameters.
$\Delta$ in then introduced to represent the complete error budget of the time-evolution operator, while $0<\alpha_\mathrm{Trot.}<1$ represents the fraction allocated to Trotterization error and $0<\alpha_\mathrm{Newt.}<1$ represents the fraction allocated to Newton-truncation errors.
The fraction of $\Delta$ allocated for $R^Z$ synthesis error is then $\alpha_\mathrm{synth.}=1-\alpha_\mathrm{Trot.}-\alpha_\mathrm{Newt.}$ and must be positive.
The truncation to $\Lambda<\infty$ is an additional known source of error that we do not quantify in this work.

With a target Trotterization error bound of $\alpha_\mathrm{Trot.} \Delta$, the minimal number of Trotter steps $s$ is derived as described above.
Similarly, the target error bound on Newton's method truncation error implies certain values of $n$ and $m$ according to the T-gate optimization scheme outlined above.
Lastly, the $R^Z$ gate accuracy $\epsilon$ is derivable given the error budget of $(1-\alpha_\mathrm{Trot.}-\alpha_\mathrm{Newt.}) \Delta$, once $s$ and $m$ are known.

The derived values of $s$, $n$, $m$, and $\epsilon$ are then inserted into the complete time-evolution cost formulas.
The complete T-gate count is
\begin{align}
    &2s \biggl[-64 L \max (m,2 \eta )-128 L n \max (2 m,4 \eta +2)+64 \max (m,2 \eta )+128 n \max (2 m,4 \eta +2)\nonumber\\
    &-784 \eta ^2-928 \eta +784 \eta ^2 L+928 \eta  L+512 L m^2 n+256 \eta  L m+2560 \eta  L m n+1408 L m n\nonumber\\
    &+256 \eta  L n-64 L n+8 L-512 m^2 n-256 \eta  m-2560 \eta  m n-1408 m n-256 \eta  n+64 n-8 \nonumber\\
    & \hspace{6.25 cm}+ \bigl(-3 \eta +3 \eta  L+16 L m+21 L-16 m-19\bigr) \mathcal{C}_z(\epsilon)
 \biggr].
\end{align}
The ancilla-qubit count is just the ancilla-qubit count of a hopping propagator as it is the largest of all ancillary registers, and was found to be
\begin{align}
    \max (m,2 \eta +2)+27 \eta +3 m n+10 m+4 \eta  n+3 n+15 .
\end{align}
Some example values for the qubit and T-gate counts at given values of parameters are provided in Table~\ref{tab:errorBoundedCosts} of Appendix~\ref{app:cost-tables}. Furthermore, the T-gate count as a function of $x$ and $L$ are plotted in Fig.~\ref{fig:T-cost} for better visualization of the resource requirement toward the continuum and bulk limits. 
\begin{figure}[t!]
  \begin{subfigure}{0.44\textwidth}
    \includegraphics[width=\textwidth]{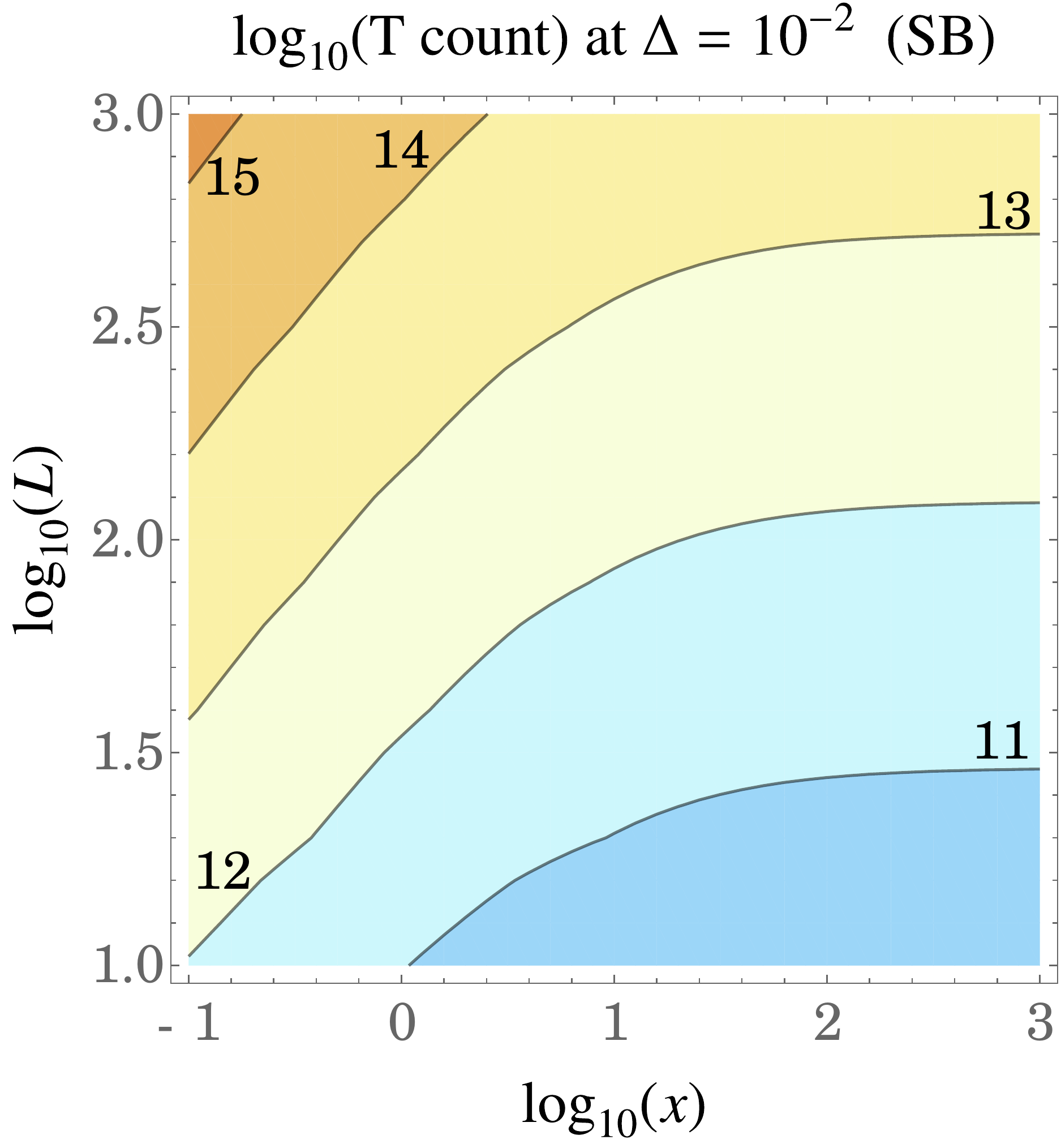}
    \caption{
    }
  \end{subfigure}
  \hfill
  \begin{subfigure}{0.44\textwidth}
    \includegraphics[width=\textwidth]{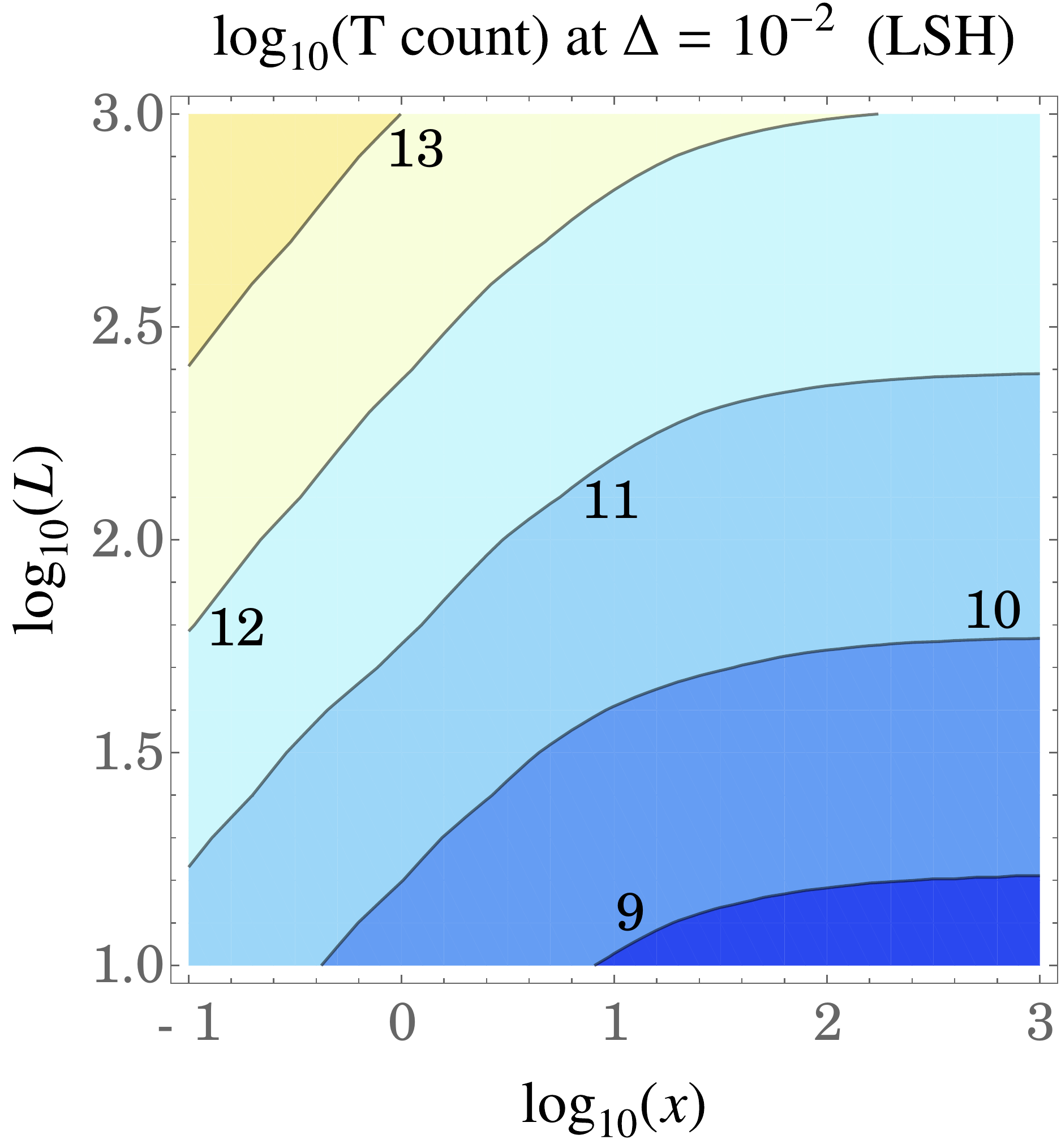}
    \caption{
    }
  \end{subfigure}
  \begin{subfigure}{0.44\textwidth}
    \includegraphics[width=\textwidth]{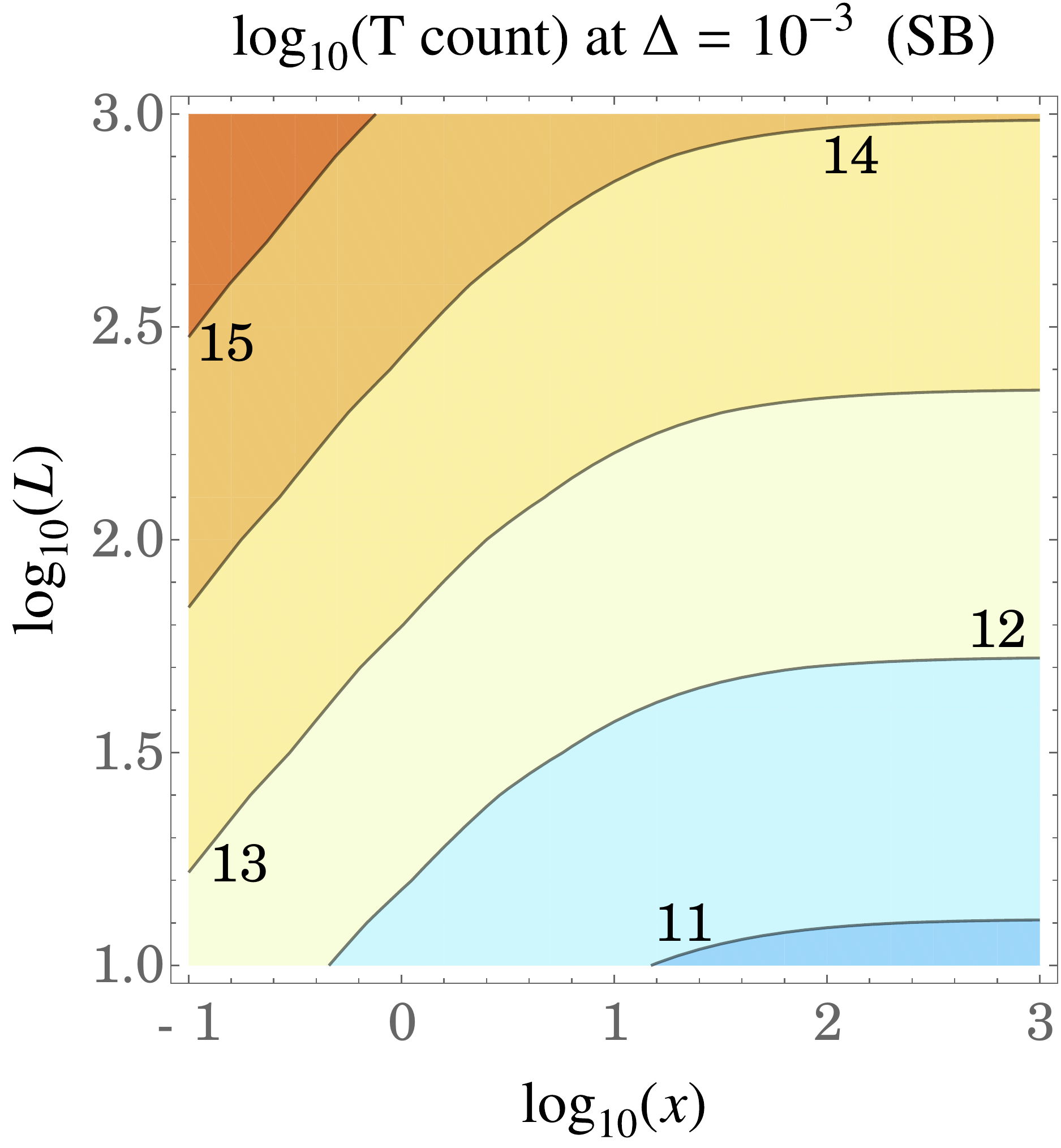}
    \caption{
    }
  \end{subfigure}
  \hfill
  \begin{subfigure}{0.44\textwidth}
    \includegraphics[width=\textwidth]{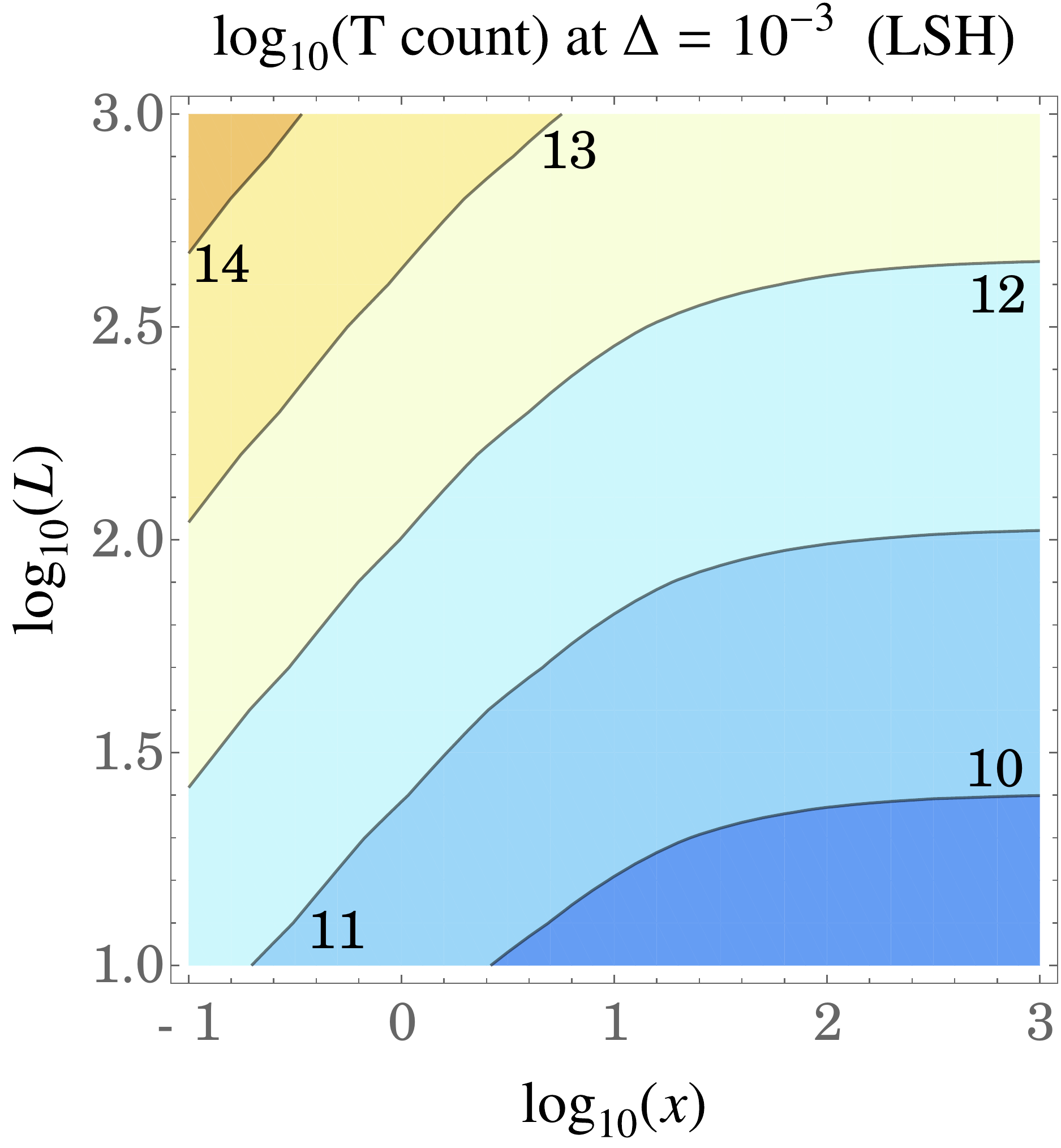}
    \caption{
    }
  \end{subfigure}
\caption{
    T-gate costs at fixed $m/g=1$.
    Other simulation parameters not explicitly shown are $\eta=8$, $t/a_s=1$, $\alpha_{\mathrm{Trot}}=90\% $, $\alpha_{\rm Newt.}=9\% $, and $\alpha_{\rm synth.}=1\% $.
    \label{fig:T-cost}
}
\end{figure}

\subsubsection{Loop-string-hadron formulation}
\noindent
Deriving the total error-bounded simulation costs in the near- and far-term scenarios for the LSH formulation follows closely that for the Schwinger-boson simulation. Without repeating the procedure, we only state the differences:

\vspace{0.2 cm}\noindent
\emph{Trotterization error.}---Instead of function $\rho^{\rm SB}$, one uses $\rho^{\rm LSH}$ derived in Appendix~\ref{eq:2ndTrotter}:
\begin{align}
    \rho^{\rm LSH} (x, \eta, \mu ) \equiv 
    \frac{47 \sqrt{2} x^3}{3} + 2 \Lambda  x^2 + \frac{25 \mu  x^2}{3} + 3 x^2 + \frac{\Lambda ^2 x}{24 \sqrt{2}} + \frac{\Lambda  \mu  x}{3 \sqrt{2}} + \frac{\Lambda  x}{8 \sqrt{2}} + \frac{5 \mu ^2 x}{6 \sqrt{2}} + \frac{\mu  x}{2 \sqrt{2}} + \frac{3 x}{32 \sqrt{2}}
     ,
\end{align}
with $\Lambda=2^\eta-1$.

\vspace{0.2 cm}
\noindent
\emph{Near-term simulation costs.}---The total number of CNOT gates required is:
\begin{align}
2s(L-1)\left(16\times 2^\eta+\frac{33}{2}\eta^2+\frac{42}{2}\eta+53\right),
\end{align}
with $s$ defined in Eq.~(\ref{eq:error-bounded-s_SB}) after substituting $\rho^{\rm SB}$ by $\rho^{\rm LSH}$. Some representative values are tabulated for a range of modest simulation parameters in Table~\ref{tab:errorBoundedNISQCosts} of Appendix \ref{app:cost-tables}. Furthermore, the CNOT-gate count as a function of $\eta$ and evolution time are plotted in Fig.~\ref{fig:CNOT-cost} for better visualization.

\vspace{0.2 cm}
\noindent
\emph{Finite-precision (Newton's method) error}---The total number of diagonal-function evaluations will be $2s \, \nu^{\rm LSH}=4s$, and one needs to replace $x$ with $\sqrt{2}x$ in the expression for $\Delta_{\rm Newt.}$ Eq.~(\ref{eq:fullNewtonError_SB}). 

\vspace{0.2 cm}
\noindent
\emph{Synthesis error.}---Given a different number of $R^Z$ gates for the LSH formulation, the total synthesis error now is:
\begin{align}
    \label{eq:fullSynthesisError_LSH}
    \Delta_\mathrm{synth.} &=  s \epsilon (-6 \eta +6 \eta  L+8 L m+18 L-8 m-14) .
\end{align}
This can be rearranged like Eq.~(\ref{eq:Rz mininum gate error, SB}) to give the maximum per-gate error $\epsilon(\eta,L,s,m,\Delta_{\mathrm{synth}})$ allowed within a given error budget. 
\vspace{0.2 cm}
\noindent
\emph{Error-bounded simulation costs.}---The complete time-evolution T-gate count in the LSH formulation is
\begin{align}
    &2s \biggl[-16 L \max (m,\eta +1)-32 L n \max (2 m,2 \eta +2)+16 \max (m,\eta +1)+32 n \max (2 m,2 \eta +2)\nonumber\\
    &-48 \eta ^2-256 \eta +48 \eta ^2 L+256 \eta  L+128 L m^2 n+32 \eta  L m+320 \eta  L m n+352 L m n+32 L m \nonumber\\
    &+32 \eta  L n-16 L n+80 L-128 m^2 n-32 \eta  m-320 \eta  m n-352 m n-32 m-32 \eta  n+16 n-80 \nonumber\\
    & \hspace{6.9 cm} +\bigl(-3 \eta +3 \eta  L+4 L m+9 L-4 m-7\bigr) \mathcal{C}_z(\epsilon)\biggr] .
\end{align}
The ancilla-qubit count is just the ancilla-qubit count of a hopping propagator, which was found to be
\begin{align}
    2 \eta  n+3 m n+12 \eta + \max (\eta +2,m)+10 m+3 n+18 .
\end{align}
Some example values for the qubit and T-gate counts at given values of parameters are shown in Table \ref{tab:errorBoundedCosts} of Appendix~\ref{app:cost-tables}. Furthermore, the T-gate count as a function of $x$ and $L$ are plotted in Fig.~\ref{fig:T-cost} for better visualization of the resource requirement toward the continuum and bulk limits.

\subsection{Discussion of the results}
The cost analysis of this section resulted in the full simulation resource requirements of the SU(2) LGT in 1+1 D in two formulations. Here, we summarize the key features of the algorithms presented and the findings of this analysis:
\begin{itemize}
    \item[$\diamond$]{\emph{Subdivision of hopping terms.}---%
    One of the key insights of this work is the importance of carefully identifying subterms within a given hopping term that can be circuitized as a whole, the number of which we have referred to as $\nu$.
    Firstly, the value of $\nu$ has a tangible impact on the gate costs because the total number of diagonal-function evaluations is proportional to $\nu$. Therefore, any reductions in $\nu$ will carry over to the gate costs.
    Secondly, as shown in Appendix~\ref{app:commutators}, the second-order Trotter-error commutator bound is a cubic polynomial in the quantity $\nu x$, and thus a lower $\nu$ can dramatically tighten up the error-per-Trotter-step, especially in the weak-coupling (large-$x$) limit. In this regime, the minimal number of second-order Trotter steps (for a fixed error tolerance) scales as $\nu^{3/2} $, so reducing $\nu$ leads to a reduction in the number of required steps.
    }
    \item[$\diamond$]{\emph{Symmetry considerations.}---Aside from the clear impact on gate costs, there is a connection between $\nu$ and local symmetry preservation.
    The limiting case of an algorithm with $\nu=1$ is special in that all local symmetries would have to be preserved (because the hopping terms are gauge invariant cumulatively).
    As $\nu$ increases, the potential for individual subterms, hence the Trotterized hopping propagator, to violate symmetries also increases.
    Indeed, the LSH formulation naturally led to a $\nu=2$ splitting that conserves all local symmetries, despite the Trotterized hopping propagator not being exact.
    In the Schwinger-boson formulation, it was natural to identify a splitting with $\nu=8$, where each individual subterm can be seen to conserve $G_3(r)$ as well as the AGL, but not $G_1(r)$ and $G_2(r)$. This is an improvement compared with a naive $\nu=64$ splitting (where all individual terms in the hopping Hamiltonian are separated).
    
    It should be noted that the exact benefits afforded by conserving local symmetries are not well quantified. For example, while symmetry-breaking errors due to algorithms with bounded errors (such as Trotterization) can drive the evolution to outside the physical Hilbert space, these errors can be systematically eliminated (e.g., by increasing the number of Trotter steps). On the other hand, hardware noise can contribute to decoherence in the simulation, and so it appears that there may be a value in working with formulations such as LSH, and with simulation algorithms such as those presented here, that restrict the evolution to a substantially smaller Hilbert space.
    }
    \item[$\diamond$]{\emph{Functional dependence on bosonic variables.}---%
    Function-evaluation costs generally scale with the size and the number of arguments.
    In the near term, where we have considered naive Pauli decompositions, reducing the number of bosonic inputs called for by a function represents a potential for exponential savings.
    While the Schwinger-boson functions $\DSB$ depend on three bosonic quantum numbers ($p$, $q$, and $p'$), the LSH counterparts     $\DLSH$ were shown to depend on only one bosonic variable ($p$), given the simplifications offered by gauge-invariant local DOFs and the AGL (see Appendix~\ref{app:LSHSubtermSimplification}).
    This led to relative CNOT-gate count per Trotter step of roughly $8^{\eta}+4\eta^2$ (Schwinger boson) to $2^\eta+\eta^2$ (LSH).
    
    In the far term, in addition to the number of bosonic inputs, generally the simplicity of the functional expressions is also important.
    In the LSH (Schwinger-boson) formulation, the diagonal function $\DLSH$ ($\DSB$) could be expressed in terms of the ratio of two integers (products of integers), $p+1+n$ and $p+1+n'$ ($p \, q$ and $(p+p')(p+p'+1)$).
    This ultimately requires the inverse square root of a product of two (four) integers that was called $g^{\rm LSH}$ ($g^{\rm SB}$).
    Lemma~\ref{lem:newtit} confirms that the costly Newton iterations associated with the inverse square root scale linearly with the size of $g$, so the less bits needed for $g$, the better.
    Note that without the simplification of Appendix \ref{app:LSHSubtermSimplification}, the LSH diagonal function would have ultimately involved four integers and likely had a comparable cost to the Schwinger-boson diagonal functions.
    }
    \item[$\diamond$]{\emph{Function-evaluation choices.}---%
    There are numerous possibilities to how one may obtain the irrational value $\sqrt{p/q}$ for integers $p$ and $q$, but the approach that multiplies $q$ with $1/\sqrt{p \, q}$ requires only a single approximate function evaluation: the inverse square root.
    It is expected that less function evaluations will generally lead to lower costs, although there are other factors to consider, such as the sizes of arguments and the rates of convergence for different functions. 
    Additionally, the LSH diagonal functions may be better suited to other approximation methods since alternatively it can be written as
    \begin{align}
    \DLSH(p,n,n')=\sqrt{1+\frac{n-n'}{p+1+n'}},
    \end{align}
    where $|(n-n')/(p+1+n')|\leq1$ provides a ``small parameter''.
    Efficient schemes for evaluating the square root of one plus a small number in the far term may improve on the Newton's method approach taken in this paper.
    This perspective could especially be helpful in the near-term, e.g., algorithms that allow for some arithmetic manipulations but exclude full-blown Newton's method iterations that are costly.
    }
    \item[$\diamond$]{\emph{Quantitative cost comparisons.}---According to the results of Sec.~\ref{sec:bounds}, the time evolution in the SU(2) LGT in 1+1 D for a range of parameters suitable for performing continuum- and bulk-limit extrapolations could be viable with $\sim 10^{7}$ (high-fidelity) CNOT gates. Over a wide range of realistic parameters, the LSH formulation reduces the minimum number of near-term Trotter steps by about 3-4 times. At low (high) cutoffs $\eta \leq 5$ ($\eta\geq10$), the LSH formulation gives a $\sim 0.4\times4^\eta$-fold ($\sim4^\eta$-fold) reduction in CNOT-gate count per Trotter step over a wide range of parameters. Far-term time evolution could be viable with $\sim 10^{12}$ T gates. Finally, over a wide range of parameters, the LSH formulation gives a $\sim20$-fold reduction in the T-gate count over the Schwinger-boson formulation.
    }
\end{itemize}
%

\section{Conclusions and outlook
\label{sec:conclusion}}
\noindent
Quantum simulation presents one of the most promising applications of digital quantum computers. As a result, research into algorithm and hardware technology that would enable scalable quantum simulation of physical systems continues to be a significant endeavor in the field of quantum information sciences and its intersection with other disciplines. With this motivation, we have provided a scalable algorithm that enables simulation of a class of physical Hamiltonians in a straightforward and systematic manner. The Hamiltonians studied are those involving operators that couple the quantum numbers of multiple (spin, bosonic, and/or fermionic) degrees of freedom. The algorithm works by diagonalizing the time-evolution operator in a chosen basis based on an algebraic approach rooted in singular-value decomposition. Strategies for simulating the diagonal functions were then presented. The algorithm is applied to a problem in nuclear and high-energy physics, namely simulating the non-Abelian SU(2) lattice gauge theory in 1+1 dimensions in two different formulations. The analysis of this work provides explicit circuit decompositions for the time-evolution operator in this theory using product formulas, and considers the near- and far-term scenarios that require different resource optimizations. The full error budget of the algorithms has been identified, and a comparative study of the two formulations are enabled. Such a study sets the standard for future algorithmic progress in gauge-theory simulations of relevance to the Standard Model of particle physics such as the SU(3) LGT in 3+1 D, and other physical models.

To conclude this work, we summarize some of our findings in relation to gauge-theory simulations, compare and contrast our approach with the existing developments, and remark on future directions and extensions:
\begin{itemize}
\item[$\diamond$]{\emph{Summary of pros and cons of the two formulations of the 1+1-D SU(2) theory.}---The LSH formulation offers significant advantages over the Schwinger-boson formulation, and likely other formulations, as concluded in Ref.~\cite{Davoudi:2020yln} for classical Hamiltonian-simulation methods as well. LSH, being an alternative avenue to the same gauge-invariant dynamics contained by the other formulations, shares many features with Schwinger bosons, such as site-local DOFs, the decomposition of hopping terms into multiple structurally-identical and non-commuting subterms, and the need to deal with costly square-root functions. The important distinction is in simulating the hopping terms.
In the near term, the minimal number of bosonic variables involved with each hopping operator makes the Pauli decomposition substantially cheaper.
In the far term, the relative simplicity of the arguments to the square root means cheaper numerical evaluations.
In either case, the ability to Trotterize a hopping term with $\nu^{\rm LSH}=2$ instead of $\nu^{\rm SB}=8$ subterms helps to reduce costly subroutines and make better use of the error budget.
Lastly, for the 1+1-dimensional theory, LSH requires fewer qubits than either of the Schwinger-boson and the angular-momentum formulations, with a gain in simplicity of the interaction terms.
Finally, the LSH Trotterization exactly conserves all local symmetries---a feature that has not been found for either the angular-momentum (e.g., the work of Ref.~\cite{Kan:2021xfc}) or Schwinger-boson (this work) formulation.
}

\item[$\diamond$]{\emph{Comparison with another algorithm for the SU(2) theory.}---The formulations in this paper have been analyzed rather differently than how the pioneering work of Ref.~\cite{Kan:2021xfc} proposed simulating  the Kogut-Susskind formulation in the standard angular-momentum basis. Our algorithms incorporate the steps necessary to eliminate mixing (``periodic wrapping'') between the upper and lower cutoff states---steps that reduce the errors in time evolution. Furthermore, the hopping-term simulation algorithm of Ref.~\cite{Kan:2021xfc} can be identified with a $\nu=64$ splitting that increases the number of costly diagonal-function evaluations, enlarges the Trotter error, and breaks Gauss's laws associated with $G_1(r)$ and $G_2(r)$ (and we suspect $G_3(r)$ as well, via wrap-around effects in the $m^L$ and $m^R$ quantum numbers).

While we have not applied our method to the angular-momentum formulation to allow a side-by-side comparison of resource requirements compared with Ref.~\cite{Kan:2021xfc}, we remark that a simple reorganization of the angular-momentum  basis could be done to cast this formulation into a form that is essentially that of Schwinger bosons---but involving one less bosonic mode per link.
With that reorganization, changing between the angular-momentum and the Schwinger-boson formulations is largely a matter of redefining quantum numbers and reshuffling where some subroutines are executed within circuits. Explicitly, instead of storing $(J,m^L,m^R)$ (with $m^L$ and $m^R$ shifted to be non-negative), one could store $(J,J+m^L,J+m^R)$. This is slightly more space-efficient, and the Clebsch-Gordan coefficients are functions of $J\pm m^L$ and $J\pm m^R$ anyway ($J-m^L$ and $J-m^R$ would still be derived quantities). But the quantities $J+m^L$ and $J+m^R$ are identical to two of the four Schwinger-boson occupation numbers on the link. We suspect that the intrinsic operational differences between the two digitization schemes, when considered as a whole, are negligible because the dominating subroutines in each formulation will be so similar. The only concrete drawback to the Schwinger-boson formulation is the greater number of logical qubits than the angular-momentum basis (reorganized or not).
}

\item[$\diamond$]{\emph{On the generalization of the algorithm to higher dimensions.}---In two or three space dimensions, one must simulate magnetic interactions, conventionally expressed as trace of four link operators around a square plaquette. 
Let us consider the Schwinger-boson formulation to demonstrate the generalization of the algorithms of this work to simulating plaquette operators, and further remark on how it compares with the algorithm of Ref.~\cite{Kan:2021xfc}. Generalization to higher dimensions for the LSH formulation can be worked out similarly.

Within the Schwinger-boson formulation, the plaquette operator expands out to $2^7$ distinct terms, which can be understood as follows. A factor $2^4$ comes from the 4 choices of indices $a, b, c, d$ in the product of four link operators, $U_{ab} U_{bc} U_{cd} U_{da}$. Another factor of $2^4$ comes from the choice of $J \to J \pm 1/2$ on each of the four links. Finally, a factor of $1/2$ is required to avoid double-counting terms that are related by Hermitian conjugation.
Now, each of the $2^7$ terms counted above involves a product of eight harmonic-oscillator ladder operators acting on distinct modes around the plaquette.
The SVD algorithm of this work requires splitting only one of these ladder operators to operators acting on disjoint even and odd Hilbert spaces of the corresponding quantum number. The diagonalization procedure will follow straightforwardly after that, as demonstrated in the example of a plaquette interaction in the U(1) case in Sec.~\ref{sec:SVD}. Compared to the hopping interactions, therefore, the plaquette interactions not only involve a larger number of diagonal phases to be evaluated, but also the phases are more complex functions and depend on occupation numbers of twelve modes (three per link). 

Finally, we remark that the algorithm of Ref.~\cite{Kan:2021xfc}, in contrast, splits each of the eight ladder operators to even and odd parts, hence requiring $2^8$ diagonal phases to be evaluated for each and every one of the $2^7$ terms. This, and the increased Trotter error given such a dramatic splitting to non-commuting terms, contribute to a substantial increase in the computational cost compared to the method of this work. Further details are left to future work~\cite{DavoudiStryker2023}.
}

\item[$\diamond$]{\emph{Classical pre-processing for near-term benefit.}---Research will need to continue to offer more efficient function-evaluation schemes, perhaps with the use of hybrid classical-quantum routines, to make the simulation of non-Abelian gauge theories more suitable for near-term quantum computing. Truncating Pauli decompositions by their coefficient size appears to not be an efficient near-term cost-reduction scheme for the functions encountered in this work.
This is evident from the slow reduction in the spectral-norm error as more diagonal Pauli operators are kept until almost all the operators are being simulated. The same slow convergence was observed in the alternate scheme that truncates the size of the input registers into the diagonal functions before a complete Pauli decomposition is performed.
There may exist more sophisticated truncation schemes that converge faster than ordering the Pauli operators by decreasing coefficient size, and those need to be explored in future work.
    
Another strategy that avoids calculating the non-trivial non-Abelian functions is to instead hard-code in the circuits the value of the Clebsch-Gordan coefficients associated with all allowed transitions. This approach requires as many controlled operations as there are for possibilities of such transitions, and the amplitude for those transitions is encoded in the rotation angles that depend on the corresponding Clebsch-Gordan coefficients. This approach is applied to the case of the pure SU(3) LGT in 2+1 D in Ref.~\cite{Ciavarella:2021nmj}. Future work should determine a cost comparison between this and other near-term approaches, such as those outlined here for the similar theories.
}

\item[$\diamond$]{
\emph{Comparison to \textrm{U(1) and remarks on simulation in the group-element basis}.}---%
Compared to the Schwinger model, simulating the SU(2) model is substantially more complicated due to the properties of the non-Abelian link operator.
In the electric basis, whose truncation and digitization are better understood, one must deal with Clebsch-Gordan coefficients, in one guise or another, and their irrational dependence on electric quantum numbers is a major departure from the trivial ones that characterize all non-vanishing matrix elements of the U(1) link operator.
In either the near- or far-term scenario, the vast majority of gate costs are confirmed to result from the handling of the SU(2) link operator in the hopping terms.
In principle, one might try avoiding the Clebsch-Gordan coefficients by working in the group-element basis, but truncation schemes are less developed when the gauge group is infinite as with the SU(2) theory, see e.g., Refs.~\cite{Hackett:2018cel,Alexandru:2019nsa,Hartung:2022hoz} for recent studies. General algorithms exist for simulating gauge theories in the group-element basis~\cite{Lamm:2019bik} but more work is needed to augment those studies with a comparable level of algorithmic detail as the present work for the SU(2) theory.
Furthermore, parametrization of the SU(2) group manifold could require function evaluations that may still ultimately dominate the gate cost. Finally, the electric Hamiltonian involves more complicated operator structure in the group-element basis, and will be more costly to implement than in the electric-field (irrep) basis. One could change the basis between group-element and irrep bases to only require implementing diagonal operators, nonetheless a discrete non-Abelian quantum Fourier transform for the digitized truncated gauge group must first be developed.
}
\item[$\diamond$]{\emph{Remarks on post-Trotterization schemes.}---While this work considered product-formula algorithms for implementing the time evolution, there are two other major classes of simulation algorithms based on linear combinations of unitaries (LCU) and quantum signal processing (QSP). An in-depth comparison of quantum-resource requirements for simulating the Heisenberg model using each of these quantum algorithms in Ref.~\cite{childs2018toward} revealed that product-formula algorithms, according to their ``true performance'', will likely outperform LCU- and QSP-based algorithms in simulating spin Hamiltonians. The ``true performance'' was obtained by extrapolating the exact performance deduced from small system sizes, and so may not offer a rigorous performance guarantee. This performance is also considered a ``worst case'' in the sense that it uses no assumptions about the input state. Assumptions about the input state and symmetries are indeed shown to improve the error bounds~\cite{Su:2020gzf, yi2022spectral, an2021time, Zhao:2021gtg}. The toolbox one would need to access the ``true performance'' of product-formula simulations (whether rigorously or heuristically) via incorporating all the physically-motivated assumptions needs to be developed for all interesting problems including gauge theories. 

The resource requirements of LCU- and QSP-based methods scale with the norm of the Hamiltonian, which is a disadvantage in simulating LGTs that exhibit unbounded Hamiltonians, requiring significant resources towards the continuum limit. One generic advantage of product formulas is that their error depends on the commutation among Hamiltonian terms as opposed to their norm, which is bounded by locality of the Hamiltonian, as was shown in this paper for the case of the SU(2) LGT and in previous work~\cite{Shaw:2020udc, Kan:2021xfc}. Hybrid approaches that combine two or more simulation algorithms to achieve better performance have been recently proposed~\cite{Rajput:2021khs}. Another desirable feature of product-formula-based algorithms is that they are simply an approximation to the exact time evolution, and will likely be of broader applicability in post-classical but pre-fault-tolerant era of quantum simulation where analog or hybrid digital-analog simulators will lead the way, see e.g., Ref.~\cite{Davoudi:2021ney} for an example of a hybrid approach in the context of LGTs.
}
\item[$\diamond$]{\emph{Remarks on systematic uncertainties and the choice of algorithms' error tolerance.}---The error budget identified in this study is associated with the approximations made in the time digitization (Trotterization), function evaluation, and gate synthesis, but the value of total error tolerance in the select examples shown was chosen somewhat arbitrarily. Nonetheless, the value of the error tolerance cannot be chosen without regard for other systematic uncertainties in the simulation. For example, as one considers reproducing physical values in future quantum simulations of gauge theories of relevance to nature, one needs to account for uncertainties arising from finite discretization of space, finite extent of physical volume, finite truncation of bosonic fields, and other systematics associated with choices of input parameters that may deviate from those in nature. As a result, choosing an ultra-high algorithmic accuracy, which leads to substantial increase in the cost, may not be necessary given that other systematics will dominate the error. As a result, one ultimately needs a holistic approach to the resource-requirement analysis of LGTs, in which all sources of systematic uncertainties are taken into account as one takes the continuum infinite-volume limits, see e.g., Refs.~\cite{Carena:2021ltu,Ciavarella:2020vqm,Briceno:2020rar,Davoudi:2020yln,Ciavarella:2021nmj,Hackett:2018cel,Alexandru:2019nsa,Hartung:2022hoz} for initial discussions on quantifying and controlling such systematics in the context of quantum simulation of lattice field theories.
}
\end{itemize}
%

\section*{Acknowledgments}
\noindent
We acknowledge fruitful discussions with Andrew Childs, Martin Savage, and Nathan Wiebe. We further acknowledge valuable discussions on a range of topics at two Quantum Simulation of Strong Interactions (QuaSI) workshops held virtually at the InQubator for Quantum Simulation (IQuS) at the University of Washington, Seattle in Spring and Summer of 2021, which solidified our vision for this work.

ZD, AFS, and JRS were supported by the U.S. Department of Energy's Office of Science Early Career Award, under award DE-SC0020271, for theoretical developments for simulating lattice gauge theories on quantum computers. ZD and JRS were further supported by the U.S. Department of Energy's Office of Science, Office of Advanced Scientific Computing Research, Accelerated Research in Quantum Computing program award DE-SC0020312, for algorithmic developments for quantum simulation. AFS was further supported by the Lanczos Fellowship of the University of Maryland and the National Institute for Standard and Technology's Joint Center for Quantum Information and Computes Science (QuICS), and National Science Foundation's Graduate Research Fellowship Program (GRFP). The IQuS workshops were supported by the U.S. Department of Energy's Office of Science, Office of Nuclear Physics, InQubator for Quantum Simulation, under award DE-SC0020970.

\newpage
\begin{appendices}
\section{Arithmetic algorithmic routines
\label{app:arithmetic}}
\noindent
This appendix contains explicit circuit constructions for various arithmetic routines used throughout this paper. While some of these routines are known, some are application or extension of the related existing results. Multiplication and addition are among the primary operations in the circuits presented in the main text. There are many algorithms for these operations, the choice of which depends on the goal of analysis, e.g., minimizing T-gate or ancilla-qubit count. We choose algorithms with low T-gate counts and employ them in subsequent arithmetic evaluations. These routines are described in the following lemmas, along with a near-term implementation of incrementer that is ancilla free. The rest of the appendix presents a step-by-step guide to evaluating the Schwinger-boson and the LSH diagonal phases via Newton's method. 

\begin{lemma}[Ancilla-free incrementer] There exists a quantum circuit that increments by one an $\eta$-qubit integer in the computational basis with a CNOT-gate cost $2\eta (\eta - 1)$ and without using any ancilla qubit.
\label{lem:inc-near}
\end{lemma}
\begin{proof}
A straightforward inspection of the circuit shown in Fig.~\ref{fig:inc-near-term} shows that using a QFT circuit and its inverse, plus a number of $Z$ rotations, an $\eta$-qubit integer can be incremented by one. Since each QFT (and its inverse) on an $\eta$-qubit register can be implemented with $\eta (\eta - 1)$ CNOT gates~\cite{nielsen2002quantum}, the CNOT-gate count  of the incrementer is $2\eta (\eta - 1)$. 
\end{proof}

\begin{figure}[h!]
\centering
    \adjustbox{width=0.45\linewidth,center}{%
    \begin{tikzcd}[row  sep={9mm,between  origins},transparent]  
& \gate[6,nwires={3}]{{\rm QFT}^{-1}} & \gate{\exp \bigl( i \frac{ \pi }{2^\eta}2^{\eta-1} Z \bigr)} & \gate[6,nwires={3}]{{\rm QFT}} & \qw \\
&  &\gate{\exp \bigl( i \frac{ \pi  }{2^\eta}2^{\eta -2} Z \bigr)} &  &   \qw \\
\raisebox{5pt}{\vdots} & &\raisebox{5pt}{\vdots} & &\raisebox{5pt}{\vdots} \\
&  & \gate{\exp \bigl( i \frac{ \pi  }{2^\eta}2^2 Z \bigr)} &  &\qw \\
&  & \gate{\exp \bigl( i  \frac{ \pi  }{2^\eta}2 Z \bigr)} &  &\qw \\
&  & \gate{\exp \bigl( i   \frac{ \pi  }{2^\eta} Z \bigr)} &  &\qw
\end{tikzcd} 
    }
\caption{
The circuit implementation of the ancilla-free incrementer on an $\eta$-qubit register based on the known quantum Fourier transform (QFT) circuit~\cite{Shaw:2020udc}.}
\label{fig:inc-near-term}
\end{figure}
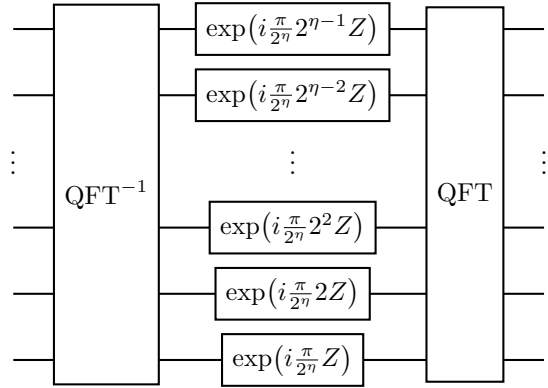
\begin{lemma}[Incrementer by a fixed integer in the fault-tolerant model] Assume $\ket{x}$ is a computational basis vector encoded in a register of $\eta$ qubits, i.e., $x\in \{0,\ldots,2^{\eta}-1\}$, and $y$ is a fixed integer such that $y \in \{0,\ldots,2^{\eta}-1\}$. There exists a quantum circuit that maps $\ket{x} \to \ket{x+y \mod 2^\eta}$ with a T-gate cost of $4\eta-8$ and an ancilla-gate cost of $\eta-3$. Additionally, an extra ancilla qubit may be needed for the execution of all the Toffoli gates.
\label{lem:inc-far}
\end{lemma}
\begin{proof}
A circuit description for incrementation by a fixed integer can be obtained via slight modifications of the addition circuit presented in Ref.~\cite{gidney2018halving}. A proof is provided below for consistency of the presentation. For examples, the reader can refer to Fig.~\ref{fig:inc-far-term}.

If $x$ is the number being incremented, the output of the inceremeter circuit should obtain in the $k$th-significant bit
\begin{equation}
    x_k \oplus y_k \oplus \text{carry}_{k},
\label{eq:kth-bit-inc}
\end{equation}
where $\text{carry}_{k} = x_{k-1}y_{k-1} \oplus (x_{k-1}\oplus y_{k-1})\text{carry}_{k-1}$, and $\text{carry}_{0} = 0 $. To construct such a circuit, one first computes each $\text{carry}_{k}$ for $k=2,\ldots,\eta-2$ into an ancillary qubit. Note that $\text{carry}_{0} = 0$, and $\text{carry}_{1} = x_0 y_0$ which is either $x_0$ or zero, so $\text{carry}_0$ and $\text{carry}_1$ do not need to be computed. Furthermore, the most significant bit of the output can be evaluated from $\text{carry}_{\eta-2}$ and $x_{\eta-2}$, thus $\text{carry}_{\eta-1}$ does not need to be computed and stored.

The value of $\text{carry}_{k}$ for $k\geq 2$ can be computed as following. If $y_{k-1} = 0$, using a logical AND operation, the value of $x_{k-1}\text{carry}_{k-1}$, which is equal to $\text{carry}_{k}$ in this case, is outputed to the ancilla qubit assigned to $\text{carry}_{k}$. If $y_{k-1} = 1$, first $\bar x_{k-1}\text{carry}_{k-1}$, with $\bar x_{k-1}$ being the binary negation of $x_{k-1}$, is outputed to the ancilla qubit, then $x_{k-1}$ is added to its value using a CNOT gate, obtaining $\text{carry}_k$. Computing all $\text{carry}_{k}$ sequentially, starting from $k=2$ up to $k=\eta-2$, requires $\eta - 3$ logical AND computations, $\eta - 3$ logical AND uncomputations, and $\eta -3$ ancilla qubits.

Before the value of the carry qubits are reset by logical AND uncomputation, they are used to calculate the $k$th bit of the incremented output according to Eq.~(\ref{eq:kth-bit-inc}). This requires a CONT gate involving $\text{carry}_k$ qubit and the qubit holding either $x_k$ or its binary negation depending on the value of $y_k$. This outputs in the qubits originally assigned to $x_k$ for $k=0,\cdots,\eta-2$ the corresponding incremented values in the end. For the most significant bit associated with $k=\eta -1$, one does not need to store $\text{carry}_{\eta-1}$ in a designated ancilla qubit since $x_{\eta-1}\oplus y_{\eta-1} \oplus \text{carry}_{\eta-1}$ can be evaluated using a Toffoli gate controlled upon the values of $x_{\eta-2}$ (or its binary negation) and $\text{carry}_{\eta-2}$ (plus an additional CNOT gate if needed). Examples of incrementing a five-qubit integer by one and three are shown in the left-hand side of Fig.~\ref{fig:inc-far-term}.\footnote{This addition circuit and its $\eta$-qubit generalization perfomrs modular addition. For the regular addition, one simply pads $\ket{x}$ with an extra qubit set in zero in the left, and performs the $\eta+1$ modular fixed-number addition.}

The T-gate cost the circuit can be obtained by noting that as proposed by Gidney~\cite{gidney2018halving}, uncomputing logical ANDs can proceed via a T-gate-free scheme that relies on the measurement of the ancilla qubits. Computing each local AND amounts to performing a Toffoli gate, and the Toffoli-gate compilation of Ref.~\cite{jones2013low} requires 4 T gates and a single ancilla qubit that can be reused. Adding the T-gate and ancilla-qubit counts gives the total cost stated in the Lemma.
\end{proof}

\begin{figure}[t!]
\centering
    \includegraphics[scale=0.675]{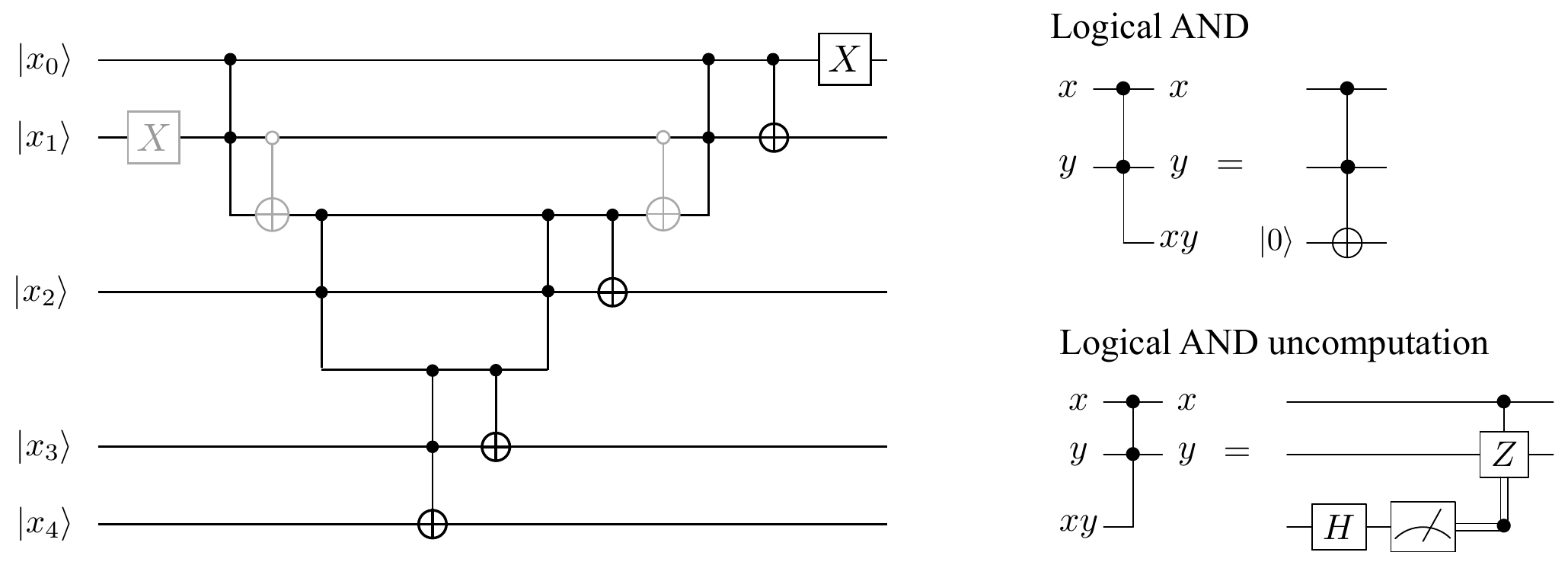}
\caption{The circuit implementation of the incrementer by one (by three when including the gates in gray) on a five-qubit register in the fault-tolerant scenario is depicted in the left, with $x_0,~\cdots, x_5$ representing the binary representation of an integer from the least to the most significant digit. The wire corner denotes a logical AND, that is equivalent to a Toffoli gate with the ancilla qubit initialized in $\ket{0}$ as shown in the left. The uncomputation of logical AND proceeds by measuring the ancilla qubit as shown, hence requiring no T gates. The circuits are adopted from Ref.~\cite{gidney2018halving} upon necessary modifications.
}
\label{fig:inc-far-term}
\end{figure}

\begin{lemma}[Multiplication of binary numbers in the fault-tolerant model]
Let $a$ and $b$ be integers such that $0 \leq a < 2^m $, and $ 0 \leq b < 2^n $.
Then there exists a quantum circuit that performs
\begin{equation}
    \ket{a}\ket{b}\ket{0}^{\otimes( m+n + \gamma)} \to \ket{a}\ket{b}\ket{a \, b}\ket{0}^{\otimes \gamma }
\end{equation}
(where $0\leq a\,b < 2^{m+n}$), using at most $\gamma=2n$ ancillary qubits and $ 8 m n - 4n
$ T gates. Additionally, an extra ancilla qubit may be needed for the execution of all the Toffoli gates.
\label{lem:multi}
\end{lemma}

\begin{proof}
The basic idea is to compute $a \, b$ using `schoolbook multiplication' in the manner suggested by the following identity:
\begin{align*}
    a \, b = \left( \sum_{i=0}^{m-1} 2^i a_i \right) b = a_0 \cdot b + 2 a_1 \cdot b + \cdots + 2^{m-1} a_{m-1} \cdot b.
\end{align*}
The above corresponds to $m-1$ applications of an $n$-bit addition subroutine.

For addition, one may use the T-gate-count-improved variation of the in-place, ripple-carry adder~\cite{Cuccaro:2004xxx} as proposed by Gidney in Ref.~\cite{gidney2018halving}. The adder performs
\begin{equation}
    \ket{0}\ket{c}\ket{d}\ket{0}^{\otimes \eta} \to \ket{c+d}\ket{d}\ket{0}^{\otimes\eta}, 
\end{equation}
where $c$ and $d$ are $\eta$-bit numbers, at a cost of $4\eta$ T gates and $\eta$ workspace ancilla qubits. Note that since Gidney's adder performs modular addition, we have padded $\ket{c}$ with a zero register from the left to effectively obtain a regular addition by performing an $(\eta+1)$-bit modular addition.

The $m+n$ bits which will ultimately contain $\ket{a\, b}$ will be referred to as the `output register.'
The algorithm is initiated by copying $b$ to the $n$ least significant bits of the output register, conditioned on the value of $a_0$:
\begin{equation}
    \ket{a}\ket{b}\ket{0}^{\otimes (m+n+\gamma)} \to \ket{a}\ket{b}\big(\ket{0}^{\otimes m}\ket{a_0\cdot b}\big)\ket{0}^{\otimes \gamma}.
    \label{eq:step0}
\end{equation}
Conditionally copying an $\eta$-qubit register to another $\eta$-qubit register calls for $\eta$ Toffoli gates. 

Schoolbook multiplication then proceeds with $m-1$ rounds of the following procedure, iterating $j$ from $1$ to $m-1$:
\begin{enumerate}
    \item Copy $b$ to an auxiliary register of size $n$, conditioned on the value of $a_j$.
    \item Use the $n$-bit, in-place adder to add the value of the ancillary register from Step 1 to the leading $n$ bits of the currently-computed output string according to
    \begin{eqnarray}
      && \ket{a}\ket{b} \ket{0}^{\otimes (m-j)}\Ket{\Sigma_{k=0}^{j-1}(a_k \cdot b)2^k} \ket{a_j \cdot b} \ket{0}^{\otimes \gamma-n}
       \nonumber\\
      &&  \to \ket{a}\ket{b} \ket{0}^{\otimes (m-j-1)}\Ket{\Sigma_{k=0}^{j}(a_k \cdot b)2^k} \ket{a_j \cdot b} \ket{0}^{\otimes \gamma-n}.
    \end{eqnarray}
    \item Undo Step 1.
\end{enumerate}
After round $j=m-1$, the output register will indeed hold the value $a\,b$. 

The ancillary qubit count of $\gamma = 2n$ is obtained as follows:
Step (1) calls for $n$ ancillary qubits, which temporarily hold logical ANDs until their uncomputation in Step (3).
The adder of Step (2) calls for $n$ ancilla qubits, which are also reset for use in the next iteration.
Thus, $n+n=2n$ ancillary qubits suffice to execute the complete multiplication algorithm.

The T-gate count is a combination of $4n$ for the Toffoli gates in the initialization operation of Eq.~(\ref{eq:step0}),
plus $(m-1) \times 4n$ for the Toffoli gates that implement all applications of Step 1 (noting that uncomputing temporary logical ANDs in Step 3 does not require any T gates according to Ref.~\cite{gidney2018halving}),
plus $(m-1)\times 4n $ for all applications of the addition circuit.
This results in the total T count of $8mn- 4n$, as stated.
\end{proof}
\begin{lemma}[Newton-iteration cost for an approximate evaluation of $1/\sqrt{g}$ in the fault-tolerant model]
\label{lem:newtit}
Let $g$ be any $k$-bit positive binary integer and $y_n$ be a positive $w$-bit binary number, generated by the recurrence relation 
\begin{eqnarray}
y_{n+1} = \frac{y_n(3 - y_n^2 g)}{2},
\label{eq:iteration}
\end{eqnarray}
with $\frac{|y_0-1/\sqrt{g}|}{1/\sqrt{g}} < 1.$ The circuit depicted in Fig.~\ref{fig:newton-iteration},
which implements
\begin{equation}
    \ket{g}\ket{y_n}\ket{0}^{\otimes (
    3w + k+1+
\beta)} \to \ket{g} \ket{y_n} \ket{y_{n+1}} \ket{0}^{\otimes \beta},
\end{equation}
may be constructed using
$$32w^2+40kw+4k+8w-8\max(k,2w)-12
$$ T gates and $$\beta=9w+3k+3
$$ auxiliary qubits. Additionally, an extra ancilla qubit may be needed for the execution of all the Toffoli gates.
\end{lemma}
\begin{proof}
With the multiplication cost obtained in Lemma \ref{lem:multi}, we only need to cost the other three operations appearing in Fig.~\ref{fig:newton-iteration}: Copy, $(3-\#)$, and $\div 2$. 

Copy only needs CNOT gates and no T gates. Similarly, the operation $\div 2$ can simply be done via SWAP gates and one ancilla qubit, as the numbers are in binary, which require no T gates. Alternatively, one just relabels the qubits at no computational cost, which is what is done here.

The operation $(3-\#)$ is slightly more involved. Recognize that the input $\#$ is always such that $0 < \# < 3$ (and generally non-integer), so we will never have to work with signed arithmetic. This operation can be done using two's-complement subtraction. As $\#$ is stored in bits labeled from the $0$th binary place to the $-(2w+k-1)$th binary place, but 3 ranges over the $1$st and $0$th place, we must pad a single qubit on the register containing $\#$ in the 2nd place in order to use Lemma \ref{lem:inc-far}. The Lemma \ref{lem:inc-far} is used to add $3+2^{-(2w+k-1)}$ to the one's complement of $\#$ (where both numbers are interpreted as scaled by $2^{(2w+k-1)}$), which gives us $(3-\#)$ in $2w+k+1$ qubits that feed into the middle multiplication algorithm in Fig.~\ref{fig:newton-iteration}. The cost of this operation is $8w+4k - 4$ T~gates and $2w+k-2
$ ancilla qubits, plus a single ancilla qubit that may be needed to execute the Toffoli gates. Note that since the overflow bit is neglected in two's-complement subtraction, the use of modular addition as in Lemma \ref{lem:inc-far} is acceptable.
\begin{table}
\centering
\begin{tabular}
{>{\arraybackslash}p{5.05cm} >{\arraybackslash}p{6.15cm} l l }
Routine & T gates & Work bits & Scratch bits \\
\hline
\hline
Copy & 0 & 0 & $w$ \\
$y_n^2$: $w$-bit squaring & $8w^2-4w$ & $2w$ & $2w$ \\
$y_n^2g$: $(2w \times k)$-bit mult. & $16kw -4 \max (k, 2w)$ & $2 \max (k, 2w)$ & $2w+k$ \\
$3-y_n^2g$ operation & $8w+4k-4$ & $2w+k-2$ & $1$ \\
$y_n(3-y_n^2g)$: $ w \times (2w+k+1) $-bit mult. & $16w^2+8kw-4k-4$ & $4w+2k+2$ & $0$ \\
Division by 2 & 0 & 0 & 0 \\
\hline
Overall circuit & $32w^2+40kw+4k+8w{-8\max(k,2w)}-12$ & $4w+2k+2$ & $5w+k+1$
\end{tabular}
\caption{Costs associated with an iteration of Newton's method, i.e., the circuit depicted in Fig.~\ref{fig:newton-iteration} of the main text. Note that the overall T cost of the circuit accounts for two implementations of each of the $w$-bit squaring, $(2w \times k)$-bit multiplication, and the $3-\#$ operation.}
\label{tab:iteration-cost}
\end{table}

Let us now determine the ancilla-qubit cost of each iteration of Newton's method. The dominant cost is either associated with evaluating $(3-\#)$ as discussed above or from the multiplication circuit which multiplies the numbers with the most bits.
All other subroutines may use these workspace qubits. The ancilla cost of multiplications according to Lemma \ref{lem:multi} is $2w$ for the $y_n^2$ operation, $2\max(k,2w)$ for the $y_n^2g$ multiplication, and $4w+2k+2$ for the $y_n(3-y_n^2g)$ multiplication. This must be compared with the ancilla-qubit cost of the $(3-\#)$ operation, that is $2w+k-2$. So the number of needed ancilla qubits is $4w+2k+2$. Finally, one needs to add to this the number of ancilla qubits in the scratch space needed to store intermediate values, see Fig.~\ref{fig:newton-iteration} and Table~\ref{tab:iteration-cost}, which adds to $5w+k+1$. Combining these values gives the value of $\beta$ as state above. Note that the scratch space is uncomputed in the end of each iteration and can be reused in other steps, hence the scratch space is counted toward the ancilla cost.

Finally, the T-gate cost is that of the five multiplications and two $(3-\#)$ operations. Each $y_n^2$ multiplication costs $ 8 w^2 - 4w$ T gates, each $y_n^2g$ multiplication costs $ 16 wk -4\max(k,2w)$ T gates, and the $y_n(3-y_n^2g)$ multiplication costs $ 16w^2+8kw-4k-4$ T gates. The T-gate count of each $(3-\#)$ is stated before. Adding all these contributions, the T-gate count comes to the value stated above. The cost breakdown of this Lemma is summarized in Table~\ref{tab:iteration-cost}.
\end{proof}
\begin{lemma}[The initial guess for an approximate evaluation of $1/\sqrt{g}$ via Newton's method]\label{lem:guess}
A suitable initial value, $y_0$, for Newton's method iterations in Eq.~(\ref{eq:iteration}) is the inverse square root of the $s^{\rm th}$ power of four such that $s$ is the closest integer to the logarithm of the input $g$ in base four. In other words,
\begin{figure}
\centering
    \includegraphics[scale=1]{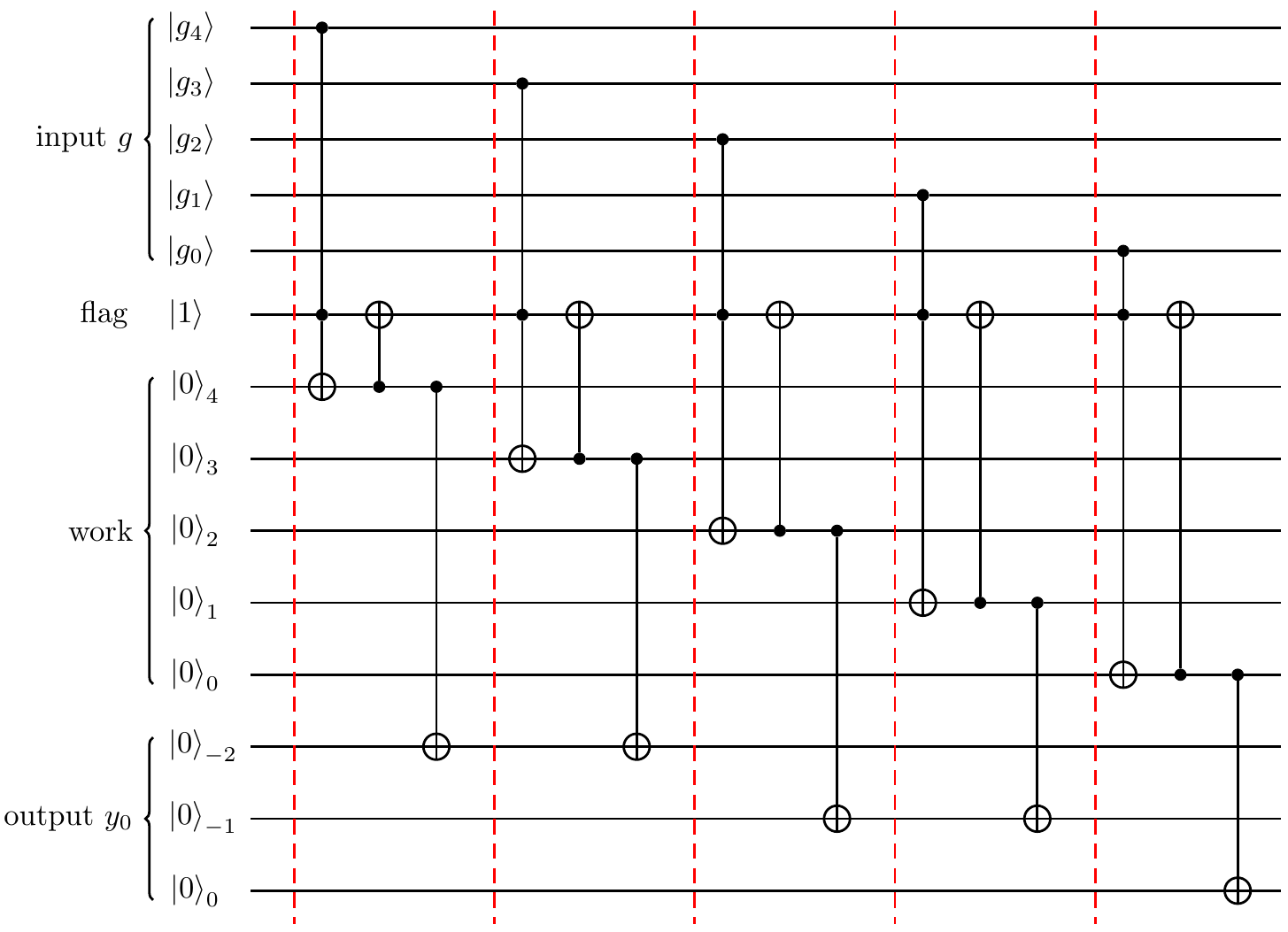}
    \caption{Circuit example for computing $y_0$ in Eq.~\eqref{eq:y0}, that is the zeroth approximation to inverse square root of an integer. Here, the input is encoded in $\eta = 5$ qubits and in each collection of qubit registers, the subscripts refer to the power of two in the binary representation of the integer values the registers hold.
    }
    \label{fig:initial-value-circuit}
\end{figure}
\begin{align}
    y_0(g) = \frac{1}{\sqrt{4^{s(g)}}}~~\text{with}~~s(g) = \left\lfloor \log_4g + \frac{1}{2} \right\rfloor.
\label{eq:y0}
\end{align}
This is sufficient to guarantee an approximately quadratic convergence to the inverse square root of an $\eta$-qubit integer input $g$. There exists a quantum circuit which maps
\begin{equation*}
 \ket{g}\ket{0}^{\otimes (1+\lfloor \eta/2 \rfloor+\ell + \beta)} \xrightarrow[]{} \ket{g}\ket{y_0}\ket{\mathrm{junk}}\ket{0}^{\otimes \beta}.
\end{equation*}
Here, $1 + \lfloor \eta/2 \rfloor$ qubits are required to store the initial guess, the state $\ket{\rm junk}$ occupies an $\ell$-qubit
register with $\ell=\eta + 1$. Furthermore, the circuit is implemented using $4 \eta$ T gates, and $\beta=1$ ancilla qubit is needed to implement all the Toffoli gates.
When the input register is in the state $g=0$, the circuit returns $y_0(0)=0$.
\end{lemma}
\begin{proof}
As discussed in the main text, if the relative error $|\delta_0|$ in the initial value for Newton's method is less than one, a roughly quadratic convergence is guaranteed, see Eq.~(\ref{eq:convergence}).

Now let $y_0 = 2^{-s} $ be the initial value and let $1/\sqrt{g}$ be the exact value. So $g = 4^s + r$ where, based on the definition, one can obtain the relation $-\frac{1}{2} 4^{s} \leq r < 4^s$. This implies that
\begin{align}
\label{eq:delta0}
    |\delta_0| = \frac{|y_0 - 1/\sqrt{g}|}{1/\sqrt{g}} = \bigl | \sqrt{\frac{g}{4^s}} - 1 \bigr | 
    = \bigl | \sqrt{ 1 + \frac{r}{4^s}} - 1 \bigr | 
    \leq|\sqrt{2} - 1 | < 1.
\end{align}

To compute this initial guess, one can use the generalization of the circuit shown in Fig.~\ref{fig:initial-value-circuit}. The input $g$ considered in the figure is a five-bit integer. To represent the zeroth iterate of an $\eta$-bit integer input requires that the output precision is at least $1+\lfloor \eta/2\rfloor$ bits.
The circuit works by inspecting the integer input $g$ and finding the first occurrence of $1$ from left to right, i.e., from the most significant bit to the least significant bit, in the binary form of $g$. This bit is flagged by a `flag' qubit initialized in state $\ket{1}$. Once the first $1$ bit is found, the shown operations in the figure allow to set the correct bit in the initially $y_0=0$ register to one, providing the final $y_0=2^{-s}$ value. Another $\eta$-sized register, noted as `work' register in the figure, is used to ensure that the flag is set to zero once the first bit of the input with value $1$ is found, which stops the rest of the operations and refrains from accumulating any value in the remaining $y_0$ bits. The circuit involves $\eta$ Toffoli gates which can be compiled with $4\eta$ T gates and one ancilla qubit. A similar procedure for finding the initial-value guess for both integer and non-integer inputs is discussed in Ref.~\cite{Haner:2018yea}.
\end{proof}

\begin{lemma}[Bounding the absolute error in the evaluation of $1/\sqrt{g}$ via a fixed-bit-precision application of Newton's method]
\label{lem:bound}
Let $g \geq 1$ and let $n$ and $m$ be two positive integers. Furthermore, let $y_k$ for $1 \leq k \leq n$ denote the $k$th exact output of Newton's iteration in Eq.~\eqref{eq:iteration}, and $\tilde y_k$ denote the truncated output of the Newton's iteration applied to $\tilde y_{k-1}$ such that $\tilde y_k$ is an $m$-bit number. Additionally, assume that the initial value $y_0$ is chosen via Lemma~\ref{lem:guess} and is kept exact, i.e., $\tilde y_0 = y_0$. Then
\begin{equation}
\label{eq:error-newton}
    \bigg|\frac{1}{\sqrt{g}}- \tilde y_n\bigg| \leq2^n (\sqrt{2}-1)^{2^n} +  2^{2-m}\biggl[\biggl(\frac{3}{2}\biggr)^n-1\biggr].
\end{equation}
\end{lemma}
\begin{proof}
First note that for $n=0$, Eq.~(\ref{eq:error-newton}) reproduces Eq.~(\ref{eq:delta0}) considering that $1/\sqrt{g} \leq 1$. To find a bound on $|1/\sqrt{g}- \tilde y_n|$ for $n>0$, that is the absolute error between the exact value and the value computed by $n$ repetitions of the circuit in Fig.~\ref{fig:newton-iteration}, we follow a derivation based on Ref.~\cite{bhaskar2015quantum}. First, one notes that $|1/\sqrt{g}- \tilde y_n| \leq |1/\sqrt{g}- y_n| + |y_n- \tilde y_n|$, and by separately bounding $|1/\sqrt{g} - y_n|$ and $|y_n - \tilde y_n|$, the full error can be bounded.

To bound $|1/\sqrt{g} - y_n|$, recall that for the relative error $\delta_k = \frac{y_k - 1/\sqrt{g}}{1/\sqrt{g}}$, Eq.~(\ref{eq:convergence}) states that $\delta_{k+1} = -\frac{3}{2}\delta_k^2 - \frac{1}{2} \delta_k^3$, which leads to $|\delta_n| \leq 2^n |\delta_0|^{2^n}$ for the relative error on the last output of Newton iterations. The convergence to the exact value $1/\sqrt{g}$ is achieved as long as $|\delta_0| <1$. This condition is guaranteed using the initial guess from Lemma~\ref{lem:guess} since $|\delta_0| \leq \sqrt{2}-1<1$. The bound on $|\delta_n|$ gives:
\begin{equation}
     \bigg|\frac{1}{\sqrt{g}} - y_n\bigg| = \frac{|\delta_n|}{\sqrt{g}} \leq \frac{2^n (\sqrt{2}-1)^{2^n}}{\sqrt{g}} \leq 2^n (\sqrt{2}-1)^{2^n},
\label{eq:bound1}
\end{equation}
where the last inequality holds because $g \geq 1$.

To bound $ |y_n - \tilde y_n|$, let $\mathcal N_g$ be the Newton-iteration function such that $\mathcal N_g(y_k) = y_{k+1}$ according to Eq.~\eqref{eq:iteration}. Denoting $\xi_k$ to be the additive error occurring due to truncating $\tilde y_k$ to only $m$ bits of precision, that is
\begin{align}
    \tilde y_0 &= y_0, \nonumber\\
    \tilde y_1 &= \mathcal N_g(\tilde y_0) + \xi_1, \nonumber\\
    \tilde y_2 &= \mathcal N_g(\tilde y_1) + \xi_2, \nonumber\\
    &\ \ \vdots \nonumber\\
    \tilde y_n, &= \mathcal N_g(\tilde y_{n-1}) +\xi_n ,\nonumber
\end{align}
one obtains the following relation:
\begin{equation}
\label{eq:yk-tildeyk}
    |y_k - \tilde y_k| \leq |\xi_k| + |y_{k-1} - \tilde y_{k-1}|\bigl|\max_{x\in \mathbb{D}} \mathcal N_g'(x) \bigr|,
\end{equation}
for some domain $\mathbb{D}$ over which Newton iterations are restricted during the application of each step.

A sufficient domain is $\mathbb{D} = [0,1/\sqrt{g}]$. This can be seen by noting that: i) if $0 \leq y_k \leq 1/\sqrt{g}$, then $0 \leq y_{k+1} \leq 1/\sqrt{g}$, and ii) $0 \leq y_1 \leq 1/\sqrt{g}$. Combining i) and ii) implies the input to Newton-iteration function $\mathcal N_g$ can only be within $[0,1/\sqrt{g}]$. Number i) follows from the fact that $\mathcal N_g'(x) \geq 0$ over $\mathbb{D}$, implying $\mathcal N_g$ takes its minimum at 0 and its maximum at $1/\sqrt{g}$ in this domain. Since $0 \leq y_k \leq 1/\sqrt{g}$ by assumption, then $\mathcal N_g(0) \leq \mathcal N_g(y_k) \leq \mathcal N_g(1/\sqrt{g})$, which implies $0 \leq y_{k+1} \leq 1/\sqrt{g}$. Number ii) follows from the fact that $y_1 - 1/\sqrt{g} =\tfrac{y_0}{2} (3-y_0^2g - \tfrac{2}{y_0 \sqrt{g}}) \leq 0$ for $\tfrac{1}{2} \leq y_0^2 g < 2$. This last condition can be shown to follow from the bound on the closeness of the initial guess to the exact function stated in Lemma~\ref{lem:guess}, so $y_1 \leq 1/\sqrt{g}$ is established. Similarly, since $y_1 =\tfrac{y_0}{2} (3-y_0^2 g)$, the initial guess implies $ 0 \leq y_1$.

With the sufficient domain $\mathbb{D}$ found, it is easy to see that $|\max_{x\in \mathbb{D}} \mathcal N_g'(x)| =3/2$. Furthermore, since $0 \leq y_k \leq 1/\sqrt{g} \leq 1$, each $y_k$ has only one significant bit at the one's place, and the rest of the bits occur after the decimal point. This means that the $m$-bit-truncated $y_k$ incur an error $\xi \equiv \xi_k= 2^{1-m}$. With these information, Eq.~\eqref{eq:yk-tildeyk} can be further processed:
\begin{align}
    |y_n - \tilde y_n| &\leq \xi + \frac{3}{2}|y_{n-1} - \tilde y_{n-1}| \nonumber\\
                    &\leq \xi + \frac{3}{2}\xi + \biggl ( \frac{3}{2} \biggr)^2 |y_{n-2} - \tilde y_{n-2}| \nonumber\\
                    &\ \ \vdots \nonumber\\
                    &\leq \xi + \frac{3}{2}\xi + \biggl ( \frac{3}{2}  \biggr)^2 \xi+ \ldots + \biggl ( \frac{3}{2} \biggr)^{n-1}\xi = 2^{2-m}\biggl[\biggl(\frac{3}{2}\biggr)^n - 1\biggr].
\label{eq:bound2}
\end{align}
Ultimately, combining Eqs.~\eqref{eq:bound1} and \eqref{eq:bound2} gives a bound on the total absolute error: 
\begin{equation}
    |\frac{1}{\sqrt{g}}- \tilde y_n| \leq |\frac{1}{\sqrt{g}} - y_n|+|y_n - \tilde y_n| \leq2^n (\sqrt{2}-1)^{2^n} +  2^{2-m}\biggl[\biggl(\frac{3}{2}\biggr)^n-1\biggr],
\end{equation}
as stated.
\end{proof}

\begin{lemma}[Computing the Schwinger-boson kinetic phase in the fault-tolerant model]\label{lem:kineticphasekickSB}
Let $\eta$, $m$, and $n$ be positive integers, with $\eta$ denoting the size of the $\ket{p}$, $\ket{q}$, and $\ket{p'}$ registers, i.e., $0\leq p,q,p' \leq 2^\eta-1$.
There exists a quantum circuit for $\mathscr{U}_{\widetilde{\mathcal{D}}}$ which maps
\begin{equation*}
 \ket{p}\ket{q}\ket{p'}\ket{0}^{\otimes (1+m + \ell + \beta)} \xrightarrow[]{\mathscr{U}_{\widetilde{\mathcal{D}}}}  \ket{p}\ket{q}\ket{p + p' + 1}\ket{\widetilde{\mathcal{D}}(p,q,p')}\ket{\mathrm{junk}}\ket{0}^{\otimes \beta}.
\end{equation*}
Here, the state $\ket{\widetilde{\mathcal{D}}(p,q,p')}$ denotes an $m$-bit fixed-precision approximation to
\[
\mathcal{D}(p, q , p') = \sqrt{\frac{ p  q }{(p + p') ( p + p'+ 1)}}
\]
within absolute error $|\delta^{\rm Abs.}_{n,m}|\leq 2^n (\sqrt{2}-1)^{2^n}+2^{2-m}((3/2)^n-1)$, where $n$ denotes the number of Newton iterations used to evaluate the function to within this error. Furthermore, the output of each iteration step is truncated to $m$ bits of precision as well.
The state $\ket{{\rm junk}}$ occupies at most an $\ell$-qubit register, where $$\ell =3mn + 4\eta n + 15\eta + 3n + 6 + \max(2\eta+2,m).
$$ Furthermore, the circuit can be implemented using a number of T gates equal to
$$ 
160 \eta m n + 32m^2 n + 48\eta^2 + 16\eta m + 16\eta n + 88mn - 8n\max(4\eta+2,2m) + 48\eta - 4\max(2\eta,m) - 4n + 8.
$$
The number of ancillary workspace qubits is given by: $$\beta=12\eta+9m+9.
$$
Additionally,
an extra ancilla qubit may be needed for the execution of all the Toffoli gates.
\end{lemma}
\begin{table}
\centering
\begin{tabular}{>{\arraybackslash}p{4.8cm} >{\arraybackslash}p{4.65cm} l >{\arraybackslash}p{3.5cm}}
Routine & T gates & Work bits & Junk bits \\
\hline
\hline
$pq:\eta\times\eta$ mult. & $8\eta^2-4\eta$ & $2\eta$ & $2\eta$ \\
$p+p':\eta$-bit adder & $4\eta$ & $\eta$ & $0$ \\
$(pq)(p+p'):2\eta \times (\eta+1)$ mult. & $16\eta^2+8\eta$ & $4\eta$ & $3\eta+1$ \\
$p+p'+1:(\eta+1)$-bit incr. & $4\eta-4$ & $\eta-2$ & $0$ \\
$(p+p')(pq(p+p'+1)):(\eta+1)\times (3\eta+1)$ mult. & $24\eta^2+20\eta+4$ & $6\eta+2$ & $k=4\eta+2$ \\
Initial-value~evaluation & $16\eta+8$ & $0$ & $\max(2\eta+2,m)+4\eta+3$ \\
$n$ steps of Newton iteration & $n ( 160 \eta m + 32m^2 + 16 \eta + 88 m - 8 \max(4\eta+2,2m)-4) $ & $12\eta+9m+9$ & $n(4\eta+3m+3)$ \\
$pq \, \tilde y_n:2\eta\times m$ mult. & $16 \eta m - 4 \max(2\eta,m)$ & $2 \max(2\eta, m)$ & $2\eta$ \\
\hline
Overall circuit & $ 160 \eta m n + 32m^2 n + 48\eta^2 + 16\eta m + 16\eta n + 88mn - 8n\max(4\eta+2,2m) + 48\eta - 4\max(2\eta,m) - 4n + 8$ & $12\eta+9m+9$ & $3mn + 4\eta n + 15\eta + 3n + 6 + \max(2\eta+2,m)$
\end{tabular}
\caption{Costs associated with the $\UfSB$ circuit for the Schwinger-boson formulation.}
\label{tab:Uf_SB}
\end{table}
\begin{proof}
The size of the workspace that is needed to store all the intermediate values required to approximately compute $\sqrt{pq/((p+p')(p+p'+1))}$ can be read off from the circuit shown in Fig.~\ref{fig:Uf_SB}, and is summarized for each element of the circuit in Table~\ref{tab:Uf_SB}. These add up to the value $\ell$ as stated. This register, denoted as $\ket{\rm junk}$, is not reset at the end of the $\mathscr{U}_{\widetilde{\mathcal{D}}}$ evaluation circuit, and hence it will not be usable for later computations. Note that for an input size $k=4\eta+2$, the initial-value circuit needs $1 + \lfloor k/2 \rfloor = 2\eta +2$ qubits to store the value of $y_0$. However, if this is smaller than the $m$-bit register size required to store all the Newton iterates, one should pad the bit representation of $y_0$ to enlarge it to an $m$-bit register. 

The T-gate and ancillary-qubit cost of the circuit can be evaluated as following. Inspecting the circuit in Fig.~\ref{fig:Uf_SB} from left to right, one first encounters: i) a multiplication circuit for evaluating $pq$ which requires $8\eta^2-4\eta$ T gates and $2\eta$ ancilla qubits, ii) an addition circuit that evaluates $p+p'$ and requires $4\eta$ T gates and $\eta$ ancilla qubits, iii) a multiplication circuit for evaluating $(pq)(p+p')$ which requires $16\eta^2+8\eta$ T gates and $4\eta$ ancilla qubits, iv) an incrementer which increases by one the value of $p+p'$ and requires $4\eta-4$ T gates and $\eta-2$ ancilla qubits, v) a multiplication circuit for evaluating $(p+p')(pq(p+p'+1))$ which requires $24\eta^2+20\eta+4$ T gates and $6\eta+2$ ancilla qubits, vi) a circuit for computing the initial value for the subsequent application Newton's method with the input $pq(p+p')(p+p'+1)$ with size $k=4\eta+2$ which requires $4k$ T gates and no ancilla qubit (beside one for evaluating the Toffoli gates which we take into account later), vii) $n$ applications of Newton-iteration steps with input size $k$ and output size $m$, which requires in total $n(160 \eta m + 32m^2 + 16 \eta + 88 m - 8 \max(4\eta+2,2m)-4) )$ T gates and $2\max(2\eta,m)$ ancilla qubits, and finally viii) a multiplication circuit for evaluating $pq \, \tilde y_n$ with $\tilde y_n$ being the $m$-bit precision output of the last Newton iteration, which requires $16\eta m-4\max(2\eta,m)$ T gates and $2\max(2\eta,m)$ ancilla qubit.

The total T-gate cost of the circuit is the sum of the cost of each subcircuit above and adds up to the value stated. The total auxiliary qubit cost of the circuit is that of the subcircuit which requires the largest workspace, that is the workspace of the Newton iterations. This gives the value of $\beta$ as stated above. Note that an extra ancilla qubit may be needed throughout the circuit for implementing all Toffoli gates.
\end{proof}
\begin{lemma}[Computing the LSH kinetic phase in the fault-tolerant model]
\label{lem:kineticphasekickLSH}
Let $\eta$, $m$, $n$, and $n_\ell$ be positive integers with $0 \leq n_\ell \leq 2^\eta - 1 $, and $n_q$ and $n_q^\prime$ be single-bit numbers 0 or 1.
There exists a quantum circuit for $\mathscr{U}_{\widetilde{\mathcal{D}}}$ which maps
\begin{equation*}
 \ket{n_q} \ket{n_q^\prime} \ket{n_\ell} \ket{0}^{\otimes (1+m+\ell+\beta)} \xrightarrow[]{\mathscr{U}_{\widetilde{\mathcal{D}}}}  \ket{n_q} \ket{n_q^\prime} \ket{n_\ell + 1 + n_q} \ket{\widetilde{\mathcal{D}}(n_\ell , n_q , n_q^\prime )}\ket{\rm junk}\ket{0}^{\otimes \beta} .
\end{equation*}
Here, the state $\ket{\widetilde{\mathcal{D}}(n_\ell , n_q , n_q^\prime )}$ denotes an $m$-bit fixed-precision approximation to
\[
\mathcal{D}(n_\ell , n_q , n_q^\prime ) = \sqrt{\frac{n_\ell + 1 + n_q }{n_\ell + 1 + n_q^\prime }}
\]
within absolute error $|\delta^{\rm Abs.}_{n,m}|\leq \sqrt{2} \bigl( 2^n (\sqrt{2}-1)^{2^n}+2^{2-m}((3/2)^n-1) \bigr)$, where $n$ denotes the number of
Newton iterations used to evaluate the function to within this error. Furthermore, the output of
each iteration step is truncated to $m$ bits of precision as well.
The state $\ket{\mathrm{junk}}$ occupies at most an $\ell$-qubit register, where $$\ell = 2\eta n + 3mn + 6\eta + \max(\eta+2,m) + 3n + 7. $$
Furthermore, the circuit can be implemented using a number of T gates equal to
$$ 80\eta mn + 32 m^2 n + 8\eta^2 + 8\eta m + 8 \eta n + 88mn - 8n \max(2\eta+2,2m)+32\eta-4\max (\eta+1,m)+8m-4n+8.
$$
The number of ancillary workspace qubits is given by $$ \beta = 6\eta+9m+9.$$
Additionally,
an extra ancilla qubit may be needed for the execution of all the Toffoli gates.
\end{lemma}
\begin{proof}
The size of the workspace that is needed to store all the intermediate values required to approximately compute $\sqrt{(n_\ell+1+n_q)/(n_\ell+1+n_q')}$ can be read off from the circuit shown in Fig.~\ref{fig:Uf_LSH}, and is summarized for each element of the circuit in Table~\ref{tab:Uf_LSH}, adding up to the value $\ell$ as stated. In particular, the copy operation, the first multiplier, the initial-value circuit, each Newton block, and
the final multiplier are the operations that need such workspace. These lead to the state $\ket{\rm junk}$ that is not reset at the end of the $\mathscr{U}_{\widetilde{\mathcal{D}}}$ evaluation circuit, and hence it will not be usable for later computations.
Note that for the input register size $k=2\eta+2$, if $1 + \lfloor k/2 \rfloor  = \eta+2 < m$, the output of the initial-value circuit must be padded to $m$ bits in order to apply Lemma~\ref{lem:newtit}.
\begin{table}
\centering
\begin{tabular}{>{\arraybackslash}p{5.05cm} >{\arraybackslash}p{4.5cm} l >{\arraybackslash}p{3.5cm}}
Routine & T gates & Work bits & Junk bits \\
\hline
\hline
$(\eta+1)$-bit incr. & $4\eta-4$ & $\eta-2$ & 0 \\
Copy $\eta+1$ bits & 0 & 0 & $\eta+1$ \\
Two controlled $(\eta+2)$-bit incrs. & $8\eta$ & $\eta-1$ & $0$ \\
$(n_\ell+1+n_q)(n_\ell+1+n_q'):(\eta+1)\times(\eta+1)$ mult. & $8\eta^2+12\eta+4$ & $2\eta+2$ & $2\eta+2$ \\
Initial-value evaluation & $8\eta+8$ & 0 & $\max (\eta +2,m)+2 \eta +3$ \\
$n$ steps of Newton iteration & $n(80\eta m + 32m^2 + 8\eta + 88m -16 \max (\eta +1,m)-4)$ & $6\eta+9m+9$ & $n(2\eta + 3m + 3)$ \\
$(n_\ell+1+n_q')\tilde y_n:(\eta+1)\times m$ mult. & $8m \eta+8m-4\max (\eta+1,m)$ & $2 \max (\eta+1,m)$ & $\eta+1$ \\
\hline
Overall circuit & $80\eta mn + 32 m^2 n + 8\eta^2 + 8\eta m + 8 \eta n + 88mn - 8n \max(2\eta+2,2m)+32\eta-4\max (\eta+1,m)+8m-4n+8$ & $6\eta+9m+9$ & $2\eta n + 3mn + 6\eta + \max(\eta+2,m) + 3n + 7$
\end{tabular}
\caption{Costs associated with the $\UfLSH$ circuit for the LSH formulation.}
\label{tab:Uf_LSH}
\end{table}

The T-gate and ancillary-qubit cost of the circuit can be evaluated as following. Inspecting the circuit in Fig.~\ref{fig:Uf_LSH} from left to right, one first encounters: i) an incrementer that is applied to an $(\eta+1)$-bit register, with the associated T-gate count $4\eta-4$ and ancilla-qubit count $\eta-2$ according to Lemma~\ref{lem:inc-far}, a copy operation which can be done with CNOT gates alone, costing no T gates and ancilla qubits, iii) two controlled incrementers each of which can be implemented using an ordinary incrementer on $\eta+2$ bits (followed by a bit flip on the ``control'' qubit), costing $4 \eta$ T gates and $\eta-1$ ancilla qubits, iv) a multiplier which is applied to two $(\eta+1)$-bit integers producing the value $(n_\ell+1+n_q)(n_\ell+1+n_q')$ at a cost of $8\eta^2 + 8\eta+4$ T gates and $2\eta+2$ ancilla qubits according to Lemma~\ref{lem:multi}, v) the initial-value circuit that is applied to an argument of size $
k= 2\eta+2$ at a cost of $8\eta+8$ T gates and no ancilla qubit (beside one reusable one for implementing all Toffoli gates) according to Lemma~\ref{lem:guess}, vi) $n$ steps of Newton iteration, each with a cost of $80\eta m + 32m^2 + 8\eta + 88m -16 \max (\eta +1,m)-4$ T gates and $6\eta+9m+9$ ancilla qubits as implied by Lemma~\ref{lem:newtit}, and finally vii) a multiplier whose inputs are the leading $m$ bits of the last Newton iterate $\tilde y_n$ and the $(\eta+1)$-bit integer $n_\ell+1+n_q$, which costs $8m \eta+8m-4\max (\eta+1,m)$ T gates and $2\max(\eta+1,m)$ ancilla qubits.

The total T-gate cost of the circuit is the sum of the cost of each subcircuit above and adds up to the value stated. The total auxiliary qubit cost of the circuit is that of the subcircuit which requires the largest workspace, that is the workspace of the Newton iterations. This gives the value of $\beta$ as stated above. Note that an extra ancilla qubit may be needed throughout the circuit for implementing all Toffoli gates.
\end{proof}
%

\section{Deriving the loop-string-hadron hopping subterms\label{app:LSHSubtermSimplification}}
\noindent
In this appendix, we show that the LSH hopping Hamiltonian written as the two subterms given in Eqs.~(\ref{eq:HI-subterms-LSH}) are equivalent to the original form in Eq.~(43) of Ref.~\cite{Raychowdhury:2019iki}, that is
\begin{align}
    &H_I^{{\rm LSH}}(r) = 
    H_I^{(o)}(r) + H_I^{(i)}(r) , \\
    &
    H_I^{(o)}(r) 
    \equiv x\frac{1}{\sqrt{ N^L(r) + 1 }} \mathcal{S}_{\rm out}^{++} (r) \mathcal{S}_{\rm in}^{+-} (r+1) \frac{1}{\sqrt{ N^R(r+1) + 1}} + \mathrm{H.c.} , \\
    &
    H_I^{(i)}(r) 
    \equiv x\frac{1}{\sqrt{ N^L(r) + 1 }} \mathcal{S}_{\rm out}^{+-} (r) \mathcal{S}_{\rm in}^{--} (r+1) \frac{1}{\sqrt{ N^R(r+1) + 1 }} + \mathrm{H.c.} ,
\end{align}
where $N^{L/R}$ are the flux quantum numbers used throughout the Schwinger-boson and LSH sections and $\mathcal{S}_{\rm out/in}$ are site-local string operators in the LSH formulation defined in Eqs.~(\ref{eq:S-def}).
One can check, by expanding out the definitions of $N^{L/R}$ and the string operators, that these subterms can be expressed as
\begin{align}
    H_I^{(o)} 
    &= 
    x\left[ \chi^\dagger_{o}(r) \, \lshladder^{\dagger}(r)^{n_i(r)} \right] \left[ \chi_{o}(r+1)  \, \lshladder^{\dagger}(r+1) ^{1-n_i(r+1)} \right]
    \\
    &\hspace{3 cm}\times\sqrt{\frac{n_\ell(r) + 1}{n_\ell(r) + 1 + n_i(r)}} \sqrt{\frac{n_\ell(r+1)  + 1 + n_i(r+1)}{n_\ell(r+1) + 1}} + \mathrm{H.c.} , \label{eq:HIoLSH_expanded}\\
    H_I^{(i)} 
    &= 
    x\left[ \chi_{i}^\dagger(r+1) \, \lshladder^{\dagger}(r+1)^{n_o(r+1)} \right] \left[ \chi_{i}(r) \, \lshladder^{\dagger}(r)^{1-n_o(r)} \right]
    \\
    &\hspace{3 cm} \times\sqrt{\frac{n_\ell(r)  + 1 + n_o(r)}{n_\ell(r) + 1}} \sqrt{\frac{n_\ell(r+1) + 1}{n_\ell(r+1) + 1 + n_o(r+1)}} + \mathrm{H.c.} \label{eq:HIiLSH_expanded}
\end{align}
The square-root functions in Eqs. (\ref{eq:HIoLSH_expanded}) and (\ref{eq:HIiLSH_expanded}) can be simplified using the AGL.
For example, consider the square roots in Eq.~(\ref{eq:HIoLSH_expanded}).
The AGL on link $r$ can be used to solve for $n_\ell(r+1)$ as
\begin{align*}
    n_\ell(r+1) &= n_\ell(r) + n_o(r) - n_i(r) n_o(r) - n_i(r+1) + n_i(r+1) n_o(r+1).
\end{align*}
However, since this will need to be evaluated on the right of $\chi_{o}^\dagger (r) \chi_{o} (r+1)$, one can effectively set $n_o(r) \to 0$ and $n_o(r+1) \to 1$ in the formula for $n_\ell(r+1)$. Thus,
\begin{align}
    &n_\ell(r+1) \rightarrow n_\ell(r) + 0 - 0 - n_i(r+1) + n_i(r+1)= n_\ell(r) , \nonumber \\
    &\Rightarrow \sqrt{\frac{n_\ell(r) + 1}{n_\ell(r) + 1 + n_i(r)}} \sqrt{\frac{n_\ell(r+1)  + 1 + n_i(r+1)}{n_\ell(r+1) + 1}} \to \sqrt{\frac{n_\ell(r)  + 1 + n_i(r+1)}{n_\ell(r) + 1 + n_i(r)}} . 
\end{align}
An analogous calculation can be done to simplify Eq.~(\ref{eq:HIoLSH_expanded}) by first solving for $n_\ell(r)$ using the AGL and then substituting this into the square-root functions in Eq.~(\ref{eq:HIoLSH_expanded}).
The result is
\begin{align*}
    \sqrt{\frac{n_\ell(r)  + 1 + n_o(r)}{n_\ell(r) + 1}} \sqrt{\frac{n_\ell(r+1) + 1}{n_\ell(r+1) + 1 + n_o(r+1)}} &\to \sqrt{\frac{n_\ell(r+1)  + 1 + n_o(r)}{n_\ell(r+1) + 1 + n_o(r+1)}}.
\end{align*}
To summarize, using the AGL, one can equivalently use
\begin{align}
     H_I^{(o)}(r) 
    &\equiv x\left[ \chi^\dagger_{o}(r) \, \lshladder^{\dagger}(r)^{n_i(r)} \right] \left[ \chi_{o}(r+1)  \, \lshladder^{\dagger}(r+1) ^{1-n_i(r+1)} \right] \sqrt{\frac{n_\ell(r+1)  + 1 + n_i(r+1)}{n_\ell(r+1) + 1 + n_i(r)}} + \mathrm{H.c.},
    \\
    H_I^{(i)}(r) 
    &\equiv x\left[ \chi_{i}^\dagger(r+1) \, \lshladder^{\dagger}(r+1)^{n_o(r+1)} \right] \left[ \chi_{i}(r) \, \lshladder^{\dagger}(r)^{1-n_o(r)} \right] \sqrt{\frac{n_\ell(r)  + 1 + n_o(r)}{n_\ell(r) + 1 + n_o(r+1)}} + \mathrm{H.c.},
\end{align}
from which $H_I^{{\rm LSH}}$ in Eq.~\eqref{eq:HI-subterms-LSH} is reproduced.

\section{Second-order commutator bounds
\label{app:commutators}}
\noindent
In this appendix, we compute the second-order Trotter error bound of Ref. \cite{childs2021theory}:
\begin{align}
    \bigl \| V(\theta) - e^{-i\theta H} \bigr \| & \leq \frac{\theta^3}{24} \sum_{\gamma_1 = 1}^\Upsilon \bigl \| \biggl[ H_{\gamma_1} , \biggl [H_{\gamma_1},\sum_{\gamma_2 = \gamma_1 + 1}^\Upsilon  H_{\gamma_2} \biggr ] \biggr] \bigr \| \nonumber \\
    & \quad + \frac{\theta^3}{12} \sum_{\gamma_1 = 1}^\Upsilon \bigl \| \biggl [\sum_{\gamma_3 = \gamma_1 + 1}^\Upsilon H_{\gamma_3} , \biggl [\sum_{\gamma_2 = \gamma_1 + 1}^\Upsilon H_{\gamma_2}, H_{\gamma_1} \biggr ]\biggr] \bigr \| .
\end{align}
for both the Schwinger-boson and the LSH formulations. The ordered operators, $H_{\gamma}$, are chosen as
\begin{align}
    \left\{ \cdots , H_M(r), H_E(r), \{ H_I^{(j)}(r) \}_{j=1}^{\nu} , H_M(r+1), H_E(r+1), \{ H_I^{(j)}(r+1) \}_{j=1}^{\nu} , \cdots \right\} ,
\end{align}
for both formulations, where $\nu=8$ for the Schwinger-boson formulation and $\nu=2$ for the LSH formulation.

The first step one may take is to move the summations to the outside of the operator norms, making use of the triangle inequality:
\begin{align}
    \label{eq:2ndOrderTrotter_SumsOutside}
    \bigl \| V(\theta) - e^{-i\theta H} \bigr \| & \leq \frac{\theta^3}{24} \sum_{\gamma_1 = 1}^\Upsilon \sum_{\gamma_2 = \gamma_1 + 1}^\Upsilon \bigl \| \biggl[ H_{\gamma_1} , \biggl [H_{\gamma_1}, H_{\gamma_2} \biggr ] \biggr] \bigr \| \nonumber \\
    & \quad + \frac{\theta^3}{12} \sum_{\gamma_1 = 1}^\Upsilon \sum_{\gamma_2 = \gamma_1 + 1}^\Upsilon \sum_{\gamma_3 = \gamma_1 + 1}^\Upsilon \bigl \| \biggl [H_{\gamma_3} , \biggl [H_{\gamma_2}, H_{\gamma_1} \biggr ]\biggr] \bigr \|.
\end{align}
This is because we generally do not see cancellations take place by keeping the sums inside the commutator brackets, and whenever a norm of the form $\| A + B \|$ is encountered, a better upper bound on this than $\|A\|+\|B\|$ is often not known.

The second-order Trotter-error-bound calculation has now been turned into an exercise in evaluating various double commutators of simulatable Hamiltonian terms. One first identifies non-vanishing commutators in the two sums, noting that many of the commutators are identically zero due to the commutativity of operators from distant sites/links.
In the double sum, $\sum_{\gamma_1} \sum_{\gamma_2 > \gamma_1}$, the five non-vanishing operator forms are
\begin{align}
    \mathcal{O}_{1} &= \left[ H_M(r) , [ H_M(r) , H_I^{(j)}(r)   ] \right], \\
    \mathcal{O}_{2} &= \left[ H_E(r) , [ H_E(r) , H_I^{(j)}(r)   ] \right], \\
    \mathcal{O}_{3} &= \left[ H_I^{(j)}(r) , [ H_I^{(j)}(r) , H_I^{(k)}(r)   ] \right] \quad (k>j), \\
    \mathcal{O}_{4} &= \left[ H_I^{(j)}(r) , [ H_I^{(j)}(r) , H_M(r+1) ] \right], \\
    \mathcal{O}_{5} &= \left[ H_I^{(j)}(r) , [ H_I^{(j)}(r) , H_I^{(k)}(r+1) ] \right]. \\
\end{align}
In the triple sum, $\sum_{\gamma_1} \sum_{\gamma_2 > \gamma_1 } \sum_{\gamma_3 > \gamma_1 }$, the nineteen non-vanishing operator forms are
\begingroup
\allowdisplaybreaks
\begin{align}
    \mathcal{O}_{6 } &= \left[ H_E(r)   , [ H_I^{(j)}(r) , H_M(r)   ] \right], \\
    \mathcal{O}_{7 } &= \left[ H_I^{(k)}(r)   , [ H_I^{(j)}(r) , H_M(r)   ] \right], \\
    \mathcal{O}_{8 } &= \left[ H_M(r+1) , [ H_M(r) , H_I^{(j)}(r)   ] \right], \\
    \mathcal{O}_{9 } &= \left[ H_I^{(k)}(r+1) , [ H_I^{(j)}(r) , H_M(r)   ] \right], \\
\\    
    \mathcal{O}_{10} &= \left[ H_I^{(k)}(r)   , [ H_I^{(j)}(r) , H_E(r)   ] \right], \\
    \mathcal{O}_{11} &= \left[ H_M(r+1) , [ H_I^{(j)}(r) , H_E(r)   ] \right], \\
    \mathcal{O}_{12} &= \left[ H_I^{(k)}(r+1) , [ H_I^{(j)}(r) , H_E(r)   ] \right], \\
\\   
    \mathcal{O}_{13} &= \left[ H_I^{(l)}(r)   , [ H_I^{(k)}(r) , H_I^{(j)}(r)   ] \right] \quad (k>j,l>j), \\
    \mathcal{O}_{14} &= \left[ H_M(r+1) , [ H_I^{(k)}(r) , H_I^{(j)}(r)   ] \right] \quad (k>j), \\
    \mathcal{O}_{15} &= \left[ H_I^{(l)}(r+1) , [ H_I^{(k)}(r) , H_I^{(j)}(r)   ] \right] \quad (k>j), \\
\\
    \mathcal{O}_{16} &= \left[ H_I^{(k)}(r)   , [  H_M(r+1) , H_I^{(j)}(r) ] \right] \quad (k>j), \\
    \mathcal{O}_{17} &= \left[ H_M(r+1) , [  H_M(r+1) , H_I^{(j)}(r) ] \right], \\
    \mathcal{O}_{18} &= \left[ H_I^{(k)}(r+1) , [  H_M(r+1) , H_I^{(j)}(r) ] \right] \quad (k>j), \\
\\ 
    \mathcal{O}_{19} &= \left[ H_I^{(l)}(r)   , [  H_I^{(k)}(r+1) , H_I^{(j)}(r) ] \right] \quad (l>j), \\
    \mathcal{O}_{20} &= \left[ H_M(r+1) , [  H_I^{(k)}(r+1) , H_I^{(j)}(r) ] \right], \\
    \mathcal{O}_{21} &= \left[ H_E(r+1) , [  H_I^{(k)}(r+1) , H_I^{(j)}(r) ] \right], \\
    \mathcal{O}_{22} &= \left[ H_I^{(l)}(r+1) , [  H_I^{(k)}(r+1) , H_I^{(j)}(r) ] \right], \\
    \mathcal{O}_{23} &= \left[ H_M(r+2) , [  H_I^{(k)}(r+1) , H_I^{(j)}(r) ] \right], \\
    \mathcal{O}_{24} &= \left[ H_I^{(l)}(r+2) , [  H_I^{(k)}(r+1) , H_I^{(j)}(r) ] \right].
\end{align}
\endgroup
Strictly speaking, some of the above operator forms are not relevant at the boundaries of the lattice. For simplicity, the Trotterization error bound will be slightly inflated by ignoring the OBCs and acting as if every one of these forms appears $L$ times.

Turning now to the evaluations, some of the commutators are evaluated (or partially evaluated) by first noting that each hopping term in either formulations has the form
\begin{align}
    H_I^{(j)}(r) &=  x(K^{(j)}(r) + {K^{(j)}}^\dagger(r)) ,
\end{align}
where $K^{(j)}$ (${K^{(j)}}^\dagger$) can be read from the hopping Hamiltonians in Eq.~(\ref{eq:HI-subterms-LSH}). They contain both bosonic and fermionic creation and/or annihilation operators. Then one can show that generally
\begin{subequations} \label{eq:commsWithI}
\begin{align}
    &[ \psi^\dagger (r) \psi (r) ,  H_I^{(j)} (r) ] = \pm x \, ( {K^{(j)}}^\dagger (r) - K^{(j)} (r) ) , \label{eq:IMcomm1} \\
    &[ \psi^\dagger(r+1) \psi (r+1) ,  H_I^{(j)} (r) ] = \mp x \,  ( {K^{(j)}}^\dagger (r) - K^{(j)} (r) ) ,
\end{align}
\end{subequations}
where the $\pm$ and $\mp$ are correlated and depend on the specific $H_I^{(j)}$ being considered. Note that the commutators in Eqs.~(\ref{eq:commsWithI}) hold whether or not the $K^{(j)}$ has a cutoff built in. Furthermore, the electric Hamiltonian at each link in either formulation can be written as $J^L (r) (J^L(r) + 1)$ where $J^L=N^L/2$ in either formulation. Then one can show that
\begin{align}
    [ J^L (r) ( J^L(r) + 1) , H_I^{(j)} (r) ] &= \left( J^L(r) + \frac{1}{4}\right) x{K^{(j)}}^\dagger (r) - xK^{(j)} (r) \left( J^L(r) + \frac{1}{4}\right) , \label{eq:IEcomm}
\end{align}
for both formulations. Following whatever simplifications are available, one may eventually resort to using derivatives of the triangle inequality of the form
\begin{align}
    \| [ A , B ] \| &\leq 2 \|A\| \|B\| , \label{eq:ABIneq} \\
    \| [ A , [ B , C ] ] \| &\leq 4 \|A\| \|B\| \|C\| . \label{eq:ABCIneq}
\end{align}
In this work we do not find useful simplifications to commutators of the form $[H_I^{(j)}(r), H_I^{(k)}(r')]$ to apply before invoking one of Eqs.~\eqref{eq:ABIneq} or \eqref{eq:ABCIneq}.

In summary, the norm evaluation will always boil down to some combinations of $\| H_M(r) \|$, $\| J^L(r) + 1/4 \| $, and/or $\| x {K^{(j)}}^\dagger \pm x K^{(j)} \| $. While $\| H_M(r) \| \leq \mu$ for both formulations, the bounds on $\| J^L(r) + 1/4 \|$ and $\| x {K^{(j)}}^\dagger \pm x K^{(j)} \|$ are formulation dependent, which we evaluate in the respective subsections below. Nonetheless, a general intermediate bound can be still placed on $\| x {K^{(j)}}^\dagger \pm x K^{(j)} \|$ as follows.

One first notes that $K^{(j)}$ contains at least one fermionic operator, e.g., $\psi_\reg{x}^\dagger$. Let $\ket{\Psi} = c_0 \ket{\phi_0} + c_1 \ket{\phi_1}$ denote an eigenvector of $ K^{(j)} + {K^{(j)}}^\dagger $ with the largest eigenvalue, where $\psi_\reg{x}^\dagger \psi_\reg{x} \ket{\phi_b} = b \ket{\phi_b}$ for $b=0,1$, $\braket{\phi_{b'}|\phi_{b}}=\delta_{b'\,b}$, and hence $|c_0|^2 + |c_1|^2=1$. Now,
\begin{align*}
    \| K^{(j)} \pm {K^{(j)}}^\dagger \| &= \left| ( K^{(j)} \pm {K^{(j)}}^\dagger ) \ket{\Psi} \right| \\
    &= \left| c_0 K^{(j)} \ket{\phi_0} \pm c_1 {K^{(j)}}^\dagger \ket{\phi_1} \right| \\
    &= \sqrt{ |c_0 |^2 \bra{\phi_0} {K^{(j)}}^\dagger K^{(j)} \ket{\phi_0} + |c_1 |^2 \bra{\phi_1} K^{(j)} {K^{(j)}}^\dagger \ket{\phi_1} } \\
    &\leq \sqrt{( |c_0|^2 + |c_1|^2 ) \max \left( \bra{\phi_0} {K^{(j)}}^\dagger K^{(j)} \ket{\phi_0} , \bra{\phi_1} K^{(j)} {K^{(j)}}^\dagger \ket{\phi_1} \right)} \\
    &= \sqrt{\max \left( \bra{\phi_0} {K^{(j)}}^\dagger K^{(j)} \ket{\phi_0} , \bra{\phi_1} K^{(j)} {K^{(j)}}^\dagger \ket{\phi_1} \right)} \\
    &= \max \left( \sqrt{\bra{\phi_0} {K^{(j)}}^\dagger K^{(j)} \ket{\phi_0}} , \sqrt{\bra{\phi_1} K^{(j)} {K^{(j)}}^\dagger \ket{\phi_1}} \right) \\
    &\leq \max ( \| K^{(j)} \| , \| {K^{(j)}}^\dagger \| ) .
\end{align*}
From this, it immediately follows that 
\begin{align}
\| x {K^{(j)}}^\dagger \pm x K^{(j)} \| \leq \max ( \| x \, K^{(j)} \|, \| x \, {K^{(j)}}^\dagger \| ).
\label{eq:KpKdagged-bound}
\end{align}
The bound on $\| K^{(j)} \|$ and $\| {K^{(j)}}^\dagger \|$ will be discussed in the following for each formulation.
\newcommand{\elecNormSB}{\Lambda + \frac{1}{4}}
\newcommand{\elecNormLSH}{\frac{\Lambda}{2} + \frac{3}{4}}

\subsubsection{Schwinger-boson commutator bounds
\label{app:commutators_SB}}

Recall the generalized Schwinger-boson hopping term as written in Eq.~\eqref{eq:generalizedHISubterm-JW_SB},
\begin{align*}
    H_I^{\mathrm{SB}(j)} =\ &\ \pm x \, \psi^\dagger_\reg{x} \psi_\reg{y} \lambda^-_\reg{p} \lambda^-_\reg{q} \DSB( p, q, p') + \mathrm{H.c.}
\end{align*}
Above, one can identify
\begin{subequations}
\begin{align}
    K^{(j)} &= \pm \psi^\dagger_\reg{x} \psi_\reg{y} \lambda^-_\reg{p} \lambda^-_\reg{q} \DSB( p, q, p') , \\
    {K^{(j)}}^\dagger &= \pm \psi^\dagger_\reg{y} \psi_\reg{x} \lambda^+_\reg{p} \lambda^+_\reg{q} \DSB( p+1, q+1, p') ,
\end{align}
\end{subequations}
where $\DSB (p,q,p') = \sqrt{pq/((p+p')(p+p'+1))}$. Given Eq.~(\ref{eq:KpKdagged-bound}), the next step is the evaluation of $\max \left( \| x \, K^{(j)} \|, \| x \, {K^{(j)}}^\dagger \| \right)$ which boils down to the bounds on $\| \DSB \|$. It is easy to see that $\| \DSB (p,q,p') \| \leq 1$. This is because the AGL implies $q\leq p+p'$, and therefore,
$$ \frac{pq}{(p+p')(p+p'+1)} \leq \frac{p}{p+p'+1} \leq \frac{p}{p+1} < 1 . $$
A similar argument can be made to show that $\| \DSB (p+1,q+1,p') \| \leq 1$ as well.
In either case, the norm upper bound immediately follows upon reinserting the factor of $x$,\footnote{Note that in our convention, $x \geq 0$.} 
\begin{align}
\| x {K^{(j)}}^\dagger \pm x K^{(j)} \| \leq x.
\end{align}

As for the commutators involving $H^{\rm SB}_E(r)$, one notes that for the Schwinger-boson formulation,
\begin{align}
J^L(r) = \frac{1}{2}(N_1^L(r)+N_2^L(r)),
\end{align}
hence
\begin{align}
\bigl \| J^L(r) +\frac{1}{4} \bigr \| \leq \Lambda+\frac{1}{4}.
\end{align}

Equipped with these bounds, one can proceed to evaluating a bound on the double commutators. When the double commutators involve only one $H_I^{(j)}$ operator, Eqs.~(\ref{eq:commsWithI}) help one to `fully' evaluate the norms as
\begin{alignat}{4}
    \| \mathcal{O}_{1 } \| &= \bigl \| \left[ H_M(r)   , [ H_M(r) , H_I^{(j)}(r)   ] \right] \bigr \| & & \leq x \mu^2 , \\
    \| \mathcal{O}_{2 } \| &= \bigl \| \left[ H_E(r)   , [ H_E(r) , H_I^{(j)}(r)   ] \right] \bigr \| & & \leq x \left(\elecNormSB\right)^2 , \\
    \| \mathcal{O}_{6 }\|  &= \bigl \| \left[ H_E(r)   , [ H_I^{(j)}(r) , H_M(r)   ] \right] \bigr\| & & \leq x \left(\elecNormSB\right) \mu , \\
    \| \mathcal{O}_{8 } \| &= \bigl\| \left[ H_M(r+1) , [ H_I^{(j)}(r) , H_M(r)   ] \right] \bigr \| & & \leq x \mu^2 , \\
    \| \mathcal{O}_{11} \| &= \bigl\| \left[ H_M(r+1) , [ H_I^{(j)}(r) , H_E(r)   ] \right] \bigr\| & & \leq x \left( \elecNormSB \right) \mu , \\
    \| \mathcal{O}_{17} \| &= \bigl\| \left[ H_M(r+1) , [ H_M(r+1) , H_I^{(j)}(r) ] \right] \bigr\| & & \leq x \mu^2 ,
\end{alignat}
without ever appealing to Eqs.~\eqref{eq:ABIneq} or \eqref{eq:ABCIneq}.
When the `inner' commutator and `outer' commutator each contain a $H_I^{(j)}$, one can use the bounds obtained on $K^{(j)}$ and $J^L$, in combination with Eq.~(\ref{eq:ABIneq}), to obtain
\begin{alignat}{4}
    \| \mathcal{O}_{4 } \| &= \| \, [ H_I^{(j)}(r)   , [ H_I^{(j)}(r) , H_M(r+1) ] ] \, \| & & \leq 2 x^2 \mu , \\
    \| \mathcal{O}_{7 } \| &= \| \, [ H_I^{(k)}(r)   , [ H_I^{(j)}(r) , H_M(r)   ] ] \, \| & & \leq 2 x^2 \mu , \\
    \| \mathcal{O}_{9 } \| &= \| \, [ H_I^{(k)}(r+1) , [ H_I^{(j)}(r) , H_M(r)   ] ] \, \| & & \leq 2 x^2 \mu , \\
    \| \mathcal{O}_{10} \| &= \| \, [ H_I^{(k)}(r)   , [ H_I^{(j)}(r) , H_E(r)   ] ] \, \| & & \leq 2 x^2 \left( \elecNormSB \right) , \\
    \| \mathcal{O}_{12} \| &= \| \, [ H_I^{(k)}(r+1) , [ H_I^{(j)}(r) , H_E(r)   ] ] \, \| & & \leq 2 x^2 \left( \elecNormSB \right) , \\
    \| \mathcal{O}_{16} \| &= \| \, [ H_I^{(k)}(r)   , [ H_M(r+1) , H_I^{(j)}(r) ] ] \, \| & & \leq 2 x^2 \mu \quad (k>j) , \\
    \| \mathcal{O}_{18} \| &= \| \, [ H_I^{(k)}(r+1) , [ H_M(r+1) , H_I^{(j)}(r) ] ] \, \| & & \leq 2 x^2 \mu \quad (k>j) .
\end{alignat}
When the `inner' commutator involves two $H_I^{(j)}$ terms, one usually appeals directly to Eq.~(\ref{eq:ABCIneq}), obtaining
\begin{alignat}{4}
    \| \mathcal{O}_{3 } \| &= \bigl \| \left[ H_I^{(j)}(r)   , [ H_I^{(j)}(r) , H_I^{(k)}(r)   ] \right] \bigr \| & & \leq 4 x^3 \quad (k>j), \\
    \| \mathcal{O}_{5 } \| &= \bigl \| \left[ H_I^{(j)}(r)   , [ H_I^{(j)}(r) , H_I^{(k)}(r+1) ] \right] \bigr \| & & \leq 4 x^3 , \\
    \| \mathcal{O}_{13} \| &= \bigl \| \left[ H_I^{(l)}(r)   , [ H_I^{(k)}(r) , H_I^{(j)}(r)   ] \right] \bigr \| & & \leq 4 x^3 \quad (k>j,l>j) , \\
    \| \mathcal{O}_{14} \| &= \bigl \| \left[ H_M(r+1) , [ H_I^{(k)}(r) , H_I^{(j)}(r)   ] \right] \bigr \| & & \leq 4 x^2 \mu \quad (k>j) , \\
    \| \mathcal{O}_{15} \| &= \bigl \| \left[ H_I^{(l)}(r+1) , [ H_I^{(k)}(r) , H_I^{(j)}(r)   ] \right] \bigr \| & & \leq 4 x^3 \quad (k>j) , \\
    \| \mathcal{O}_{19} \| &= \bigl \| \left[ H_I^{(l)}(r)   , [ H_I^{(k)}(r+1) , H_I^{(j)}(r)  ] \right] \bigr \| & & \leq 4 x^3 \quad (l>j) , \\
    \| \mathcal{O}_{20} \| &= \bigl \| \left[ H_M(r+1) , [ H_I^{(k)}(r+1) , H_I^{(j)}(r) ] \right] \bigr \| & & \leq 4 x^2 \mu , \\
    \| \mathcal{O}_{22} \| &= \bigl \| \left[ H_I^{(l)}(r+1) , [ H_I^{(k)}(r+1) , H_I^{(j)}(r) ] \right] \bigr \| & & \leq 4 x^3 , \\
    \| \mathcal{O}_{24} \| &= \bigl \| \left[ H_I^{(l)}(r+2) , [ H_I^{(k)}(r+1) , H_I^{(j)}(r) ] \right] \bigr \| & & \leq 4 x^3 .
\end{alignat}
In the cases of $\mathcal{O}_{21} $ and $\mathcal{O}_{23} $, however, it is observed that the Jacobi identity could be used in combination with Eq.~(\ref{eq:ABIneq}) to derive
\begin{alignat}{4}
    \| \mathcal{O}_{21} \| &= \bigl \| \left[ H_E(r+1) , [ H_I^{(k)}(r+1) , H_I^{(j)}(r) ] \right] \bigr \| \leq 2 x^2 \left( \elecNormSB \right) , \\
    \| \mathcal{O}_{23} \| &= \bigl \| \left[ H_M(r+2) , [ H_I^{(k)}(r+1) , H_I^{(j)}(r) ] \right] \bigr \| \leq 2 x^2 \mu .
\end{alignat}

With upper bounds derived for all commutator norms, the only remaining step is to add them all up with appropriate weights.
Every operator form involves at least one $H_I^{(j)}$ operator, and hence belongs to one or more sums over the $\nu$ distinct subterms, leading to a certain multiplicity associated with each $\mathcal{O}_i$.
These multiplicities are recorded in Table~\ref{tab:commutatorBoundSummary}.
When the $\|\mathcal{O}_i\|$ are summed up with appropriate coefficients and multiplicities, the final result is
\begin{align}
    \bigl \| V^{\rm SB}(\theta) - e^{-i\theta H^{\rm SB}} \bigr \| &\leq 
     L \theta^3 \rho^{\rm SB} (x, \Lambda, \mu),
\end{align}
with
\begin{align}
    \rho^{\rm SB} (x, \Lambda, \mu ) &\equiv 
    \frac{1658 x^3}{3} + 32 \Lambda  x^2 + \frac{218 \mu  x^2}{3} + 8 x^2 + \frac{\Lambda ^2 x}{3} + \frac{4 \Lambda  \mu  x}{3} + \frac{\Lambda  x}{6} + \frac{5 \mu ^2 x}{3} + \frac{\mu  x}{3} + \frac{x}{48}.
\end{align}

\begingroup
\begin{table}
\renewcommand{\arraystretch}{1.1}
\begin{center}
\begin{tabular}{ c | c | l | l | l }
Operator form & Coefficient & $\mathcal{O}_i$ multiplicity & $\mathcal{O}_i$ bound (SB) & $\mathcal{O}_i$ bound (LSH) \\
\hline
\hline
 $\mathcal{O}_{1 }$ & $\frac{1}{24}$ & $\nu $ & $\mu ^2 x$ & $\sqrt{2} \mu ^2 x$ \\
 $\mathcal{O}_{2 }$ & $\frac{1}{24}$ & $\nu $ & $\left(\Lambda +\frac{1}{4}\right)^2 x$ & $\sqrt{2} \left(\frac{\Lambda }{2}+\frac{3}{4}\right)^2 x$ \\
 $\mathcal{O}_{3 }$ & $\frac{1}{24}$ & $\frac{1}{2} (\nu -1) \nu $ & $4 x^3$ & $8 \sqrt{2} x^3$ \\
 $\mathcal{O}_{4 }$ & $\frac{1}{24}$ & $\nu $ & $2 \mu x^2$ & $4 \mu x^2$ \\
 $\mathcal{O}_{5 }$ & $\frac{1}{24}$ & $\nu ^2$ & $4 x^3$ & $8 \sqrt{2} x^3$ \\
 $\mathcal{O}_{6 }$ & $\frac{1}{12}$ & $\nu $ & $\left(\Lambda +\frac{1}{4}\right) \mu x$ & $\sqrt{2} \left(\frac{\Lambda }{2}+\frac{3}{4}\right) \mu x$ \\
 $\mathcal{O}_{7 }$ & $\frac{1}{12}$ & $\nu ^2$ & $2 \mu x^2$ & $4 \mu x^2$ \\
 $\mathcal{O}_{8 }$ & $\frac{1}{12}$ & $\nu $ & $\mu ^2 x$ & $\sqrt{2} \mu ^2 x$ \\
 $\mathcal{O}_{9 }$ & $\frac{1}{12}$ & $\nu ^2$ & $2 \mu x^2$ & $4 \mu x^2$ \\
 $\mathcal{O}_{10}$ & $\frac{1}{12}$ & $\nu ^2$ & $2 \left(\Lambda +\frac{1}{4}\right) x^2$ & $4 \left(\frac{\Lambda }{2}+\frac{3}{4}\right) x^2$ \\
 $\mathcal{O}_{11}$ & $\frac{1}{12}$ & $\nu $ & $\left(\Lambda +\frac{1}{4}\right) \mu x$ & $\sqrt{2} \left(\frac{\Lambda }{2}+\frac{3}{4}\right) \mu x$ \\
 $\mathcal{O}_{12}$ & $\frac{1}{12}$ & $\nu ^2$ & $2 \left(\Lambda +\frac{1}{4}\right) x^2$ & $4 \left(\frac{\Lambda }{2}+\frac{3}{4}\right) x^2$ \\
 $\mathcal{O}_{13}$ & $\frac{1}{12}$ & $\frac{1}{6} (\nu -1) \nu (2 \nu -1)$ & $4 x^3$ & $8 \sqrt{2} x^3$ \\
 $\mathcal{O}_{14}$ & $\frac{1}{12}$ & $\frac{1}{2} (\nu -1) \nu $ & $4 \mu x^2$ & $8 \mu x^2$ \\
 $\mathcal{O}_{15}$ & $\frac{1}{12}$ & $\frac{1}{2} (\nu -1) \nu ^2$ & $4 x^3$ & $8 \sqrt{2} x^3$ \\
 $\mathcal{O}_{16}$ & $\frac{1}{12}$ & $\frac{1}{2} (\nu -1) \nu $ & $2 \mu x^2$ & $4 \mu x^2$ \\
 $\mathcal{O}_{17}$ & $\frac{1}{12}$ & $\nu $ & $\mu ^2 x$ & $\sqrt{2} \mu ^2 x$ \\
 $\mathcal{O}_{18}$ & $\frac{1}{12}$ & $\frac{1}{2} (\nu -1) \nu $ & $2 \mu x^2$ & $4 \mu x^2$ \\
 $\mathcal{O}_{19}$ & $\frac{1}{12}$ & $\frac{1}{2} (\nu -1) \nu ^2$ & $4 x^3$ & $8 \sqrt{2} x^3$ \\
 $\mathcal{O}_{20}$ & $\frac{1}{12}$ & $\nu ^2$ & $4 \mu x^2$ & $8 \mu x^2$ \\
 $\mathcal{O}_{21}$ & $\frac{1}{12}$ & $\nu ^2$ & $2 \left(\Lambda +\frac{1}{4}\right) x^2$ & $4 \left(\frac{\Lambda }{2}+\frac{3}{4}\right) x^2$ \\
 $\mathcal{O}_{22}$ & $\frac{1}{12}$ & $\nu ^3$ & $4 x^3$ & $8 \sqrt{2} x^3$ \\
 $\mathcal{O}_{23}$ & $\frac{1}{12}$ & $\nu ^2$ & $2 \mu x^2$ & $4 \mu x^2$ \\
 $\mathcal{O}_{24}$ & $\frac{1}{12}$ & $\nu ^3$ & $4 x^3$ & $8 \sqrt{2} x^3$ \\
\end{tabular}
\caption{\label{tab:commutatorBoundSummary}
Summary of non-zero commutator upper bounds, including the pure number coefficients from Eq.~\eqref{eq:2ndOrderTrotter_SumsOutside} and multiplicities of each operator form.}
\end{center}
\end{table}
\endgroup

\subsubsection{Loop-string-hadron commutator bounds
\label{app:commutators_LSH}}
The derivation of the second-order Trotter error for the LSH formulation closely follows that of the Schwinger-boson formulation.
One difference concerns the norms $\| x K^{(j)} \pm x {K^{(j)}}^\dagger \|$.
Specifically, one of the generalized LSH hopping subterms is as written in Eq.~\eqref{eq:generalizedHISubtermLSHTruncated}:
\begin{align*}
    H_I^{\mathrm{LSH}(j)} =\ &\ x \, \chi_{\reg{y'}}^\dagger \chi_\reg{x'} ( \lshladder^{\dagger}_\reg{q} )^{n_\reg{y}} ( \lshladder^{\dagger}_\reg{p} )^{1-n_\reg{x}} \DLSH (p,n_\reg{x},n_\reg{y}) + \mathrm{H.c.}
\end{align*}
Above, one identifies
\begin{subequations}
\begin{align}
    &{K^{(j)}}^\dagger = \chi_{\reg{y'}}^\dagger \chi_\reg{x'} ( \lshladder^{\dagger}_\reg{q} )^{n_\reg{y}} ( \lshladder^{\dagger}_\reg{p} )^{1-n_\reg{x}} \DLSH (p,n_\reg{x},n_\reg{y}) , \\
    &K^{(j)} = \chi^\dagger_\reg{x'} \chi_{\reg{y'}} ( \lshladder_\reg{q} )^{n_\reg{y}} ( \lshladder_\reg{p} )^{1-n_\reg{x}} \DLSH (p-1+n_\reg{x},n_\reg{y},n_\reg{x}),
\end{align}
\end{subequations}
where $\DLSH (p,n,n') = \sqrt{(p+1+n)/(p+1+n')}$. The norm bound reduces to the evaluation of $\| \DLSH \|$ as before since $\| x {K^{(j)}}^\dagger \pm x K^{(j)} \| \leq \max \left( \| x \, K^{(j)} \|, \| x \, {K^{(j)}}^\dagger \| \right)$. It is easy to see that $\| \DLSH (p,n,n') \| \leq \sqrt{2}$. This is because the numerator is larger when $n=1$ and the denominator is smaller when $n'=0$. The expression $\sqrt{(p+2)/(p+1)}$ is then largest for $p=0$, giving the norm of $\sqrt{2}$ as was claimed. Similar reasoning can be used to deduce that $\| \DLSH (p-1+n',n,n') \| \leq \sqrt{2}$ as well.
In either case, one arrives at $$\| x {K^{(j)}}^\dagger \pm x K^{(j)} \| \leq \sqrt{2} \ x.$$

As for the commutators involving $H^{\rm LSH}_E(r)$, one notes that for the LSH formulation,
5
\begin{align}
J^L(r) = \frac{1}{2}\big(n_\ell(r)+n_o(r)(1-n_i(r))\big),
\end{align}
hence
\begin{align}
\bigl \| J^L(r) +\frac{1}{4} \bigr \| \leq \frac{\Lambda}{2}+\frac{3}{4}.
\end{align}

Fortunately, the forms of non-vanishing commutators in the LSH formulation are the same as in the Schwinger-boson formulation.
That is, there are twenty four $\mathcal{O}_i$ expressions from before, although the operator forms of the Hamiltonian terms are different. As the evaluations proceed the same way, rather than repeating the explanations, we simply state what differences arise in the norm evaluations:
\begin{enumerate}
    \item One should use $\| J^L + 1/4 \| \leq \Lambda/2+3/4$ for the LSH formulation instead of $\Lambda+1/4$ for the Schwinger-boson formulation.
    \item Wherever a factor of $x$ is acquired in the Schwinger-boson results, one now picks up an additional factor of $\sqrt{2}$.
    \item The $\mathcal{O}_i$ have different multiplicities arising from $\nu=2$ for the LSH formulation, as opposed to $\nu=8$ for the Schwinger-boson formulation.
\end{enumerate}
Table~\ref{tab:commutatorBoundSummary} summarizes the LSH commutator upper bounds.
When the $\|O_i\|$ are summed up with appropriate coefficients and multiplicities, the final result is
\begin{align}
    \bigl \| V^{\rm LSH}(\theta) - e^{-i\theta H^{\rm LSH}} \bigr \| &\leq 
     L \theta^3 \rho^{\rm LSH} (x, \Lambda, \mu ),
\end{align}
with
\begin{align}
    \rho^{\rm LSH} (x, \Lambda, \mu ) &\equiv 
    \frac{47 \sqrt{2} x^3}{3} + 2 \Lambda  x^2 + \frac{25 \mu  x^2}{3} + 3 x^2 + \frac{\Lambda ^2 x}{24 \sqrt{2}} + \frac{\Lambda  \mu  x}{3 \sqrt{2}} + \frac{\Lambda  x}{8 \sqrt{2}} + \frac{5 \mu ^2 x}{6 \sqrt{2}} + \frac{\mu  x}{2 \sqrt{2}} + \frac{3 x}{32 \sqrt{2}} .
\end{align}
%

\section{Gate-cost tables
\label{app:cost-tables}}
\noindent
In this appendix, we tabulate the gate counts for near- and far-term approaches to the Trotterized time evolution as detailed in Sec.~\ref{sec:bounds}, over a range of possible simulation parameters.
In the near-term case, we envisage small lattices (no more than 20 staggered sites), harsh electric truncations (not exceeding four qubits per bosonic DOF), total spectral-norm-error bound no smaller than 5\%, and bare couplings that are not particularly weak ($x\leq 1$).
In the far-term case, we envisage large lattices (up to 1000 staggered sites), modestly larger bosonic registers (up to eight-qubit registers), per-cent and per-mille spectral-norm-error tolerances, as well as weaker bare couplings ($x\geq 1$).

\newcolumntype{R}{>{$}r<{$}}
\newcolumntype{C}{>{$}c<{$}}
\begin{table}[t!]
    \centering
\[
\begin{array}{rrrrrr|rrr|rrr}
&&&&&& \multicolumn{3}{c|}{\text{Schwinger bosons}} & \multicolumn{3}{c}{\text{LSH}} \\
m/g & \Delta_{\mathrm{Trot}} & x & L & \eta & t/a_s & \text{Qubits} & \text{Min.~$s$} & \text{Min.~CNOTs} & \text{Qubits} & \text{Min.~$s$} & \text{Min.~CNOTs}  \\
\hline
\hline

 1 & 10\% & 0.1 & 10 & 2 & 1 & 92 & 186 & 4.8613\times 10^6 & 40 & 63 & 2.63088\times 10^5 \\
 1 & 5\% & 0.1 & 10 & 2 & 1 & 92 & 262 & 6.84763\times 10^6 & 40 & 89 & 3.71664\times 10^5 \\
 1 & 10\% & 1 & 10 & 2 & 1 & 92 & 102 & 2.66587\times 10^6 & 40 & 26 & 1.08576\times 10^5 \\
 1 & 5\% & 1 & 10 & 2 & 1 & 92 & 144 & 3.76358\times 10^6 & 40 & 37 & 1.54512\times 10^5 \\
 1 & 10\% & 0.1 & 10 & 4 & 1 & 164 & 433 & 5.21403\times 10^8 & 60 & 136 & 1.64261\times 10^6 \\
 1 & 5\% & 0.1 & 10 & 4 & 1 & 164 & 613 & 7.38153\times 10^8 & 60 & 193 & 2.33105\times 10^6 \\
 1 & 10\% & 1 & 10 & 4 & 1 & 164 & 129 & 1.55337\times 10^8 & 60 & 34 & 4.10652\times 10^5 \\
 1 & 5\% & 1 & 10 & 4 & 1 & 164 & 182 & 2.19158\times 10^8 & 60 & 48 & 5.79744\times 10^5 \\
 1 & 10\% & 0.1 & 20 & 2 & 1 & 192 & 262 & 1.44561\times 10^7 & 80 & 89 & 7.84624\times 10^5 \\
 1 & 5\% & 0.1 & 20 & 2 & 1 & 192 & 371 & 2.04703\times 10^7 & 80 & 126 & 1.11082\times 10^6 \\
 1 & 10\% & 1 & 20 & 2 & 1 & 192 & 144 & 7.94534\times 10^6 & 80 & 37 & 3.26192\times 10^5 \\
 1 & 5\% & 1 & 20 & 2 & 1 & 192 & 203 & 1.12007\times 10^7 & 80 & 52 & 4.58432\times 10^5 \\
 1 & 10\% & 0.1 & 20 & 4 & 1 & 344 & 613 & 1.55832\times 10^9 & 120 & 193 & 4.92111\times 10^6 \\
 1 & 5\% & 0.1 & 20 & 4 & 1 & 344 & 866 & 2.20148\times 10^9 & 120 & 272 & 6.93546\times 10^6 \\
 1 & 10\% & 1 & 20 & 4 & 1 & 344 & 182 & 4.62667\times 10^8 & 120 & 48 & 1.2239\times 10^6 \\
 1 & 5\% & 1 & 20 & 4 & 1 & 344 & 257 & 6.53326\times 10^8 & 120 & 68 & 1.73386\times 10^6 \\
 1 & 10\% & 0.1 & 10 & 2 & 5 & 92 & 2072 & 5.41538\times 10^7 & 40 & 702 & 2.93155\times 10^6 \\
 1 & 5\% & 0.1 & 10 & 2 & 5 & 92 & 2929 & 7.65523\times 10^7 & 40 & 993 & 4.14677\times 10^6 \\
 1 & 10\% & 1 & 10 & 2 & 5 & 92 & 1133 & 2.96121\times 10^7 & 40 & 288 & 1.20269\times 10^6 \\
 1 & 5\% & 1 & 10 & 2 & 5 & 92 & 1602 & 4.18699\times 10^7 & 40 & 407 & 1.69963\times 10^6 \\
 1 & 10\% & 0.1 & 10 & 4 & 5 & 164 & 4841 & 5.82936\times 10^9 & 60 & 1519 & 1.83465\times 10^7 \\
 1 & 5\% & 0.1 & 10 & 4 & 5 & 164 & 6846 & 8.24371\times 10^9 & 60 & 2149 & 2.59556\times 10^7 \\
 1 & 10\% & 1 & 10 & 4 & 5 & 164 & 1432 & 1.72436\times 10^9 & 60 & 375 & 4.52925\times 10^6 \\
 1 & 5\% & 1 & 10 & 4 & 5 & 164 & 2024 & 2.43723\times 10^9 & 60 & 531 & 6.41342\times 10^6 \\
 1 & 10\% & 0.1 & 20 & 2 & 5 & 192 & 2929 & 1.61611\times 10^8 & 80 & 993 & 8.75429\times 10^6 \\
 1 & 5\% & 0.1 & 20 & 2 & 5 & 192 & 4143 & 2.28594\times 10^8 & 80 & 1404 & 1.23777\times 10^7 \\
 1 & 10\% & 1 & 20 & 2 & 5 & 192 & 1602 & 8.8392\times 10^7 & 80 & 407 & 3.58811\times 10^6 \\
 1 & 5\% & 1 & 20 & 2 & 5 & 192 & 2266 & 1.25029\times 10^8 & 80 & 575 & 5.0692\times 10^6 \\
 1 & 10\% & 0.1 & 20 & 4 & 5 & 344 & 6846 & 1.74034\times 10^{10} & 120 & 2149 & 5.47952\times 10^7 \\
 1 & 5\% & 0.1 & 20 & 4 & 5 & 344 & 9682 & 2.46128\times 10^{10} & 120 & 3038 & 7.74629\times 10^7 \\
 1 & 10\% & 1 & 20 & 4 & 5 & 344 & 2024 & 5.14526\times 10^9 & 120 & 531 & 1.35394\times 10^7 \\
 1 & 5\% & 1 & 20 & 4 & 5 & 344 & 2863 & 7.2781\times 10^9 & 120 & 750 & 1.91235\times 10^7 \\
 
\end{array}
\]
    \caption{
    \label{tab:errorBoundedNISQCosts}
    Near-term simulation costs as a function of Hamiltonian parameters ($m/g,x,L$, and $\eta$), evolution time ($t/a_s=2xT$), and desired bound on the controlled sources of error ($\Delta_{\mathrm{Trot}}$).
    Qubit counts are the register size of the lattice and exclude possible ancilla qubits (which are insignificant in the near-term circuits cost).
    Other tabulated costs are the minimal required number of second-order Trotter steps (based on the decomposition of Eq.~\eqref{eq:2ndTrotter}) and the associated CNOT-gate count (for the naive circutization approach based on the full Pauli decomposition of diagonal phase functions) in the zero-noise limit.
    }
\end{table}

For a comparison between the formulations, one can compute the relative CNOT cost of the LSH formulation to that of the Schwinger-boson formulation. Using the data from Table \ref{tab:errorBoundedNISQCosts}, one observes
the relative CNOT cost of the LSH to Schwinger bosons to be most strongly influenced by $\eta$, and to a lesser degree by $x$.
For example, the reduction is
about 18.4 times at $(\eta,x)=(2,0.1)$, 
24.5 times at $(\eta,x)=(2,1)$, 
317 times at $(\eta,x)=(4,0.1)$, 
and 380 times at $(\eta,x)=(4,1)$.
In a similar vein, using Table \ref{tab:errorBoundedCosts}, one observes that the relative T-gate cost of LSH to Schwinger bosons to be most strongly influenced by $x$, with the reduction being about 21-fold at $x=1$, and 24-fold at $x=10$.
\begin{table}[t!]
    \centering
\[
\begin{array}{cccccccc|cr|cr}
&&&&&&&& \multicolumn{2}{c|}{\text{Schwinger bosons}} & \multicolumn{2}{c}{\text{LSH}} \\
m/g & x & \eta & L & t/a_s & \Delta & \alpha_\mathrm{Trot.} & \alpha_\mathrm{Newt.} & \text{Qubits} & \text{T gates} & \text{Qubits} & \text{T gates} \\
\hline
\hline
 1 & 1 & 4 & 100 & 1 & 0.01 & 90\% & 9\% & 2626 & 8.19713\times 10^{11} & 1319 & 3.91817\times 10^{10} \\
 1 & 1 & 4 & 100 & 1 & 0.001 & 90\% & 9\% & 2704 & 3.09951\times 10^{12} & 1397 & 1.5172\times 10^{11} \\
 1 & 10 & 4 & 100 & 1 & 0.01 & 90\% & 9\% & 2626 & 5.7715\times 10^{11} & 1319 & 2.31055\times 10^{10} \\
 1 & 10 & 4 & 100 & 1 & 0.001 & 90\% & 9\% & 2704 & 2.18285\times 10^{12} & 1397 & 8.94211\times 10^{10} \\
 1 & 1 & 8 & 100 & 1 & 0.01 & 90\% & 9\% & 4398 & 5.79224\times 10^{12} & 1807 & 2.72735\times 10^{11} \\
 1 & 1 & 8 & 100 & 1 & 0.001 & 90\% & 9\% & 4476 & 2.1482\times 10^{13} & 1885 & 1.03709\times 10^{12} \\
 1 & 10 & 8 & 100 & 1 & 0.01 & 90\% & 9\% & 4398 & 1.33288\times 10^{12} & 1807 & 5.40046\times 10^{10} \\
 1 & 10 & 8 & 100 & 1 & 0.001 & 90\% & 9\% & 4476 & 4.94102\times 10^{12} & 1885 & 2.05007\times 10^{11} \\
 1 & 1 & 4 & 1000 & 1 & 0.01 & 90\% & 9\% & 18904 & 3.12769\times 10^{13} & 6797 & 1.53099\times 10^{12} \\
 1 & 1 & 4 & 1000 & 1 & 0.001 & 90\% & 9\% & 19008 & 1.22564\times 10^{14} & 6875 & 5.81562\times 10^{12} \\
 1 & 10 & 4 & 1000 & 1 & 0.01 & 90\% & 9\% & 18904 & 2.2027\times 10^{13} & 6797 & 9.02342\times 10^{11} \\
 1 & 10 & 4 & 1000 & 1 & 0.001 & 90\% & 9\% & 19008 & 8.63137\times 10^{13} & 6875 & 3.42924\times 10^{12} \\
 1 & 1 & 8 & 1000 & 1 & 0.01 & 90\% & 9\% & 35076 & 2.16773\times 10^{14} & 10885 & 1.04652\times 10^{13} \\
 1 & 1 & 8 & 1000 & 1 & 0.001 & 90\% & 9\% & 35180 & 8.30098\times 10^{14} & 10963 & 3.91414\times 10^{13} \\
 1 & 10 & 8 & 1000 & 1 & 0.01 & 90\% & 9\% & 35076 & 4.98595\times 10^{13} & 10885 & 2.06871\times 10^{12} \\
 1 & 10 & 8 & 1000 & 1 & 0.001 & 90\% & 9\% & 35180 & 1.9092\times 10^{14} & 10963 & 7.73624\times 10^{12} \\
 1 & 1 & 4 & 100 & 10 & 0.01 & 90\% & 9\% & 2704 & 3.0993\times 10^{13} & 1397 & 1.51643\times 10^{12} \\
 1 & 1 & 4 & 100 & 10 & 0.001 & 90\% & 9\% & 2808 & 1.2146\times 10^{14} & 1475 & 5.76229\times 10^{12} \\
 1 & 10 & 4 & 100 & 10 & 0.01 & 90\% & 9\% & 2704 & 2.18258\times 10^{13} & 1397 & 8.94083\times 10^{11} \\
 1 & 10 & 4 & 100 & 10 & 0.001 & 90\% & 9\% & 2808 & 8.55326\times 10^{13} & 1475 & 3.39741\times 10^{12} \\
 1 & 1 & 8 & 100 & 10 & 0.01 & 90\% & 9\% & 4476 & 2.14816\times 10^{14} & 1885 & 1.03705\times 10^{13} \\
 1 & 1 & 8 & 100 & 10 & 0.001 & 90\% & 9\% & 4580 & 8.22615\times 10^{14} & 1963 & 3.87886\times 10^{13} \\
 1 & 10 & 8 & 100 & 10 & 0.01 & 90\% & 9\% & 4476 & 4.94053\times 10^{13} & 1885 & 2.04958\times 10^{12} \\
 1 & 10 & 8 & 100 & 10 & 0.001 & 90\% & 9\% & 4580 & 1.89192\times 10^{14} & 1963 & 7.66615\times 10^{12} \\
 1 & 1 & 4 & 1000 & 10 & 0.01 & 90\% & 9\% & 19008 & 1.22564\times 10^{15} & 6875 & 5.81468\times 10^{13} \\
 1 & 1 & 4 & 1000 & 10 & 0.001 & 90\% & 9\% & 19086 & 4.48657\times 10^{15} & 6979 & 2.29217\times 10^{14} \\
 1 & 10 & 4 & 1000 & 10 & 0.01 & 90\% & 9\% & 19008 & 8.63103\times 10^{14} & 6875 & 3.4283\times 10^{13} \\
 1 & 10 & 4 & 1000 & 10 & 0.001 & 90\% & 9\% & 19086 & 3.15948\times 10^{15} & 6979 & 1.35149\times 10^{14} \\
 1 & 1 & 8 & 1000 & 10 & 0.01 & 90\% & 9\% & 35180 & 8.30094\times 10^{15} & 10963 & 3.91412\times 10^{14} \\
 1 & 1 & 8 & 1000 & 10 & 0.001 & 90\% & 9\% & 35258 & 2.99214\times 10^{16} & 11067 & 1.5154\times 10^{15} \\
 1 & 10 & 8 & 1000 & 10 & 0.01 & 90\% & 9\% & 35180 & 1.90912\times 10^{15} & 10963 & 7.73585\times 10^{13} \\
 1 & 10 & 8 & 1000 & 10 & 0.001 & 90\% & 9\% & 35258 & 6.88164\times 10^{15} & 11067 & 2.99513\times 10^{14} \\
\end{array}
\]
    \caption{Far-term simulation costs as a function of Hamiltonian parameters  ($m/g,x,L$, and $\Lambda=2^\eta-1$), evolution time ($t/a_s$), and desired bound on the controlled sources of error ($\Delta$). Qubit counts are the sum of qubits needed to represent lattice DOFs and ancilla qubits used for implementing the time evolution.}
    \label{tab:errorBoundedCosts}
\end{table}
\end{appendices}

\bibliographystyle{quantum}
\bibliography{main}

\end{document}